\documentclass[11pt]{article}

\usepackage{color, amsmath, amssymb, amsthm, graphicx, bbm, footnote, mathtools, bm, enumitem, longtable, subcaption}

\usepackage[top=1in, bottom=1in, left=1in, right=1in]{geometry}

\usepackage[toc,page,header]{appendix}
\usepackage{titletoc}
\usepackage{tikz}
\usepackage{titling}
\allowdisplaybreaks

\usepackage{setspace}
\usepackage{adjustbox}
\usepackage[colorlinks,allcolors=blue]{hyperref}
\usepackage[round]{natbib}  
\usepackage{siunitx}

\usepackage{authblk}

\usepackage{apptools}
\AtAppendix{\counterwithin{theorem}{section}}
\AtAppendix{\counterwithin{remark}{section}}

\newtheorem{theorem}{Theorem}

\newtheorem{proposition}[theorem]{Proposition}
\newtheorem{corollary}[theorem]{Corollary}
\newtheorem{lemma}[theorem]{Lemma}
\newtheorem{remark}{Remark}
\newtheorem{assump}{Assumption}

\newtheorem{innercustomthm}{Theorem}
\newenvironment{customthm}[1]
{\renewcommand\theinnercustomthm{#1}\innercustomthm}
{\endinnercustomthm}

\newtheorem{innercustomprop}{Proposition}
\newenvironment{customprop}[1]
{\renewcommand\theinnercustomprop{#1}\innercustomprop}
{\endinnercustomprop}

\newtheorem{innercustomlem}{Lemma}
\newenvironment{customlem}[1]
{\renewcommand\theinnercustomlem{#1}\innercustomlem}
{\endinnercustomlem}

\newtheorem{innercustomcor}{Corollary}
\newenvironment{customcor}[1]
{\renewcommand\theinnercustomcor{#1}\innercustomcor}
{\endinnercustomcor}

\newtheorem{example}{Example}

\DeclareMathOperator*{\argmin}{arg\,min}

\newcommand{\norm}[1]{\left|\left| #1 \right|\right|}
\newcommand{\Diag}{\mathtt{Diag}}
\renewcommand{\hat}{\widehat}
\renewcommand{\tilde}{\widetilde}

\renewcommand{\P}{\mbox{$\mathrm{P}$}}
\newcommand{\E}{\mbox{$\mathbb{E}$}}

\definecolor{emerald}{rgb}{0.31, 0.78, 0.47}

\title{\textbf{Nonparametric Inference on Dose-Response Curves Without the Positivity Condition}}
\author{Yikun Zhang$^{1,*}$, \; Yen-Chi Chen$^{1,\dagger}$, \;\text{and}\, Alexander Giessing$^2$}
\date{\normalsize
	\textit{$^1$Department of Statistics, University of Washington}\\
	$^*$\href{mailto:yikun@uw.edu}{yikun@uw.edu}; $^{\dagger}$\href{mailto:yenchic@uw.edu}{yenchic@uw.edu}\\
	\textit{$^2$Department of Statistics and Data Science, National University of Singapore}\\~\\	
	\today}
\begin{document}
	\maketitle

\begin{abstract}
Existing statistical methods in causal inference often assume the positivity condition, where every individual has some chance of receiving any treatment level regardless of covariates. This assumption could be violated in observational studies with continuous treatments. In this paper, we develop identification and estimation theories for causal effects with continuous treatments (\emph{i.e.}, dose-response curves) without relying on the positivity condition. 
Our approach identifies and estimates the derivative of the treatment effect for each observed sample, integrating it to the treatment level of interest to mitigate bias from the lack of positivity. The method is grounded in a weaker assumption, satisfied by additive confounding models. We propose a fast and reliable numerical recipe for computing our integral estimator in practice and derive its asymptotic properties. To enable valid inference on the dose-response curve and its derivative, we use the nonparametric bootstrap and establish its consistency. The performances of our proposed estimators are validated through simulation studies and an analysis of the effect of air pollution exposure (PM$_{2.5}$) on cardiovascular mortality rates.
		\\~\\
\noindent \textbf {Keywords:} {Causal Inference; Dose-Response Curve; Positivity; Kernel Smoothing.}
\end{abstract}
	
\section{Introduction}
\label{sec:Intro}

In observational studies, causal effects do not always result from standard binary interventions but rather arise from continuous treatments or exposures. The causal effect of a continuous treatment variable $T\in \mathcal{T} \subset \mathbb{R}$ on an outcome $Y\in \mathcal{Y}\subset \mathbb{R}$ is described by the (causal) dose-response curve, which characterizes the average outcome if all individuals in the population were assigned to a treatment level $T=t$. Typically, an additional set of covariates $\bm{S} \in \mathcal{S} \subset \mathbb{R}^d$ is observed, capturing potential confounding variables that influence both the treatment $T$ and the outcome $Y$. Under appropriate regularity conditions, the dose-response curve coincides with the covariate-adjusted regression function $t\mapsto \mathbb{E}\left[\mu(t,\bm{S})\right]$, where $\mu(t,\bm{s})=\mathbb{E}\left(Y|T=t,\bm{S}=\bm{s}\right)$, making it identifiable from the observed data for any $t\in \mathcal{T}$ \citep{robins2000marginal,neugebauer2007nonparametric,diaz2013targeted,kennedy2017non,takatsu2022debiased}.

\begin{figure}[t]
	\captionsetup[subfigure]{justification=centering}
	\begin{subfigure}[t]{0.32\linewidth}
		\centering
		\includegraphics[width=1\linewidth]{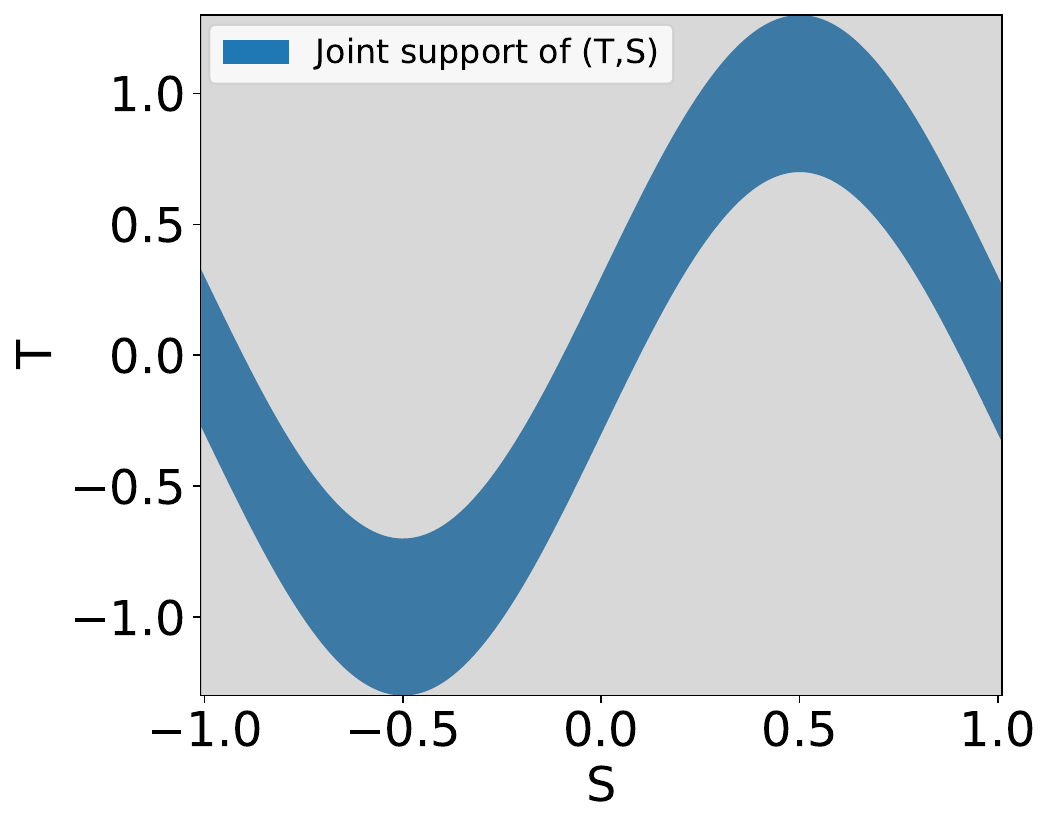}
	\end{subfigure}
	\hfil
	\begin{subfigure}[t]{0.32\linewidth}
		\centering		\includegraphics[width=1\linewidth]{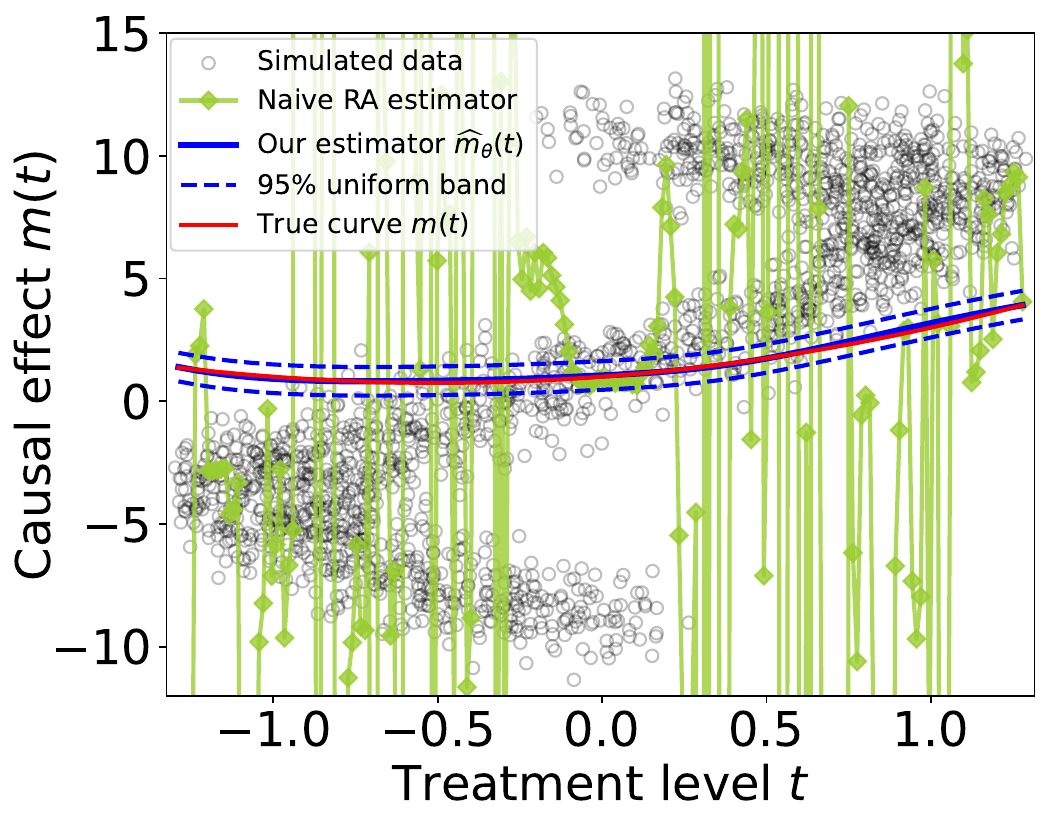}
	\end{subfigure}
	\hfil
	\begin{subfigure}[t]{0.32\linewidth}
		\centering		\includegraphics[width=1\linewidth]{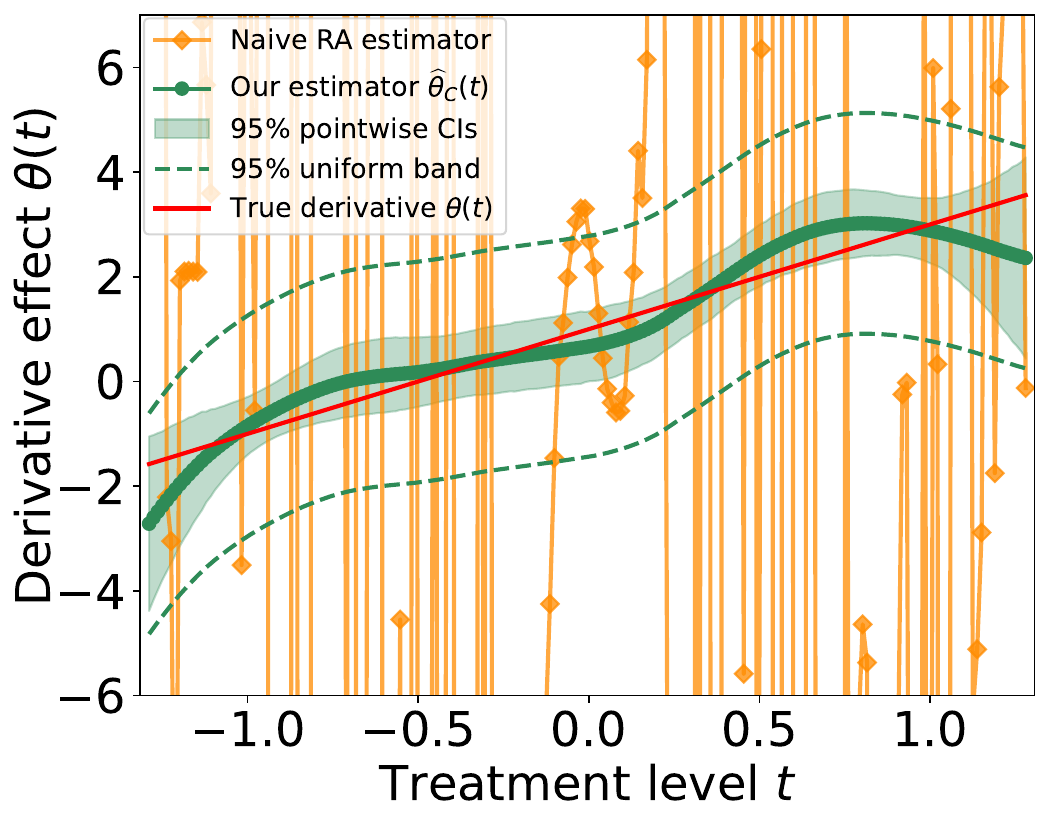}
	\end{subfigure}
	\caption{Simulation results under the single confounder model \eqref{single_conf} with $n=2000$. {\bf Left:} The support of the joint distribution of $(T,S)$. {\bf Middle:} Estimated dose-response curves using naive regression adjustment (RA) and the proposed integral estimators, overlaid with the true one $m(t)$. {\bf Right:} Estimated derivatives of the dose-response curve by naive RA and our proposed localized estimators, overlaid with the true derivative $\theta(t)$. The middle and right panels also present the 95\% confidence intervals (CIs) and/or uniform confidence bands from our proposed estimators, shown as shaded regions and dashed lines, respectively.}
	\label{fig:single_conf}
\end{figure}

However, the regularity conditions for identifying and estimating the dose-response curve may not always hold in observational studies. In particular, a key assumption in this context is the following positivity condition, which ensures sufficient variability in treatment assignment across all strata of covariates.

\begin{assump}[Positivity]
	\label{assump:positivity}
	The conditional density $p(t|\bm{s})$ is bounded above and away from zero almost surely for all $t\in \mathcal{T}$ and $\bm{s}\in \mathcal{S}$.
\end{assump}

The positivity condition \eqref{assump:positivity} may be violated in observational studies for two main reasons: (i) theoretically, it is impossible for some individuals with specific covariate values to receive certain treatment levels, and (ii) practically, a finite data sample may fail to collect individuals exposed to certain treatment levels; see \cite{cole2008constructing,westreich2010invited,petersen2012diagnosing} for related discussions. These issues are particularly pervasive in the context of continuous treatments. For instance, air pollution levels may vary significantly across larger regions while remaining relatively consistent within smaller areas, resulting in individuals in the same location being exposed to only one level of pollution. This discrepancy, where spatial covariates change at a finer scale than the exposure measurement, leads to violations of the positivity condition \citep{paciorek2010importance,schnell2020,keller2020selecting}. As an illustrative example, consider a single confounder model:
\begin{equation}
	\label{single_conf}
	Y=T^2+T+1+10S+\epsilon, \quad T=\sin(\pi S) + E, \quad \text{ and } \quad S\sim \text{Uniform}[-1,1] \subset \mathbb{R},
\end{equation}
where $E\sim \text{Uniform}[-0.3,0.3]$ is an independent treatment variation and $\epsilon\sim \mathcal{N}(0,1)$ is an independent noise variable. Here, the marginal supports of $T$ and $S$ are $\mathcal{T}=[-1.3, 1.3]$ and $\mathcal{S}=[-1,1]$ respectively, but the joint support of $(T,S)$ covers only a narrow band within the product space $\mathcal{T}\times \mathcal{S}$; see the left panel of \autoref{fig:single_conf}. In this scenario, the conditional density $p(t|s)$ for any $s\in \mathcal{S}$ is 0 in the gray regions, violating the positivity condition \eqref{assump:positivity}. Without \eqref{assump:positivity}, existing approaches for estimating the dose-response curve $t\mapsto m(t)$ and its derivative $t \mapsto m'(t):=\theta(t)$ can be very unstable at certain treatment levels, as illustrated by the usual regression adjustment estimators (defined in \eqref{m_RA} and Remark~\ref{remark:theta_RA} below) in the middle and right panels of \autoref{fig:single_conf}.

To address the issues arising from violations of the positivity condition, we introduce a novel identification strategy for the dose-response curve $m(t)$ and its derivative $\theta(t)$ based on conditional expectations, thereby eliminating reliance on \eqref{assump:positivity}. Leveraging this strategy, we propose an integral estimator of $m(t)$ and a localized estimator of $\theta(t)$, both of which remain consistent even when the positivity condition \eqref{assump:positivity} fails in some regions of $\mathcal{T}\times \mathcal{S}$; see the middle and right panels of \autoref{fig:single_conf}. Our key contributions and the paper outline are as follows:

{\bf 1. Identification and Inference:} We illuminate identification conditions for $m(t)$ and $\theta(t)$ without requiring the positivity condition \eqref{assump:positivity} in \autoref{sec:prelim}. In \autoref{subsec:integral_est}, we introduce our integral estimator of $m(t)$, constructed by integrating a localized estimator of $\theta(t)$ around the observed data, enabling extrapolation to any treatment level $t\in \mathcal{T}$ of interest. Its practical utility is facilitated by Riemann sum approximations in \autoref{subsec:fast_algo} and reliable inference through nonparametric bootstrap in \autoref{subsec:bootstrap_infer}.

{\bf 2. Asymptotic Theory:} In \autoref{subsec:consistency_int_est}, we establish the uniform consistency of our integral and localized derivative estimators within the kernel smoothing framework. Furthermore, the validity of nonparametric bootstrap inference is rigorously proved in \autoref{subsec:bootstrap_validity}.

{\bf 3. Experiments:} We evaluate the finite-sample performances of our proposed estimators through simulations and a case study of the effects of fine particulate matter (PM$_{2.5}$) on cardiovascular mortality rates in \autoref{sec:experiments}. Code for all our experiments is publicly available at \url{https://github.com/zhangyk8/npDoseResponse/tree/main/Paper_Code}, together with both Python package ``\href{https://pypi.org/project/npDoseResponse/}{\texttt{npDoseResponse}}'' and R package \citep{zhang2024package}.

\subsection{Other Related Works}
\label{subsec:related_work}


Estimating the dose-response curve is technically challenging in causal inference because it is not pathwise differentiable and cannot be consistently estimated in a $\sqrt{n}$ rate \citep{chamberlain1986asymptotic,van1991differentiable}. Parametrically, \cite{robins2000marginal} pioneered a marginal structural model to estimate $m(t)$, later extended by \cite{laan2003unified,neugebauer2007nonparametric} to relax dependence on parametric model specification through projection methods. Nonparametrically, \cite{newey1994kernel} developed a regression adjustment approach using two-stage kernel smoothing, later adapted to dose-response curve estimation by \cite{flores2007estimation}. However, the consistency and asymptotic normality of their kernel-based estimators require undersmoothing bandwidths, which complicates practical bandwidth selection; see Section 5.7 in \cite{wasserman2006all}. On the contrary, our proposed integral estimator of $m(t)$, leveraging kernel-based derivative estimation, avoids the need for undersmoothing or explicit bias correction \citep{calonico2018effect,takatsu2022debiased}, allowing standard bandwidth selection methods in nonparametric regression. 


There are many existing works about the estimation of average derivative effects in the literature; see \cite{hardle1989investigating,powell1989semiparametric,newey1993efficiency,cattaneo2010robust,hirshberg2020debiased} and references therein. Under regularity conditions, average derivative effects are identical to the so-called incremental treatment effects, which are closely related to the derivatives $\theta(t)=m'(t)$ of dose-response curves; see Proposition 1 in \cite{rothenhausler2019incremental} and Section 6.1 in \cite{hines2023optimally}.


Most existing methods for estimating $m(t)$ and $\theta(t)$ assume the positivity condition \eqref{assump:positivity}. For binary treatments, prior research has examined the impact of positivity violations on various estimators from the perspectives of convergence rates \citep{khan2010irregular,d2021overlap} and empirical performances \citep{busso2014new,leger2022causal}. Approaches to address positivity violations include trimming extreme propensity score values \citep{dehejia1999causal,crump2009dealing,yang2018asymptotic,branson2023causal} and developing robust estimators \citep{rothe2017robust,ma2020robust}. More recently, \cite{kennedy2019nonparametric,rothenhausler2019incremental,schindl2024incremental} proposed incremental treatment effects as alternative estimands under positivity violations. To the best of our knowledge, no existing work provides nonparametric inference methods for $m(t)$ and $\theta(t)$ without assuming \eqref{assump:positivity}.

\subsection{Notations}
\label{subsec:notations}

Throughout the paper, we consider an outcome variable $Y\in \mathcal{Y}\subset \mathbb{R}$, univariate continuous treatment $T\in \mathcal{T}\subset \mathbb{R}$, and a vector of covariates $\bm{S}=(S_1,...,S_d)\in \mathcal{S}\subset \mathbb{R}^d$ with a fixed dimension $d$. The data sample consists of independent and identically distributed (i.i.d.) observations $\bm{U}_i=(Y_i,T_i,\bm{S}_i), i=1,...,n$ with the common distribution $\P$ and Lebesgue density $p(y,t,\bm{s}) = p(y|t,\bm{s})\cdot p(t|\bm{s}) \cdot p_S(\bm{s})$. Here, $p_T(t)$ and $p_S(\bm{s})$ are the marginal densities of $T$ and $\bm{S}$, respectively, and $p(t|\bm{s}) = \frac{\partial}{\partial t} \P\left(T\leq t|\bm{S} = \bm{s}\right)$ is the conditional density of $T$ given covariates $\bm{S}=\bm{s}$. We also denote the joint density of $(T,\bm{S})$ by $p(t,\bm{s}) = p(t|\bm{s})\cdot p_S(\bm{s}) = p(\bm{s}|t)\cdot p_T(t)$, whose support is $\mathcal{E}\subset \mathcal{T}\times \mathcal{S}$. Let $\mathbb{P}_n$ be the empirical measure so that $\mathbb{P}_n g = \int g(\bm{u}) \,d\mathbb{P}_n(\bm{u}) = \frac{1}{n}\sum_{i=1}^n g(\bm{U}_i)$ for some measurable function $g$. At the population level, we denote $\P g= \mathbb{E}\left[g(\bm{U})\right] = \int g(\bm{u}) \,d\P(\bm{u})$ for $\P$-integrable function $g$. In addition, we define the empirical process evaluated at a $\P$-integrable function $g$ as $\mathbb{G}_n(g) = \sqrt{n}\left(\mathbb{P}_n - \P\right) g =\frac{1}{\sqrt{n}} \sum_{i=1}^n \left[g(\bm{U}_i) - \mathbb{E}\left(g(\bm{U}_i)\right)\right]$. Finally, we use $\mathbbm{1}_A$ to denote the indicator function of a set $A$.

We use the big-$O$ notation $h_n=O(g_n)$ if $|h_n|$ is upper bounded by a positive constant multiple of $g_n >0$ when $n$ is sufficiently large. In contrast, $h_n=o(g_n)$ when $\lim_{n\to\infty} \frac{|h_n|}{g_n}=0$. For random variables, the notation $o_P(1)$ is short for a sequence of random variables converging to zero in probability, while the expression $O_P(1)$ denotes the sequence that is bounded in probability. We also use the notation $a_n\lesssim b_n$ or $b_n\gtrsim a_n$ when there exists an absolute constant $A>0$ such that $a_n\leq A b_n$ when $n$ is large. If $a_n\gtrsim b_n$ and $a_n \lesssim b_n$, then $a_n,b_n$ are asymptotically equal and it is denoted by $a_n\asymp b_n$.

\section{Basic Setup and Identification}
\label{sec:prelim}

Following the potential outcome framework \citep{rubin1974estimating}, we let $Y(t)$ denote the potential outcome that would have been observed under treatment level $T=t$. The causal effect for the continuous treatment variable (\emph{i.e.}, dose-response curve) can be written as $t\mapsto m(t)=\mathbb{E}\left[Y(t)\right]$, and its derivative effect curve is $t\mapsto \theta(t)=\frac{d}{dt}\mathbb{E}\left[Y(t)\right]$. We first state some basic identification assumptions.

\begin{assump}[Basic identification conditions]
	\label{assump:identify_cond}
	For any $t\in \mathcal{T}$, we have that
	\begin{enumerate}[label=(\alph*)]
		\item (Consistency) $T=t$ implies that $Y(t) = Y$.
		\item (Ignorability or unconfoundedness) $Y(t)$ is conditionally independent of $T$ given $\bm{S}$.
		\item (Treatment variation) The conditional variance of $T$ given any $\bm{S}=\bm{s}\in \mathcal{S}$ is strictly positive, \emph{i.e.}, $\mathrm{Var}(T|\bm{S} = \bm{s}) > 0$ for any $\bm{s}\in \mathcal{S}$.
	\end{enumerate}
\end{assump}

Assumption~\ref{assump:identify_cond}(a) is a continuous treatment version of the stable unit treatment value assumption (SUTVA; Page 19 of \citealt{cox1958planning} and \citealt{rubin1980randomization}), ensuring no interference between units and no multiple versions of treatments. When both the treatment $T$ and the outcome $Y$ are continuous, the required consistency condition is essentially the equality of conditional distributions $\P\left(Y(T) \leq y|T=t\right) = \P\left(Y(t)\leq y|T=t\right)$ for all $t\in \mathcal{T}$ and $y\in \mathcal{Y}$  \citep{gill2001causal}. Assumption~\ref{assump:identify_cond}(b), the ignorability condition, was first generalized to continuous treatments by \cite{hirano2004propensity} and implies that the potential outcome variable is independent of the treatment within any specific strata of covariates. In spatial confounding, this condition holds when spatial locations are adjusted for \citep{gilbert2023causal}. It further ensures that the mean potential outcome under $T=t$ is invariant across treatment levels, conditioned on $\bm{S}=\bm{s}$. Finally, Assumption~\ref{assump:identify_cond}(c), the treatment variation condition, is crucial for identifying the conditional mean outcome (or regression) function $\mu(t,\bm{s})=\E\left(Y|T=t,\bm{S}=\bm{s}\right)$ over a non-degenerate region of $\mathcal{T}\times \mathcal{S}$. We demonstrate how $\mathrm{Var}(T|\bm{S}=\bm{s})=0$ for some $\bm{s}\in \mathcal{S}$ can lead to ambiguous definitions of the associated dose-response curves in Example~\ref{example:id} of \autoref{app:np_bounds}. In cases where Assumption~\ref{assump:identify_cond}(c) fails, we also derive nonparametric bounds on $m(t)$ and its derivative $\theta(t)$ in \autoref{app:np_bounds} under an additive confounding model \eqref{add_conf_model} stated below.

Unfortunately, Assumption~\ref{assump:identify_cond} alone is insufficient to identify the dose-response curve $m(t)=\E\left[Y(t)\right]$ with observable data in the standard manner, such as through the covariate-adjusted regression function $\E\left[\mu(t,\bm{S})\right] = \int_{\mathcal{S}} \mu(t,\bm{s})\cdot p_S(\bm{s})\, d\bm{s}$. The key reason is that without the positivity condition \eqref{assump:positivity}, the conditional mean outcome function $\mu(t,\bm{s}) = \mathbb{E}\left(Y|T=t,\bm{S}=\bm{s}\right)$ is ill-defined in those regions of $\mathcal{T}\times \mathcal{S}$ outside the support $\mathcal{E}$ of the joint density $p(t,\bm{s})$. To resolve this identification issue, we introduce an additional assumption to identify the derivative curve $\theta(t)=\frac{d}{dt}\E\left[Y(t)\right]$, subsequently linking the identification of $m(t)$ to $\theta(t)$ via an integral formula.

\begin{assump}[Extrapolation condition]
	\label{assump:diff_inter}
	Suppose that at least one of the following assumptions are valid.
	\begin{enumerate}[label=(\alph*)]
		\item The function $\mathbb{E}\left[Y(t)|\bm{S}=\bm{s}\right]$ is continuously differentiable with respect to $t$ for any $(t,\bm{s})\in \mathcal{T}\times \mathcal{S}$ such that $p(\bm{s}|t)>0$, and the following equalities hold true:
		\begin{equation}
			\label{theta_equi}
			\theta(t) = \E\left[\frac{\partial}{\partial t} \mathbb{E}\left[Y(t)|\bm{S}\right]\right] = \E\left[\frac{\partial}{\partial t} \mathbb{E}\left[Y(t)|\bm{S}\right]\Big|T=t\right].
		\end{equation}
		
		\item The potential outcome $Y(t)$ is continuously differentiable with respect to $t$, and the following equalities hold true:
		\begin{equation}
			\label{theta_equi2}
			\theta(t) = \E\left[\mathbb{E}\left[\frac{\partial}{\partial t}Y(t)\Big|\bm{S}\right]\right] = \E\left[\mathbb{E}\left[\frac{\partial}{\partial t}Y(t) \Big|\bm{S}\right]\Big|T=t\right].
		\end{equation}
	\end{enumerate}
Additionally, it holds true that $\E(Y) = \E\left[m(T)\right]$.
\end{assump}

The first equalities of \eqref{theta_equi} and \eqref{theta_equi2} essentially reflect the interchangeability of expectation and partial derivative operations, which hold under mild conditions. In particular, these equalities are valid if $\left|\frac{\partial}{\partial t} \mathbb{E}\left[Y(t)|\bm{S}\right]\right|$ or $\left|\frac{\partial}{\partial t} Y(t)\right|$ are bounded by some integrable functions $\bar{\mu}_1(\bm{S}),\bar{\mu}_2(\bm{S})$ with respect to the distribution of $\bm{S}$, respectively; see Theorem 1.1 and Example 1.8 in \cite{shao2003mathematical}. Additionally, the smoothness condition on $Y(t)$ in Assumption~\ref{assump:diff_inter}(b) is not new, as it has been introduced in Assumption 3 of \cite{rothenhausler2019incremental} for identifying incremental effects. The key identification conditions in Assumption~\ref{assump:diff_inter} are captured by the second equalities of \eqref{theta_equi} and \eqref{theta_equi2}. As demonstrated by Proposition~\ref{prop:iden_theta}, they bridge the connection between the hypothetical potential outcome regime and the observable data distribution.

\begin{proposition}[Identification for $\theta(t)$]
\label{prop:iden_theta}
Under Assumption~\ref{assump:identify_cond}, we have that
\begin{equation}
	\label{theta_C}
	\E\left[\frac{\partial}{\partial t} \mathbb{E}\left[Y(t)|\bm{S}\right]\Big|T=t\right] = \E\left[\frac{\partial}{\partial t} \mathbb{E}\left(Y|T=t,\bm{S}\right)\Big|T=t\right] :=\theta_C(t),
\end{equation}
If, in addition, $\mathbb{E}\left[\left|\frac{\partial}{\partial t}Y(t)\right| \Big|\bm{S}=\bm{s}\right] < \infty$ for all $\bm{s}\in \mathcal{S}$, then 
\begin{equation}
	\label{theta_C2}
\E\left[\mathbb{E}\left[\frac{\partial}{\partial t}Y(t) \Big|\bm{S}\right]\Big|T=t\right] = \E\left[\frac{\partial}{\partial t} \mathbb{E}\left[Y(t)|\bm{S}\right]\Big|T=t\right]=\theta_C(t).
\end{equation}
\end{proposition}

\begin{proof}[Proof of Proposition~\ref{prop:iden_theta}]
The equality \eqref{theta_C} is a natural consequence of Assumption~\ref{assump:identify_cond} as:
$$\E\left[\frac{\partial}{\partial t} \mathbb{E}\left[Y(t)|\bm{S}\right]\Big|T=t\right] \stackrel{\text{(*)}}{=} \E\left[\frac{\partial}{\partial t} \mathbb{E}\left[Y(t)|T=t,\bm{S}\right]\Big|T=t\right] \stackrel{\text{(**)}}{=} \E\left[\frac{\partial}{\partial t} \mathbb{E}\left[Y|T=t,\bm{S}\right]\Big|T=t\right],$$
where (*) follows from Assumption~\ref{assump:identify_cond}(b) and (**) is due to Assumption~\ref{assump:identify_cond}(a).

In addition, if $\mathbb{E}\left[\left|\frac{\partial}{\partial t}Y(t)\right| \Big|\bm{S}=\bm{s}\right] < \infty$ for all $\bm{s}\in \mathcal{S}$, then the interchangeability of partial derivative $\frac{\partial}{\partial t}$ and conditional expectation $\E(\cdot|\bm{S})$ holds, and \eqref{theta_C2} follows.
\end{proof}

Under Assumption~\ref{assump:diff_inter} and Proposition~\ref{prop:iden_theta}, we can identify the derivative curve $\theta(t)$ through the form $\theta_C(t) = \E\left[\frac{\partial}{\partial t} \mu(t,\bm{S}) \Big| T=t\right]$, and we only require $\frac{\partial}{\partial t} \mu(t,\bm{s}) = \frac{\partial}{\partial t} \mathbb{E}\left(Y|T=t,\bm{S}=\bm{s}\right)$ to be well-defined whenever $p(\bm{s}|t)>0$ for any $t\in \mathcal{T}$. This serves as our key technique to bypass the restrictive positivity condition \eqref{assump:positivity} in the context of continuous treatments. 

Once $\theta(t)$ is identifiable from observable data under Assumptions~\ref{assump:identify_cond} and \ref{assump:diff_inter}, we can refer the identification and estimation back to the dose-response curve $m(t)$ as follows. For any $t\in \mathcal{T}$, the fundamental theorem of calculus reveals that 
$$m(t) = m(T) + \int_{\tilde{t}=T}^{\tilde{t}=t} m'(\tilde{t})\, d\tilde{t} = m(T)+ \int_{\tilde{t}=T}^{\tilde{t}=t} \theta(\tilde{t})\, d\tilde{t}.$$
Taking the expectation on both sides of the above equality under Assumption~\ref{assump:diff_inter} yields that
\begin{align}
	\label{integral_rel}
	\begin{split}
		m(t) = \E\left[m(T) + \int_{\tilde{t}=T}^{\tilde{t}=t} \theta(\tilde{t})\, d\tilde{t}\right] &= \E(Y) + \E\left\{\int_{\tilde{t}=T}^{\tilde{t}=t} \E\left[\frac{\partial}{\partial t}\mu(\tilde{t},\bm{S})\Big|T=\tilde{t}\right] \, d\tilde{t}\right\}.
	\end{split}
\end{align}
All the quantities on the right-hand of \eqref{integral_rel} are identifiable from the observed data $\left\{(Y_i,T_i,\bm{S}_i)\right\}_{i=1}^n$ even without the positivity condition \eqref{assump:positivity}; see \autoref{subsec:integral_est} for more details. To further demystify Assumption~\ref{assump:diff_inter}, we now present several structural equation modeling examples under which Assumption~\ref{assump:diff_inter} holds.

\subsection{Motivating Example: Additive Confounding Model}
\label{subsec:additive}

The key running example in this paper is the additive confounding model defined as:
\begin{align}
	\label{add_conf_model}
	\begin{split}
		Y(t) &= \bar{m}(t) + \eta(\bm{S}) +\epsilon,
	\end{split}
\end{align}
where $\bar{m}(t)$ is the primary treatment effect of interest, $\eta(\bm{S})$ is the random effect, and $\epsilon\in \mathbb{R}$ is an independent noise variable with $\E(\epsilon)=0$. 

Under model \eqref{add_conf_model} and Assumption~\ref{assump:identify_cond}(a), the outcome model becomes $Y=\bar{m}(T) + \eta(\bm{S})+\epsilon$. Such an additive form is a common working model in spatial confounding problems \citep{paciorek2010importance,schnell2020} and also known as the geoadditive structural equation model \citep{kammann2003geoadditive,thaden2018structural,wiecha2024two}, where $\bm{S}\in \mathbb{R}^d$ are the spatial locations (usually with $d=2$) or other spatially correlated covariates that affect both the treatment $T$ and the outcome $Y$. We summarize some key properties of the additive confounding model \eqref{add_conf_model} in the following proposition, whose proof is in \autoref{app:proof_prop1}. In particular, Proposition~\ref{prop:add_conf_prop}(c,d) verify Assumption~\ref{assump:diff_inter} under model \eqref{add_conf_model}.

\begin{proposition}[Properties of the additive confounding model]
	\label{prop:add_conf_prop}
	Let $\mu(t,\bm{s})=\mathbb{E}\left(Y|T=t,\bm{S}=\bm{s}\right)=\bar{m}(t)+\eta(\bm{s})$ under the additive confounding model \eqref{add_conf_model}. Assume that $\bar{m}(t)$ is differentiable for any $t\in \mathcal{T}$. The following results hold for any $t\in \mathcal{T}$ under Assumption~\ref{assump:identify_cond}:
	\begin{enumerate}[label=(\alph*)]
		\item $\mathbb{E}\left[\mu(t,\bm{S})\right] = \bar{m}(t)$ when $\E\left[\eta(\bm{S})\right]=0$.
		
		
		\item $\theta(t) \neq \frac{d}{d t} \E\left[\mu(t,\bm{S})|T=t\right]$ when $\frac{d}{dt} \E\left[\eta(\bm{S})|T=t\right] \neq 0$.
		
		\item $\theta(t) =\mathbb{E}\left[\frac{\partial}{\partial t} \mu(t,\bm{S})\right] = \mathbb{E}\left[\frac{\partial}{\partial t} \mu(t,\bm{S}) \big| T=t\right] = \theta_C(t)$.
		
		\item $\E(Y)=\mathbb{E}\left[\mu(T,\bm{S})\right] = \mathbb{E}\left[m(T)\right]$, where $m(t)=\E\left[Y(t)\right]$.
	\end{enumerate}
\end{proposition}

In addition to the additive confounding model \eqref{add_conf_model}, Assumption~\ref{assump:diff_inter} holds under the following more general model:
\begin{equation}
\label{add_conf_general}
Y(t)=\bar{m}(t) +\eta(\bm{S}_1,\bm{S}_2) + q(\bm{S}_1,t) \cdot \Phi(\bm{S}_2) + \epsilon,
\end{equation}
where the covariate vector is $\bm{S}=(\bm{S}_1,\bm{S}_2)^T \in \mathcal{S}_1\times \mathcal{S}_2\subset \mathbb{R}^d$, and $\eta:\mathcal{S}_1\times \mathcal{S}_2 \to \mathbb{R}, q:\mathcal{S}_1\times \mathcal{T}\to \mathbb{R}, \Phi: \mathcal{S}_2\to \mathbb{R}$ are deterministic functions. Furthermore, the random variable $\Phi(\bm{S}_2)$ needs to satisfy that
$$\E\left[\Phi(\bm{S}_2) \big|\bm{S}_1=\bm{s}_1,T=t\right] = \E\left[\Phi(\bm{S}_2) \big|\bm{S}_1=\bm{s}_1\right] = 0$$ 
for all $t\in \mathcal{T}$ and $\bm{s}_1\in \mathcal{S}_1$.

\section{Nonparametric Inference Without the Positivity Condition}
\label{sec:method}

Under the positivity condition \eqref{assump:positivity} and Assumption~\ref{assump:identify_cond}, the dose-response curve $m(t)$ can be identified through the covariate-adjusted regression function $\E\left[\mu(t,\bm{S})\right]$ from the observed data $\left\{(Y_i,T_i,\bm{S}_i)\right\}_{i=1}^n$, suggesting the following regression adjustment (RA) or G-computation estimator \citep{robins1986new,gill2001causal}:
\begin{equation}
	\label{m_RA}
	\hat{m}_{RA}(t)  = \frac{1}{n}\sum_{i=1}^n \hat{\mu}(t,\bm{S}_i),
\end{equation}
where $\hat{\mu}(t,\bm{s})$ is any consistent estimator of the conditional mean outcome function $\mu(t,\bm{s})=\E\left(Y|T=t,\bm{S}=\bm{s}\right)$. As discussed in \autoref{sec:prelim}, the above estimator \eqref{m_RA} becomes invalid and inconsistent for estimating $m(t)$ when the positivity condition \eqref{assump:positivity} fails in some region of $\mathcal{T}\times \mathcal{S}$.
Similarly, other existing methods for estimating $m(t)$, such as the inverse probability weighted (IPW) or its augmented variants, also rely on \eqref{assump:positivity} for their consistency and empirical performances \citep{diaz2013targeted,kennedy2017non,kallus2018policy,colangelo2020double}. The same issue incurred by the failure of \eqref{assump:positivity} applies to the estimation of $\theta(t) = m'(t)$.

In this section, utilizing our identification assumption \eqref{assump:diff_inter} and the integral formula \eqref{integral_rel}, we propose a novel integral estimator to address the inconsistency of RA estimators for $m(t)$ and $\theta(t)$ in the absence of \eqref{assump:positivity}. 
Furthermore, we outline a fast algorithm for computing our integral estimator of $m(t)$ in practice and delineate bootstrap inference procedures for both estimators of $\theta(t)$ and $m(t)$.

\subsection{Proposed Integral Estimator of $m(t)$}
\label{subsec:integral_est}

Given the observed data $\left\{(Y_i,T_i,\bm{S}_i)\right\}_{i=1}^n$, our estimation strategy for addressing violations the positivity condition \eqref{assump:positivity} relies on three critical insights.
\begin{itemize}
	\item {\bf Insight 1: Estimation of $\mu(t,\bm{s})$ and $\frac{\partial}{\partial t}\mu(t,\bm{s})$.} Since the observed data $(T_i,\bm{S}_i),i=1,...,n$ typically fall within high density regions of $p(t,\bm{s})$, the conditional mean outcome function $\mu(t,\bm{s})$ is well-defined and estimable at each observation $(T_i,\bm{S}_i)$, enabling consistent estimation of the partial derivative $\frac{\partial}{\partial t} \mu(t,\bm{s})$ at $(T_i,\bm{S}_i)$ as well.
	
	\item {\bf Insight 2: Estimation of $\theta(t)$ via $\theta_C(t)$.}
	The localized form $\theta_C(t) =\E\left[\frac{\partial}{\partial t}\mu(t,\bm{S})\Big|T=t\right]$ only requires accurate estimation of $\frac{\partial}{\partial t}\mu(t,\bm{s})$ at those covariate vectors $\bm{s}$ with high conditional density $p(\bm{s}|t)$.
	
	\item {\bf Insight 3: Integral Relation Between $\theta(t)$ and $m(t)$.}
	The integral formula \eqref{integral_rel} suggests that as long as we have consistent estimators of $\frac{\partial}{\partial t}\mu(t,\bm{s})$ at each $(T_i,\bm{S}_i)$ and $\theta(t)$ near $T_i$, we can extrapolate the estimation to $m(t)$ for any $t\in \mathcal{T}$ via integration even when $p(t,\bm{s})=0$ for some $\bm{s} \in \mathcal{S}$.
\end{itemize}

Building on these insights, we propose an \emph{integral estimator} of the dose-response curve $m(t)$ as:
\begin{equation}
	\label{simp_integral}
	\hat{m}_\theta(t) = \frac{1}{n}\sum_{i=1}^n \left[Y_i + \int_{\tilde{t}=T_i}^{\tilde{t}=t} \hat{\theta}_C(\tilde{t})\, d\tilde{t} \right],
\end{equation}
where $\hat{\theta}_C(t)$ is a consistent estimator of $\theta_C(t) = \E\left[\frac{\partial}{\partial t}\mu(t,\bm{S})\Big|T=t\right] = \int \frac{\partial}{\partial t} \mu(t,\bm{s})\, d\P(\bm{s}|t)$. The construction of $\hat{\theta}_C(t)$ requires estimation of two nuisance functions: 
\begin{enumerate}[label=(\roman*)]
	\item the partial derivative $\frac{\partial}{\partial t} \mu(t,\bm{s})$ of the conditional mean outcome function $\mu(t,\bm{s})$;
	
	\item the conditional cumulative distribution function (CDF) $P(\bm{s}|t)$.
\end{enumerate}
In this paper, we employ the following kernel smoothing methods for estimating these two nuisance functions and leave other possibilities for future exploration.

\subsubsection{Local Polynomial Regression Estimator of $\frac{\partial}{\partial t} \mu(t,\mathbf{s})$}
\label{subsec:partial_deriv_locpoly}

We consider estimating $\frac{\partial}{\partial t} \mu(t,\bm{s})$ by local polynomial regression \citep{fan1996local} due to its robustness around the boundary of support, known as the automatic kernel carpentry \citep{hastie1993local}. Let $K_T:\mathbb{R}\to [0,\infty),\, K_S:\mathbb{R}^d \to [0,\infty)$ be two symmetric kernel functions, and let $h,b>0$ be two smoothing bandwidth parameters. Some common univariate kernel functions include the Epanechnikov kernel $K(u) = \frac{3}{4}\left(1-u^2\right)\cdot \mathbbm{1}_{\{|u|\leq 1\}}$ and Gaussian kernel $K(u)=\frac{1}{\sqrt{2\pi}} \exp\left(-\frac{u^2}{2}\right)$. For a multivariate kernel function, we use the product kernel technique as $K_S(\bm{u}) = \prod_{i=1}^d K(u_i)$ for $\bm{u}\in \mathbb{R}^d$. 

To estimate $\frac{\partial}{\partial t} \mu(t,\bm{s})$ from the observed data $\{(Y_i,T_i,\bm{S}_i)\}_{i=1}^n$, we fit a partial local polynomial regression of order $q\geq 1$, with monomials in $\{(T_i-t)\}_{i=1}^n$ as the polynomial basis for the treatment variable $T$ and a local linear function for the covariate vector $\bm{S}$ \citep{ruppert1994multivariate}. Specifically, we define $\bm{X}_i(t,\bm{s}) = \left(1, (T_i-t),...,(T_i-t)^q, (S_{i,1}-s_1),...,(S_{i,d}-s_d) \right)^T\in\mathbb{R}^{q+1+d}$ for $i=1,...,n$ and consider 
\begin{equation}
		\label{localpoly2}
		\begin{aligned}
			&\left(\hat{\bm{\beta}}(t,\bm{s}), \hat{\bm{\alpha}}(t,\bm{s}) \right)^T \\
			&= \argmin_{(\bm{\beta},\bm{\alpha})^T \in \mathbb{R}^{q+1}\times \mathbb{R}^d} \sum_{i=1}^n \left[Y_i-\sum_{j=0}^q\beta_j (T_i-t)^q - \sum_{\ell=1}^d\alpha_{\ell}(S_{i,\ell}-s_{\ell})\right]^2 K_T\left(\frac{T_i-t}{h}\right)K_S\left(\frac{\bm{S}_i-\bm{s}}{b}\right)\\
			&= \argmin_{(\bm{\beta},\bm{\alpha})^T \in \mathbb{R}^{q+1}\times \mathbb{R}^d} \sum_{i=1}^n \left(Y_i-\bm{X}_i(t,\bm{s})^T\begin{pmatrix}
				\bm{\beta}\\
				\bm{\alpha}
			\end{pmatrix}\right)^2 K_T\left(\frac{T_i-t}{h}\right)K_S\left(\frac{\bm{S}_i-\bm{s}}{b}\right).
		\end{aligned}
\end{equation}
For simplicity, we use the same bandwidth parameter $b$ for each coordinate in $K_S$. One can generalize the above method and related theoretical results to a general bandwidth matrix for $K_S$ with little effort. Let $\bm{X}(t,\bm{s}) \in\mathbb{R}^{n\times (q+1+d)}$ be a matrix with the $j$-th row as $\bm{X}_j(t,\bm{s})^T$ and $\bm{Y} = (Y_1,..., Y_n)^T\in \mathbb{R}^n$. We define a diagonal weight matrix as:
$$\bm{W}(t,\bm{s})  = \Diag\left( K_T\left(\frac{T_1-t}{h}\right)K_S\left(\frac{\bm{S}_1-\bm{s}}{b}\right),...,
	K_T\left(\frac{T_n-t}{h}\right)K_S\left(\frac{\bm{S}_n-\bm{s}}{b}\right) \right) \in \mathbb{R}^{n\times n}.$$
Thus, \eqref{localpoly2} is a weighted least square problem and has a closed-form solution as:
\begin{equation}
	\label{localpoly3}
	\left(\hat{\bm{\beta}}(t,\bm{s}), \hat{\bm{\alpha}}(t,\bm{s}) \right)^T = \left[\bm{X}^T(t,\bm{s})\bm{W}(t,\bm{s}) \bm{X}(t,\bm{s})\right]^{-1} \bm{X}(t,\bm{s})^T\bm{W}(t,\bm{s}) \bm{Y}. 
\end{equation}
The second component $\hat{\beta}_2(t,\bm{s})$ of the fitted coefficient $\hat{\bm{\beta}}(t,\bm{s})$ provides a natural estimator of $\beta_2(t,\bm{s}) := \frac{\partial}{\partial t}\mu(t,\bm{s})$. In practice, we recommend choosing $q$ to be an even number in \eqref{localpoly2}, because there is an increment to the asymptotic variance of $\hat{\beta}_2(t,\bm{s})$ when $q$ changes from an even number to the consecutive odd number. Additionally, fitting (partial) local polynomial regressions of higher orders often give rise to a possible reduction of bias but also a substantial increase of the variability; see Chapter 3.3 in \cite{fan1996local}. We will mainly focus on (partial) local quadratic regression with $q=2$ when constructing our derivative estimator $\hat \beta_2(t,\bm{s})$ in the subsequent analysis.

\subsubsection{Kernel-Based Conditional CDF Estimator of $\P(\mathbf{s}|t)$}
\label{subsec:cond_CDF_est}

We consider estimating $\P(\bm{s}|t)$ through a Nadaraya-Watson conditional CDF estimator \citep{hall1999methods} defined as:
\begin{equation}
	\label{loc_CDF_est}
	\hat{P}_{\hslash}(\bm{s}|t) = \frac{\sum_{i=1}^n  \mathbbm{1}_{\{\bm{S}_i\leq \bm{s}\}} \cdot \bar{K}_T\left(\frac{T_i-t}{\hslash}\right)}{\sum_{j=1}^n \bar{K}_T\left(\frac{T_j-t}{\hslash}\right)},
\end{equation}
where $\bar{K}_T:\mathbb{R}\to [0,\infty)$ is again a kernel function and $\hslash>0$ is a smoothing bandwidth parameter that needs not be the same as the bandwidth $h>0$ used in estimating $\beta_2(t,\bm{s}) = \frac{\partial}{\partial t}\mu(t,\bm{s})$ via local polynomial regression. Practically, several strategies for selecting $\hslash>0$ in \eqref{loc_CDF_est} are discussed by \cite{bashtannyk2001bandwidth,holmes2012fast}.\\

Combining the partial derivative estimator $\hat{\beta}_2(t,\bm{s})$ from \eqref{localpoly3} with the conditional CDF estimator $\hat{P}_{\hslash}(\bm{s}|t)$ from \eqref{loc_CDF_est}, we obtain the final localized estimator of $\theta_C(t) = \int \frac{\partial}{\partial t} \mu(t,\bm{s})\, d\P(\bm{s}|t)$ as:
\begin{equation}
	\label{theta_C_est}
	\begin{aligned}
		\hat{\theta}_C(t) &= \int \hat{\beta}_2(t,\bm{s})\, d\hat{P}_{\hslash}(\bm{s}|t) = \frac{\sum_{i=1}^n \hat{\beta}_2(t,\bm{S}_i) \cdot \bar{K}_T\left(\frac{T_i-t}{\hslash}\right)}{\sum_{j=1}^n \bar{K}_T\left(\frac{T_j-t}{\hslash}\right)}.
	\end{aligned}
\end{equation}
In essence, $\hat{\theta}_C(t)$ is a regression adjustment (RA) estimator with two nuisance functions: $\beta_2(t,\bm{s}) = \frac{\partial}{\partial t} \mu(t,\bm{s})$ and $P(\bm{s}|t)$. We have showcased how to use the (partial) local polynomial regression \eqref{localpoly2} to estimate $\beta_2(t,\bm{s})$ and Nadaraya-Watson conditional CDF estimator \eqref{loc_CDF_est} to estimate $P(\bm{s}|t)$.

\begin{remark}[RA estimator of $\theta(t)$]
	\label{remark:theta_RA}
	Under the additive confounding model \eqref{add_conf_model}, one can directly estimate $\theta(t)$ via $\hat{\beta}_2(t,\bm{s})$ or the RA estimator $\hat{\theta}_{RA}(t)=\frac{1}{n}\sum_{i=1}^n \hat{\beta}_2(t,\bm{S}_i)$, because, as shown in Proposition~\ref{prop:add_conf_prop},
	$$\theta(t) = \frac{d}{dt} \mathbb{E}\left[\bar{m}(t) + \eta(\bm{S}) \right] = \bar{m}'(t) = \frac{\partial}{\partial t} \mu(t,\bm{s}) = \mathbb{E}\left[\frac{\partial}{\partial t}\mu(t,\bm{S})\right].$$
	However, the estimator $\hat{\beta}_2(t,\bm{s})$ is suboptimal, as it only uses information from a single location.
	Furthermore, $\hat{\theta}_{RA}(t)$ may not be stable and consistent when the positivity condition fails at certain $(t,\bm{S}_i) \in \mathcal{T}\times \mathcal{S}$; see \autoref{fig:single_conf} for an illustration. In contrast, our proposed estimator in \eqref{theta_C_est}, which relies on the conditional CDF estimator, remains valid and consistent even when the positivity condition \eqref{assump:positivity} does not hold.
\end{remark}

\begin{remark}[Linear smoother]
	It is worth noting that our integral estimator \eqref{simp_integral} under the kernel-based estimator \eqref{theta_C_est} is indeed a linear smoother. Let $\bm{e}_2 = (0,1,0,...,0)^T \in \mathbb{R}^{q+1+d}$. From \eqref{localpoly3} and $\hat{\beta}_2(t,\bm{s}) = \bm{e}_2^T \hat{\bm{\beta}}(t,\bm{s})$, we can express $\hat{\theta}_C(t)$ in \eqref{theta_C_est} as:
	$$\hat{\theta}_C(t) = \left\{\int \bm{e}_2^T\left[\bm{X}^T(t,\bm{s})\bm{W}(t,\bm{s}) \bm{X}(t,\bm{s})\right]^{-1} \bm{X}(t,\bm{s})^T\bm{W}(t,\bm{s}) \,d \hat{P}_{\hslash}(\bm{s}|t) \right\} \bm{Y}.$$
	This implies that
	\begin{align*}
		\hat m_\theta(t) &= \frac{1}{n}\sum_{i=1}^n Y_i +   \left\{\frac{1}{n}\sum_{i=1}^n \int_{\tilde{t}=T_i}^{\tilde{t}=t} \int \bm{e}_2^T\left[\bm{X}^T(t,\bm{s})\bm{W}(t,\bm{s}) \bm{X}(t,\bm{s})\right]^{-1} \bm{X}(t,\bm{s})^T\bm{W}(t,\bm{s}) \,d \hat{P}_{\hslash}(\bm{s}|t)\, d\tilde{t}\right\}\bm{Y}\\
		&\equiv \sum_{i=1}^n l_i(t) Y_i.
	\end{align*}
	As a consequence, we can apply the theory of linear smoothers to derive the effective degrees of freedom and fine-tune the smoothing bandwidth parameters \citep{buja1989linear,wasserman2006all}.
\end{remark} 

\subsection{Fast Computing Algorithm for the Proposed Integral Estimator}
\label{subsec:fast_algo}

Our proposed estimator $\hat{m}_{\theta}(t)=\frac{1}{n}\sum_{i=1}^n \left[Y_i + \int_{T_i}^t \hat{\theta}_C(\tilde{t})\, d\tilde{t} \right]$ of the dose-response curve $m(t)$ involves an integral that can be analytically difficult to compute in practice. Therefore, we propose an efficient numerical approximation method to evaluate $\hat{m}_{\theta}(t)$ with an error of order at most $O_P\left(\frac{1}{n}\right)$. The key idea is to approximate the integral via Riemann sum, evaluate $\hat{m}_{\theta}(t)$ only at the observed data points $T_1,...,T_n$, and use linear interpolation to estimate the value at any arbitrary point $t \in \mathcal{T}$.

Let $T_{(1)}\leq \cdots\leq T_{(n)}$ be the order statistics of $T_1,..., T_n$ and $\Delta_j = T_{(j+1)} - T_{(j)}$ for $j=1,..., n-1$ be the consecutive differences. The integral estimator $\hat m_\theta(t)$ in \eqref{simp_integral} can be rewritten as: 
$$\hat{m}_\theta(t) = \frac{1}{n}\sum_{i=1}^n Y_i +  \frac{1}{n}\sum_{i=1}^n \int_{\tilde{t}=T_{(i)}}^{\tilde{t}=t} \hat \theta_C(\tilde{t})\,d\tilde{t}.$$
To compute $\hat{m}_{\theta}(T_{(j)})$ at the $j$-th order statistic, we approximate the integral term $\frac{1}{n}\sum_{i=1}^n \int_{T_{(i)}}^{T_{(j)}} \hat \theta_C(\tilde{t})\, d\tilde{t}$ above as follows. 
\begin{itemize}
	\item When $i<j$, the integral is approximated by $\int_{\tilde{t}=T_{(i)}}^{\tilde{t}=T_{(j)}} \hat \theta_C(\tilde{t}) \,d\tilde{t} \approx  \sum_{\ell=i}^{j-1} \hat \theta_C(T_{(\ell)}) \Delta_{\ell}$.
	
	\item When $i>j$, the integral is approximated by $\int_{\tilde{t}=T_{(i)}}^{\tilde{t}=T_{(j)}} \hat \theta_C(\tilde{t}) \,d\tilde{t} \approx  -\sum_{\ell=j}^{i-1} \hat \theta_C(T_{(\ell+1)}) \Delta_{\ell}$.
\end{itemize}
Substituting the above results into $\frac{1}{n}\sum_{i=1}^n \int_{T_{(i)}}^{T_{(j)}} \hat \theta_C(\tilde{t})\, d\tilde{t}$ for $1< j<n$, we obtain that
\begin{align*}
	\frac{1}{n}\sum_{i=1}^n \int_{\tilde{t}=T_{(i)}}^{\tilde{t}=T_{(j)}} \hat \theta_C(\tilde{t}) \,d\tilde{t}
	& \approx \frac{1}{n}\left[\sum_{i=1}^{j-1}\sum_{\ell=i}^{j-1} \hat \theta_C(T_{(\ell)}) \Delta_{\ell}
	-  \sum_{i=j}^n\sum_{\ell=j}^{i-1}\hat \theta_C(T_{(\ell+1)}) \Delta_{\ell} \right]\\
	& \stackrel{\text{(i)}}{=} \frac{1}{n}\left[\sum_{\ell=1}^{j-1} \ell \cdot \hat \theta_C(T_{(\ell)}) \Delta_{\ell}
	- \sum_{\ell=j}^{n-1} (n-\ell)\cdot \hat \theta_C(T_{(\ell+1)}) \Delta_{\ell} \right]\\
	&  = \frac{1}{n}\sum_{i=1}^{n-1} \Delta_i \left[ i \cdot \hat \theta_C(T_{(i)}) \mathbbm{1}_{\{i<j\}} - (n-i)\cdot \hat \theta_C(T_{(i+1)}) \mathbbm{1}_{\{i\geq j\}} \right],
\end{align*}
where the equality (i) follows from switching the orders of summations. Similarly, when $j=1$ or $j=n$, we also have that
\begin{align*}
	\frac{1}{n}\sum_{i=1}^n \int_{\tilde{t}=T_{(i)}}^{\tilde{t}=T_{(j)}} \hat \theta_C(\tilde{t}) \,d\tilde{t} & \approx
	\begin{cases}
		-\frac{1}{n} \sum_{i=2}^n \sum_{\ell=1}^{i-1} \hat{\theta}_C(T_{(\ell+1)}) \Delta_{\ell} & \text{ when } j=1,\\
		\frac{1}{n} \sum_{i=1}^{n-1} \sum_{\ell=i}^{n-1} \hat{\theta}_C(T_{(\ell)}) \Delta_{\ell} & \text{ when } j=n,
	\end{cases}\\
	&= \begin{cases}
		-\frac{1}{n} \sum_{\ell=1}^{n-1} (n-\ell) \cdot \hat{\theta}_C(T_{(\ell+1)}) \Delta_{\ell} & \text{ when } j=1,\\
		\frac{1}{n} \sum_{\ell=1}^{n-1} \ell \cdot \hat{\theta}_C(T_{(\ell)}) \Delta_{\ell} & \text{ when } j=n,
	\end{cases}\\
	&= \frac{1}{n}\sum_{i=1}^{n-1} \Delta_i \left[ i \cdot \hat \theta_C(T_{(i)}) \mathbbm{1}_{\{i<j\}} - (n-i)\cdot \hat \theta_C(T_{(i+1)}) \mathbbm{1}_{\{i\geq j\}} \right].
\end{align*}
Thus, the final approximation for $\hat{m}_{\theta}(T_{(j)})$ is given by:
\begin{equation}
	\label{simp_integral_approx}
	\hat m_\theta(T_{(j)}) \approx \frac{1}{n}\sum_{i=1}^n Y_i +  \frac{1}{n}\sum_{i=1}^{n-1} \Delta_i \left[ i \cdot \hat \theta_C(T_{(i)}) \mathbbm{1}_{\{i<j\}} - (n-i)\cdot \hat \theta_C(T_{(i+1)}) \mathbbm{1}_{\{i\geq j\}} \right].
\end{equation}
To evaluate $\hat{m}_{\theta}(t)$ at any arbitrary point $t$, we use linear interpolation between $\hat{m}_{\theta}(T_{(j)})$ and $\hat{m}_{\theta}(T_{(j+1)})$ over the interval $t\in \left[T_{(j)}, T_{(j+1)}\right]$. The main advantage of using the approximation formula \eqref{simp_integral_approx} is that it requires computing the derivative estimator $\hat{\theta}_C(t)$ only at the order statistics $T_{(1)},..., T_{(n)}$.
When $\hat{m}_{\theta}(t)$ and its derivative $\hat{\theta}_C(t)$ are Lipschitz continuous and the marginal density $p_T(t)$ is uniformly bounded away from 0 within the region of interest, the error of this approximation formula is at most $O_P\left(\frac{1}{n}\right)$. This approximation error is asymptotically negligible compared to the dominating estimation error of $\hat{m}_{\theta}(t)$ itself; see \autoref{thm:int_est_rate} for further details.

\subsection{Bootstrap Inference}
\label{subsec:bootstrap_infer}

Since deriving consistent estimators of the (asymptotic) variances of our integral estimator \eqref{simp_integral} and localized derivative estimator \eqref{theta_C_est} is complicated, we propose conducting inference on $m(t)$ and $\theta(t)$ through the empirical bootstrap method \citep{efron1979bootstrap} as follows. 
\begin{enumerate}
	\item Compute the integral estimator $\hat{m}_{\theta}(t)$ and localized derivative estimator $\hat{\theta}_C(t)$ on the original data $\{(Y_i,T_i,\bm{S}_i)\}_{i=1}^n$.
	
	\item Generate $B$ bootstrap samples $\left\{\left(Y_i^{*(b)},T_i^{*(b)},\bm{S}_i^{*(b)}\right)\right\}_{i=1}^n, b=1,...,B$ by sampling with replacement from the original data and compute the integral estimator $\hat{m}_{\theta}^{*(b)}(t)$ and localized derivative estimator $\hat{\theta}_C^{*(b)}(t)$ on each bootstrapped sample for $b=1,...,B$.
	
	\item Let $\alpha \in (0,1)$ be a pre-specified significance level.
	\begin{itemize}
		\item For a pointwise inference at $t_0\in \mathcal{T}$, calculate the $1-\alpha$ quantiles $\zeta_{1-\alpha}^*(t_0)$ and $\bar{\zeta}_{1-\alpha}^*(t_0)$ of $\{D_1(t_0),...,D_B(t_0)\}$ and $\{\bar{D}_1(t_0),...,\bar{D}_B(t_0)\}$ respectively, where $D_b(t_0) = \left|\hat{m}_{\theta}^{*(b)}(t_0) - \hat{m}_{\theta}(t_0)\right|$ and $\bar{D}_b(t_0) = \left|\hat{\theta}_C^{*(b)}(t_0) - \hat{\theta}_C(t_0)\right|$ for $b=1,...,B$. 
		
		\item For an uniform inference on the entire dose-response curve $m(t)$ and its derivative $\theta(t)$, compute the $1-\alpha$ quantiles $\xi_{1-\alpha}^*$ and $\bar{\xi}_{1-\alpha}^*$ of $\{D_{\sup,1},...,D_{\sup,B}\}$ and $\{\bar{D}_{\sup,1},...,\bar{D}_{\sup,B}\}$ respectively, where $D_{\sup,b} = \sup_{t\in \mathcal{T}}\left|\hat{m}_{\theta}^{*(b)}(t) - \hat{m}_{\theta}(t)\right|$ and $\bar{D}_{\sup,b} = \sup_{t\in \mathcal{T}}\left|\hat{\theta}_C^{*(b)}(t) - \hat{\theta}_C(t)\right|$ for $b=1,...,B$.
	\end{itemize}
	
	\item Define the $1-\alpha$ confidence intervals for $m(t_0)$ and $\theta(t_0)$ as:
	$$\left[\hat{m}_{\theta}(t_0) - \zeta_{1-\alpha}^*(t_0),\, \hat{m}_{\theta}(t_0) + \zeta_{1-\alpha}^*(t_0)\right] \; \text{ and } \; \left[\hat{\theta}_C(t_0) - \bar{\zeta}_{1-\alpha}^*(t_0),\, \hat{\theta}_C(t_0) + \bar{\zeta}_{1-\alpha}^*(t_0)\right],$$ 
	respectively, as well as the simultaneous $1-\alpha$ confidence bands as:
	$$\left[\hat{m}_{\theta}(t) - \xi_{1-\alpha}^*,\, \hat{m}_{\theta}(t) + \xi_{1-\alpha}^*\right] \quad \text{ and } \quad \left[\hat{\theta}_C(t) - \bar{\xi}_{1-\alpha}^*,\, \hat{\theta}_C(t) + \bar{\xi}_{1-\alpha}^*\right]$$ 
	for every $t\in \mathcal{T}$.
\end{enumerate}
In \autoref{subsec:bootstrap_validity}, we will establish the consistency of the above bootstrap inference procedures.

\section{Asymptotic Theory}
\label{sec:theory}

In this section, we study the consistency results of our integral estimator \eqref{simp_integral} and localized derivative estimator \eqref{theta_C_est}, along with the validity of bootstrap inference described in \autoref{subsec:bootstrap_infer}.

\subsection{Notations and Assumptions}

We introduce the regularity conditions for our subsequent theoretical analysis. Recall that $\mathcal{E} \subset \mathcal{T}\times \mathcal{S}$ is the support of the joint density $p(t,\bm{s})$, $\mathcal{E}^{\circ}$ is the interior of $\mathcal{E}$, and $\partial\mathcal{E}$ is the boundary of $\mathcal{E}$.

\begin{assump}[Differentiability of the conditional mean outcome function]
	\label{assump:reg_diff}\noindent
	\begin{enumerate}[label=(\alph*)]
		\item For any $(t,\bm{s}) \in \mathcal{E}^{\circ}$, the conditional mean outcome function $\mu(t,\bm{s})$ is at least $(q+1)$ times continuously differentiable with respect to $t$ and at least four times continuously differentiable with respect to $\bm{s}$, where $q$ is the order of (partial) local polynomial regression in \eqref{localpoly2}.
		
		\item All these partial derivatives of $\mu(t,\bm{s})$ are continuous up to the boundary $\partial \mathcal{E}$.
		
		\item $\mu(t,\bm{s})$ and the partial derivatives are uniformly bounded on $\mathcal{E}$.
		
		\item There exist absolute constants $\sigma,A_0 >0$ such that $\mathrm{Var}(Y|T=t,\bm{S}=\bm{s}) = \sigma^2$ and $\E|Y|^4 < A_0<\infty$ uniformly in $\mathcal{E}$.
	\end{enumerate}
\end{assump}

\begin{assump}[Differentiability of the joint density]
	\label{assump:den_diff}\noindent
	\begin{enumerate}[label=(\alph*)]
		\item The joint density $p(t,\bm{s})$ is bounded and at least twice continuously differentiable with bounded partial derivatives up to the second order on $\mathcal{E}^{\circ}$.
		
		\item All these partial derivatives of $p(t,\bm{s})$ are continuous up to the boundary $\partial \mathcal{E}$.
		
		\item $\mathcal{E}$ is compact and $p(t,\bm{s})$ is uniformly bounded away from 0 on $\mathcal{E}$.
		
		\item The marginal density $p_T(t)$ of $T$ is non-degenerate, \emph{i.e.}, its support $\mathcal{T}$ has a nonempty interior.
	\end{enumerate}
\end{assump}

Assumptions~\ref{assump:reg_diff} and \ref{assump:den_diff} are standard differentiability conditions in local polynomial regression literature \citep{ruppert1994multivariate,fan1996local}, which can be slightly relaxed by the H\"older continuity condition. These conditions ensure that the bias term of $\hat{\beta}_2(t,\bm{s})$ from the local polynomial regression \eqref{localpoly3} converges to 0 of the standard order $O(h^q) + O(b^2)$; see Lemma~\ref{lem:loc_poly_deriv_unif} below. Notice that the homogeneous error condition in Assumption~\ref{assump:reg_diff}(d) is imposed only for simplifying the proof. Our asymptotic theories hold for heterogeneous errors that are uniformly bounded away from 0 in $\mathcal{E}$. Additionally, the projection $\mathrm{proj}_T(\mathcal{E})$ of the joint density support $\mathcal{E}$ onto the domain of $T$ coincides with the marginal support $\mathcal{T}$. Hence, $\mathcal{T}$ is compact under Assumption~\ref{assump:den_diff}. 

To mitigate the boundary effects of the local polynomial regression, we impose the following conditions on $p(t,\bm{s})$, $\mu(t,\bm{s})$, and the geometric structure of $\mathcal{E}$ near its boundary $\partial \mathcal{E}$. 

\begin{assump}[Boundary conditions]
	\label{assump:boundary_cond}\noindent
	\begin{enumerate}[label=(\alph*)]	
		\item There exists some constants $r_1,r_2 \in (0,1)$ such that for any $(t,\bm{s}) \in \mathcal{E}$ and all $\delta\in (0,r_1]$, there is a point $(t',\bm{s}') \in \mathcal{E}$ satisfying 
		$$\mathcal{B}\left((t',\bm{s}'),\, r_2\delta\right) \subset \mathcal{B}\left((t,\bm{s}),\, \delta\right) \cap \mathcal{E},$$
		where $\mathcal{B}((t,\bm{s}),\, r)=\left\{(t_1,\bm{s}_1) \in \mathbb{R}^{d+1}: \norm{(t_1-t,\bm{s}_1-\bm{s})}_2 \leq r \right\}$ with $\norm{\cdot}_2$ being the standard Euclidean norm.
		
		\item For any $(t,\bm{s}) \in \partial \mathcal{E}$, the boundary of $\mathcal{E}$, it satisfies that $\frac{\partial}{\partial t} p(t,\bm{s}) = \frac{\partial}{\partial s_j} p(t,\bm{s})=0$ and $\frac{\partial^2}{\partial s_j^2} \mu(t,\bm{s}) = 0$ for all $j=1,...,d$.
		
		\item For any $\delta >0$, the Lebesgue measure of the set $\partial \mathcal{E} \oplus \delta$ satisfies $\left|\partial \mathcal{E} \oplus \delta \right| \leq A_1\cdot \delta$ for some absolute constant $A_1>0$, where $\partial \mathcal{E} \oplus \delta = \left\{\bm{z}\in \mathbb{R}^{d+1}: \inf_{\bm{x}\in \partial \mathcal{E}} \norm{\bm{z}-\bm{x}}_2 \leq \delta\right\}$.
	\end{enumerate}
\end{assump}

Assumption~\ref{assump:boundary_cond}(a) is adopted from the boundary condition (Assumption X) in \cite{fan2015multivariate}, whose primitive and stronger version also appeared as Assumption (A4) in \cite{ruppert1994multivariate}. This is a relatively mild condition that holds when the boundary $\partial\mathcal{E}$ is smooth or contains non-smooth vertices from regular structures, such as $d+1$ dimensional cubes or convex cones. Indeed, Assumption~\ref{assump:boundary_cond}(a) remains valid if any ball inside $\mathcal{E}$ that centers near a vertex of $\partial \mathcal{E}$ has its radius shrunk linearly when approaching the vertex. More examples and discussions about Assumption~\ref{assump:boundary_cond}(a) are available in \cite{ruppert1994multivariate,fan2015multivariate}. The main purpose of this support condition is to ensure sufficient observations near the boundary so that the rate of convergence for our local polynomial estimator $\hat{\beta}_2(t,\bm{s})$ at boundary points matches that at interior points of $\mathcal{E}$. Assumption~\ref{assump:boundary_cond}(b) also regularizes the slope of $p(t,\bm{s})$ and the curvature of $\mu(t,\bm{s})$ at any boundary point $(t,\bm{s}) \in \partial \mathcal{E}$, requiring $p(t,\bm{s})$ to be flat and $\mu(t,\bm{s})$ to exhibit zero curvature at $\partial\mathcal{E}$. Similar to Assumption~\ref{assump:den_diff}, the partial derivatives $\frac{\partial}{\partial t} p(t,\bm{s}), \frac{\partial}{\partial s_j} p(t,\bm{s})$ are computed as limits from nearby interior points. This condition is also essential for maintaining the rate of convergence for the bias term of $\hat{\beta}_2(t,\bm{s})$ at boundary points. Finally, Assumption~\ref{assump:boundary_cond}(c) ensures that $\partial \mathcal{E}$ does not exhibit fractal-like or irregular structures that lead to an infinite perimeter.

To establish the (uniform) consistency of our local polynomial regression estimator \eqref{localpoly3} and Nadaraya-Watson conditional CDF estimator \eqref{loc_CDF_est}, we impose the following regularity conditions on the kernel functions.

\begin{assump}[Regular kernel and VC-type conditions]
	\label{assump:reg_kernel}\noindent
	\begin{enumerate}[label=(\alph*)]
		\item The functions $K_T:\mathbb{R} \to [0,\infty)$ and $K_S:\mathbb{R}^d \to [0,\infty)$ are compactly supported and Lispchitz continuous kernels such that $\int_{\mathbb{R}} K_T(t)\,dt = \int_{\mathbb{R}^d} K_S(\bm{s}) \, d\bm{s}=1$, $K_T(t)=K_T(-t)$, and $K_S$ is radially symmetric with $\int \bm{s}\cdot K_S(\bm{s}) d\bm{s}=\bm{0}$. In addition, for all $j=1,2,...,$ and $\ell=1,...,d$, it holds that 
		\begin{align*}
			&\kappa_j^{(T)} := \int_{\mathbb{R}} u^j K_T(u) \, du < \infty,\quad \nu_j^{(T)} := \int_{\mathbb{R}} u^j K_T^2(u) \, du < \infty, \\ &\kappa_{j,\ell}^{(S)} := \int_{\mathbb{R}^d} u_{\ell}^j K_S(\bm{u})\, d\bm{u} < \infty, \quad \text{ and } \quad \nu_{j,k}^{(S)} := \int_{\mathbb{R}^d} u_{\ell}^j K_S^2(\bm{u})\, d\bm{u} < \infty.
		\end{align*}
		Finally, both $K_T$ and $K_S$ are second-order kernels, \emph{i.e.}, $\kappa_2^{(T)} >0$ and $\kappa_{2,\ell}^{(S)} >0$ for all $\ell=1,...,d$.
		
		\item Let $\mathcal{K}_{q,d} = \Big\{(y,\bm{z})\mapsto \left(\frac{y-t}{h}\right)^{\ell} \left(\frac{z_i-s_i}{b}\right)^{k_1} \left(\frac{z_j-s_j}{b}\right)^{k_2} K_T\left(\frac{y-t}{h}\right) K_S\left(\frac{\bm{z}-\bm{s}}{b}\right): (t,\bm{s})\in \mathcal{T}\times \mathcal{S}; i,j=1,...,d;\ell=0,...,2q; k_1,k_2=0,1; h, b>0\Big\}$. It holds that $\mathcal{K}_{q,d}$ is a bounded VC-type class of measurable functions on $\mathbb{R}^{d+1}$.
		
		\item The function $\bar{K}_T:\mathbb{R}\to [0,\infty)$ is a second-order, Lipschitz continuous, and symmetric kernel with a compact support, \emph{i.e.}, $\int_{\mathbb{R}} \bar{K}_T(t)\, dt=1$, $\bar{K}_T(t)=\bar{K}_T(-t)$, and $\int_{\mathbb{R}} t^2 \bar{K}_T(t) \, dt \in (0,\infty)$. 
		
		\item Let $\bar{\mathcal{K}} = \left\{y \mapsto \bar{K}_T\left(\frac{y-t}{\hslash}\right): t\in \mathcal{T}, \hslash >0 \right\}$. It holds that $\bar{\mathcal{K}}$ is a bounded VC-type class of measurable functions on $\mathbb{R}$.
	\end{enumerate}
\end{assump}

Assumption~\ref{assump:reg_kernel}(a,c) describe inherent properties of commonly used kernel functions rather than being regularity conditions. Assumption~\ref{assump:reg_kernel}(b,d) are critical for the uniform consistency of kernel-based estimators \citep{gine2002rates,einmahl2005uniform} and can be satisfied by many kernel functions, such as Gaussian and Epanechnikov kernels. For example, these assumptions hold when the kernel is a composite function of a polynomial in $d\geq 1$ variables and a real-valued function of bounded variation; see Lemma 22 in \cite{nolan1987u}. Since most of the kernel functions are bounded, we can always take constant envelope functions for $\mathcal{K}_{q,d}$ and $\bar{\mathcal{K}}$ in Assumption~\ref{assump:reg_kernel}(b,d).

\subsection{Consistency of the Integral Estimator}
\label{subsec:consistency_int_est}

Before establishing the consistency of our proposed integral estimator \eqref{simp_integral} and localized derivative estimator \eqref{theta_C_est}, we first study the uniform consistency of $\hat{\beta}_2(t,\bm{s})$ from the local polynomial regression \eqref{localpoly3} as a building block. The pointwise consistency of $\hat{\beta}_2(t,\bm{s})$ is detailed in \autoref{app:pointwise_loc_poly} instead.

\begin{lemma}[Uniform convergence of $\hat{\beta}_2(t,\bm{s})$]
	\label{lem:loc_poly_deriv_unif}
	Let $q>0$. Suppose that Assumptions~\ref{assump:reg_diff}, \ref{assump:den_diff}, \ref{assump:boundary_cond}, and \ref{assump:reg_kernel}(a,b) hold. Let $\hat \beta_2(t,\bm{s})$ be the second element of $\hat{\bm{\beta}}(t,\bm{s}) \in \mathbb{R}^{q+1}$ and $\beta_2(t,\bm{s}) = \frac{\partial}{\partial t} \mu(t,\bm{s})$. Then, as $h,b,\frac{\max\{h,b\}^4}{h}\to 0$ and $\frac{nh^3 b^d}{|\log(hb^d)|},\frac{|\log(hb^d)|}{\log\log n} \to \infty$, we have that 
	$$\sup_{(t,\bm{s})\in \mathcal{E}}\left|\hat{\beta}_2(t,\bm{s}) - \beta_2(t,\bm{s})\right| =
	O\left(h^q + b^2 + \frac{\max\{b,h\}^4}{h}\right) + O_P\left(\sqrt{\frac{|\log(hb^d)|}{nh^3 b^d}}\right).$$
\end{lemma}

The proof of Lemma~\ref{lem:loc_poly_deriv_unif} is in \autoref{app:proof_beta_2_unif}. At a high level, our rates of convergence for $\hat{\beta}_2(t,\bm{s})$ in Lemma~\ref{lem:loc_poly_deriv_unif} are consistent with standard results in local polynomial regression literature \citep{ruppert1994multivariate,fan1996local,lu1996multivariate}. The extra bias rate, $O\left(\frac{b^4}{h}\right)$ or $O\left(\frac{\max\{h,b\}^4}{h}\right)$, arises from a higher-order term in Taylor's expansion and becomes asymptotically negligible when $h\gtrsim b$ as $n\to \infty$; see \eqref{taylor_reg} and \eqref{bias_term1} in \autoref{app:proof_beta_2} for example. The stochastic variation term $\hat{\beta}_2(t,\bm{s}) - \mathbb{E}\left[\hat{\beta}_2(t,\bm{s})\right]$ can be approximated (up to a scaled factor $\sqrt{nh^3b^d}$) by an empirical process, leading to the upper bound of its uniform rate of convergence as in \cite{einmahl2005uniform}. Notably, when using the (partial) local quadratic regression ($q=2$) to obtain $\hat{\beta}_2(t,\bm{s})$, the bandwidths $h\asymp b$ that optimally trades off the bias and variance in Lemma~\ref{lem:loc_poly_deriv_unif} are of order $O\left(n^{-\frac{1}{d+7}}\right)$, resulting in the optimal rate for derivative estimation established in \cite{stone1980optimal,stone1982optimal}. 

The uniform rate of convergence in Lemma~\ref{lem:loc_poly_deriv_unif} is useful not only for studying the asymptotic behaviors of our integral estimator \eqref{simp_integral} and localized derivative estimator \eqref{theta_C_est} in \autoref{thm:int_est_rate} but also for analyzing the bootstrap consistency in \autoref{subsec:bootstrap_validity}. With Lemma~\ref{lem:loc_poly_deriv_unif}, we now present the uniform consistency results for $\hat{m}_{\theta}(t)$ in \eqref{simp_integral} and $\hat{\theta}_C(t)$ in \eqref{theta_C_est}. 

\begin{theorem}[Convergence of $\hat{\theta}_C(t)$ and $\hat{m}_{\theta}(t)$]
	\label{thm:int_est_rate}
	Let $q>0$ and $\mathcal{T}' \subset \mathcal{T}$ be a compact set so that $p_T(t)$ is uniformly bounded away from 0 within $\mathcal{T}'$. Suppose that Assumptions~\ref{assump:identify_cond}, \ref{assump:diff_inter}, \ref{assump:reg_diff}, \ref{assump:den_diff}, \ref{assump:boundary_cond}, and \ref{assump:reg_kernel} hold. Then, as $h,b,\hslash,\frac{\max\{b,h\}^4}{h}\to 0$, and $\frac{n\max\{h,\hslash\}b^d}{|\log(\hslash hb^d)|}, \frac{|\log(h\hslash b^d)|}{\log\log n}, \frac{nh^3}{|\log h|}\to \infty$, we know that
	\begin{align*}
		\sup_{t\in \mathcal{T}'} \left|\hat{\theta}_C(t) - \theta(t)\right| &= O\left(h^q + b^2 + \frac{\max\{b,h\}^4}{h}\right) + O_P\left(\sqrt{\frac{|\log h|}{nh^3}} + \hslash^2 + \sqrt{\frac{|\log \hslash|}{n\hslash}}\right)
	\end{align*}
	and
	\begin{align*}
		\sup_{t\in \mathcal{T}'}\left|\hat m_\theta(t) - m(t) \right| &= O_P\left(\frac{1}{\sqrt{n}}\right) + O\left(h^q + b^2 + \frac{\max\{b,h\}^4}{h}\right) + O_P\left(\sqrt{\frac{|\log h|}{nh^3}} + \hslash^2 + \sqrt{\frac{|\log \hslash|}{n\hslash}}\right).
	\end{align*}
\end{theorem}

The proof of \autoref{thm:int_est_rate} is in \autoref{app:sim_integral_proof}. We restrict the uniform consistency results in \autoref{thm:int_est_rate} to a compact set $\mathcal{T}'\subset \mathcal{T}$ to avoid issues with the density decay of $p_T(t)$ near the boundary of $\mathcal{T}$. If $p_T(t)$ is uniformly bounded away from 0 in its support $\mathcal{T}$, we can take $\mathcal{T}' = \mathcal{T}$. 

The (uniform) rate of convergence of $\hat{m}_{\theta}(t)$ consists of two parts. The first part $O_P\left(\frac{1}{\sqrt{n}}\right)$ arises from the sample average of $Y_i,i=1,...,n$ and is asymptotically negligible. The second dominant part is due to the integral component:
\begin{equation}
	\label{integral_part_si}
	\hat \Delta_{h,b}(t) =\frac{1}{n}\sum_{i=1}^n \int_{\tilde{t}=T_i}^{\tilde{t}=t} \hat \theta_C(\tilde{t})d\tilde{t},
\end{equation}
whose rate of convergence is determined by $\hat{\theta}_C(t)$. Furthermore, the (uniform) rate of convergence of $\hat{\theta}_C(t)$ interweaves the rate $O\left(h^q + b^2 + \frac{\max\{b,h\}^4}{h}\right) + O_P\left(\sqrt{\frac{|\log h|}{nh^3}}\right)$ for estimating a one-dimensional derivative-like quantity $\mathbb{E}\left[\frac{\partial}{\partial t} \mu(t,\bm{S}) \big| T=t\right]$ with the rate $O(\hslash^2) + O_P\left(\sqrt{\frac{|\log \hslash|}{n\hslash}}\right)$ for estimating the conditional CDF $\P(\bm{s}|t)$. Similar rates of convergence for the additive model \eqref{add_conf_model} have also been obtained by \cite{fan1998direct}.

\subsection{Validity of the Bootstrap Inference}
\label{subsec:bootstrap_validity}

Before proving the consistency of our bootstrap procedure in \autoref{subsec:bootstrap_infer}, we derive the asymptotic linearity of our integral estimator \eqref{simp_integral} and localized derivative estimator \eqref{theta_C_est} as intermediate results.

\begin{lemma}[Asymptotic linearity]
	\label{lem:asymp_linear}
	Let $q\geq 2$ in the local polynomial regression for estimating $\frac{\partial}{\partial t}\mu(t,\bm{s})$ and $\mathcal{T}' \subset \mathcal{T}$ be a compact set so that $p_T(t)$ is uniformly bounded away from 0 within $\mathcal{T}'$. Suppose that Assumptions~\ref{assump:identify_cond}, \ref{assump:diff_inter}, \ref{assump:reg_diff}, \ref{assump:den_diff}, \ref{assump:boundary_cond}, and \ref{assump:reg_kernel} hold. Then, if $b\lesssim h\asymp n^{-\frac{1}{\gamma}}$ and $\hslash \asymp n^{-\frac{1}{\varpi}}$ for some $\gamma, \varpi >0$ such that $\frac{nh^5}{\log n} \to c_1$ and $\frac{n\hslash^5}{\log n} \to c_2$ for some finite number $c_1,c_2 \geq 0$ and $\sqrt{nh^3 \max\{h, \hslash\}^4}, \frac{h^3\log n}{\hslash},\frac{\log n}{\sqrt{n\hslash}}, \frac{\log n}{n\max\{h,\hslash\} b^d} \to 0$ as $n\to \infty$, then for any $t\in \mathcal{T}'$, we have that
	$$\sqrt{nh^3} \left[\hat{\theta}_C(t) -\theta(t) \right] = \mathbb{G}_n \bar{\varphi}_t + o_P(1) \quad \text{ and } \quad \sqrt{nh^3} \left[\hat{m}_{\theta}(t) -m(t) \right] = \mathbb{G}_n \varphi_t + o_P(1),$$
	where $\bar{\varphi}_t(Y, T, \bm{S}) = \frac{C_{K_T}[Y -\mu(T,\bm{S})]}{\sqrt{h}\cdot p_T(t)} \left(\frac{T-t}{h}\right) K_T\left(\frac{T-t}{h}\right)$ for some constant $C_{K_T}>0$ that only depends on the kernel $K_T$ and 
	\begin{align}
		\label{influ_func}
		\begin{split}
			\varphi_t\left(Y,T,\bm{S}\right) &= \mathbb{E}_{T_1}\left[\int_{T_1}^t \bar{\varphi}_{\tilde{t}}(Y, T, \bm{S}) \, d\tilde{t}\right] = \mathbb{E}_{T_1}\left\{\int_{T_1}^t \frac{C_{K_T} \left[Y-\mu(T,\bm{S})\right]}{\sqrt{h}\cdot p_T(\tilde{t})} \left(\frac{T-\tilde{t}}{h}\right) K_T\left(\frac{T-\tilde{t}}{h}\right) d\tilde{t}\right\}.
		\end{split}
	\end{align}
	Furthermore, we have the following uniform results as:
	\begin{align*}
		&\left|\sqrt{nh^3}\sup_{t\in \mathcal{T}'}\left|\hat{\theta}_C(t) -\theta(t)\right| - \sup_{t\in \mathcal{T}'} |\mathbb{G}_n\bar{\varphi}_t|\right| \asymp \left|\sqrt{nh^3} \sup_{t\in \mathcal{T}'}\left|\hat{m}_{\theta}(t) -m(t)\right| - \sup_{t\in \mathcal{T}'}|\mathbb{G}_n\varphi_t|\right| \\
		&= O_P\left(\sqrt{nh^3 \max\{h, \hslash\}^4} + \sqrt{\frac{h^3\log n}{\hslash}} + \frac{\log n}{\sqrt{n\hslash}} + \sqrt{\frac{\log n}{nb^d\hslash}}\right).
	\end{align*}
\end{lemma}

The proof of Lemma~\ref{lem:asymp_linear} is in \autoref{app:asymp_linear_proof}, and we make two remarks to this crucial lemma.

\begin{remark}[Non-degeneracy and validity of pointwise confidence intervals]
	We prove in Lemma~\ref{lem:influc_func_moment_bnd} of \autoref{app:asymp_linear_proof} that the variances $\text{Var}\left[\varphi_t(Y,T,\bm{S})\right]$ and $\text{Var}\left[\bar{\varphi}_t(Y,T,\bm{S})\right]$ of the influence functions in \eqref{influ_func} are positive for each $t\in \mathcal{T}'$, provided that $\mathrm{Var}(Y|T=t,\bm{S}=\bm{s})$ is finite and uniformly bounded away from 0 for all $(t,\bm{s})\in \mathcal{E}$. This ensures that the asymptotic linearity results in Lemma~\ref{lem:asymp_linear} are non-degenerate. Similarly, the bootstrap estimates $\hat{m}_{\theta}^*(t)$ and $\hat{\theta}_C^*(t)$ are also asymptotically linear given the observed data $\left\{(Y_i,T_i,\bm{S}_i)\right\}_{i=1}^n$; see the proof of \autoref{thm:bootstrap_cons} in \autoref{app:boot_consistency}. This indicates that the pointwise bootstrap confidence intervals for $m(t)$ and $\theta(t)$ in \autoref{subsec:bootstrap_infer} are asymptotically valid; see Lemma 23.3 in \cite{VDV1998} and related results in \cite{arcones1992bootstrap,tang2024consistency}.
\end{remark}

\begin{remark}[Bandwidth selection and the curse of dimensionality]
	\label{remark:bw_sel}
	To ensure that the remainder terms of the asymptotically linear forms in Lemma~\ref{lem:asymp_linear} are of order $o_P(1)$, we can choose the bandwidths $h\asymp \hslash$ to be of order $O\left(n^{-\frac{1}{5}}\right)$, consistent with typical bandwidth selection methods for nonparametric regression \citep{wasserman2006all, schindler2011bandwidth}. Thus, we can tune $h,\hslash$ using standard bandwidth selectors for estimating the regression function $\mu(t,\bm{s})$ and conditional CDF $p(\bm{s}|t)$ without any explicit undersmoothing. For the bandwidth parameter $b$ related to the confounding variables in the local polynomial regression, it must satisfy $\left(\frac{\log n}{n}\right)^{\frac{4}{5d}}\lesssim b \lesssim n^{-\frac{1}{5}}$ to maintain fast and dimensionally independent rates of convergence for our integral and derivative estimators \eqref{simp_integral} and \eqref{theta_C_est}. This constraint implies that the dimension of confounding variables has to be $d\leq 4$ and reveals that our nonparametric inference procedures for $m(t)$ and $\theta_C(t)$ suffer from the curse of dimensionality. However, when $d\geq 5$, the consistency and asymptotic properties of our estimators $\hat{m}_{\theta}(t)$ and $\hat{\theta}_C(t)$ remain valid, though their rates of convergence become slower and dimensionally dependent.
\end{remark}

Apart from the above asymptotic linearity results, we establish couplings between $\sqrt{nh^3}\cdot \sup_{t\in \mathcal{T}'}\left|\hat{m}_{\theta}(t) -m(t)\right|$ and $\sup_{t\in \mathcal{T}'}\left|\mathbb{G}_n(\varphi_t)\right|$ (and similarly, for $\hat{\theta}_C(t)$) in Lemma~\ref{lem:asymp_linear}. These couplings are crucial for deriving the Gaussian approximations of $\hat{m}{\theta}(t)$ and $\hat{\theta}_C(t)$ in \autoref{thm:gauss_approx_si} and for establishing their bootstrap consistency. Now, consider two function classes
\begin{equation}
	\label{func_class}
	\mathcal{F} = \left\{(v,x,\bm{z}) \mapsto \varphi_t(v,x,\bm{z}): t\in \mathcal{T}'\right\} \; \text{ and } \; \mathcal{F}_{\theta} = \left\{(v,x,\bm{z}) \mapsto \bar{\varphi}_t(v,x,\bm{z}): t\in \mathcal{T}'\right\}
\end{equation} 
with $\varphi_t,\bar{\varphi}_t$ defined in \eqref{influ_func} respectively. Let $\mathbb{B}, \bar{\mathbb{B}}$ be two Gaussian processes indexed by $\mathcal{F}$ and $\mathcal{F}_{\theta}$ respectively with zero means and covariance functions
\begin{align*}
	&\mathrm{Cov}(\mathbb{B}(f_1), \mathbb{B}(f_2)) = \mathbb{E}\left[f_1(\bm{U}) \cdot f_2(\bm{U})\right] \; \text{ and } \; \mathrm{Cov}(\bar{\mathbb{B}}(g_1), \bar{\mathbb{B}}(g_2)) = \mathbb{E}\left[g_1(\bm{U}) \cdot g_2(\bm{U})\right]
\end{align*}
with $\bm{U}=(Y,T,\bm{S})$ for any $f_1,f_2 \in \mathcal{F}$ and $g_1,g_2\in \mathcal{F}_{\theta}$.

\begin{theorem}[Gaussian approximation]
	\label{thm:gauss_approx_si}
	Let $q\geq 2$ in the local polynomial regression for estimating $\frac{\partial}{\partial t}\mu(t,\bm{s})$ and $\mathcal{T}' \subset \mathcal{T}$ be a compact set so that $p_T(t)$ is uniformly bounded away from 0 within $\mathcal{T}'$. Suppose that Assumptions~\ref{assump:identify_cond}, \ref{assump:diff_inter}, \ref{assump:reg_diff}, \ref{assump:den_diff}, \ref{assump:boundary_cond}, and \ref{assump:reg_kernel} hold. If $b\lesssim h\asymp n^{-\frac{1}{\gamma}}$ and $\hslash \asymp n^{-\frac{1}{\varpi}}$ for some $\gamma\geq \varpi >0$ such that $\frac{nh^5}{\log n} \to c_1$ and $\frac{n\hslash^5}{\log n} \to c_2$ for some finite number $c_1,c_2 \geq 0$ and $\frac{\hslash}{h^3\log n}, \hslash n^{\frac{1}{3}}\log n, \frac{\sqrt{n\hslash}}{\log n}, \frac{n\max\{h,\hslash\} b^d}{\log n} \to \infty$ as $n\to \infty$, then there exist Gaussian processes $\mathbb{B}, \bar{\mathbb{B}}$ such that 
	\begin{align*}
		&\sup_{u\geq 0} \left|\P\left(\sqrt{nh^3}\cdot \sup_{t\in \mathcal{T}'}\left|\hat{m}_{\theta}(t) -m(t)\right| \leq u\right) - \P\left(\sup_{f\in \mathcal{F}} |\mathbb{B}(f)| \leq u\right) \right|\\
		& \asymp \sup_{u\geq 0} \left|\P\left(\sqrt{nh^3}\cdot \sup_{t\in \mathcal{T}'}\left|\hat{\theta}_C(t) -\theta(t)\right| \leq u\right) - \P\left(\sup_{g\in \mathcal{F}_{\theta}} |\bar{\mathbb{B}}(g)| \leq u\right) \right| \\
		&= O\left(\left(\frac{\log^5 n}{nh^3}\right)^{\frac{1}{8}} + \left(\frac{\log^2 n}{nb^d\hslash}\right)^{\frac{3}{8}}\right),
	\end{align*}
	where $\mathcal{F},\mathcal{F}_{\theta}$ are defined in \eqref{func_class}.
\end{theorem}

The proof of \autoref{thm:gauss_approx_si} is in \autoref{app:gauss_approx_proof}. It demonstrates that the distributions of $\sqrt{nh^3}\cdot \sup_{t\in \mathcal{T}'}\left|\hat{m}_{\theta}(t) -m(t)\right|$ and $\sqrt{nh^3}\cdot \sup_{t\in \mathcal{T}'}\left|\hat{\theta}_C(t) -\theta(t)\right|$ can be asymptotically approximated by the suprema of two separate Gaussian processes respectively. To establish the consistency of our bootstrap inference procedure in \autoref{subsec:bootstrap_infer}, it remains to show that the bootstrap versions of the differences $\sqrt{nh^3}\cdot \sup_{t\in \mathcal{T}'}\left|\hat{m}_{\theta}^*(t) -\hat{m}_{\theta}(t)\right|$ and $\sqrt{nh^3}\cdot \sup_{t\in \mathcal{T}'}\left|\hat{\theta}_C^*(t) -\hat{\theta}_C(t)\right|$, conditional on the observed data $\bm{U}_n=\left\{(Y_i,T_i,\bm{S}_i)\right\}_{i=1}^n$, can also be asymptotically approximated by the suprema of the same Gaussian processes, which are summarized in the following theorem.

\begin{theorem}[Bootstrap consistency]
	\label{thm:bootstrap_cons}
	Let  $\mathbb{U}_n=\left\{(Y_i,T_i,\bm{S}_i)\right\}_{i=1}^n$ be the observed data. Under the same setup of \autoref{thm:gauss_approx_si}, we have that
	\begin{align*}
		&\sup_{u\geq 0} \left|\P\left(\sqrt{nh^3} \cdot \sup_{t\in \mathcal{T}'}\left|\hat{m}_{\theta}^*(t) -\hat{m}_{\theta}(t)\right| \leq u \Big| \mathbb{U}_n\right) - \P\left(\sup_{f\in \mathcal{F}} |\mathbb{B}(f)| \leq u\right) \right| \\
		&\asymp \sup_{u\geq 0} \left|\P\left(\sqrt{nh^3}\cdot \sup_{t\in \mathcal{T}'}\left|\hat{\theta}_C^*(t) -\hat{\theta}_C(t)\right| \leq u \Big| \mathbb{U}_n\right) - \P\left(\sup_{g\in \mathcal{F}_{\theta}} |\bar{\mathbb{B}}(g)| \leq u\right) \right| \\
		&= O_P\left(\left(\frac{\log^5 n}{nh^3}\right)^{\frac{1}{8}} + \left(\frac{\log^2 n}{nb^d\hslash}\right)^{\frac{3}{8}}\right),
	\end{align*}
	where $\hat{m}_{\theta}^*(t)$ and $\hat{\theta}_C^*(t)$ are the integral estimator \eqref{simp_integral} and localized derivative estimator \eqref{theta_C_est} based on a bootstrap sample $\mathbb{U}_n^*=\left\{\left(Y_i^*,T_i^*,\bm{S}_i^*\right)\right\}_{i=1}^n$.
\end{theorem}

The proof of \autoref{thm:bootstrap_cons} is in \autoref{app:boot_consistency}. These results, together with \autoref{thm:gauss_approx_si}, imply the asymptotic validity of bootstrap uniform confidence bands in \autoref{subsec:bootstrap_infer}; see Corollary~\ref{cor:boot_ci} below with its proof in \autoref{app:conf_band}.

\begin{corollary}[Uniform confidence band]
	\label{cor:boot_ci}
	Under the setup of \autoref{thm:bootstrap_cons}, we have that
	\begin{align*}
			&\P\left(\theta(t) \in \left[\hat{\theta}_C(t) - \bar{\xi}_{1-\alpha}^*, \hat{\theta}_C(t) + \bar{\xi}_{1-\alpha}^*\right] \text{ for all } t\in \mathcal{T}'\right) = 1-\alpha + O\left(\left(\frac{\log^5 n}{nh^3}\right)^{\frac{1}{8}} + \left(\frac{\log^2 n}{nb^d\hslash}\right)^{\frac{3}{8}}\right), \\
			&\P\left(m(t) \in \left[\hat{m}_{\theta}(t) - \xi_{1-\alpha}^*, \hat{m}_{\theta}(t) + \xi_{1-\alpha}^*\right] \text{ for all } t\in \mathcal{T}'\right) = 1-\alpha + O\left(\left(\frac{\log^5 n}{nh^3}\right)^{\frac{1}{8}} + \left(\frac{\log^2 n}{nb^d\hslash}\right)^{\frac{3}{8}}\right).
	\end{align*}
\end{corollary}

\section{Experiments}
\label{sec:experiments}

In this section, we evaluate the finite-sample performances of our proposed integral estimator \eqref{simp_integral} and localized derivative estimator \eqref{theta_C_est} through simulation studies. Additionally, we apply these estimators to analyze the causal effects of fine particulate matter on the cardiovascular mortality rate in a case study.

\subsection{Parameter Setup}
\label{subsec:para_setup}

In all experiments, we use the Epanechnikov kernel for $K_T$ and $K_S$ (with the product kernel technique) in the (partial) local quadratic regression \eqref{localpoly3}. To select the default bandwidth parameters $h,b >0$, we implement the rule-of-thumb method outlined in Appendix A of \cite{yang1999multivariate} as:
$$h_{ROT} = C_h\left[\frac{\left(\nu_0^{(T)}\right)^2 \hat{R} \left(T_{(n)} - T_{(1)}\right)}{4\left(\kappa_2^{(T)}\right)^2 \hat{C}_{\mu} \cdot n}\right]^{\frac{1}{5}}, \quad \bm{b}_{ROT} = C_b \left[\frac{d\left(\nu_{0,1}^{(S)}\right)^{2d} \hat{R} \left(\bm{S}_{(n)} - \bm{S}_{(1)}\right)}{4\left(\kappa_{2,1}^{(S)}\right)^2 \hat{C}_{\mu}}\right]^{-\frac{1}{d+5}} n^{-\frac{1}{d+1}},$$
where $\nu_0^{(T)}, \nu_{0,1}^{(S)}, \kappa_2^{(T)}, \kappa_{2,1}^{(S)} >0$ are defined in Assumption~\ref{assump:reg_kernel}(a), $\hat{R} = \frac{1}{n-5}\sum_{i=1}^n \left[Y - \hat{m}(T_i,\bm{S}_i)\right]^2$, with $\hat{m}(t,\bm{s}) = \hat{\beta}_1(t,\bm{s})$ estimated using global (partial) polynomial regression \eqref{localpoly3} with $q=4$ and $K_T=K_S\equiv 1$. The differences $\bm{S}_{(n)} - \bm{S}_{(1)}$ are coordinatewise maxima and minima that approximate the integration ranges, and $\hat{C}_{\mu,T},\hat{C}_{\mu,\bm{S}}$ are the estimated density-weighted curvatures (or second-order partial derivatives) of $\mu$ obtained via coordinatewise global fourth-order polynomials averaged over the observed data $\left\{(T_i,\bm{S}_i),i=1,...,n\right\}$; see also Section 4.3 in \cite{garcia2023notes}. To account for variability in covariate scales, we use a bandwidth vector $\bm{b}_{ROT}\in [0,\infty)^d$ rather than a scalar $b>0$. In addition, unless stated otherwise, we set the scaling constants to be $C_h=10, C_b=15$ for all simulations. Finally, for fair comparisons between our proposed estimators and the naive RA estimators $\hat{m}_{RA}(t)$ in \eqref{m_RA} and $\hat{\theta}_{RA}(t)$ in Remark~\ref{remark:theta_RA}, we estimate $\hat{\mu}(t,\bm{s})$ and $\hat{\beta}_2(t,\bm{s})$ using (partial) local quadratic regression \eqref{localpoly3} with the same bandwidth parameters $h,b$.

For the conditional CDF estimator \eqref{loc_CDF_est}, we leverage the Gaussian kernel for $\bar{K}_T$ and set its default bandwidth parameter $\hslash$ to the normal reference rule in \cite{chacon2011asymptotics,chen2016comprehensive} as $\hslash_{NR} = \left(\frac{4}{3n}\right)^{\frac{1}{5}} \hat{\sigma}_T$, where $\hat{\sigma}_T$ is the sample standard deviation of $\{T_i,i=1,...,n\}$. This rule assumes a normal distribution for $p_T(t)$ and optimizes the asymptotic bias-variance trade-off, which may oversmooth the conditional CDF estimator \citep{sheather2004density}. Nevertheless, since our proposed estimators of $m(t)$ and $\theta(t)$ are relatively insensitive to the choice of $\hslash$, and the order of $\hslash_{NR}$ aligns with our theoretical requirement (see Remark~\ref{remark:bw_sel}), we adopt $\hslash_{NR}$ in the subsequent analysis. For the bootstrap inference, we set a resampling time of $B=1000$ and a significance level of $\alpha = 0.05$, targeting 95\% confidence intervals and uniform bands for inferring the true dose-response curve $m(t)$ and its derivative $\theta(t)$.

\subsection{Simulation Studies} 

We consider three different model settings under the additive confounding model \eqref{add_conf_model} for our simulation studies. For each model setting, we generate i.i.d. observations $(Y_i,T_i,\bm{S}_i),i=1,...,n$.

$\bullet$ {\bf Single Confounder Model:} The data-generating model is already described in \eqref{single_conf}.

\begin{figure}[t]
	\centering
	\includegraphics[width=1\linewidth]{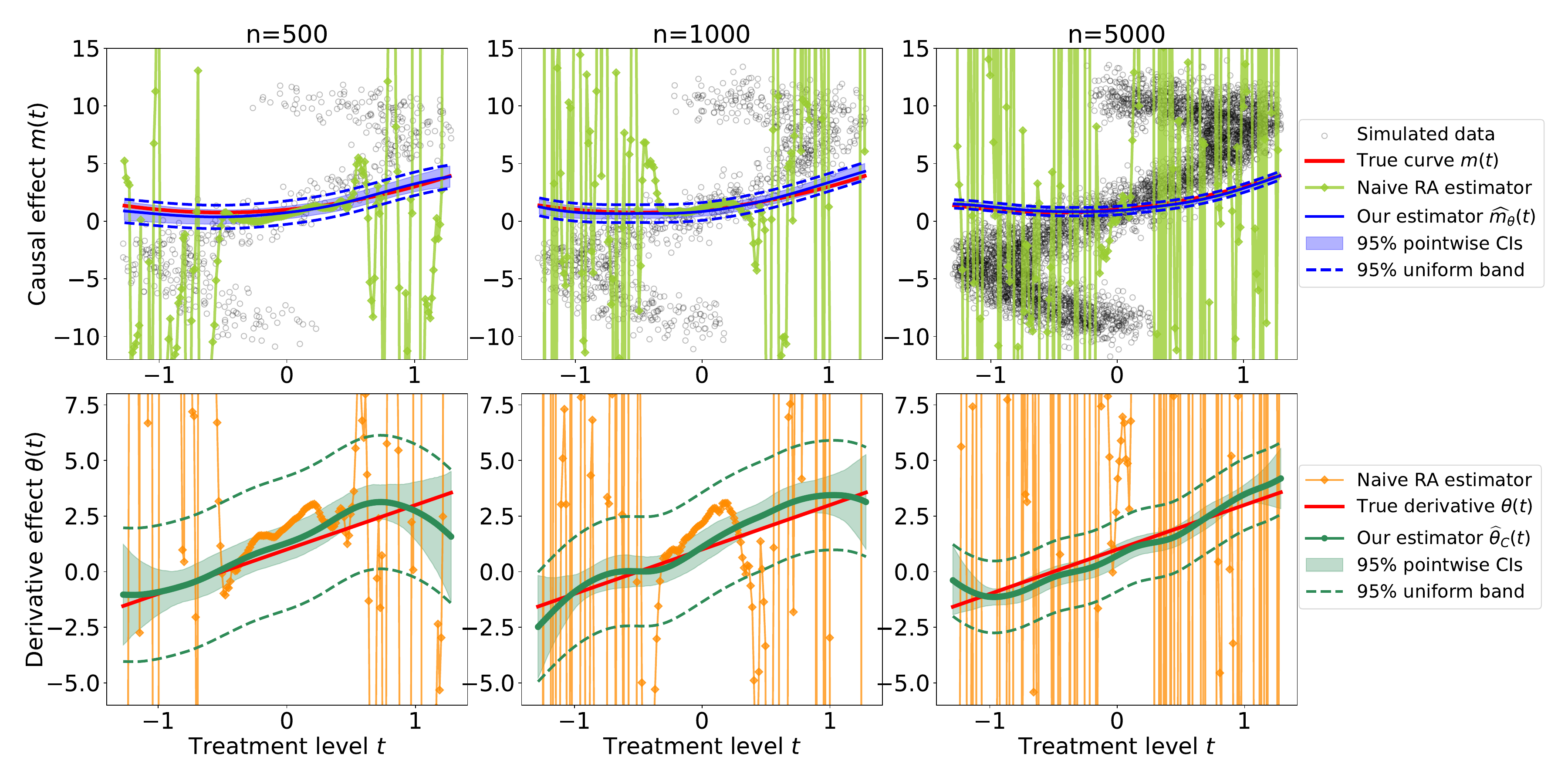}
	\caption{Comparisons between our proposed estimators and the naive RA estimators across various sample sizes under the single confounder model \eqref{single_conf}. Rows present results for estimating $m(t)$ and $\theta(t)$ respectively, while columns correspond to different values for $n$.}
	\label{fig:single_conf_full}
\end{figure}

$\bullet$ {\bf Linear Confounding Model:} We consider the following linear effect model with
\begin{equation}
	\label{linear_conf}
	Y = T+ 6S_1 + 6S_2 + \epsilon, \quad T = 2S_1+S_2+E, \quad \text{and}\quad  \bm{S}=(S_1,S_2) \sim \text{Uniform}[-1,1]^2\subset \mathbb{R}^2,
\end{equation}
where $E\sim \text{Uniform}[-0.5,0.5]$ and $\epsilon\sim \mathcal{N}(0,1)$ are two independent variables. The marginal supports of $T$ and $\bm{S}$ are $\mathcal{T}=[-3.5,3.5]$ and $\mathcal{S}=[-1,1]^2$ respectively, while the support of conditional density $p(t|\bm{s})$ for any $\bm{s}\in \mathcal{S}$ is much narrower than $\mathcal{T}$.

\begin{figure}[t]
	\centering
	\includegraphics[width=1\linewidth]{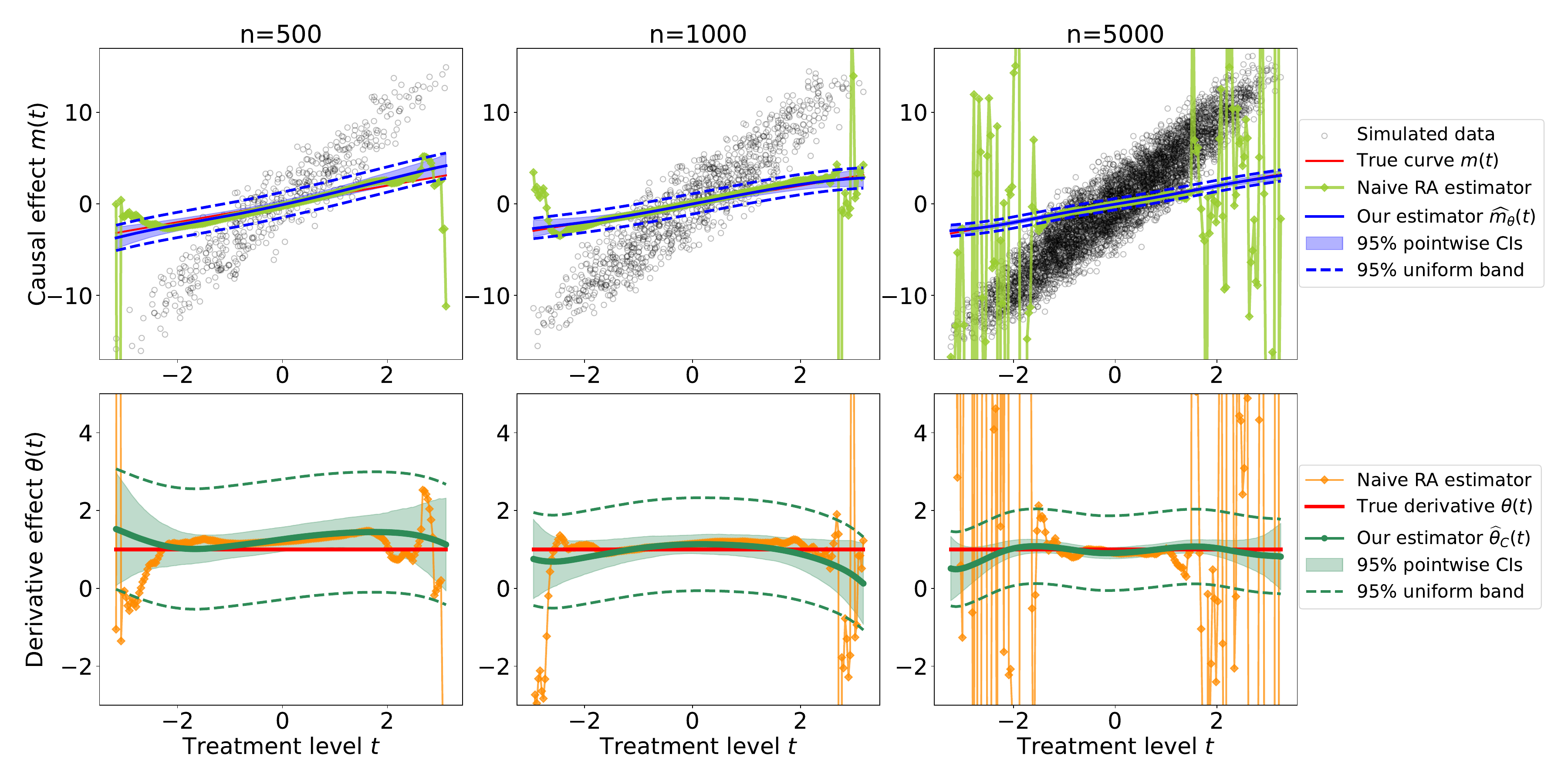}
	\caption{Comparisons between our proposed estimators and the naive RA estimators across various sample sizes under the linear confounding model \eqref{linear_conf}. Rows present results for estimating $m(t)$ and $\theta(t)$ respectively, while columns correspond to different values for $n$.}
	\label{fig:linear_conf}
\end{figure}


$\bullet$ {\bf Nonlinear Confounding Model:} We consider the following nonlinear effect model with
\begin{align}
	\label{nonlinear_conf}
	\begin{split}
		&Y = T^2+T+ 10Z + \epsilon, \quad T = \cos\left(\pi Z^3\right)+ \frac{Z}{4}+E, \quad Z=4S_1+S_2,\\ 
		&\quad \text{ and } \quad \bm{S}=(S_1,S_2) \sim \text{Uniform}[-1,1]^2\subset \mathbb{R}^2,
	\end{split}
\end{align}
where $E\sim \text{Uniform}[-0.1,0.1]$ and $\epsilon\sim \mathcal{N}(0,1)$. Due to the nonlinear confounding effect $\cos(\pi Z)$ and limited treatment variation $E$, the positivity condition \eqref{assump:positivity} is violated on the generated data from \eqref{nonlinear_conf}. Notably, the debiased estimator of $m(t)$ proposed by \cite{takatsu2022debiased} (and similarly, \citealt{kennedy2017non}) fails to recover $m(t)$ under this data model, because some of their pseudo-outcomes have invalid values caused by near-zero estimated conditional densities of $p(t|\bm{s})$.

\begin{figure}[t]
	\centering
	\includegraphics[width=1\linewidth]{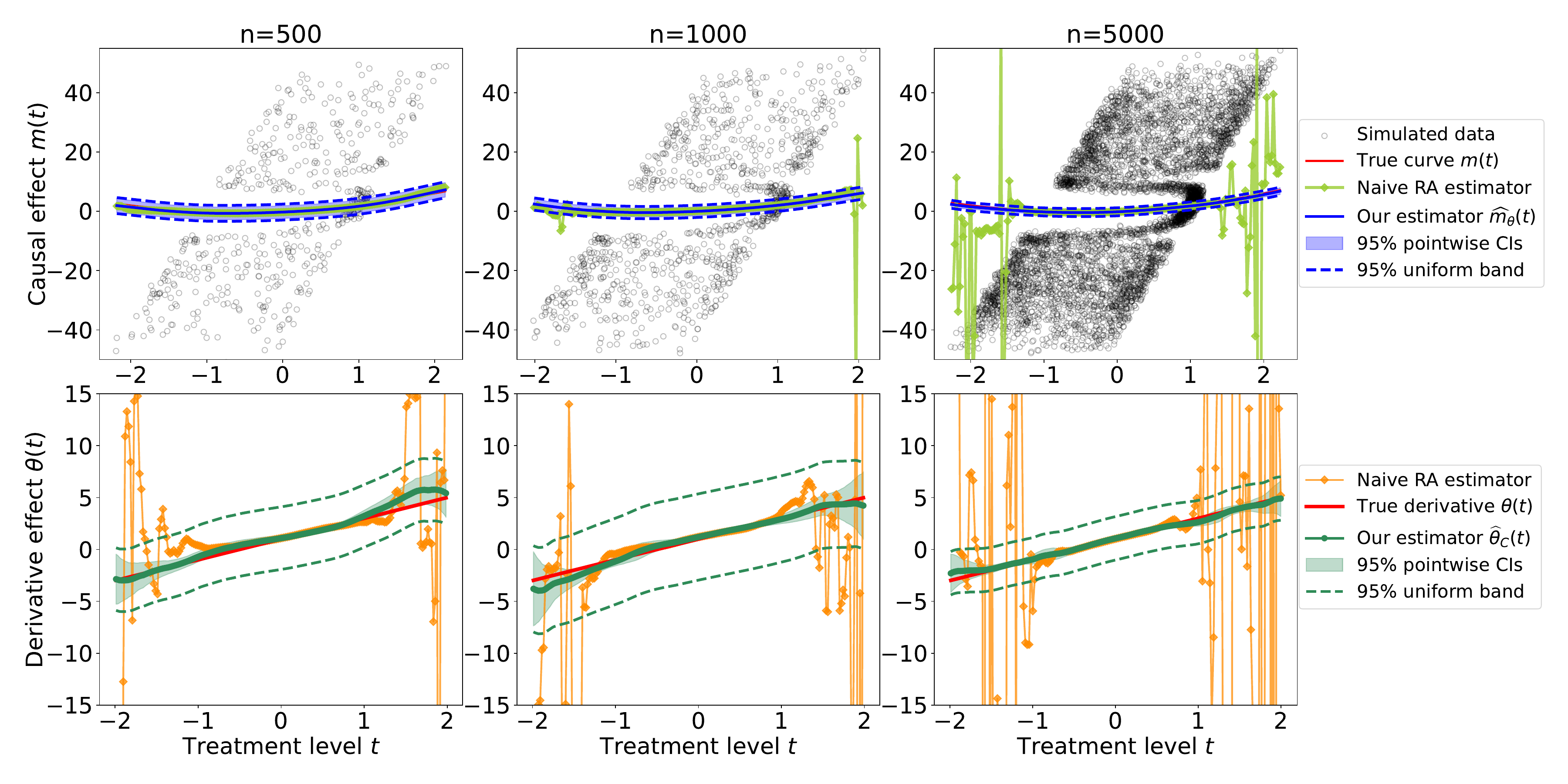}
	\caption{Comparisons between our proposed estimators and the naive RA estimators across various sample sizes under the linear confounding model \eqref{nonlinear_conf}. Rows present results for estimating $m(t)$ and $\theta(t)$ respectively, while columns correspond to different values for $n$.}
	\label{fig:nonlin_conf}
\end{figure}


In the presence of confounding variables, regressing $Y_i$'s to $T_i$'s in the generated data from each model only leads to biased estimates of $m(t)$. We apply our proposed integral estimator $\hat{m}_{\theta}(t)$ and localized derivative estimator $\hat{\theta}_C(t)$ with nonparametric bootstrap inference in \autoref{subsec:bootstrap_infer} to the data generated by each model above with varying sample sizes $n$. We also compare our proposed estimators to the naive RA estimators in \eqref{m_RA} and Remark~\ref{remark:theta_RA}. The results are presented in \autoref{fig:single_conf_full}, \autoref{fig:linear_conf}, and \autoref{fig:nonlin_conf}. As sample size increases, the violations of positivity become more pronounced, causing naive RA estimators to exhibit greater instability, especially near the boundary of the marginal support $\mathcal{T}$. Because of the positivity violations, fine-tuning the bandwidth parameters $h,b$ does not improve the stability or consistency of the naive RA estimators. In contrast, our proposed estimators approach the true dose-response curves and its derivatives as sample size increases, with narrower pointwise confidence intervals and uniform bands. These results demonstrate the consistency of our proposed estimators under violations of the positivity condition \eqref{assump:positivity}.

\subsection{Case Study: Effect of PM$_{2.5}$ on Cardiovascular Mortality Rate}

Air pollution, especially fine particulate matter with diameters less than 2.5 $\si{\um}$ (PM$_{2.5}$), is known to be associated with increased cardiovascular disease risk \citep{brook2010particulate,krittanawong2023pm2}. Recent studies have also demonstrated a positive association between PM$_{2.5}$ levels ($\si{\ug}/m^3$) and county-level cardiovascular mortality rates (CMR; deaths/100,000 person-years) in the United States, with socioeconomic factors being well-controlled \citep{wyatt2020contribution}. 

To showcase the applicability of our proposed integral and localized derivative estimators, we apply them to analyzing the PM$_{2.5}$-CMR data from \cite{wyatt2020annual}, which contain average annual CMRs (outcome variable $Y$) and PM$_{2.5}$ concentrations (exposure variable $T$) across $n=2132$ counties from 1990 to 2010. These data were sourced from the US National Center for Health Statistics and Community Multiscale Air Quality modeling system, respectively. The covariate vector $\bm{S}\in \mathbb{R}^{10}$ consists of two parts. The first part includes two spatial confounding variables, latitude and longitude of each county, to capture spatial dependence. The second part comprises eight county-level socioeconomic factors from the U.S. census: population in 2000, civilian unemployment rate, median household income, percentage of female households without spouses, percentage of vacant housing units, percentage of owner occupied housing units, percentage of high school graduates or above, and percentage of households in poverty. For each county in a given year, we assign the values of socioeconomic variables based on the nearest U.S. Census year (1990, 2000, or 2010) and compute the averages of $Y,T,\bm{S}$ for each county over the 21 years to construct the final dataset. Following \cite{takatsu2022debiased}, we focus on PM$_{2.5}$ levels between 2.5 $\si{\ug}/m^3$ and 11.5 $\si{\ug}/m^3$ to minimize boundary effects.

We apply our proposed integral estimator \eqref{simp_integral} and localized derivative estimator \eqref{theta_C_est} to the final data with bandwidth parameters specified in \autoref{subsec:para_setup}. To account for scale differences between treatment $T$ and covariates in $\bm{S}$, we multiply the (coordinatewise) standard deviations of $T$ and $\bm{S}$ on the data to $h_{ROT},\bm{b}_{ROT}$ respectively and set $C_h=16, C_b=23$ for smoothing the estimated derivatives. To examine confounding effects, we compare three modeling settings: 
\begin{enumerate}[label=(\roman*)]
	\item regress $Y$ on $T$ alone via local quadratic regression estimators;
	
	\item regress $Y$ on $T$ with spatial locations;
	
	\item regress $Y$ on $T$ with both spatial and socioeconomic covariates.
\end{enumerate}
Setting (i) ignores confounding effects. Setting (ii) assumes that spatial locations fully account for confounding effects, and our proposed estimators are valid if, for example, no interactions exist between PM$_{2.5}$ levels ($T$) and spatial locations ($\bm{S}$), as in model \eqref{add_conf_model}. Setting (iii) incorporates additional socioeconomic factors as confounding variables, with our estimators remaining valid if interaction effects between PM$_{2.5}$ levels ($T$) and all confounding variables ($\bm{S}_1,\bm{S}_2$) are present but can be resolved by conditioning on a subset of key confounding variables ($\bm{S}_1$), as in model \eqref{add_conf_general}.  

Results in the left panel of \autoref{fig:PM25_app} show that without adjusting for any confounding variables, the fitted curve is non-monotonic, peaking at around 8 $\si{\ug}/m^3$. After controlling for spatial locations, the relationship flatten but remains non-monotonic when the PM$_{2.5}$ level is above 8 $\si{\ug}/m^3$. Only when we incorporate both spatial and socioeconomic covariates can the estimated curve between PM$_{2.5}$ and CMR become monotonically increasing. The estimated changing rate of CMR with respect to PM$_{2.5}$ (\emph{i.e.}, estimated derivative curve) in the right panel of \autoref{fig:PM25_app} confirms this trend. Notably, after adjusting for all confounding variables, the changing rate of CMR and its uniform 95\% confidence band stabilize and remain consistently above 0 when PM$_{2.5}$ levels are below 8 $\si{\ug}/m^3$, indicating a statistically significant positive association between PM$_{2.5}$ and CMR. However, for PM$_{2.5}$ levels above 10 $\si{\ug}/m^3$, the association becomes less clear, potentially due to competing risks \citep{leiser2019evaluation}. Our findings align with those of \cite{wyatt2020contribution,takatsu2022debiased}, while offering additional insights. Specifically, we estimate not only the relationship between PM$_{2.5}$ and CMR but also its derivative. Furthermore, our proposed estimators effectively address potential violations of the positivity condition \eqref{assump:positivity}, yielding more robust conclusions in this observational study.

\begin{figure}
	\captionsetup[subfigure]{justification=centering}
	\begin{subfigure}[t]{0.49\linewidth}
		\centering		\includegraphics[width=1\linewidth]{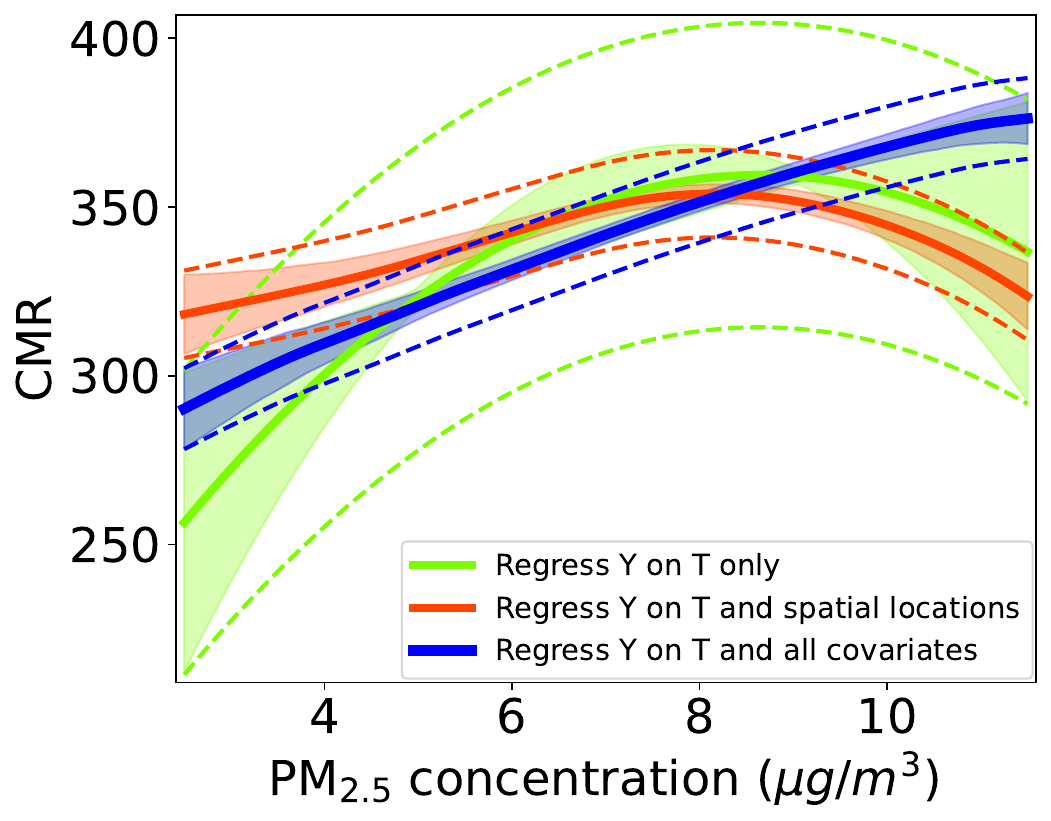}
	\end{subfigure}
	\hfil
	\begin{subfigure}[t]{0.49\linewidth}
		\centering		\includegraphics[width=1\linewidth]{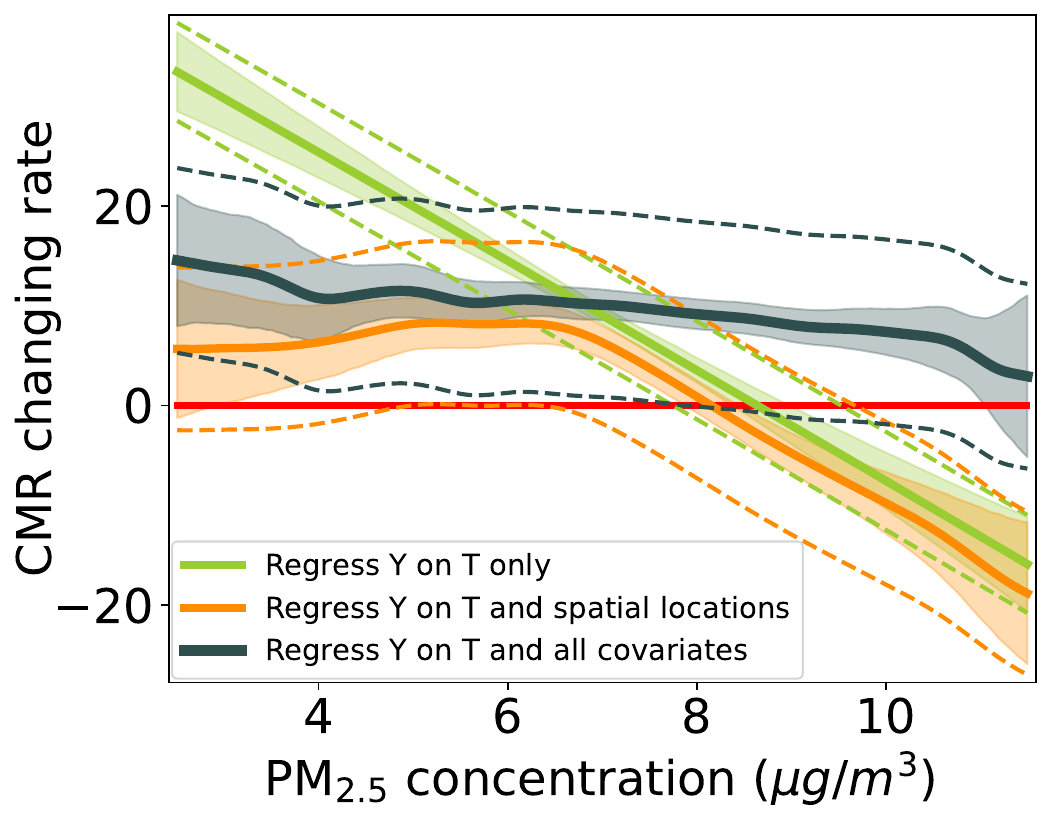}
	\end{subfigure}
	\caption{Estimated relationships between the PM$_{2.5}$ concentration and CMR or its changing rate at the county level. {\bf Left:} The estimated CMR with respect to the PM$_{2.5}$ concentration. {\bf Right:} The estimated changing rates of CMR with respect to the PM$_{2.5}$ concentration. We also present the 95\% confidence intervals and uniform confidence bands as shaded regions and dashed lines respectively for each regression scenario.}
	\label{fig:PM25_app}
\end{figure}

\section{Discussion}
\label{sec:discussion}

In summary, this paper studies nonparametric identification and inference on the dose-response curve and its derivative via innovative integral and localized strategies. Our proposed estimators consistently recover and infer these curves without assuming the overly restrictive positivity condition \eqref{assump:positivity} for continuous treatments. We establish the consistency and bootstrap validity of the proposed estimators when nuisance functions are estimated by kernel smoothing methods, without requiring explicit undersmoothing of bandwidth parameters. Simulations and real-world applications demonstrate the effectiveness of our approach in addressing the limitation of \eqref{assump:positivity}. There are several future directions that can further advance the impacts of our proposal.

{\bf 1. Estimation of Nuisance Parameters/Functions:} The proposed integral estimator \eqref{simp_integral} and localized derivative estimator \eqref{theta_C_est} require estimating two nuisance functions: the partial derivative $\frac{\partial}{\partial t} \mu(t,\bm{s})$ and the conditional CDF $P(\bm{s}|t)$. When we focus on the kernel smoothing methods, the finite-sample performances of these estimators depend on the selection of three bandwidth parameters $h,b,\hslash>0$. While this paper employs rule-of-thumb and normal reference bandwidth selection methods, exploring other bandwidth choices through the plug-in rule \citep{ruppert1995effective} or cross-validation \citep{li2004cross} remains an interesting area for future research. More broadly, alternative approaches for estimating the nuisance functions, such as regression splines for $\frac{\partial}{\partial t} \mu(t,\bm{s})$ \citep{friedman1991multivariate,zhou2000derivative} or nearest neighbor-type and local logistic methods for $P(\bm{s}|t)$ \citep{stute1986conditional,hall1999methods}, could further enhance the empirical performances of our proposed estimators $\hat{\theta}_C(t)$ and $\hat{m}_{\theta}(t)$.

{\bf 2. IPW and Doubly Robust Estimation:} Our proposed estimators \eqref{simp_integral} and \eqref{theta_C_est} are based on the idea of regression adjustment. An interesting avenue for future research would be exploring how our integral and localized techniques could address the positivity requirement in existing generalized propensity score-based approaches \citep{hirano2004propensity,imai2004causal} and doubly robust methods \citep{kennedy2017non,westling2020causal,colangelo2020double,semenova2021debiased,bonvini2022fast} for estimating the dose-response curve and its derivative.

{\bf 3. Additive Model Diagnostics:} Assumption~\ref{assump:diff_inter} is vital to the consistency of our proposed estimators, though its validity beyond the additive confounding model \eqref{add_conf_model} remains unclear. Nevertheless, our estimators can be utilized to test the additive structure in \eqref{add_conf_model}. Under the positivity condition \eqref{assump:positivity}, this can be achieved by estimating the absolute difference $\left|\mathbb{E}\left[\frac{\partial}{\partial t} \mu(t,\bm{S})\right] - \mathbb{E}\left[\frac{\partial}{\partial t} \mu(t,\bm{S}) \big| T=t\right] \right|$ from 0 using the RA estimator $\hat{\theta}_{RA}(t)$ in Remark~\ref{remark:theta_RA} and our estimator $\hat{\theta}_C(t)$ in \eqref{theta_C_est}. Other testing procedures, such as the marginal integration regression method \citep{linton1995kernel}, are also applicable. When the positivity condition \eqref{assump:positivity} is violated, the partial derivative $\beta_2(t,\bm{s}) = \frac{\partial}{\partial t} \mu(t,\bm{s})$ remains independent of $\bm{s}$ for any $t\in \mathcal{T}$ under \eqref{add_conf_model}. Consequently, our estimator $\hat{\beta}_2(t,\bm{s})$ in \eqref{localpoly3} should asymptotically not depend on $\bm{s}$. This suggests that statistical quantities, such as $\Theta(t) = \sup_{\bm{s}\in \mathcal{S}(t)} \left|\hat{\beta}_2(t,\bm{s}) - \hat{\theta}_C(t)\right|$ and $\Theta = \int_{\mathcal{T}}  \Theta(t) \, dt$, should converge to 0 as $n\to\infty$. Thus, by deriving the limiting distribution of $\Theta(t)$ or $\Theta$, one can develop procedures to statistically test the additive confounding model \eqref{add_conf_model} without relying on \eqref{assump:positivity}. In contrast, existing tests for additivity \citep{eubank1995testing,gozalo2001testing} generally require the positivity condition \eqref{assump:positivity} and may not be directly applicable here.

{\bf 4. Violation of Ignorability (Assumption~\ref{assump:identify_cond}(b)):} In observational studies, unmeasured confounding variables beyond the covariate vector $\bm{S}\in \mathcal{S}\subset \mathbb{R}^d$ can bias inference \citep{vanderweele2008sign}. Hence, it is crucial to analyze the sensitivity of the dose-response curve and its derivative to violations of ignorability. In this context, approaches such as the instrumental variable model \citep{kilbertus2020class} and the Riesz-Frechet representation technique \citep{chernozhukov2022long} could be useful for the analysis. Alternatively, one could incorporate an additional high-dimensional covariate vector $\bm{Z}\in \mathbb{R}^{d'}$ into model \eqref{add_conf_model} to mitigate unmeasured confounding variables \citep{guo2019decorrelated}. We will leave a rigorous investigation of these directions as future work.

\section*{Acknowledgement}

We thank Alex Luedtke, Andrea Rotnitzky, Marco Carone, Pawel Morzywolek, Zhichao Jiang, and Daniel Suen for their insightful comments on the earlier version of this paper. YZ is supported in part by YC's NSF grant DMS-2141808. YC is supported by NSF grants DMS-1952781, 2112907, 2141808, and NIH U24-AG07212.

\vspace{5mm}
\singlespacing
\bibliography{Nonpar_Conf}

\begin{thebibliography}{110}
\providecommand{\natexlab}[1]{#1}
\providecommand{\url}[1]{\texttt{#1}}
\expandafter\ifx\csname urlstyle\endcsname\relax
  \providecommand{\doi}[1]{doi: #1}\else
  \providecommand{\doi}{doi: \begingroup \urlstyle{rm}\Url}\fi

\bibitem[Arcones and Gine(1992)]{arcones1992bootstrap}
M.~A. Arcones and E.~Gine.
\newblock {On the Bootstrap of $U$ and $V$ Statistics}.
\newblock \emph{The Annals of Statistics}, 20\penalty0 (2):\penalty0 655 --
  674, 1992.

\bibitem[Bahadur(1966)]{bahadur1966note}
R.~R. Bahadur.
\newblock A note on quantiles in large samples.
\newblock \emph{The Annals of Mathematical Statistics}, 37\penalty0
  (3):\penalty0 577--580, 1966.

\bibitem[Bashtannyk and Hyndman(2001)]{bashtannyk2001bandwidth}
D.~M. Bashtannyk and R.~J. Hyndman.
\newblock Bandwidth selection for kernel conditional density estimation.
\newblock \emph{Computational Statistics \& Data Analysis}, 36\penalty0
  (3):\penalty0 279--298, 2001.

\bibitem[Bonvini and Kennedy(2022)]{bonvini2022fast}
M.~Bonvini and E.~H. Kennedy.
\newblock Fast convergence rates for dose-response estimation.
\newblock \emph{arXiv preprint arXiv:2207.11825}, 2022.

\bibitem[Branson et~al.(2023)Branson, Kennedy, Balakrishnan, and
  Wasserman]{branson2023causal}
Z.~Branson, E.~H. Kennedy, S.~Balakrishnan, and L.~Wasserman.
\newblock Causal effect estimation after propensity score trimming with
  continuous treatments.
\newblock \emph{arXiv preprint arXiv:2309.00706}, 2023.

\bibitem[Brook et~al.(2010)Brook, Rajagopalan, Pope, Brook, Bhatnagar,
  Diez-Roux, Holguin, Hong, Luepker, Mittleman, Peters, Siscovick, Smith,
  Whitsel, and Kaufman]{brook2010particulate}
R.~D. Brook, S.~Rajagopalan, C.~A. Pope, J.~R. Brook, A.~Bhatnagar, A.~V.
  Diez-Roux, F.~Holguin, Y.~Hong, R.~V. Luepker, M.~A. Mittleman, A.~Peters,
  D.~Siscovick, S.~C. Smith, L.~Whitsel, and J.~D. Kaufman.
\newblock Particulate matter air pollution and cardiovascular disease: an
  update to the scientific statement from the american heart association.
\newblock \emph{Circulation}, 121\penalty0 (21):\penalty0 2331--2378, 2010.

\bibitem[Buja et~al.(1989)Buja, Hastie, and Tibshirani]{buja1989linear}
A.~Buja, T.~Hastie, and R.~Tibshirani.
\newblock Linear smoothers and additive models.
\newblock \emph{The Annals of Statistics}, 17\penalty0 (2):\penalty0 453 --
  510, 1989.

\bibitem[Busso et~al.(2014)Busso, DiNardo, and McCrary]{busso2014new}
M.~Busso, J.~DiNardo, and J.~McCrary.
\newblock New evidence on the finite sample properties of propensity score
  reweighting and matching estimators.
\newblock \emph{Review of Economics and Statistics}, 96\penalty0 (5):\penalty0
  885--897, 2014.

\bibitem[Calonico et~al.(2018)Calonico, Cattaneo, and
  Farrell]{calonico2018effect}
S.~Calonico, M.~D. Cattaneo, and M.~H. Farrell.
\newblock On the effect of bias estimation on coverage accuracy in
  nonparametric inference.
\newblock \emph{Journal of the American Statistical Association}, 113\penalty0
  (522):\penalty0 767--779, 2018.

\bibitem[Cattaneo et~al.(2010)Cattaneo, Crump, and Jansson]{cattaneo2010robust}
M.~D. Cattaneo, R.~K. Crump, and M.~Jansson.
\newblock Robust data-driven inference for density-weighted average
  derivatives.
\newblock \emph{Journal of the American Statistical Association}, 105\penalty0
  (491):\penalty0 1070--1083, 2010.

\bibitem[Chac{\'o}n et~al.(2011)Chac{\'o}n, Duong, and
  Wand]{chacon2011asymptotics}
J.~E. Chac{\'o}n, T.~Duong, and M.~Wand.
\newblock Asymptotics for general multivariate kernel density derivative
  estimators.
\newblock \emph{Statistica Sinica}, pages 807--840, 2011.

\bibitem[Chamberlain(1986)]{chamberlain1986asymptotic}
G.~Chamberlain.
\newblock Asymptotic efficiency in semi-parametric models with censoring.
\newblock \emph{Journal of Econometrics}, 32\penalty0 (2):\penalty0 189--218,
  1986.

\bibitem[Chen et~al.(2015)Chen, Genovese, and Wasserman]{chen2015asymptotic}
Y.-C. Chen, C.~R. Genovese, and L.~Wasserman.
\newblock Asymptotic theory for density ridges.
\newblock \emph{The Annals of Statistics}, 43\penalty0 (5):\penalty0
  1896--1928, 2015.

\bibitem[Chen et~al.(2016)Chen, Genovese, and Wasserman]{chen2016comprehensive}
Y.-C. Chen, C.~R. Genovese, and L.~Wasserman.
\newblock {A comprehensive approach to mode clustering}.
\newblock \emph{Electronic Journal of Statistics}, 10\penalty0 (1):\penalty0
  210 -- 241, 2016.

\bibitem[Chen et~al.(2017)Chen, Genovese, and Wasserman]{chen2017density}
Y.-C. Chen, C.~R. Genovese, and L.~Wasserman.
\newblock Density level sets: Asymptotics, inference, and visualization.
\newblock \emph{Journal of the American Statistical Association}, 112\penalty0
  (520):\penalty0 1684--1696, 2017.

\bibitem[Cheng and Chen(2019)]{cheng2019nonparametric}
G.~Cheng and Y.-C. Chen.
\newblock {Nonparametric inference via bootstrapping the debiased estimator}.
\newblock \emph{Electronic Journal of Statistics}, 13\penalty0 (1):\penalty0
  2194 -- 2256, 2019.

\bibitem[Chernozhukov et~al.(2014)Chernozhukov, Chetverikov, and
  Kato]{chernozhukov2014gaussian}
V.~Chernozhukov, D.~Chetverikov, and K.~Kato.
\newblock Gaussian approximation of suprema of empirical processes.
\newblock \emph{The Annals of Statistics}, 42\penalty0 (4):\penalty0
  1564--1597, 2014.

\bibitem[Chernozhukov et~al.(2022)Chernozhukov, Cinelli, Newey, Sharma, and
  Syrgkanis]{chernozhukov2022long}
V.~Chernozhukov, C.~Cinelli, W.~Newey, A.~Sharma, and V.~Syrgkanis.
\newblock Long story short: Omitted variable bias in causal machine learning.
\newblock Technical report, National Bureau of Economic Research, 2022.

\bibitem[Colangelo and Lee(2020)]{colangelo2020double}
K.~Colangelo and Y.-Y. Lee.
\newblock Double debiased machine learning nonparametric inference with
  continuous treatments.
\newblock \emph{arXiv preprint arXiv:2004.03036}, 2020.

\bibitem[Cole and Hern{\'a}n(2008)]{cole2008constructing}
S.~R. Cole and M.~A. Hern{\'a}n.
\newblock Constructing inverse probability weights for marginal structural
  models.
\newblock \emph{American Journal of Epidemiology}, 168\penalty0 (6):\penalty0
  656--664, 2008.

\bibitem[Cox(1958)]{cox1958planning}
D.~R. Cox.
\newblock \emph{Planning of Experiments}.
\newblock Wiley, 1958.

\bibitem[Crump et~al.(2009)Crump, Hotz, Imbens, and Mitnik]{crump2009dealing}
R.~K. Crump, V.~J. Hotz, G.~W. Imbens, and O.~A. Mitnik.
\newblock Dealing with limited overlap in estimation of average treatment
  effects.
\newblock \emph{Biometrika}, 96\penalty0 (1):\penalty0 187--199, 2009.

\bibitem[Dehejia and Wahba(1999)]{dehejia1999causal}
R.~H. Dehejia and S.~Wahba.
\newblock Causal effects in nonexperimental studies: Reevaluating the
  evaluation of training programs.
\newblock \emph{Journal of the American statistical Association}, 94\penalty0
  (448):\penalty0 1053--1062, 1999.

\bibitem[D{\'\i}az and van~der Laan(2013)]{diaz2013targeted}
I.~D{\'\i}az and M.~J. van~der Laan.
\newblock Targeted data adaptive estimation of the causal dose--response curve.
\newblock \emph{Journal of Causal Inference}, 1\penalty0 (2):\penalty0
  171--192, 2013.

\bibitem[Dvoretzky et~al.(1956)Dvoretzky, Kiefer, and
  Wolfowitz]{dvoretzky1956asymptotic}
A.~Dvoretzky, J.~Kiefer, and J.~Wolfowitz.
\newblock {Asymptotic Minimax Character of the Sample Distribution Function and
  of the Classical Multinomial Estimator}.
\newblock \emph{The Annals of Mathematical Statistics}, 27\penalty0
  (3):\penalty0 642 -- 669, 1956.

\bibitem[D’Amour et~al.(2021)D’Amour, Ding, Feller, Lei, and
  Sekhon]{d2021overlap}
A.~D’Amour, P.~Ding, A.~Feller, L.~Lei, and J.~Sekhon.
\newblock Overlap in observational studies with high-dimensional covariates.
\newblock \emph{Journal of Econometrics}, 221\penalty0 (2):\penalty0 644--654,
  2021.

\bibitem[Efron(1979)]{efron1979bootstrap}
B.~Efron.
\newblock Bootstrap methods: Another look at the jackknife.
\newblock \emph{The Annals of Statistics}, 7\penalty0 (1):\penalty0 1 -- 26,
  1979.

\bibitem[Einmahl and Mason(2005)]{einmahl2005uniform}
U.~Einmahl and D.~M. Mason.
\newblock {Uniform in bandwidth consistency of kernel-type function
  estimators}.
\newblock \emph{The Annals of Statistics}, 33\penalty0 (3):\penalty0 1380 --
  1403, 2005.

\bibitem[Eubank et~al.(1995)Eubank, Hart, Simpson, and
  Stefanski]{eubank1995testing}
R.~Eubank, J.~D. Hart, D.~G. Simpson, and L.~A. Stefanski.
\newblock Testing for additivity in nonparametric regression.
\newblock \emph{The Annals of Statistics}, 23\penalty0 (6):\penalty0
  1896--1920, 1995.

\bibitem[Fan and Gijbels(1996)]{fan1996local}
J.~Fan and I.~Gijbels.
\newblock \emph{Local polynomial modelling and its applications}, volume~66.
\newblock Chapman \& Hall/CRC, 1996.

\bibitem[Fan et~al.(1996)Fan, Gijbels, Hu, and Huang]{fan1996study}
J.~Fan, I.~Gijbels, T.-C. Hu, and L.-S. Huang.
\newblock A study of variable bandwidth selection for local polynomial
  regression.
\newblock \emph{Statistica Sinica}, 6:\penalty0 113--127, 1996.

\bibitem[Fan et~al.(1998)Fan, H{\"a}rdle, and Mammen]{fan1998direct}
J.~Fan, W.~H{\"a}rdle, and E.~Mammen.
\newblock Direct estimation of low-dimensional components in additive models.
\newblock \emph{The Annals of Statistics}, 26\penalty0 (3):\penalty0 943--971,
  1998.

\bibitem[Fan and Guerre(2015)]{fan2015multivariate}
Y.~Fan and E.~Guerre.
\newblock Multivariate local polynomial estimators: Boundary properties and
  uniform asymptotic linear representation.
\newblock \emph{Advances in Econometrics}, 2015.

\bibitem[Flores(2007)]{flores2007estimation}
C.~Flores.
\newblock Estimation of dose-response functions and optimal doses with a
  continuous treatment.
\newblock Technical report, Department of Economics, University of Miami, 2007.
\newblock URL \url{https://core.ac.uk/download/pdf/7169663.pdf}.

\bibitem[Friedman(1991)]{friedman1991multivariate}
J.~H. Friedman.
\newblock Multivariate adaptive regression splines.
\newblock \emph{The Annals of Statistics}, 19\penalty0 (1):\penalty0 1--67,
  1991.

\bibitem[Garc{\'\i}a~Portugu{\'e}s(2023)]{garcia2023notes}
E.~Garc{\'\i}a~Portugu{\'e}s.
\newblock Notes for nonparametric statistics, 2023.
\newblock URL \url{https://bookdown.org/egarpor/NP-UC3M/}.
\newblock (Accessed on March 26, 2024).

\bibitem[Giessing(2023)]{giessing2023gaussian}
A.~Giessing.
\newblock Gaussian and bootstrap approximations for suprema of empirical
  processes.
\newblock \emph{arXiv preprint arXiv:2309.01307}, 2023.

\bibitem[Gilbert et~al.(2023)Gilbert, Datta, Casey, and
  Ogburn]{gilbert2023causal}
B.~Gilbert, A.~Datta, J.~A. Casey, and E.~L. Ogburn.
\newblock A causal inference framework for spatial confounding.
\newblock \emph{arXiv preprint arXiv:2112.14946}, 2023.

\bibitem[Gill and Robins(2001)]{gill2001causal}
R.~D. Gill and J.~M. Robins.
\newblock Causal inference for complex longitudinal data: the continuous case.
\newblock \emph{Annals of Statistics}, 29\penalty0 (6):\penalty0 1785--1811,
  2001.

\bibitem[Gin{\'e} and Guillou(2002)]{gine2002rates}
E.~Gin{\'e} and A.~Guillou.
\newblock Rates of strong uniform consistency for multivariate kernel density
  estimators.
\newblock \emph{Annales de l'Institut Henri Poincare (B) Probability and
  Statistics}, 38\penalty0 (6):\penalty0 907--921, 2002.

\bibitem[Gozalo and Linton(2001)]{gozalo2001testing}
P.~L. Gozalo and O.~B. Linton.
\newblock Testing additivity in generalized nonparametric regression models
  with estimated parameters.
\newblock \emph{Journal of Econometrics}, 104\penalty0 (1):\penalty0 1--48,
  2001.

\bibitem[Guo et~al.(2019)Guo, Yuan, and Zhang]{guo2019decorrelated}
Z.~Guo, W.~Yuan, and C.-H. Zhang.
\newblock Decorrelated local linear estimator: Inference for non-linear effects
  in high-dimensional additive models.
\newblock \emph{arXiv preprint arXiv:1907.12732}, 2019.

\bibitem[Hall et~al.(1999)Hall, Wolff, and Yao]{hall1999methods}
P.~Hall, R.~C. Wolff, and Q.~Yao.
\newblock Methods for estimating a conditional distribution function.
\newblock \emph{Journal of the American Statistical Association}, 94\penalty0
  (445):\penalty0 154--163, 1999.

\bibitem[H{\"a}rdle and Stoker(1989)]{hardle1989investigating}
W.~H{\"a}rdle and T.~M. Stoker.
\newblock Investigating smooth multiple regression by the method of average
  derivatives.
\newblock \emph{Journal of the American statistical Association}, 84\penalty0
  (408):\penalty0 986--995, 1989.

\bibitem[Hastie and Loader(1993)]{hastie1993local}
T.~Hastie and C.~Loader.
\newblock Local regression: Automatic kernel carpentry.
\newblock \emph{Statistical Science}, pages 120--129, 1993.

\bibitem[Hines et~al.(2023)Hines, Diaz-Ordaz, and
  Vansteelandt]{hines2023optimally}
O.~Hines, K.~Diaz-Ordaz, and S.~Vansteelandt.
\newblock Optimally weighted average derivative effects.
\newblock \emph{arXiv preprint arXiv:2308.05456}, 2023.

\bibitem[Hirano and Imbens(2004)]{hirano2004propensity}
K.~Hirano and G.~W. Imbens.
\newblock \emph{The Propensity Score with Continuous Treatments}, chapter~7,
  pages 73--84.
\newblock John Wiley \& Sons, Ltd, 2004.

\bibitem[Hirshberg and Wager(2020)]{hirshberg2020debiased}
D.~A. Hirshberg and S.~Wager.
\newblock Debiased inference of average partial effects in single-index models:
  Comment on wooldridge and zhu.
\newblock \emph{Journal of Business \& Economic Statistics}, 38\penalty0
  (1):\penalty0 19--24, 2020.

\bibitem[Holmes et~al.(2012)Holmes, Gray, and Isbell]{holmes2012fast}
M.~P. Holmes, A.~G. Gray, and C.~L. Isbell.
\newblock Fast nonparametric conditional density estimation.
\newblock \emph{arXiv preprint arXiv:1206.5278}, 2012.

\bibitem[Imai and van Dyk(2004)]{imai2004causal}
K.~Imai and D.~A. van Dyk.
\newblock Causal inference with general treatment regimes: Generalizing the
  propensity score.
\newblock \emph{Journal of the American Statistical Association}, 99\penalty0
  (467):\penalty0 854--866, 2004.

\bibitem[Kallus and Zhou(2018)]{kallus2018policy}
N.~Kallus and A.~Zhou.
\newblock Policy evaluation and optimization with continuous treatments.
\newblock In \emph{International Conference on Artificial Intelligence and
  Statistics}, pages 1243--1251. PMLR, 2018.

\bibitem[Kammann and Wand(2003)]{kammann2003geoadditive}
E.~Kammann and M.~P. Wand.
\newblock Geoadditive models.
\newblock \emph{Journal of the Royal Statistical Society Series C: Applied
  Statistics}, 52\penalty0 (1):\penalty0 1--18, 2003.

\bibitem[Keller and Szpiro(2020)]{keller2020selecting}
J.~P. Keller and A.~A. Szpiro.
\newblock Selecting a scale for spatial confounding adjustment.
\newblock \emph{Journal of the Royal Statistical Society Series A: Statistics
  in Society}, 183\penalty0 (3):\penalty0 1121--1143, 2020.

\bibitem[Kennedy(2019)]{kennedy2019nonparametric}
E.~H. Kennedy.
\newblock Nonparametric causal effects based on incremental propensity score
  interventions.
\newblock \emph{Journal of the American Statistical Association}, 114\penalty0
  (526):\penalty0 645--656, 2019.

\bibitem[Kennedy et~al.(2017)Kennedy, Ma, McHugh, and Small]{kennedy2017non}
E.~H. Kennedy, Z.~Ma, M.~D. McHugh, and D.~S. Small.
\newblock Nonparametric methods for doubly robust estimation of continuous
  treatment effects.
\newblock \emph{Journal of the Royal Statistical Society Series B: Statistical
  Methodology}, 79\penalty0 (4):\penalty0 1229--1245, 2017.

\bibitem[Khan and Tamer(2010)]{khan2010irregular}
S.~Khan and E.~Tamer.
\newblock Irregular identification, support conditions, and inverse weight
  estimation.
\newblock \emph{Econometrica}, 78\penalty0 (6):\penalty0 2021--2042, 2010.

\bibitem[Kilbertus et~al.(2020)Kilbertus, Kusner, and
  Silva]{kilbertus2020class}
N.~Kilbertus, M.~J. Kusner, and R.~Silva.
\newblock A class of algorithms for general instrumental variable models.
\newblock \emph{Advances in Neural Information Processing Systems},
  33:\penalty0 20108--20119, 2020.

\bibitem[Kong et~al.(2010)Kong, Linton, and Xia]{kong2010uniform}
E.~Kong, O.~Linton, and Y.~Xia.
\newblock Uniform bahadur representation for local polynomial estimates of
  m-regression and its application to the additive model.
\newblock \emph{Econometric Theory}, 26\penalty0 (5):\penalty0 1529--1564,
  2010.

\bibitem[Krittanawong et~al.(2023)Krittanawong, Qadeer, Hayes, Wang, Thurston,
  Virani, and Lavie]{krittanawong2023pm2}
C.~Krittanawong, Y.~K. Qadeer, R.~B. Hayes, Z.~Wang, G.~D. Thurston, S.~Virani,
  and C.~J. Lavie.
\newblock Pm2.5 and cardiovascular diseases: State-of-the-art review.
\newblock \emph{International Journal of Cardiology Cardiovascular Risk and
  Prevention}, 19:\penalty0 200217, 2023.

\bibitem[L{\'e}ger et~al.(2022)L{\'e}ger, Chatton, Le~Borgne, Pirracchio,
  Lasocki, and Foucher]{leger2022causal}
M.~L{\'e}ger, A.~Chatton, F.~Le~Borgne, R.~Pirracchio, S.~Lasocki, and
  Y.~Foucher.
\newblock Causal inference in case of near-violation of positivity: comparison
  of methods.
\newblock \emph{Biometrical Journal}, 64\penalty0 (8):\penalty0 1389--1403,
  2022.

\bibitem[Leiser et~al.(2019)Leiser, Smith, VanDerslice, Glotzbach, Farrell, and
  Hanson]{leiser2019evaluation}
C.~L. Leiser, K.~R. Smith, J.~A. VanDerslice, J.~P. Glotzbach, T.~W. Farrell,
  and H.~A. Hanson.
\newblock Evaluation of the sex-and-age-specific effects of pm2. 5 on hospital
  readmission in the presence of the competing risk of mortality in the
  medicare population of utah 1999--2009.
\newblock \emph{Journal of Clinical Medicine}, 8\penalty0 (12):\penalty0 2114,
  2019.

\bibitem[Li and Racine(2004)]{li2004cross}
Q.~Li and J.~Racine.
\newblock Cross-validated local linear nonparametric regression.
\newblock \emph{Statistica Sinica}, pages 485--512, 2004.

\bibitem[Linton and Nielsen(1995)]{linton1995kernel}
O.~Linton and J.~P. Nielsen.
\newblock A kernel method of estimating structured nonparametric regression
  based on marginal integration.
\newblock \emph{Biometrika}, 82\penalty0 (1):\penalty0 93--100, 1995.

\bibitem[Lu(1996)]{lu1996multivariate}
Z.-Q. Lu.
\newblock Multivariate locally weighted polynomial fitting and partial
  derivative estimation.
\newblock \emph{Journal of Multivariate Analysis}, 59\penalty0 (2):\penalty0
  187--205, 1996.

\bibitem[Ma and Wang(2020)]{ma2020robust}
X.~Ma and J.~Wang.
\newblock Robust inference using inverse probability weighting.
\newblock \emph{Journal of the American Statistical Association}, 115\penalty0
  (532):\penalty0 1851--1860, 2020.

\bibitem[Manski(1990)]{manski1990nonparametric}
C.~F. Manski.
\newblock Nonparametric bounds on treatment effects.
\newblock \emph{The American Economic Review}, 80\penalty0 (2):\penalty0
  319--323, 1990.

\bibitem[Massart(1990)]{massart1990tight}
P.~Massart.
\newblock The tight constant in the dvoretzky-kiefer-wolfowitz inequality.
\newblock \emph{The Annals of Probability}, 18\penalty0 (3):\penalty0
  1269--1283, 1990.

\bibitem[Neugebauer and van~der Laan(2007)]{neugebauer2007nonparametric}
R.~Neugebauer and M.~van~der Laan.
\newblock Nonparametric causal effects based on marginal structural models.
\newblock \emph{Journal of Statistical Planning and Inference}, 137\penalty0
  (2):\penalty0 419--434, 2007.

\bibitem[Newey(1994)]{newey1994kernel}
W.~K. Newey.
\newblock Kernel estimation of partial means and a general variance estimator.
\newblock \emph{Econometric Theory}, 10\penalty0 (2):\penalty0 1--21, 1994.

\bibitem[Newey and Stoker(1993)]{newey1993efficiency}
W.~K. Newey and T.~M. Stoker.
\newblock Efficiency of weighted average derivative estimators and index
  models.
\newblock \emph{Econometrica}, 61\penalty0 (5):\penalty0 1199--1223, 1993.

\bibitem[Nolan and Pollard(1987)]{nolan1987u}
D.~Nolan and D.~Pollard.
\newblock {$U$-Processes: Rates of Convergence}.
\newblock \emph{The Annals of Statistics}, 15\penalty0 (2):\penalty0 780 --
  799, 1987.

\bibitem[Paciorek(2010)]{paciorek2010importance}
C.~J. Paciorek.
\newblock The importance of scale for spatial-confounding bias and precision of
  spatial regression estimators.
\newblock \emph{Statistical Science}, 25\penalty0 (1):\penalty0 107--125, 2010.

\bibitem[Petersen et~al.(2012)Petersen, Porter, Gruber, Wang, and Van
  Der~Laan]{petersen2012diagnosing}
M.~L. Petersen, K.~E. Porter, S.~Gruber, Y.~Wang, and M.~J. Van Der~Laan.
\newblock Diagnosing and responding to violations in the positivity assumption.
\newblock \emph{Statistical methods in medical research}, 21\penalty0
  (1):\penalty0 31--54, 2012.

\bibitem[Powell et~al.(1989)Powell, Stock, and
  Stoker]{powell1989semiparametric}
J.~L. Powell, J.~H. Stock, and T.~M. Stoker.
\newblock Semiparametric estimation of index coefficients.
\newblock \emph{Econometrica}, pages 1403--1430, 1989.

\bibitem[Robins(1986)]{robins1986new}
J.~Robins.
\newblock A new approach to causal inference in mortality studies with a
  sustained exposure period—application to control of the healthy worker
  survivor effect.
\newblock \emph{Mathematical modelling}, 7\penalty0 (9-12):\penalty0
  1393--1512, 1986.

\bibitem[Robins et~al.(2000)Robins, Hernan, and Brumback]{robins2000marginal}
J.~M. Robins, M.~A. Hernan, and B.~Brumback.
\newblock Marginal structural models and causal inference in epidemiology.
\newblock \emph{Epidemiology}, 11\penalty0 (5):\penalty0 550--560, 2000.

\bibitem[Rothe(2017)]{rothe2017robust}
C.~Rothe.
\newblock Robust confidence intervals for average treatment effects under
  limited overlap.
\newblock \emph{Econometrica}, 85\penalty0 (2):\penalty0 645--660, 2017.

\bibitem[Rothenh{\"a}usler and Yu(2019)]{rothenhausler2019incremental}
D.~Rothenh{\"a}usler and B.~Yu.
\newblock Incremental causal effects.
\newblock \emph{arXiv preprint arXiv:1907.13258}, 2019.

\bibitem[Rubin(1974)]{rubin1974estimating}
D.~B. Rubin.
\newblock Estimating causal effects of treatments in randomized and
  nonrandomized studies.
\newblock \emph{Journal of Educational Psychology}, 66\penalty0 (5):\penalty0
  688--701, 1974.

\bibitem[Rubin(1980)]{rubin1980randomization}
D.~B. Rubin.
\newblock Randomization analysis of experimental data: The fisher randomization
  test comment.
\newblock \emph{Journal of the American Statistical Association}, 75\penalty0
  (371):\penalty0 591--593, 1980.

\bibitem[Ruppert and Wand(1994)]{ruppert1994multivariate}
D.~Ruppert and M.~P. Wand.
\newblock Multivariate locally weighted least squares regression.
\newblock \emph{The Annals of Statistics}, 22\penalty0 (3):\penalty0
  1346--1370, 1994.

\bibitem[Ruppert et~al.(1995)Ruppert, Sheather, and Wand]{ruppert1995effective}
D.~Ruppert, S.~J. Sheather, and M.~P. Wand.
\newblock An effective bandwidth selector for local least squares regression.
\newblock \emph{Journal of the American Statistical Association}, 90\penalty0
  (432):\penalty0 1257--1270, 1995.

\bibitem[Schindl et~al.(2024)Schindl, Shen, and
  Kennedy]{schindl2024incremental}
K.~Schindl, S.~Shen, and E.~H. Kennedy.
\newblock Incremental effects for continuous exposures.
\newblock \emph{arXiv preprint arXiv:2409.11967}, 2024.

\bibitem[Schindler(2011)]{schindler2011bandwidth}
A.~Schindler.
\newblock \emph{Bandwidth selection in nonparametric kernel estimation}.
\newblock PhD thesis, G{\"o}ttingen, Georg-August Universit{\"a}t, Diss., 2011.

\bibitem[Schnell and Papadogeorgou(2020)]{schnell2020}
P.~Schnell and G.~Papadogeorgou.
\newblock Mitigating unobserved spatial confounding when estimating the effect
  of supermarket access on cardiovascular disease deaths.
\newblock \emph{Annals of Applied Statistics}, 14:\penalty0 2069--2095, 12
  2020.

\bibitem[Semenova and Chernozhukov(2021)]{semenova2021debiased}
V.~Semenova and V.~Chernozhukov.
\newblock Debiased machine learning of conditional average treatment effects
  and other causal functions.
\newblock \emph{The Econometrics Journal}, 24\penalty0 (2):\penalty0 264--289,
  2021.

\bibitem[Shao(2003)]{shao2003mathematical}
J.~Shao.
\newblock \emph{Mathematical Statistics}.
\newblock Springer Science \& Business Media, 2003.

\bibitem[Sheather(2004)]{sheather2004density}
S.~J. Sheather.
\newblock Density estimation.
\newblock \emph{Statistical Science}, 19\penalty0 (4):\penalty0 588--597, 2004.

\bibitem[Shieh(2014)]{shieh2014u}
G.~S. Shieh.
\newblock {U-and V-statistics}.
\newblock \emph{Wiley StatsRef: Statistics Reference Online}, 2014.

\bibitem[Stone(1980)]{stone1980optimal}
C.~J. Stone.
\newblock Optimal rates of convergence for nonparametric estimators.
\newblock \emph{The Annals of Statistics}, 8\penalty0 (6):\penalty0 1348--1360,
  1980.

\bibitem[Stone(1982)]{stone1982optimal}
C.~J. Stone.
\newblock Optimal global rates of convergence for nonparametric regression.
\newblock \emph{The Annals of Statistics}, 10\penalty0 (4):\penalty0
  1040--1053, 1982.

\bibitem[Stute(1986)]{stute1986conditional}
W.~Stute.
\newblock Conditional empirical processes.
\newblock \emph{The Annals of Statistics}, 14\penalty0 (2):\penalty0 638--647,
  1986.

\bibitem[Takatsu and Westling(2022)]{takatsu2022debiased}
K.~Takatsu and T.~Westling.
\newblock Debiased inference for a covariate-adjusted regression function.
\newblock \emph{arXiv preprint arXiv:2210.06448}, 2022.

\bibitem[Tang and Westling(2024)]{tang2024consistency}
Z.~Tang and T.~Westling.
\newblock Consistency of the bootstrap for asymptotically linear estimators
  based on machine learning.
\newblock \emph{arXiv preprint arXiv:2404.03064}, 2024.

\bibitem[Thaden and Kneib(2018)]{thaden2018structural}
H.~Thaden and T.~Kneib.
\newblock Structural equation models for dealing with spatial confounding.
\newblock \emph{The American Statistician}, 72\penalty0 (3):\penalty0 239--252,
  2018.

\bibitem[van~der Laan and Robins(2003)]{laan2003unified}
M.~J. van~der Laan and J.~M. Robins.
\newblock \emph{Unified methods for censored longitudinal data and causality}.
\newblock Springer, 2003.

\bibitem[van~der Vaart(1991)]{van1991differentiable}
A.~van~der Vaart.
\newblock On differentiable functionals.
\newblock \emph{The Annals of Statistics}, 19\penalty0 (1):\penalty0 178--204,
  1991.

\bibitem[van~der Vaart and Wellner(1996)]{van1996weak}
A.~van~der Vaart and J.~Wellner.
\newblock \emph{Weak Convergence and Empirical Processes: With Applications to
  Statistics}.
\newblock Springer Series in Statistics. Springer, 1996.

\bibitem[van~der Vaart(1998)]{VDV1998}
A.~W. van~der Vaart.
\newblock \emph{Asymptotic Statistics}.
\newblock Cambridge Series in Statistical and Probabilistic Mathematics.
  Cambridge University Press, 1998.

\bibitem[VanderWeele(2008)]{vanderweele2008sign}
T.~J. VanderWeele.
\newblock The sign of the bias of unmeasured confounding.
\newblock \emph{Biometrics}, 64\penalty0 (3):\penalty0 702--706, 2008.

\bibitem[Wasserman(2006)]{wasserman2006all}
L.~Wasserman.
\newblock \emph{All of Nonparametric Statistics}.
\newblock Springer Science \& Business Media, 2006.

\bibitem[Westling et~al.(2020)Westling, Gilbert, and
  Carone]{westling2020causal}
T.~Westling, P.~Gilbert, and M.~Carone.
\newblock Causal isotonic regression.
\newblock \emph{Journal of the Royal Statistical Society Series B: Statistical
  Methodology}, 82\penalty0 (3):\penalty0 719--747, 2020.

\bibitem[Westreich and Cole(2010)]{westreich2010invited}
D.~Westreich and S.~R. Cole.
\newblock Invited commentary: positivity in practice.
\newblock \emph{American Journal of Epidemiology}, 171\penalty0 (6):\penalty0
  674--677, 2010.

\bibitem[Wiecha and Reich(2024)]{wiecha2024two}
N.~Wiecha and B.~J. Reich.
\newblock Two-stage spatial regression models for spatial confounding.
\newblock \emph{arXiv preprint arXiv:2404.09358}, 2024.

\bibitem[Wyatt et~al.(2020{\natexlab{a}})Wyatt, Peterson, Wade, Neas, and
  Rappold]{wyatt2020contribution}
L.~H. Wyatt, G.~C. Peterson, T.~J. Wade, L.~M. Neas, and A.~G. Rappold.
\newblock The contribution of improved air quality to reduced cardiovascular
  mortality: Declines in socioeconomic differences over time.
\newblock \emph{Environment international}, 136:\penalty0 105430,
  2020{\natexlab{a}}.

\bibitem[Wyatt et~al.(2020{\natexlab{b}})Wyatt, Peterson, Wade, Neas, and
  Rappold]{wyatt2020annual}
L.~H. Wyatt, G.~C.~L. Peterson, T.~J. Wade, L.~M. Neas, and A.~G. Rappold.
\newblock Annual pm2. 5 and cardiovascular mortality rate data: Trends modified
  by county socioeconomic status in 2,132 us counties.
\newblock \emph{Data in Brief}, 30:\penalty0 105318, 2020{\natexlab{b}}.

\bibitem[Yang and Tschernig(1999)]{yang1999multivariate}
L.~Yang and R.~Tschernig.
\newblock Multivariate bandwidth selection for local linear regression.
\newblock \emph{Journal of the Royal Statistical Society Series B: Statistical
  Methodology}, 61\penalty0 (4):\penalty0 793--815, 1999.

\bibitem[Yang and Ding(2018)]{yang2018asymptotic}
S.~Yang and P.~Ding.
\newblock Asymptotic inference of causal effects with observational studies
  trimmed by the estimated propensity scores.
\newblock \emph{Biometrika}, 105\penalty0 (2):\penalty0 487--493, 2018.

\bibitem[Zhang and Chen(2024)]{zhang2024package}
Y.~Zhang and Y.-C. Chen.
\newblock \emph{Package ‘npDoseResponse’}, 2024.
\newblock URL \url{https://cran.r-project.org/package=npDoseResponse}.
\newblock R package version 0.1.

\bibitem[Zhou and Wolfe(2000)]{zhou2000derivative}
S.~Zhou and D.~A. Wolfe.
\newblock On derivative estimation in spline regression.
\newblock \emph{Statistica Sinica}, 10\penalty0 (1):\penalty0 93--108, 2000.

\end{thebibliography}

\newpage 
\onehalfspacing

\begin{center}
{\LARGE \textbf{Supplementary Materials to ``Nonparametric Inference on Dose-Response Curves Without the Positivity Condition''}}\\~\\

\end{center}

\appendix

\setcounter{page}{1}

\noindent{\bf\LARGE Contents}

\startcontents[sections]
\printcontents[sections]{l}{1}{\setcounter{tocdepth}{2}}

\vspace{1cm}

%
%

\section{Nonparametric Bounds Under Zero Treatment Variations}
\label{app:np_bounds}

In this section, we study nonparametric bounds on the dose-response curve $m(t)$ and its derivative function $\theta(t)$ under the additive confounding model \eqref{add_conf_model} when the intrinsic treatment variation has zero variance. Specifically,
\begin{equation}
\label{add_conf2}
Y(t) = \bar{m}(t) + \eta(\bm{S}) +\epsilon \quad \text{ and } \quad T=f(\bm{S},E),
\end{equation}
where $E\in \mathbb{R}$ is an independent treatment variation random variable with $\mathbb{E}(E)=0$. When the treatment variation has zero variance, it means that $\mathrm{Var}(E)=0$.
The derivations of nonparametric bounds for more general models are left as future work.

To motivate this discussion, we first present an illustrative example demonstrating how zero treatment variation for certain covariate levels $\bm{s}\in \mathcal{S}$ can lead to ambiguities in defining the corresponding dose-response curves.

\begin{example}[Necessity of Assumption~\ref{assump:identify_cond}(c)]
	\label{example:id}
	Suppose that the positivity condition \eqref{assump:positivity} holds, and let $T=f(\bm{S},E) =S_1$ almost surely so that $\mathrm{Var}(E)=\mathrm{Var}(T|\bm{S}=\bm{s})=0$ for any $S_1=s_1$. Here, $S_1$ is the first component of $\bm{S}\in \mathcal{S}\subset \mathbb{R}^d$, and $E \in \mathbb{R}$ is an independent treatment variation random variable. We further assume that $\mathbb{E}(S_1) = 0$ and consider two equivalent conditional mean outcome functions:
	\begin{align*}
		\mu_1(T,\bm{S}) \equiv T+ 2S_1 \quad \text{ and } \quad \mu_2(T,\bm{S}) \equiv 2T + S_1,
	\end{align*}
	both of which are well-defined under the positivity condition \eqref{assump:positivity} and equal to $3S_1$ almost surely. While they agree on the support $\left\{(t,\bm{s})\in \mathcal{T}\times \mathcal{S}: t=f(\bm{s})=s_1 \right\}$, these two conditional mean outcome functions lead to two distinct covariate-adjusted functions (or equivalently, dose-response curves):
	$$m_1(t) =\E\left[Y_1(t)\right]= \E\left[\mu_1(t,\bm{S})\right] = t \quad \text{ and } \quad m_2(t) =\E\left[Y_2(t)\right] = \E\left[\mu_2(t,\bm{S})\right] = 2t,$$
	whose derivatives are different as well. 
\end{example}

\subsection{Nonparametric Bound on $m(t)$}

Under the additive confounding model \eqref{add_conf2}, Example~\ref{example:id} above demonstrates that when $\mathrm{Var}(E)=0$, the conditional mean outcome function $\mu(t,\bm{s})$ is identifiable only on the lower-dimensional surface $\left\{(t,\bm{s})\in \mathcal{T}\times \mathcal{S}: t=f(\bm{s},0)\equiv f(\bm{s})\right\}$, where 
\begin{equation}
	\label{mu_surf}
	\mu\left(f(\bm{s}),\bm{s}\right) = \bar{m}\left(f(\bm{s})\right) + \eta(\bm{s}) = m\left(f(\bm{s})\right) + \eta(\bm{s})
\end{equation}
when $\E\left[\eta(\bm{S})\right]=0$ by Proposition~\ref{prop:add_conf_prop}. 

Since $\mathrm{Var}(E)=0$, the relation $T=f(\bm{S})$ can always be estimated from the observed data. Hence, we assume that the function $f$ is known for any $\bm{s}\in \mathcal{S}$ in subsequent analyses. To obtain nonparametric bounds for all possible values of $m(t)$, we impose the following assumption on the magnitude of the random effect function $\eta:\mathcal{S}\to \mathbb{R}$ in \eqref{add_conf2}; see \cite{manski1990nonparametric} for related discussions.

\begin{assump}[Bounded random effect]
	\label{assump:bound_random_eff}
	Let $L_f(t) = \left\{\bm{s}\in \mathcal{S}: f(\bm{s}) = t\right\}$ be a level set of the function $f:\mathcal{S}\to \mathbb{R}$ at $t\in \mathcal{T}$. There exists a constant $\rho_1>0$ such that
	$$\rho_1\geq \max\left\{\sup_{t\in \mathcal{T}}\sup_{\bm{s}\in L_f(t)} |\eta(\bm{s})|,\; \frac{\sup_{t\in \mathcal{T}}\sup_{\bm{s}\in L_f(t)}\mu(f(\bm{s}), \bm{s}) - \inf_{t\in \mathcal{T}}\inf_{\bm{s}\in L_f(t)}\mu(f(\bm{s}), \bm{s})}{2} \right\}.$$
\end{assump}

By \eqref{mu_surf} and the first lower bound on $\rho_1\geq \sup_{t\in \mathcal{T}}\sup_{\bm{s}\in L_f(t)} |\eta(\bm{s})|$ in Assumption~\ref{assump:bound_random_eff}, we know that $$\left|\mu(f(\bm{s}),\bm{s}) - m(t)\right| = |\eta(\bm{s})| \leq \rho_1$$
for any $\bm{s}\in L_f(t)$. This also implies that
\begin{align}
	\label{nonp_bound_m}
	\begin{split}
		m(t) &\in \bigcap_{\bm{s}\in L_f(t)} \left[\mu(f(\bm{s}),\bm{s}) - \rho_1,\, \mu(f(\bm{s}),\bm{s})+\rho_1 \right] \\
		&= \left[\sup_{\bm{s}\in L_f(t)}\mu(f(\bm{s}), \bm{s}) -\rho_1,\, \inf_{\bm{s}\in L_f(t)}\mu(f(\bm{s}), \bm{s}) + \rho_1 \right],
	\end{split}
\end{align}
which is the nonparametric bound on $m(t)$ that contains all the possible values of $m(t)$ for any fixed $t\in \mathcal{T}$ when $\mathrm{Var}(E)=0$. Notice that this bound \eqref{nonp_bound_m} is well-defined and nonempty under the second lower bound on $\rho_1$ in Assumption~\ref{assump:bound_random_eff}.

\subsection{Nonparametric Bound on $\theta(t)$}

Since the functions $\mu(f(\bm{s}),\bm{s})$ and $f(\bm{s})$ are identifiable even when $\mathrm{Var}(E)=0$, their derivatives with respect to $\bm{s}$ are also identifiable. Hence, we assume that the gradients $\nabla_{\bm{s}} \mu(f(\bm{s}), \bm{s})$ and $\nabla f(\bm{s})$ are known for any $\bm{s}\in \mathcal{S}$ in the following analysis. The nonparametric bound on the dose-response curve $m(t)$ relies on an upper bound on the random effect function $\eta:\mathcal{S}\to \mathbb{R}$ (recall Assumption~\ref{assump:bound_random_eff}), while the nonparametric bound on the derivative function $\theta(t)$ requires some constraints on the gradient $\nabla\eta(\bm{s})$ of the random effect function as follows.

\begin{assump}[Constraints on the random effect gradient]
	\label{assump:bound_grad_random_eff}
	Let $L_f(t) = \left\{\bm{s}\in \mathcal{S}: f(\bm{s}) = t\right\}$ be a level set of the function $f:\mathcal{S}\to \mathbb{R}$ at $t\in \mathcal{T}$. There exist constants $\rho_2,\rho_3>0$ such that
	$$\rho_3\leq \inf_{t\in \mathcal{T}} \inf_{\bm{s}\in L_f(t)} \norm{\nabla \eta(\bm{s})}_{\min} \leq \sup_{t\in \mathcal{T}} \sup_{\bm{s}\in L_f(t)} \norm{\nabla \eta(\bm{s})}_{\max} \leq \rho_2$$
	and 
	$$\sup_{t\in \mathcal{T}} \sup_{\bm{s}\in L_f(t)} \max_{j=1,...,d} \left[\frac{v_j(\bm{s}) - \mathrm{sign}(g_j(\bm{s})) \cdot \rho_2}{g_j(\bm{s})}\right] \leq \inf_{t\in \mathcal{T}} \inf_{\bm{s}\in L_f(t)} \min_{j=1,...,d} \left[\frac{v_j(\bm{s}) + \mathrm{sign}(g_j(\bm{s})) \cdot \rho_2}{g_j(\bm{s})}\right],$$
	where $\norm{\nabla\eta(\bm{s})}_{\min} = \min_{j=1,...,d}\left|\frac{\partial}{\partial s_j} \eta(\bm{s})\right|$, $\norm{\nabla\eta(\bm{s})}_{\max} = \max_{j=1,...,d}\left|\frac{\partial}{\partial s_j} \eta(\bm{s})\right|$, $v_j(\bm{s}) = \frac{\partial}{\partial s_j} \mu(f(\bm{s}),\bm{s})$, and $g_j(\bm{s}) = \frac{\partial}{\partial s_j} f(\bm{s})$ for $j=1,...,d$.
\end{assump}

The first inequality in Assumption~\ref{assump:bound_grad_random_eff} ensures that we can derive a finite and meaningful nonparametric bound on $\theta(t)$ for any $t\in \mathcal{T}$. In particular, a direct calculation shows that
\begin{align*}
	\nabla_{\bm{s}}\mu(f(\bm{s}), \bm{s}) &= \nabla_{\bm{s}} m(f(\bm{s})) + \nabla \eta(\bm{s})\\
	&= m'(f(\bm{s}))\cdot \nabla f(\bm{s}) + \nabla \eta(\bm{s})\\
	&= \theta(t) \cdot \nabla f(\bm{s}) + \nabla \eta(\bm{s})
\end{align*}
for all $\bm{s}\in L_f(t)$. This indicates that
\begin{align*}
	\left|\frac{\partial}{\partial s_j} \mu(f(\bm{s}),\bm{s}) - \theta(t) \cdot \frac{\partial}{\partial s_j} f(\bm{s})\right| = \left|v_j(\bm{s}) - \theta(t) \cdot g_j(\bm{s})\right| = \left|\frac{\partial}{\partial s_j} \eta(\bm{s})\right|\leq \rho_2
\end{align*}
for $j=1,...,d$ and therefore, a nonparametric bound on $\theta(t)$ is given by
\begin{equation}
	\label{nonp_bound_theta}
	\theta(t) \in \bigcap_{j=1}^d \left[\frac{v_j(\bm{s}) - \mathrm{sign}(g_j(\bm{s})) \cdot \rho_2}{g_j(\bm{s})},\; \frac{v_j(\bm{s}) + \mathrm{sign}(g_j(\bm{s})) \cdot \rho_2}{g_j(\bm{s})} \right].
\end{equation}
This bound is well-defined and nonempty under the second inequality in Assumption~\ref{assump:bound_grad_random_eff}. More importantly, the nonparametric bound \eqref{nonp_bound_theta} highlights the pivotal role of the covariate effect function $f$ of $\bm{S}$ on treatment $T$ in bounding the possible value of $\theta(t)$. Specifically, as the variation in $f$ increases (\emph{i.e.}, the magnitude of $\nabla f(\bm{s})$ becomes larger), the nonparametric bound \eqref{nonp_bound_theta} becomes tighter, providing more precise constraints on $\theta(t)$.

\section{Pointwise Consistency of $\hat{\beta}_2(t,\mathbf{s})$ from \eqref{localpoly3}}
\label{app:pointwise_loc_poly}

In this section, we establish the pointwise consistency of $\hat{\beta}_2(t,\bm{s})$ from local polynomial regression \eqref{localpoly3}. 

\subsection{Additional Notations}

We introduce some notations that will be used in the following proofs of the consistency results of $\hat{\beta}_2(t,\bm{s})$ from local polynomial regression \eqref{localpoly3} (Lemmas~\ref{lem:loc_poly_deriv} and \ref{lem:loc_poly_deriv_unif}). Under the notations in Assumption~\ref{assump:reg_kernel}(a), we define a matrix 
\begin{equation}
	\label{M_q}
	\bm{M}_q = 
	\begin{pmatrix}
		\left(\kappa_{i+j-2}^{(T)} \right)_{1\leq i,j\leq q+1} & \bm{0} \\
		\bm{0} & \left(\kappa_{2,i-q-1}^{(S)} \mathbbm{1}_{\{i=j\}} \right)_{q+1 < i,j \leq q+1+d}
	\end{pmatrix} \in \mathbb{R}^{(q+1+d) \times (q+1+d)}.
\end{equation}
Notice that $\bm{M}_q$ only depends on the kernel functions $K_T,K_S$ in the local polynomial regression. For any $(t,\bm{s})\in \mathcal{T}\times \mathcal{S}$, we also define the functions $\bm{\Psi}_{t,\bm{s}}: \mathbb{R}\times \mathbb{R}\times \mathbb{R}^d \to \mathbb{R}^{q+1+d}$ and $\bm{\psi}_{t,\bm{s}}: \mathbb{R}\times \mathbb{R}^d \to \mathbb{R}^{q+1+d}$ as:
\begin{align}
	\label{psi_prod}
	\begin{split}
		\bm{\Psi}_{t,\bm{s}}(y,z,\bm{v}) &= \begin{bmatrix}
			\left(y\cdot \left(\frac{z-t}{h}\right)^{j-1} K_T\left(\frac{z-t}{h}\right) K_S\left(\frac{\bm{v}-\bm{s}}{b}\right)\right)_{1\leq j\leq q+1}\\
			\left(y\cdot \left(\frac{v_{j-q-1}- s_{j-q-1}}{b}\right) K_T\left(\frac{z-t}{h}\right) K_S\left(\frac{\bm{v}-\bm{s}}{b}\right)\right)_{q+1< j\leq q+1+d}
		\end{bmatrix}\\
		&= y \cdot \begin{bmatrix}
			\left(\left(\frac{z-t}{h}\right)^{j-1} K_T\left(\frac{z-t}{h}\right) K_S\left(\frac{\bm{v}-\bm{s}}{b}\right)\right)_{1\leq j\leq q+1}\\
			\left(\left(\frac{v_{j-q-1}- s_{j-q-1}}{b}\right) K_T\left(\frac{z-t}{h}\right) K_S\left(\frac{\bm{v}-\bm{s}}{b}\right)\right)_{q+1< j\leq q+1+d}
		\end{bmatrix}\equiv y\cdot \bm{\psi}_{t,\bm{s}}(z,\bm{v}).
	\end{split}
\end{align}

\subsection{Consistency Statement and its Proof}
\label{app:proof_beta_2}

The detailed asymptotic expressions for conditional variance and bias of $\hat{\beta}_2(t,\bm{s})$ are stated in Lemma~\ref{lem:loc_poly_deriv} below. Asymptotically, under the assumptions in Lemma~\ref{lem:loc_poly_deriv}, we know that as $h,b,\frac{\max\{h,b\}^4}{h}\to 0$, and $nh^3b^d \to\infty$,
\begin{align*}
	\hat{\beta}_2(t,\bm{s}) -  \beta_2(t,\bm{s}) &=
	\begin{cases}
		O\left(h^{q+1} + b^2 + \frac{b^4}{h}\right) + O_P\left(\sqrt{\frac{1}{nh^3 b^d}}\right) & \text{ if } q \text{ is odd and } (t,\bm{s})\in \mathcal{E}^{\circ},\\
		O\left(h^q + b^2 + \frac{b^4}{h}\right) + O_P\left(\sqrt{\frac{1}{nh^3 b^d}}\right) & \text{ if } q \text{ is even and } (t,\bm{s})\in \mathcal{E}^{\circ},\\
		O\left(h^q + \frac{\max\{h,b\}^4}{h}\right) + O_P\left(\sqrt{\frac{1}{nh^3 b^d}}\right) & \text{ if } q \text{ is an integer and } (t,\bm{s})\in \partial \mathcal{E}
	\end{cases}
\end{align*}
for any fixed $(t,\bm{s}) \in \mathcal{E}$ and integer $q>0$.

\begin{lemma}[Pointwise convergence of $\hat{\beta}_2(t,\bm{s})$]
	\label{lem:loc_poly_deriv}
	Suppose that Assumptions~\ref{assump:reg_diff}, \ref{assump:den_diff}, \ref{assump:boundary_cond}, and \ref{assump:reg_kernel}(a) hold. Let $\hat \beta_2(t,\bm{s}) $ be the second element of $\hat{\bm{\beta}}(t,\bm{s}) \in \mathbb{R}^{q+1}$, $\beta_2(t,\bm{s}) = \frac{\partial}{\partial t} \mu(t,\bm{s})$, and $\mathbb{X} = \left\{(T_i,\bm{S}_i)\right\}_{i=1}^n$. Then, for any $(t,\bm{s}) \in \mathcal{E}^{\circ}$ and $q>0$, as $h,b,\frac{b^4}{h}\to 0$ and $nh^3 b^d \to\infty$, we know that the asymptotic conditional covariance of $\hat{\beta}_2(t,\bm{s}) \in \mathbb{R}$ is given by
	\begin{align*}
		\mathrm{Var}\left[\hat{\beta}_2(t,\bm{s}) \big| \mathbb{X} \right] &= \frac{\sigma^2}{nh^3 b^d \cdot p(t,\bm{s})} \left[\bm{e}_2^T \bm{M}_q^{-1} \bm{M}_q^* \bm{M}_q^{-1} \bm{e}_2 + O\left(\max\{h,b\}\right) + O_P\left(\sqrt{\frac{1}{nhb^d}}\right) \right],
	\end{align*}
	and the asymptotic conditional bias is given by
	\begin{align*}
		\mathrm{Bias}\left[\hat{\beta}_2(t,\bm{s}) \big| \mathbb{X}\right] =
		\begin{cases}
			\frac{1}{p(t,\bm{s})} \left[h^{q+1} \tau_q^{\text{odd}} + b^2 \tau_q^* + O\left(\frac{b^4}{h}\right) + o_P\left(\sqrt{\frac{1}{nhb^d}}\right) \right] & q \text{ is odd},\\
			\frac{1}{p(t,\bm{s})} \left[h^q \tau_q^{\text{even}} + b^2 \tau_q^* + O\left(\frac{b^4}{h}\right) + o_P\left(\sqrt{\frac{1}{nhb^d}}\right) \right] & q \text{ is even}.
		\end{cases}
	\end{align*}
	Here, $\mathcal{E}^{\circ}$ is the interior of the support $\mathcal{E}\subset \mathcal{T}\times \mathcal{S}$, $\bm{e}_2=(0,1,0,...,0)^T \in \mathbb{R}^{q+1}$, $\bm{M}_q = \left(\kappa_{i+j-2}^{(T)} \right)_{1\leq i,j\leq q+1} \in \mathbb{R}^{(q+1)\times (q+1)}$, $\bm{M}_q^*=\left(\nu_{i+j-2}^{(T)} \nu_0^{(S)} \right)_{1\leq i,j\leq q+1} \in \mathbb{R}^{(q+1)\times (q+1)}$, $\tau_q^{\text{odd}} = \frac{\bm{e}_2^T\bm{M}_q^{-1} \bm{\tau}_q}{h}$ when $q$ is odd and $\tau_q^{\text{even}} = \bm{e}_2^T\bm{M}_q^{-1} \bm{\tau}_q$ when $q$ is even, as well as $\tau_q^* = \frac{\bm{e}_2^T\bm{M}_q^{-1} \bm{\tau}_q^*}{b}$, where
	\begin{align*}
		\bm{\tau}_q = \left[\frac{\partial^{q+1}}{\partial t^{q+1}} \mu(t,\bm{s}) \right] \frac{1}{(q+1)!} 
		\begin{bmatrix}
			\begin{pmatrix}
				\kappa_{q+j}^{(T)} \cdot p(t,\bm{s}) + h\cdot \kappa_{q+j+1}^{(T)} \cdot \frac{\partial}{\partial t} p(t,\bm{s})
			\end{pmatrix}_{1 \leq j \leq q+1}\\
			\begin{pmatrix}
				b\kappa_{q+1}^{(T)} \kappa_{2,j-q-1}^{(S)} \cdot \frac{\partial}{\partial s_{j-q-1}} p(t,\bm{s})
			\end{pmatrix}_{q+1 < j \leq q+1+d}
		\end{bmatrix} \in \mathbb{R}^{(q+1+d)\times (q+1+d)},
	\end{align*}
	\begin{align*}
		\bm{\tau}_q^* = 
		\begin{bmatrix}
			\begin{pmatrix}
				b\sum_{\ell=1}^d \left[\frac{\partial^2}{\partial t \partial s_{\ell}} \mu(t,\bm{s})\right] \kappa_j^{(T)} \kappa_{2,\ell}^{(S)} \cdot \frac{\partial}{\partial s_{\ell}} p(t,\bm{s})
			\end{pmatrix}_{1\leq j \leq q+1}\\
			\begin{pmatrix}
				h\left[\frac{\partial^2}{\partial t \partial s_{j-q-1}} \mu(t,\bm{s})\right] \kappa_2^{(T)} \kappa_{2,j-q-1}^{(S)} \cdot \frac{\partial}{\partial t} p(t,\bm{s})
			\end{pmatrix}_{q+1 < j \leq q+1+d}
		\end{bmatrix} \in \mathbb{R}^{(q+1+d)\times (q+1+d)},
	\end{align*}
	and
	\begin{align*}
		\tilde{\bm{\tau}}_q = 
		\begin{bmatrix}
			\begin{pmatrix}
				\kappa_{j-1}^{(T)} p(t,\bm{s}) \sum_{\ell=1}^d \frac{\kappa_{2,\ell}^{(S)}}{2} \left(\frac{\partial^2}{\partial s_{\ell}^2} \mu(t,\bm{s})\right)
			\end{pmatrix}_{1\leq j \leq q+1}\\
			\begin{pmatrix}
				b \sum_{\ell=1}^d \frac{\kappa_{2,j-q-1,\ell}^{(S)}}{2} \left[\frac{\partial^2}{\partial s_{\ell}^2} \mu(t,\bm{s})\right] \cdot \frac{\partial}{\partial s_{j-q-1}} p(t,\bm{s})
			\end{pmatrix}_{q+1 < j \leq q+1+d}
		\end{bmatrix} \in \mathbb{R}^{(q+1+d)\times (q+1+d)}.
	\end{align*}
	Furthermore, as $h,b,\frac{\max\{h,b\}^4}{h}\to 0$, and $nh^3b^d \to\infty$, we have that
	\begin{align*}
		&\hat{\beta}_2(t,\bm{s}) -  \beta_2(t,\bm{s}) \\
		&=
		\begin{cases}
			O\left(h^{q+1} + b^2 + \frac{b^4}{h}\right) + O_P\left(\sqrt{\frac{1}{nh^3 b^d}}\right) & \text{ if } q \text{ is odd and } (t,\bm{s})\in \mathcal{E}^{\circ},\\
			O\left(h^q + b^2 + \frac{b^4}{h}\right) + O_P\left(\sqrt{\frac{1}{nh^3 b^d}}\right) & \text{ if } q \text{ is even and } (t,\bm{s})\in \mathcal{E}^{\circ},\\
			O\left(h^q + \frac{\max\{h,b\}^4}{h}\right) + O_P\left(\sqrt{\frac{1}{nh^3 b^d}}\right) & \text{ if } q \text{ is an integer and } (t,\bm{s})\in \partial \mathcal{E}.
		\end{cases}
	\end{align*}
\end{lemma}

\begin{proof}[Proof of Lemma~\ref{lem:loc_poly_deriv}]
	Recall from \eqref{localpoly3} that 
	\begin{align*}
		\left(\hat{\bm{\beta}}(t,\bm{s}), \hat{\bm{\alpha}}(t,\bm{s}) \right)^T &= \left[\bm{X}^T(t,\bm{s})\bm{W}(t,\bm{s}) \bm{X}(t,\bm{s})\right]^{-1} \bm{X}(t,\bm{s})^T\bm{W}(t,\bm{s}) \bm{Y} \\
		&\equiv \left(\bm{X}^T\bm{W} \bm{X}\right)^{-1} \bm{X}^T\bm{W} \bm{Y}.
	\end{align*}
	The proof of Lemma~\ref{lem:loc_poly_deriv} has two major components, where we consider the cases when $(t,\bm{s})$ is an interior point or a boundary point of the support $\mathcal{E}$. When $(t,\bm{s})$ is an interior point of $\mathcal{E}$, we also divide the arguments into three steps that deal with the rates of convergence for the term $\bm{X}^T\bm{W} \bm{X}$, the conditional covariance term $\mathrm{Cov}\left[\left(\hat{\bm{\beta}}(t,\bm{s}), \hat{\bm{\alpha}}(t,\bm{s}) \right)^T\right]$, and the conditional bias term $\mathrm{Bias}\left[\left(\hat{\bm{\beta}}(t,\bm{s}), \hat{\bm{\alpha}}(t,\bm{s}) \right)^T \Big| \mathbb{X}\right]$ separately.\\~\\

	We first consider the case when $(t,\bm{s})$ is an interior point of the support $\mathcal{E}$, \emph{i.e.}, for any $(t_1,\bm{s}_1)\in \mathcal{T}\times \mathcal{S}$, $\frac{t_1-t}{h}$ and  $\frac{\bm{s}_1-\bm{s}}{b}$ lie in the supports of $K_T$ and $K_S$ respectively when $h,b$ are small. 
	Notice that the conditional mean outcome function $\mu(t,\bm{s})=\E\left(Y|T=t,\bm{S}=\bm{s}\right)$ is well-defined in $\mathcal{E}$, and we have that
	\begin{align*}
		\mathrm{Cov}\left[\left(\hat{\bm{\beta}}(t,\bm{s}), \hat{\bm{\alpha}}(t,\bm{s}) \right)^T \Big| \mathbb{X}\right] &= \left(\bm{X}^T\bm{W} \bm{X} \right)^{-1} \bm{X}^T \bm{\Sigma} \bm{X} \left(\bm{X}^T\bm{W} \bm{X} \right)^{-1}
	\end{align*}
	and
	\begin{align*}
		\mathrm{Bias}\left[\left(\hat{\bm{\beta}}(t,\bm{s}), \hat{\bm{\alpha}}(t,\bm{s}) \right)^T \Big| \mathbb{X}\right] &= \mathbb{E}\left[\left(\hat{\bm{\beta}}(t,\bm{s}), \hat{\bm{\alpha}}(t,\bm{s}) \right)^T \right] - \left(\bm{\beta}(t,\bm{s}), \bm{\alpha}(t,\bm{s}) \right)^T \\
		&= \left(\bm{X}^T\bm{W} \bm{X} \right)^{-1} \bm{X}^T \bm{W} \left[ \begin{pmatrix}
			\mu(T_1,\bm{S}_1)\\
			\vdots\\
			\mu(T_n,\bm{S}_n)
		\end{pmatrix} - \bm{X} \begin{pmatrix}
			\bm{\beta}(t,\bm{s})\\
			\bm{\alpha}(t,\bm{s})
		\end{pmatrix}\right],
	\end{align*}
	where $\bm{\Sigma} \in \mathbb{R}^{n\times n}$ is a diagonal matrix defined as:
	\begin{align*}
		\bm{\Sigma} &= \mathrm{Var}\left(\bm{Y}|\mathbb{X}\right) \cdot \Diag\left(K_T^2\left(\frac{T_1-t}{h}\right) K_S^2\left(\frac{\bm{S}_1-\bm{s}}{b}\right),..., K_T^2\left(\frac{T_n-t}{h}\right) K_S^2\left(\frac{\bm{S}_n-\bm{s}}{b}\right) \right) \\
		&= \sigma^2 \cdot \Diag\left(K_T^2\left(\frac{T_1-t}{h}\right) K_S^2\left(\frac{\bm{S}_1-\bm{s}}{b}\right),..., K_T^2\left(\frac{T_n-t}{h}\right) K_S^2\left(\frac{\bm{S}_n-\bm{s}}{b}\right) \right)
	\end{align*}
under Assumption~\ref{assump:reg_diff} with  $\mathrm{Var}\left(\bm{Y}|\mathbb{X}\right)=\Diag\Big(\mathrm{Var}(Y_1|T_1,\bm{S}_1),...,\mathrm{Var}(Y_n|T_n,\bm{S}_n)\Big) \in \mathbb{R}^{n\times n}$. In addition, by Assumption~\ref{assump:reg_diff} and Taylor's expansion, we know that
	\[
	\left(\bm{\beta}(t,\bm{s}), \bm{\alpha}(t,\bm{s}) \right)^T \equiv \left(
	\mu(t,\bm{s}),
	\frac{\partial}{\partial t} \mu(t,\bm{s}),..., \frac{1}{q!} \cdot \frac{\partial^q}{\partial t^q} \mu(t,\bm{s}), \frac{\partial}{\partial s_1} \mu(t,\bm{s}),..., \frac{\partial}{\partial s_d} \mu(t,\bm{s})\right)^T.
	\]
	
	\noindent {\bf Step 1: Common term $\bm{X}^T\bm{W} \bm{X}$.} Before deriving the asymptotic behaviors of the above conditional covariance matrix and bias, we first study the rates of convergence of $\bm{X}^T\bm{W} \bm{X} \in \mathbb{R}^{(q+1+d) \times (q+1+d)}$.\\
	
	\noindent By direct calculations, we know that
	\begin{align}
		\label{X_W_mat}
		\begin{split}
			&\left(\bm{X}^T\bm{W} \bm{X}\right)_{i,j} \\
			&= 
			\begin{cases}
				\sum\limits_{k=1}^n (T_k-t)^{i+j-2} K_T\left(\frac{T_k-t}{h}\right) K_S\left(\frac{\bm{S}_k -\bm{s}}{b}\right), \quad\quad 1\leq i,j \leq q+1,\\
				\sum\limits_{k=1}^n (T_k-t)^{i-1} \left(S_{k,j-q-1} - s_{j-q-1}\right) K_T\left(\frac{T_k-t}{h}\right) K_S\left(\frac{\bm{S}_k -\bm{s}}{b}\right), \quad 1\leq i \leq q+1 \text{ and } q+1< j\leq q+1+d,\\
				\sum\limits_{k=1}^n \left(S_{k,i-q-1} - s_{i-q-1}\right) (T_k-t)^{j-1} K_T\left(\frac{T_k-t}{h}\right) K_S\left(\frac{\bm{S}_k -\bm{s}}{b}\right), \quad q+1< i\leq q+1+d \text{ and } 1\leq j \leq q+1,\\
				\sum\limits_{k=1}^n \left(S_{k,i-q-1} - s_{i-q-1}\right) \left(S_{k,j-q-1} - s_{j-q-1}\right) K_T\left(\frac{T_k-t}{h}\right) K_S\left(\frac{\bm{S}_k -\bm{s}}{b}\right), \quad q+1< i,j\leq q+1+d.
			\end{cases}
		\end{split}
	\end{align}
	Here, each sample $\bm{S}_k\in \mathcal{S} \subset \mathbb{R}^d$ is written as $\bm{S}_k = \left(S_{k,1},...,S_{k,d}\right)^T \in \mathbb{R}^d$. We now derive the asymptotic rates of convergence of the expectation and variance for each term in \eqref{X_W_mat} under Assumptions~\ref{assump:reg_diff}, \ref{assump:den_diff}, and \ref{assump:reg_kernel}.\\
	
	\noindent $\bullet$ {\bf Case I:} $1\leq i,j\leq q+1$. We compute that
	\begin{align*}
		&\mathbb{E}\left[\left(\bm{X}^T\bm{W} \bm{X}\right)_{i,j} \right] \\
		&= n\int_{\mathbb{R} \times \mathbb{R}^d} \left(\tilde{t} - t\right)^{i+j-2} K_T\left(\frac{\tilde{t} -t}{h}\right) K_S\left(\frac{\tilde{\bm{s}} - \bm{s}}{b}\right) p(\tilde{t},\tilde{\bm{s}}) \, d\tilde{t} d\tilde{\bm{s}}\\
		&\stackrel{\text{(i)}}{=} nh^{i+j-1} b^d \int_{\mathbb{R} \times \mathbb{R}^d} u^{i+j-2} K_T(u) K_S(\bm{v}) \cdot p(t+uh, \bm{s} + b\bm{v}) \, du d\bm{v}\\
		&\stackrel{\text{(ii)}}{=} nh^{i+j-1} b^d \int_{\mathbb{R} \times \mathbb{R}^d} u^{i+j-2} K_T(u) K_S(\bm{v}) \left[p(t,\bm{s}) + uh \cdot \frac{\partial}{\partial t} p(t,\bm{s}) + b\bm{v}^T \frac{\partial}{\partial \bm{s}} p(t,\bm{s}) + O\left(\max\{h,b\}^2\right)\right] du d\bm{v}\\
		&\stackrel{\text{(iii)}}{=} nh^{i+j-1} b^d \left[\kappa_{i+j-2}^{(T)} \cdot p(t,\bm{s}) + h \cdot \kappa_{i+j-1}^{(T)} \cdot \frac{\partial}{\partial t} p(t,\bm{s}) + O\left(\max\{h,b\}^2\right)\right],
	\end{align*}
	where (i) utilizes the changes of variables $u=\frac{\tilde{t}-t}{h}$ and $\bm{v}=\frac{\tilde{\bm{s}}-\bm{s}}{b}$, (ii) leverages the differentiability of $p(t,\bm{s})$ and apply Taylor's expansion, as well as (iii) uses the symmetric properties of $K_T,K_S$ with notations in Assumption~\ref{assump:reg_kernel}(a). In addition, we calculate that
	\begin{align*}
		&\mathrm{Var}\left[\left(\bm{X}^T\bm{W} \bm{X}\right)_{i,j} \right]\\
		&= n \cdot \mathrm{Var}\left[(T_1-t)^{i+j-2} K_T\left(\frac{T_1-t}{h}\right) K_S\left(\frac{\bm{S}_1 -\bm{s}}{b}\right) \right]\\
		&\leq n \int_{\mathbb{R} \times \mathbb{R}^d} (\tilde{t}-t)^{2i+2j-4} K_T^2\left(\frac{\tilde{t} -t}{h}\right) K_S^2\left(\frac{\tilde{\bm{s}} -\bm{s}}{b}\right) p(\tilde{t},\tilde{\bm{s}}) \, d\tilde{t} d\tilde{\bm{s}}\\
		&= nh^{2i+2j-3} b^d \int_{\mathbb{R} \times \mathbb{R}^d} u^{2i+2j-4} K_T^2(u) K_S^2(v) p(t+uh, \bm{s}+b\bm{v}) \, du d\bm{v} \\
		&= nh^{2i+2j-3} b^d \int_{\mathbb{R} \times \mathbb{R}^d} u^{2i+2j-4} K_T^2(u) K_S^2(v) \left[p(t,\bm{s}) + uh \cdot \frac{\partial}{\partial t} p(t,\bm{s}) + b\bm{v}^T \frac{\partial}{\partial \bm{s}} p(t,\bm{s}) + O\left(\max\{h,b\}^2\right)\right] du d\bm{v} \\
		&= nh^{2i+2j-3} b^d \left[\nu_{2i+2j-4}^{(T)} \nu_0^{(S)} \cdot p(t,\bm{s}) + h \cdot \nu_{2i+2j-3}^{(T)} \nu_0^{(S)} \cdot \frac{\partial}{\partial t} p(t,\bm{s}) + O\left(\max\{h,b\}^2\right) \right].
	\end{align*}
	The above calculations on $\mathbb{E}\left[\left(\bm{X}^T\bm{W} \bm{X}\right)_{i,j} \right]$ and $\mathrm{Var}\left[\left(\bm{X}^T\bm{W} \bm{X}\right)_{i,j} \right]$ imply that
	\begin{align*}
		\left(\bm{X}^T\bm{W} \bm{X}\right)_{i,j} &= \mathbb{E}\left[\left(\bm{X}^T\bm{W} \bm{X}\right)_{i,j} \right] + O_P\left(\sqrt{\mathrm{Var}\left[\left(\bm{X}^T\bm{W} \bm{X}\right)_{i,j} \right]}\right) \\
		&= nh^{i+j-1} b^d \left[\kappa_{i+j-2}^{(T)} \cdot p(t,\bm{s}) + h \cdot \kappa_{i+j-1}^{(T)} \cdot \frac{\partial}{\partial t} p(t,\bm{s}) + O\left(\max\{h,b\}^2\right) + O_P\left(\sqrt{\frac{1}{nhb^d}} \right)\right]
	\end{align*}
	for any $1\leq i,j \leq q+1$.\\
	
	\noindent $\bullet$ {\bf Case II:} $1\leq i\leq q+1$ and $q+1 < j \leq q+1+d$. We compute that
	\begin{align*}
		&\mathbb{E}\left[\left(\bm{X}^T\bm{W} \bm{X}\right)_{i,j} \right]\\
		&= n\int_{\mathbb{R} \times \mathbb{R}^d} \left(\tilde{t} - t\right)^{i-1} (\tilde{s}_{j-q-1} - s_{j-q-1}) K_T\left(\frac{\tilde{t} -t}{h}\right) K_S\left(\frac{\tilde{\bm{s}} - \bm{s}}{b}\right) p(\tilde{t},\tilde{\bm{s}}) \, d\tilde{t} d\tilde{\bm{s}}\\
		&=nh^ib^{d+1} \int_{\mathbb{R} \times \mathbb{R}^d} u^{i-1} v_{j-q-1} K_T(u) K_S(\bm{v}) p(t+hu, \bm{s}+b\bm{v})\, dud\bm{v}\\
		&= nh^ib^{d+1} \int_{\mathbb{R} \times \mathbb{R}^d} u^{i-1} v_{j-q-1} K_T(u) K_S(\bm{v}) \left[p(t,\bm{s}) + uh \cdot \frac{\partial}{\partial t} p(t,\bm{s}) + b\bm{v}^T \frac{\partial}{\partial \bm{s}} p(t,\bm{s}) + O\left(\max\{h,b\}^2\right)\right] dud\bm{v}\\
		&= nh^ib^{d+1} \left[b\cdot \kappa_{i-1}^{(T)} \kappa_{2,j-q-1}^{(S)} \cdot \frac{\partial}{\partial s_{j-q-1}} p(t,\bm{s}) + O\left(\max\{h,b\}^2\right) \right]
	\end{align*}
	and
	\begin{align*}
		&\mathrm{Var}\left[\left(\bm{X}^T\bm{W} \bm{X}\right)_{i,j} \right]\\
		&= n \cdot \mathrm{Var}\left[(T_1-t)^{i-1} (S_{1,j-q-1} -s_{j-q-1}) K_T\left(\frac{T_1-t}{h}\right) K_S\left(\frac{\bm{S}_1 -\bm{s}}{b}\right) \right]\\
		&\leq n \int_{\mathbb{R} \times \mathbb{R}^d} (\tilde{t}-t)^{2i-2} (\tilde{s}_{j-q-1} - s_{j-q-1})^2 K_T^2\left(\frac{\tilde{t} -t}{h}\right) K_S^2\left(\frac{\tilde{\bm{s}} -\bm{s}}{b}\right) p(\tilde{t},\tilde{\bm{s}}) \, d\tilde{t} d\tilde{\bm{s}}\\
		&= nh^{2i-1} b^{d+2} \int_{\mathbb{R} \times \mathbb{R}^d} u^{2i-2} v_{j-q-1}^2 K_T^2(u) K_S^2(\bm{v}) \left[p(t,\bm{s}) + uh \cdot \frac{\partial}{\partial t} p(t,\bm{s}) + b\bm{v}^T \frac{\partial}{\partial \bm{s}} p(t,\bm{s}) + O\left(\max\{h,b\}^2\right)\right] dud\bm{v}\\
		&= nh^{2i-1} b^{d+2} \left[\nu_{2i-2}^{(T)} \nu_{2,j-q-1}^{(S)} \cdot p(t,\bm{s}) + O(h)+ O\left(\max\{h,b\}^2\right)\right].
	\end{align*}
	These two terms indicate that
	\begin{align*}
		\left(\bm{X}^T\bm{W} \bm{X}\right)_{i,j} &= \mathbb{E}\left[\left(\bm{X}^T\bm{W} \bm{X}\right)_{i,j} \right] + O_P\left(\sqrt{\mathrm{Var}\left[\left(\bm{X}^T\bm{W} \bm{X}\right)_{i,j} \right]}\right) \\
		&= nh^i b^{d+1} \left[b\cdot \kappa_{i-1}^{(T)} \kappa_{2,j-q-1}^{(S)} \cdot \frac{\partial}{\partial s_{j-q-1}} p(t,\bm{s}) + O\left(\max\{h,b\}^2\right) + O_P\left(\sqrt{\frac{1}{nhb^d}} \right)\right]
	\end{align*}
	for any $1\leq i\leq q+1$ and $q+1 < j \leq q+1+d$.\\
	
	\noindent $\bullet$ {\bf Case III:} $q+1 < i \leq q+1+d$ and $1\leq j\leq q+1$. By swapping the roles of $i$ and $j$ in our calculations for {\bf Case 2}, we obtain that
	\begin{align*}
		\left(\bm{X}^T\bm{W} \bm{X}\right)_{i,j} &= nh^j b^{d+1} \left[b\cdot \kappa_{j-1}^{(T)} \kappa_{2,i-q-1}^{(S)} \cdot \frac{\partial}{\partial s_{i-q-1}} p(t,\bm{s}) + O\left(\max\{h,b\}^2\right) + O_P\left(\sqrt{\frac{1}{nhb^d}} \right)\right]
	\end{align*}
	for any $q+1 < i \leq q+1+d$ and $1\leq j\leq q+1$.\\
	
	\noindent $\bullet$ {\bf Case IV:} $q+1 < i,j \leq q+1+d$. We compute that
	\begin{align*}
		&\mathbb{E}\left[\left(\bm{X}^T\bm{W} \bm{X}\right)_{i,j} \right]\\
		&= n\int_{\mathbb{R} \times \mathbb{R}^d} (\tilde{s}_{i-q-1} - s_{i-q-1}) (\tilde{s}_{j-q-1} - s_{j-q-1}) K_T\left(\frac{\tilde{t} -t}{h}\right) K_S\left(\frac{\tilde{\bm{s}} - \bm{s}}{b}\right) p(\tilde{t},\tilde{\bm{s}}) \, d\tilde{t} d\tilde{\bm{s}}\\
		&= nhb^{d+2} \int_{\mathbb{R} \times \mathbb{R}^d} v_{i-q-1} \cdot v_{j-q-1} K_T(u) K_S(\bm{v}) \cdot p(t+uh,\bm{s}+b\bm{v})\, dud\bm{v}\\
		&= nhb^{d+2} \int_{\mathbb{R} \times \mathbb{R}^d} v_{i-q-1} \cdot v_{j-q-1} K_T(u) K_S(\bm{v}) \left[p(t,\bm{s}) + uh \cdot \frac{\partial}{\partial t} p(t,\bm{s}) + b\bm{v}^T \frac{\partial}{\partial \bm{s}} p(t,\bm{s}) + O\left(\max\{h,b\}^2\right)\right] dud\bm{v}\\
		&= nhb^{d+2} \left[\kappa_{2,i-q-1}^{(S)} \mathbbm{1}_{\{i=j\}} \cdot p(t,\bm{s}) + O\left(\max\{h,b\}^2\right) \right]
	\end{align*}
	and 
	\begin{align*}
		&\mathrm{Var}\left[\left(\bm{X}^T\bm{W} \bm{X}\right)_{i,j} \right]\\
		&= n \cdot \mathrm{Var}\left[(S_{1,i-q-1} -s_{i-q-1}) (S_{1,j-q-1} -s_{j-q-1}) K_T\left(\frac{T_1-t}{h}\right) K_S\left(\frac{\bm{S}_1 -\bm{s}}{b}\right) \right]\\
		&\leq n \int_{\mathbb{R} \times \mathbb{R}^d} (\tilde{s}_{i-q-1} - s_{i-q-1})^2 (\tilde{s}_{j-q-1} - s_{j-q-1})^2 K_T^2\left(\frac{\tilde{t} -t}{h}\right) K_S^2\left(\frac{\tilde{\bm{s}} -\bm{s}}{b}\right) p(\tilde{t},\tilde{\bm{s}}) \, d\tilde{t} d\tilde{\bm{s}}\\
		&= nhb^{d+4} \int_{\mathbb{R} \times \mathbb{R}^d} v_{i-q-1}^2 v_{j-q-1}^2 K_T^2(u) K_S^2(\bm{v}) \left[p(t,\bm{s}) + uh \cdot \frac{\partial}{\partial t} p(t,\bm{s}) + b\bm{v}^T \frac{\partial}{\partial \bm{s}} p(t,\bm{s}) + O\left(\max\{h,b\}^2\right)\right] dud\bm{v}\\
		&= nhb^{d+4} \left[\nu_{2,i-q-1}^{(S)} \nu_{2,j-q-1}^{(S)} \cdot p(t,\bm{s}) + O\left(\max\{h,b\}^2\right)\right].
	\end{align*}
	The above calculations imply that
	\begin{align*}
		\left(\bm{X}^T\bm{W} \bm{X}\right)_{i,j} &= \mathbb{E}\left[\left(\bm{X}^T\bm{W} \bm{X}\right)_{i,j} \right] + O_P\left(\sqrt{\mathrm{Var}\left[\left(\bm{X}^T\bm{W} \bm{X}\right)_{i,j} \right]}\right) \\
		&= nh b^{d+2} \left[\kappa_{2,i-q-1}^{(S)} \mathbbm{1}_{\{i=j\}} \cdot p(t,\bm{s}) + O\left(\max\{h,b\}^2\right) + O_P\left(\sqrt{\frac{1}{nhb^d}} \right)\right]
	\end{align*}
	for any $q+1 < i,j \leq q+1+d$.\\
	
	Therefore, we summarize all the above cases as:
	\begin{align*}
		&\left(\bm{X}^T\bm{W} \bm{X}\right)_{i,j} \\
		&= \begin{cases}
			nh^{i+j-1} b^d \left[\kappa_{i+j-2}^{(T)} \cdot p(t,\bm{s}) + h \cdot \kappa_{i+j-1}^{(T)} \cdot \frac{\partial}{\partial t} p(t,\bm{s}) + O\left(\max\{h,b\}^2\right) + O_P\left(\sqrt{\frac{1}{nhb^d}} \right)\right], \\
			\hspace{70mm}\quad 1\leq i,j \leq q+1,\\
			nh^i b^{d+1} \left[b\cdot \kappa_{i-1}^{(T)} \kappa_{2,j-q-1}^{(S)} \cdot \frac{\partial}{\partial s_{j-q-1}} p(t,\bm{s}) + O\left(\max\{h,b\}^2\right) + O_P\left(\sqrt{\frac{1}{nhb^d}} \right)\right], \\
			\hspace{70mm}\quad  1\leq i \leq q+1 \text{ and } q+1< j\leq q+1+d,\\
			nh^j b^{d+1} \left[b\cdot \kappa_{j-1}^{(T)} \kappa_{2,i-q-1}^{(S)} \cdot \frac{\partial}{\partial s_{i-q-1}} p(t,\bm{s}) + O\left(\max\{h,b\}^2\right) + O_P\left(\sqrt{\frac{1}{nhb^d}} \right)\right], \\
			\hspace{70mm}\quad  q+1< i\leq q+1+d \text{ and } 1\leq j \leq q+1,\\
			nh b^{d+2} \left[\kappa_{2,i-q-1}^{(S)} \mathbbm{1}_{\{i=j\}} \cdot p(t,\bm{s}) + O\left(\max\{h,b\}^2\right) + O_P\left(\sqrt{\frac{1}{nhb^d}} \right)\right], \quad q+1< i,j\leq q+1+d.
		\end{cases}
	\end{align*}
	Let $\bm{H} = \Diag(1,h,...,h^q,b,...,b) \in \mathbb{R}^{(q+1+d)\times (q+1+d)}$. We also denote 
	\begin{equation*}
		\bm{M}_q = 
		\begin{pmatrix}
			\left(\kappa_{i+j-2}^{(T)} \right)_{1\leq i,j\leq q+1} & \bm{0} \\
			\bm{0} & \left(\kappa_{2,i-q-1}^{(S)} \mathbbm{1}_{\{i=j\}} \right)_{q+1 < i,j \leq q+1+d}
		\end{pmatrix} \in \mathbb{R}^{(q+1+d) \times (q+1+d)}
	\end{equation*}
	and
	{\small\begin{align*}
			&\tilde{\bm{M}}_{q,h,b} \\
			&= 
			\begin{pmatrix}
				h\cdot \frac{\partial}{\partial t} p(t,\bm{s}) \left(\kappa_{i+j-1}^{(T)} \right)_{1\leq i,j\leq q+1} & b\cdot \Diag\left(\frac{\partial}{\partial \bm{s}} p(t,\bm{s}) \right)\left(\kappa_{i-1}^{(T)} \kappa_{2,j-q-1}^{(S)}\right)_{1\leq i \leq q+1, q+1< j\leq q+1+d} \\
				b\cdot \Diag\left(\frac{\partial}{\partial \bm{s}} p(t,\bm{s}) \right) \left(\kappa_{j-1}^{(T)} \kappa_{2,i-q-1}^{(S)}\right)_{q+1< i\leq q+1+d, 1\leq j \leq q+1} & \bm{0}
			\end{pmatrix} \\
			&\in \mathbb{R}^{(q+1+d) \times (q+1+d)}.
	\end{align*}}%
	Then, we can rewrite the asymptotic behaviors of $\bm{X}^T\bm{W}\bm{X}$ in its matrix form as:
	\begin{align}
		\label{X_W_mat_asymp}
		\begin{split}
			&\bm{X}^T\bm{W}\bm{X} = nhb^d \cdot p(t,\bm{s}) \cdot\bm{H} \left[\bm{M}_q + \frac{\tilde{\bm{M}}_{q,h,b}}{p(t,\bm{s})} + O\left(\max\{h,b\}^2\right) + O_P\left(\sqrt{\frac{1}{nhb^d}} \right)\right] \bm{H},
		\end{split}
	\end{align}
	where an abuse of notation is applied when we use $O\left(\max\{h,b\}^2\right) + O_P\left(\sqrt{\frac{1}{nhb}} \right)$ to denote a matrix whose entries are of this order. 
	By the matrix inversion formula 
	$$(A+ \max\{h,b\}\cdot B)^{-1} = A^{-1} - \max\{h,b\}\cdot A^{-1}BA^{-1} + O\left(\max\{h,b\}^2\right),$$ 
	we know that
	\begin{align*}
		&\left(\bm{X}^T\bm{W}\bm{X} \right)^{-1} \\
		&= \frac{1}{nhb^d \cdot p(t,\bm{s})} \cdot \bm{H}^{-1} \Bigg[\bm{M}_q^{-1} - \bm{M}_q^{-1}\cdot \frac{\tilde{\bm{M}}_{q,h,b}}{p(t,\bm{s})} \cdot \bm{M}_q^{-1} + O\left(\max\{h,b\}^2\right) + O_P\left(\sqrt{\frac{1}{nhb^d}} \right)\Bigg] \bm{H}^{-1}.
	\end{align*}
	
	\noindent {\bf Step 2: Conditional covariance term $\mathrm{Cov}\left[\left(\hat{\bm{\beta}}(t,\bm{s}), \hat{\bm{\alpha}}(t,\bm{s}) \right)^T \Big|\mathbb{X}\right]$.} Following our calculations for $\left(\bm{X}^T\bm{W}\bm{X}\right)_{i,j}$ for $1\leq i,j\leq q+1+d$, we can similarly derive that
	\begin{align*}
		&\left(\bm{X}^T\bm{\Sigma} \bm{X}\right)_{i,j} \\	
		&= \begin{cases}
			\sigma^2\sum\limits_{k=1}^n (T_k-t)^{i+j-2} K_T^2\left(\frac{T_k-t}{h}\right) K_S^2\left(\frac{\bm{S}_k -\bm{s}}{b}\right), \quad\quad 1\leq i,j \leq q+1,\\
			\sigma^2\sum\limits_{k=1}^n (T_k-t)^{i-1} \left(S_{k,j-q-1} - s_{j-q-1}\right) K_T^2\left(\frac{T_k-t}{h}\right) K_S^2\left(\frac{\bm{S}_k -\bm{s}}{b}\right), \quad 1\leq i \leq q+1 \text{ and } q+1< j\leq q+1+d,\\
			\sigma^2\sum\limits_{k=1}^n \left(S_{k,i-q-1} - s_{i-q-1}\right) (T_k-t)^{j-1} K_T^2\left(\frac{T_k-t}{h}\right) K_S^2\left(\frac{\bm{S}_k -\bm{s}}{b}\right), \quad q+1< i\leq q+1+d \text{ and } 1\leq j \leq q+1,\\
			\sigma^2\sum\limits_{k=1}^n \left(S_{k,i-q-1} - s_{i-q-1}\right) \left(S_{k,j-q-1} - s_{j-q-1}\right) K_T^2\left(\frac{T_k-t}{h}\right) K_S^2\left(\frac{\bm{S}_k -\bm{s}}{b}\right), \quad q+1< i,j\leq q+1+d,
		\end{cases}
	\end{align*}
	and its asymptotic behaviors become
	\begin{align*}
		&\left(\bm{X}^T\bm{\Sigma} \bm{X}\right)_{i,j} \\
		&= \begin{cases}
			\sigma^2 nh^{i+j-1} b^d \left[\nu_{i+j-2}^{(T)} \nu_0^{(S)} \cdot p(t,\bm{s}) + h \cdot \nu_{i+j-1}^{(T)} \nu_0^{(S)} \cdot \frac{\partial}{\partial t} p(t,\bm{s}) + O\left(\max\{h,b\}^2\right) + O_P\left(\sqrt{\frac{1}{nhb^d}} \right)\right], \\
			\quad \hspace{10cm} 1\leq i,j \leq q+1,\\
			\sigma^2 nh^i b^{d+1} \left[b\cdot \nu_{i-1}^{(T)} \nu_{2,j-q-1}^{(S)} \cdot \frac{\partial}{\partial s_{j-q-1}} p(t,\bm{s}) + O\left(\max\{h,b\}^2\right) + O_P\left(\sqrt{\frac{1}{nhb^d}} \right)\right], \\
			\quad\hspace{90mm}  1\leq i \leq q+1 \text{ and } q+1< j\leq q+1+d,\\
			\sigma^2 nh^j b^{d+1} \left[b\cdot \nu_{j-1}^{(T)} \nu_{2,i-q-1}^{(S)} \cdot \frac{\partial}{\partial s_{i-q-1}} p(t,\bm{s}) + O\left(\max\{h,b\}^2\right) + O_P\left(\sqrt{\frac{1}{nhb^d}} \right)\right], \\
			\quad\hspace{90mm}  q+1< i\leq q+1+d \text{ and } 1\leq j \leq q+1,\\
			\sigma^2 nh b^{d+2} \left[\nu_0^{(T)}\nu_{2,i-q-1}^{(S)} \mathbbm{1}_{\{i=j\}} \cdot p(t,\bm{s}) + O\left(\max\{h,b\}^2\right) + O_P\left(\sqrt{\frac{1}{nhb^d}} \right)\right], \quad q+1< i,j\leq q+1+d.
		\end{cases}
	\end{align*}
	If we denote 
	\[
	\bm{M}_q^* = 
	\begin{pmatrix}
		\left(\nu_{i+j-2}^{(T)} \nu_0^{(S)} \right)_{1\leq i,j\leq q+1} & \bm{0} \\
		\bm{0} & \left(\nu_0^{(T)}\nu_{2,i-q-1}^{(S)} \mathbbm{1}_{\{i=j\}} \right)_{q+1 < i,j \leq q+1+d}
	\end{pmatrix} \in \mathbb{R}^{(q+1+d) \times (q+1+d)}
	\]
	and
	{\small
		\begin{align*}
			&\tilde{\bm{M}}_{q,h,b}^* \\
			&= 
			\begin{pmatrix}
				h\cdot \frac{\partial}{\partial t} p(t,\bm{s})\cdot \left(\nu_{i+j-1}^{(T)} \nu_0^{(S)}\right)_{1\leq i,j\leq q+1} & b\cdot \Diag\left(\frac{\partial}{\partial \bm{s}} p(t,\bm{s}) \right) \left(\nu_{i-1}^{(T)} \nu_{2,j-q-1}^{(S)}\right)_{1\leq i \leq q+1, q+1< j\leq q+1+d} \\
				b\cdot \Diag\left(\frac{\partial}{\partial \bm{s}} p(t,\bm{s}) \right) \left(\nu_{j-1}^{(T)} \nu_{2,i-q-1}^{(S)}\right)_{q+1< i\leq q+1+d, 1\leq j \leq q+1} & \bm{0}
			\end{pmatrix} \\
			&\in \mathbb{R}^{(q+1+d) \times (q+1+d)},
	\end{align*}}%
	then the asymptotic behavior of $\bm{X}^T\bm{\Sigma}\bm{X}$ in its matrix form is
	\begin{align*}
		&\bm{X}^T\bm{\Sigma}\bm{X} = nhb^d \sigma^2 \cdot p(t,\bm{s}) \bm{H} \left[\bm{M}_q^* + \frac{\tilde{\bm{M}}_{q,h,b}^*}{p(t,\bm{s})} + O\left(\max\{h,b\}^2\right) + O_P\left(\sqrt{\frac{1}{nhb^d}} \right)\right] \bm{H}.
	\end{align*}
	With our results for $\bm{X}^T\bm{W}\bm{X}$ in \eqref{X_W_mat_asymp} and Assumption~\ref{assump:den_diff}, we conclude that the asymptotic conditional covariance matrix of $\left(\hat{\bm{\beta}}(t,\bm{s}), \hat{\bm{\alpha}}(t,\bm{s}) \right)^T \in \mathbb{R}^{q+1+d}$ is
	\begin{align*}
		&\mathrm{Cov}\left[\left(\hat{\bm{\beta}}(t,\bm{s}), \hat{\bm{\alpha}}(t,\bm{s}) \right)^T \Big| \mathbb{X}\right] \\
		&= \left(\bm{X}^T\bm{W} \bm{X} \right)^{-1} \bm{X}^T \bm{\Sigma} \bm{X} \left(\bm{X}^T\bm{W} \bm{X} \right)^{-1}\\
		&= \frac{\sigma^2}{nhb^d \cdot p(t,\bm{s})} \cdot \bm{H}^{-1} \Bigg[\bm{M}_q^{-1} - \bm{M}_q^{-1}\cdot \frac{\tilde{\bm{M}}_{q,h,b}}{p(t,\bm{s})} \cdot \bm{M}_q^{-1} + O\left(\max\{h,b\}^2\right) + O_P\left(\sqrt{\frac{1}{nhb^d}} \right)\Bigg]\\
		&\quad \times \left[\bm{M}_q^* + \tilde{\bm{M}}_{q,h,b}^* + O\left(\max\{h,b\}^2\right) + O_P\left(\sqrt{\frac{1}{nhb^d}} \right)\right] \\
		&\quad \times \left[\bm{M}_q^{-1} - \bm{M}_q^{-1} \cdot \frac{\tilde{\bm{M}}_{q,h,b}}{p(t,\bm{s})} \cdot \bm{M}_q^{-1} + O\left(\max\{h,b\}^2\right) + O_P\left(\sqrt{\frac{1}{nhb^d}} \right)\right] \bm{H}^{-1}\\
		&= \frac{\sigma^2}{nhb^d \cdot p(t,\bm{s})} \cdot \bm{H}^{-1} \left[\bm{M}_q^{-1} \bm{M}_q^* \bm{M}_q^{-1} + O\left(\max\{h,b\}\right) + O_P\left(\sqrt{\frac{1}{nhb^d}}\right)\right] \bm{H}^{-1}.
	\end{align*}
	In particular, we know that the asymptotic conditional variance of $\hat{\beta}_2(t,\bm{s})$ is given by
	\begin{align*}
		\mathrm{Var}\left[\hat{\beta}_2(t,\bm{s}) \big| \mathbb{X} \right] &= \frac{\sigma^2}{nh^3 b^d \cdot p(t,\bm{s})} \left[\bm{e}_2^T \bm{M}_q^{-1} \bm{M}_q^* \bm{M}_q^{-1} \bm{e}_2 + O\left(\max\{h,b\}\right) + O_P\left(\sqrt{\frac{1}{nhb^d}}\right) \right],
	\end{align*}
	where $\left\{\bm{e}_1,...,\bm{e}_{q+1+d} \right\}$ is the standard basis in $\mathbb{R}^{q+1+d}$. Since the only entries in $\bm{M}_q,\bm{M}_q^*$ that affects $\mathrm{Var}\left[\hat{\beta}_2(t,\bm{s}) \big| \mathbb{X} \right]$ are those in the first $(q+1+d)$ rows and columns, the results follow by restricting $\bm{M}_q,\bm{M}_q^*$ to their first $(q+1)\times (q+1)$ block matrices and taking $\bm{e}_2=(0,1,0,...,0)^T \in \mathbb{R}^{q+1}$.\\
	
	\noindent {\bf Step 3: Conditional bias term $\mathrm{Bias}\left[\left(\hat{\bm{\beta}}(t,\bm{s}), \hat{\bm{\alpha}}(t,\bm{s}) \right)^T \Big| \mathbb{X}\right]$.} By Assumption~\ref{assump:reg_diff} and Taylor's expansion, we have that
	\begin{equation}
		\label{taylor_reg}
		\mu(T,\bm{S}) = \mu(t,\bm{s}) + \sum_{j=1}^q \frac{1}{q!} \cdot \frac{\partial^j}{\partial t^j} \mu(t,\bm{s}) \cdot (T-t)^j + \sum_{\ell=1}^d \frac{\partial}{\partial s_{\ell}} \mu(t,\bm{s}) \cdot (S_{\ell} -s_{\ell}) + r(T,\bm{S}),
	\end{equation}
	where the reminder term $r(T,\bm{S})$ is given by
	\begin{align*}
		r(T,\bm{S}) &= \frac{(T-t)^{q+1}}{(q+1)!} \left[ \frac{\partial^{q+1}}{\partial t^{q+1}} \mu(t,\bm{s})\right] + \sum_{\ell=1}^d \left[\frac{\partial^2}{\partial t \partial s_{\ell}} \mu(t,\bm{s}) \right] (T-t) (S_{\ell}-s_{\ell}) \\
		&\quad + \frac{1}{2}(\bm{S}-\bm{s})^T \left[\frac{\partial^2}{\partial \bm{s} \partial \bm{s}^T} \mu(t,\bm{s})\right] (\bm{S} -\bm{s}) + o\left(|T-t|^{q+1} + |T-t| \norm{\bm{S}-\bm{s}}_2 + \norm{\bm{S} - \bm{s}}_2^2\right).
	\end{align*}
	Thus, the conditional bias can be written as:
	\begin{align*}
		\mathrm{Bias}\left[\left(\hat{\bm{\beta}}(t,\bm{s}), \hat{\bm{\alpha}}(t,\bm{s}) \right)^T \Big| \mathbb{X}\right] &= \left(\bm{X}^T\bm{W} \bm{X} \right)^{-1} \bm{X}^T \bm{W} \left[ 
		\begin{pmatrix}
			\mu(T_1,\bm{S}_1)\\
			\vdots\\
			\mu(T_n,\bm{S}_n)
		\end{pmatrix} - \bm{X} \begin{pmatrix}
			\bm{\beta}(t,\bm{s})\\
			\bm{\alpha}(t,\bm{s})
		\end{pmatrix}\right]\\
		&= \left(\bm{X}^T\bm{W} \bm{X} \right)^{-1} \bm{X}^T \bm{W} 
		\begin{pmatrix}
			r(T_1,\bm{S}_1)\\
			\vdots\\
			r(T_n,\bm{S}_n)
		\end{pmatrix}.
	\end{align*}
	Now, we note that
	{\small \begin{align*}
			&\bm{X}^T\bm{W}\begin{pmatrix}
				r(T_1,\bm{S}_1)\\
				\vdots\\
				r(T_n,\bm{S}_n)
			\end{pmatrix} \\
			&= \begin{bmatrix}
				\begin{pmatrix}
					\sum\limits_{k=1}^n \left[\frac{\partial^{q+1}}{\partial t^{q+1}} \mu(t,\bm{s}) \right] \frac{(T_k-t)^{q+1}}{(q+1)!} \cdot K_T\left(\frac{T_k-t}{h}\right) K_S\left(\frac{\bm{S}_k-\bm{s}}{b}\right) \\
					+ \sum\limits_{k=1}^n \sum\limits_{\ell=1}^d \left[\frac{\partial^2}{\partial t\partial s_{\ell}} \mu(t,\bm{s}) \right] (T_k-t) \big(S_{k,\ell} - s_{\ell}\big) K_T\left(\frac{T_k-t}{h}\right) K_S\left(\frac{\bm{S}_k-\bm{s}}{b}\right)\\
					+ \sum\limits_{k=1}^n (T_k-t)^{j-1} \cdot \frac{1}{2} (\bm{S}_k-\bm{s})^T \left[\frac{\partial^2}{\partial \bm{s} \partial \bm{s}^T} \mu(t,\bm{s}) \right] (\bm{S}_k -\bm{s}) K_T\left(\frac{T_k-t}{h}\right) K_S\left(\frac{\bm{S}_k-\bm{s}}{b}\right)
				\end{pmatrix}_{1\leq j \leq q+1} \\
				\begin{pmatrix}
					\sum\limits_{k=1}^n \left[\frac{\partial^{q+1}}{\partial t^{q+1}} \mu(t,\bm{s}) \right] \frac{(T_k-t)^{q+1}}{(q+1)!} \cdot (S_{k,j-q-1} - s_{j-q-1}) \cdot K_T\left(\frac{T_k-t}{h}\right) K_S\left(\frac{\bm{S}_k-\bm{s}}{b}\right) \\
					+ \sum\limits_{k=1}^n \sum\limits_{\ell=1}^d \left[\frac{\partial^2}{\partial t\partial s_{\ell}} \mu(t,\bm{s}) \right] (T_k-t) (S_{k,\ell} - s_{\ell}) (S_{k,j-q-1} - s_{j-q-1}) K_T\left(\frac{T_k-t}{h}\right) K_S\left(\frac{\bm{S}_k-\bm{s}}{b}\right)\\
					+ \sum\limits_{k=1}^n (S_{k,j-q-1} -s_{j-q-1}) \cdot \frac{1}{2} (\bm{S}_k-\bm{s})^T \left[\frac{\partial^2}{\partial \bm{s} \partial \bm{s}^T} \mu(t,\bm{s}) \right] (\bm{S}_k -\bm{s}) K_T\left(\frac{T_k-t}{h}\right) K_S\left(\frac{\bm{S}_k-\bm{s}}{b}\right)
				\end{pmatrix}_{q+1 < j \leq q+1+d}
			\end{bmatrix}\\
			&= \begin{bmatrix}
				\begin{pmatrix}
					nh^{q+1+j} b^d \left[\frac{\partial^{q+1}}{\partial t^{q+1}} \mu(t,\bm{s}) \right] \frac{1}{(q+1)!} \left[\kappa_{q+j}^{(T)} p(t,\bm{s}) + \kappa_{q+j+1}^{(T)} h \frac{\partial}{\partial t} p(t,\bm{s}) + O\left(\max\{h,b\}^2\right) +O_P\left(\sqrt{\frac{1}{nhb^d}}\right)\right]\\
					+ nh^{j+1} b^{d+1} \left[b\cdot \sum_{\ell=1}^d \left[\frac{\partial^2}{\partial t \partial s_{\ell}} \mu(t,\bm{s})\right] \kappa_j^{(T)} \kappa_{2,\ell}^{(S)} \cdot \frac{\partial}{\partial s_{\ell}} p(t,\bm{s}) + O\left(\max\{h,b\}^2\right) +O_P\left(\sqrt{\frac{1}{nhb^d}}\right)\right]\\
					+ nh^jb^{d+2} \left[\kappa_{j-1}^{(T)} p(t,\bm{s}) \sum_{\ell=1}^d \frac{\kappa_{2,\ell}^{(S)}}{2} \left(\frac{\partial^2}{\partial s_{\ell}^2} \mu(t,\bm{s})\right) + O\left(\max\{h,b\}^2\right) +O_P\left(\sqrt{\frac{1}{nhb^d}}\right) \right]
				\end{pmatrix}_{1\leq j \leq q+1} \\
				\begin{pmatrix}
					nh^{q+2} b^{d+1} \left[\frac{\partial^{q+1}}{\partial t^{q+1}} \mu(t,\bm{s}) \right] \frac{1}{(q+1)!} \left[b\kappa_{q+1}^{(T)} \kappa_{2,j-q-1}^{(S)} \frac{\partial}{\partial s_{j-q-1}} p(t,\bm{s}) + O\left(\max\{h,b\}^2\right) +O_P\left(\sqrt{\frac{1}{nhb^d}}\right)\right]\\
					+ nh^2 b^{d+2} \left[\frac{\partial^2}{\partial t \partial s_{j-q-1}} \mu(t,\bm{s})\right] \left[h \kappa_2^{(T)} \kappa_{2,j-q-1}^{(S)} \cdot \frac{\partial}{\partial t} p(t,\bm{s}) + O\left(\max\{h,b\}^2\right) +O_P\left(\sqrt{\frac{1}{nhb^d}}\right)\right]\\
					+ nhb^{d+3} \left[b \sum_{\ell=1}^d \frac{\kappa_{2,j-q-1,\ell}^{(S)}}{2} \left(\frac{\partial^2}{\partial s_{\ell}^2} \mu(t,\bm{s})\right) \cdot \frac{\partial}{\partial s_{j-q-1}} p(t,\bm{s}) + O\left(\max\{h,b\}^2\right) +O_P\left(\sqrt{\frac{1}{nhb^d}}\right) \right]
				\end{pmatrix}_{q+1 < j \leq q+1+d}
			\end{bmatrix}\\
			&=nhb^d \bm{H} \begin{bmatrix}
				\begin{pmatrix}
					h^{q+1} \left[\frac{\partial^{q+1}}{\partial t^{q+1}} \mu(t,\bm{s}) \right] \frac{1}{(q+1)!} \left[\kappa_{q+j}^{(T)} p(t,\bm{s}) + \kappa_{q+j+1}^{(T)} h \frac{\partial}{\partial t} p(t,\bm{s}) + O\left(\max\{h,b\}^2\right) +O_P\left(\sqrt{\frac{1}{nhb^d}}\right)\right]\\
					+ hb \left[b\cdot \sum_{\ell=1}^d \left(\frac{\partial^2}{\partial t \partial s_{\ell}} \mu(t,\bm{s})\right) \kappa_j^{(T)} \kappa_{2,\ell}^{(S)} \cdot \frac{\partial}{\partial s_{\ell}} p(t,\bm{s}) + O\left(\max\{h,b\}^2\right) +O_P\left(\sqrt{\frac{1}{nhb^d}}\right)\right]\\
					+ b^2 \left[\kappa_{j-1}^{(T)} p(t,\bm{s}) \sum_{\ell=1}^d \frac{\kappa_{2,\ell}^{(S)}}{2} \left(\frac{\partial^2}{\partial s_{\ell}^2} \mu(t,\bm{s})\right) + O\left(\max\{h,b\}^2\right) +O_P\left(\sqrt{\frac{1}{nhb^d}}\right) \right]
				\end{pmatrix}_{1\leq j \leq q+1} \\
				\begin{pmatrix}
					h^{q+1} \left[\frac{\partial^{q+1}}{\partial t^{q+1}} \mu(t,\bm{s}) \right] \frac{1}{(q+1)!} \left[b\kappa_{q+1}^{(T)} \kappa_{2,j-q-1}^{(S)} \frac{\partial}{\partial s_{j-q-1}} p(t,\bm{s}) + O\left(\max\{h,b\}^2\right) +O_P\left(\sqrt{\frac{1}{nhb^d}}\right)\right]\\
					+ hb \left[\frac{\partial^2}{\partial t \partial s_{j-q-1}} \mu(t,\bm{s})\right] \left[h \kappa_2^{(T)} \kappa_{2,j-q-1}^{(S)} \cdot \frac{\partial}{\partial t} p(t,\bm{s}) + O\left(\max\{h,b\}^2\right) +O_P\left(\sqrt{\frac{1}{nhb^d}}\right)\right]\\
					+ b^2 \left[b \sum_{\ell=1}^d \frac{\kappa_{2,j-q-1,\ell}^{(S)}}{2} \left(\frac{\partial^2}{\partial s_{\ell}^2} \mu(t,\bm{s})\right) \cdot \frac{\partial}{\partial s_{j-q-1}} p(t,\bm{s}) + O\left(\max\{h,b\}^2\right) +O_P\left(\sqrt{\frac{1}{nhb^d}}\right) \right]
				\end{pmatrix}_{q+1 < j \leq q+1+d}
			\end{bmatrix}.
	\end{align*}}%
	Therefore, we know from the above display and \eqref{X_W_mat_asymp} that
	\begin{align*}
		&\mathrm{Bias}\left[\left(\hat{\bm{\beta}}(t,\bm{s}), \hat{\bm{\alpha}}(t,\bm{s}) \right)^T \Big| \mathbb{X}\right] \\
		&= \left(\bm{X}^T\bm{W}\bm{X}\right)^{-1}\bm{X}^T\bm{W}\begin{pmatrix}
			r(T_1,\bm{S}_1)\\
			\vdots\\
			r(T_n,\bm{S}_n)
		\end{pmatrix}\\
		&= \frac{1}{p(t,\bm{s})} \cdot \bm{H}^{-1} \Bigg[\bm{M}_q^{-1} + O\left(\max\{h,b\}\right) + O_P\left(\sqrt{\frac{1}{nhb^d}} \right)\Bigg]\\
		&\quad \times \begin{bmatrix}
			\begin{pmatrix}
				h^{q+1} \left[\frac{\partial^{q+1}}{\partial t^{q+1}} \mu(t,\bm{s}) \right] \frac{1}{(q+1)!} \left[\kappa_{q+j}^{(T)} p(t,\bm{s}) + \kappa_{q+j+1}^{(T)} h \frac{\partial}{\partial t} p(t,\bm{s}) + O\left(\max\{h,b\}^2\right) +O_P\left(\sqrt{\frac{1}{nhb^d}}\right)\right]\\
				+ hb \left[b\cdot \sum_{\ell=1}^d \left[\frac{\partial^2}{\partial t \partial s_{\ell}} \mu(t,\bm{s})\right] \kappa_j^{(T)} \kappa_{2,\ell}^{(S)} \cdot \frac{\partial}{\partial s_{\ell}} p(t,\bm{s}) + O\left(\max\{h,b\}^2\right) +O_P\left(\sqrt{\frac{1}{nhb^d}}\right)\right]\\
				+ b^2 \left[\kappa_{j-1}^{(T)} p(t,\bm{s}) \sum_{\ell=1}^d \frac{\kappa_{2,\ell}^{(S)}}{2} \left(\frac{\partial^2}{\partial s_{\ell}^2} \mu(t,\bm{s})\right) + O\left(\max\{h,b\}^2\right) +O_P\left(\sqrt{\frac{1}{nhb^d}}\right) \right]
			\end{pmatrix}_{1\leq j \leq q+1} \\
			\begin{pmatrix}
				h^{q+1} \left[\frac{\partial^{q+1}}{\partial t^{q+1}} \mu(t,\bm{s}) \right] \frac{1}{(q+1)!} \left[b\kappa_{q+1}^{(T)} \kappa_{2,j-q-1}^{(S)} \frac{\partial}{\partial s_{j-q-1}} p(t,\bm{s}) + O\left(\max\{h,b\}^2\right) +O_P\left(\sqrt{\frac{1}{nhb^d}}\right)\right]\\
				+ hb \left[\frac{\partial^2}{\partial t \partial s_{j-q-1}} \mu(t,\bm{s})\right] \left[h \kappa_2^{(T)} \kappa_{2,j-q-1}^{(S)} \cdot \frac{\partial}{\partial t} p(t,\bm{s}) + O\left(\max\{h,b\}^2\right) +O_P\left(\sqrt{\frac{1}{nhb^d}}\right)\right]\\
				+ b^2 \left[b \sum_{\ell=1}^d \frac{\kappa_{2,j-q-1,\ell}^{(S)}}{2} \left(\frac{\partial^2}{\partial s_{\ell}^2} \mu(t,\bm{s})\right) \cdot \frac{\partial}{\partial s_{j-q-1}} p(t,\bm{s}) + O\left(\max\{h,b\}^2\right) +O_P\left(\sqrt{\frac{1}{nhb^d}}\right) \right]
			\end{pmatrix}_{q+1 < j \leq q+1+d}
		\end{bmatrix}\\
		&=\frac{1}{p(t,\bm{s})} \bm{H}^{-1} \left[\bm{M}_q^{-1} + O\left(\max\{h,b\}\right) + O_P\left(\sqrt{\frac{1}{nhb^d}} \right)\right]\\
		&\quad \times \left[h^{q+1} \bm{\tau}_q + hb \bm{\tau}_q^* + b^2 \tilde{\bm{\tau}}_q + O\left(\max\{h,b\}^4\right) + o_P\left(\sqrt{\frac{1}{nhb^d}}\right) \right],
	\end{align*}
	where
	\begin{align*}
		\bm{\tau}_q = \frac{\left[\frac{\partial^{q+1}}{\partial t^{q+1}} \mu(t,\bm{s}) \right]}{(q+1)!} 
		\begin{bmatrix}
			\begin{pmatrix}
				\kappa_{q+j}^{(T)} \cdot p(t,\bm{s}) + h\cdot \kappa_{q+j+1}^{(T)} \cdot \frac{\partial}{\partial t} p(t,\bm{s})
			\end{pmatrix}_{1 \leq j \leq q+1}\\
			\begin{pmatrix}
				b\kappa_{q+1}^{(T)} \kappa_{2,j-q-1}^{(S)} \cdot \frac{\partial}{\partial s_{j-q-1}} p(t,\bm{s})
			\end{pmatrix}_{q+1 < j \leq q+1+d}
		\end{bmatrix} \in \mathbb{R}^{(q+1+d)\times (q+1+d)},
	\end{align*}
	\begin{align*}
		\bm{\tau}_q^* = 
		\begin{bmatrix}
			\begin{pmatrix}
				b\sum_{\ell=1}^d \left[\frac{\partial^2}{\partial t \partial s_{\ell}} \mu(t,\bm{s})\right] \kappa_j^{(T)} \kappa_{2,\ell}^{(S)} \cdot \frac{\partial}{\partial s_{\ell}} p(t,\bm{s})
			\end{pmatrix}_{1\leq j \leq q+1}\\
			\begin{pmatrix}
				h\left[\frac{\partial^2}{\partial t \partial s_{j-q-1}} \mu(t,\bm{s})\right] \kappa_2^{(T)} \kappa_{2,j-q-1}^{(S)} \cdot \frac{\partial}{\partial t} p(t,\bm{s})
			\end{pmatrix}_{q+1 < j \leq q+1+d}
		\end{bmatrix} \in \mathbb{R}^{(q+1+d)\times (q+1+d)},
	\end{align*}
	and
	\begin{align*}
		\tilde{\bm{\tau}}_q = 
		\begin{bmatrix}
			\begin{pmatrix}
				\kappa_{j-1}^{(T)} p(t,\bm{s}) \sum_{\ell=1}^d \frac{\kappa_{2,\ell}^{(S)}}{2} \left(\frac{\partial^2}{\partial s_{\ell}^2} \mu(t,\bm{s})\right)
			\end{pmatrix}_{1\leq j \leq q+1}\\
			\begin{pmatrix}
				b \sum_{\ell=1}^d \frac{\kappa_{2,j-q-1,\ell}^{(S)}}{2} \left[\frac{\partial^2}{\partial s_{\ell}^2} \mu(t,\bm{s})\right] \cdot \frac{\partial}{\partial s_{j-q-1}} p(t,\bm{s})
			\end{pmatrix}_{q+1 < j \leq q+1+d}
		\end{bmatrix} \in \mathbb{R}^{(q+1+d)\times (q+1+d)}.
	\end{align*}
	Since $\kappa_{2j-1}^{(T)} = 0$ for all $j=1,2,...,q+1$ by Assumption~\ref{assump:reg_kernel}(a), we know that
	\begin{align}
		\label{M_q_inter_zero}
		\begin{split}
			\bm{M}_q &= 
			\begin{pmatrix}
				\left(\kappa_{i+j-2}^{(T)} \right)_{1\leq i,j\leq q+1} & \bm{0} \\
				\bm{0} & \left(\kappa_{2,i-q-1}^{(S)} \mathbbm{1}_{\{i=j\}} \right)_{q+1 < i,j \leq q+1+d}
			\end{pmatrix} \\
			&= 
			\begin{pmatrix}
				\kappa_0^{(T)} & 0 & \kappa_2^{(T)} & 0 & \cdots & \kappa_q^{(T)} & \\
				0 & \kappa_2^{(T)} & 0 & \kappa_4^{(T)} & \cdots & \kappa_{q+1}^{(T)} & \\
				\vdots & \multicolumn{4}{c}{\ddots} & \vdots & \bm{0}\\
				\kappa_q^{(T)} & \multicolumn{4}{c}{\cdots} & \kappa_{2q}^{(T)} &\\
				& & \bm{0} & & & &\left(\kappa_{2,i-q-1}^{(S)} \mathbbm{1}_{\{i=j\}} \right)_{q+1 < i,j \leq q+1+d}
			\end{pmatrix},
		\end{split}
	\end{align}
	\emph{i.e.}, each row/column of the upper $(q+1)\times (q+1)$ submatrix of $M_q$ has interleaving nonzero and zero entries. More importantly, given that 
	$$\bm{M}_q^{-1} = \frac{1}{\det(\bm{M}_q)} \cdot \mathrm{adj}(\bm{M}_q)$$ 
	with $\mathrm{adj}(\bm{M}_q) \in \mathbb{R}^{(q+1+d)\times (q+1+d)}$ being the adjugate of $\bm{M}_q$, it can be shown that $\bm{M}_q^{-1}$ has the identical sparsity structures as $\bm{M}_q$. In more details, any $(i,j)$-minor of $\bm{M}_q$ with $i+j$ being odd has two rows/columns that share an identical sparsity structure. Recalling the formula of calculating a determinant 
	$$\det(A) = \sum_{\sigma_q}\left(\text{sign}(\sigma_q) \prod_{i=1}^{q+1+d} A_{i\sigma_q(i)}\right),$$ 
	where the sum is over all $(q+1+d)!$ permutations $\sigma_q$, we know that each summand is zero when we compute the determinant of $(i,j)$-minor of $\bm{M}_q$ with $i+j$ being odd through the above formula. Hence, the determinant of $(i,j)$-minor of $\bm{M}_q$ with $i+j$ being odd is zero. A similar argument can be found in the proof of Theorem 2 in \cite{fan1996study}.
	
	Using this result, we derive the bias term $\mathrm{Bias}\left[\hat{\beta}_2(t,\bm{s}) \big| \mathbb{X} \right]$ when $q$ is either odd or even (recommended) as follows.
	
	\noindent $\bullet$ {\bf Case I: $q$ is odd.} Then, 
	$$\bm{e}_2^T \bm{M}_q^{-1} = \left(\underbrace{0,\star, 0, \star,...,0, \star}_{1\leq j\leq q+1}, \underbrace{\bm{0}}_{q+1 < j \leq q+1+d}\right) \in \mathbb{R}^{q+1+d},$$
	where $\star \in \mathbb{R}$ stands for some nonzero bounded value (not necessarily equal for each entry). This implies that
	\begin{align*}
		&\bm{e}_2^T\bm{M}_q^{-1} \bm{\tau}_q \\
		&= \left(\underbrace{0,\star, 0, \star,...,0, \star}_{1\leq j\leq q+1}, \underbrace{\bm{0}}_{q+1 < j \leq q+1+d}\right) 
		\frac{\left[\frac{\partial^{q+1}}{\partial t^{q+1}} \mu(t,\bm{s}) \right]}{(q+1)!} 
		\begin{bmatrix}
			\begin{pmatrix}
				\kappa_{q+1}^{(T)} \cdot p(t,\bm{s})\\
				h\cdot \kappa_{q+3}^{(T)} \cdot \frac{\partial}{\partial t} p(t,\bm{s})\\
				\vdots\\
				h\cdot \kappa_{2q+2}^{(T)} \cdot \frac{\partial}{\partial t} p(t,\bm{s})
			\end{pmatrix}_{1 \leq j \leq q+1}\\
			\begin{pmatrix}
				b\cdot \kappa_{q+1}^{(T)} \kappa_{2,j-q-1}^{(S)} \cdot \frac{\partial}{\partial s_{j-q-1}} p(t,\bm{s})
			\end{pmatrix}_{q+1 < j \leq q+1+d}
		\end{bmatrix}\\
		&= h \cdot \tau_q^{\text{odd}},
	\end{align*}
	where $\tau_q^{\text{odd}} = \frac{\bm{e}_2^T\bm{M}_q^{-1} \bm{\tau}_q}{h}$ is a dominating term independent of $h,b$. Similarly, we know that
	$$\bm{e}_2^T\bm{M}_q^{-1} \bm{\tau}_q^* = b \cdot \tau_q^*, \quad \bm{e}_2^T\bm{M}_q^{-1} \tilde{\tau}_q = 0,$$
	where $\tau_q^* = \frac{\bm{e}_2^T\bm{M}_q^{-1} \bm{\tau}_q^*}{b}$ is a constant independent of $b,h$. Hence, when $q$ is odd and $h,b\to 0, nhb^d \to \infty$, we obtain that the asymptotic conditional bias of $\hat{\beta}_2(t,\bm{s})$ is given by
	\begin{align}
		\label{bias_term1}
		\begin{split}
			\mathrm{Bias}\left[\hat{\beta}_2(t,\bm{s}) \big| \mathbb{X} \right] &= \frac{1}{h\cdot p(t,\bm{s})} \left[h^{q+1} \bm{e}_2^T\bm{M}_q^{-1}\bm{\tau}_q + hb \bm{e}_2^T\bm{M}_q^{-1}\bm{\tau}_q^* + b^2 \bm{e}_2^T\bm{M}_q^{-1}\tilde{\bm{\tau}}_q + o_P\left(\sqrt{\frac{1}{nhb^d}}\right) \right]\\
			&= \frac{1}{p(t,\bm{s})} \left[h^{q+1} \tau_q^{\text{odd}} + b^2 \tau_q^* + O\left(\frac{b^4}{h}\right) + o_P\left(\sqrt{\frac{1}{nhb^d}}\right) \right].
		\end{split}
	\end{align}
	
	\noindent $\bullet$ {\bf Case II: $q$ is even.} Then,
	$$\bm{e}_2^T \bm{M}_q^{-1} = \left(\underbrace{0,\star, 0, \star,..., 0, \star,0}_{1\leq j\leq q+1}, \underbrace{\bm{0}}_{q+1 < j \leq q+1+d}\right) \in \mathbb{R}^{q+1+d},$$
	where $\star \in \mathbb{R}$ stands for some nonzero bounded value (not necessarily equal for each entry). This indicates that
	\begin{align*}
		&\bm{e}_2^T\bm{M}_q^{-1} \bm{\tau}_q \\
		&= \left(\underbrace{0,\star, 0, \star,...,0,\star,0}_{1\leq j\leq q+1}, \underbrace{\bm{0}}_{q+1 < j \leq q+1+d}\right) 
		\left[\frac{\partial^{q+1}}{\partial t^{q+1}} \mu(t,\bm{s}) \right] \frac{1}{(q+1)!} 
		\begin{bmatrix}
			\begin{pmatrix}
				h\cdot \kappa_{q+2}^{(T)} \cdot \frac{\partial}{\partial t} p(t,\bm{s})\\
				\kappa_{q+2}^{(T)} \cdot p(t,\bm{s})\\
				\vdots\\
				h\cdot \kappa_{2q+2}^{(T)} \cdot \frac{\partial}{\partial t} p(t,\bm{s})
			\end{pmatrix}_{1 \leq j \leq q+1}\\
			\bm{0}_{q+1 < j \leq q+1+d}
		\end{bmatrix}\\
		&= \tau_q^{\text{even}},
	\end{align*}
	where $\tau_q^{\text{even}} = \bm{e}_2^T\bm{M}_q^{-1} \bm{\tau}_q$ is a dominating constant independent of $h,b$. Similarly, we know that
	$$\bm{e}_2^T\bm{M}_q^{-1} \bm{\tau}_q^* = b \cdot \tau_q^*, \quad \bm{e}_2^T\bm{M}_q^{-1} \tilde{\tau}_q = 0,$$
	where $\tau_q^* = \frac{\bm{e}_2^T\bm{M}_q^{-1} \bm{\tau}_q^*}{b}$ is a constant independent of $b,h$. Hence, when $q$ is even and $h,b\to 0, nh b^d \to \infty$, we obtain that the asymptotic conditional bias of $\hat{\beta}_2(t,\bm{s})$ is given by
	\begin{align*}
		&\mathrm{Bias}\left[\hat{\beta}_2(t,\bm{s}) \big| \mathbb{X} \right] \\
		&= \frac{1}{h\cdot p(t,\bm{s})} \left[h^{q+1} \bm{e}_2^T\bm{M}_q^{-1}\bm{\tau}_q + hb \bm{e}_2^T\bm{M}_q^{-1}\bm{\tau}_q^* + b^2 \bm{e}_2^T\bm{M}_q^{-1}\tilde{\bm{\tau}}_q + O\left(\max\{h,b\}^4\right) + o_P\left(\sqrt{\frac{1}{nhb^d}}\right) \right]\\
		&= \frac{1}{p(t,\bm{s})} \left[h^q \tau_q^{\text{even}} + b^2 \tau_q^* + O\left(\frac{b^4}{h}\right) + o_P\left(\sqrt{\frac{1}{nhb^d}}\right) \right].
	\end{align*}
	
	In summary, as $h,b\to 0$ and $nh^3b^d \to \infty$, we know that
	\begin{align*}
		\mathrm{Bias}\left[\hat{\beta}_2(t,\bm{s}) \big| \mathbb{X}\right] =
		\begin{cases}
			\frac{1}{p(t,\bm{s})} \left[h^{q+1} \tau_q^{\text{odd}} + b^2 \tau_q^* + O\left(\frac{b^4}{h}\right) + o_P\left(\sqrt{\frac{1}{nhb^d}}\right) \right] & q \text{ is odd},\\
			\frac{1}{p(t,\bm{s})} \left[h^q \tau_q^{\text{even}} + b^2 \tau_q^* + O\left(\frac{b^4}{h}\right) + o_P\left(\sqrt{\frac{1}{nhb^d}}\right) \right] & q \text{ is even}.
		\end{cases}
	\end{align*}
	\vspace{3mm}
	
	Now, we consider the case when $(t,\bm{s})$ lies in the boundary region of $\mathcal{E}$. In this case, those dominating constants in $\mathrm{Var}\left[\hat{\beta}_2(t,\bm{s}) \big|\mathbb{X}\right]$ and $\mathrm{Bias}\left[\hat{\beta}_2(t,\bm{s}) \big| \mathbb{X}\right]$ would be different but the asymptotic rates remain the same due to Assumption~\ref{assump:boundary_cond}(a). In more details, $\kappa_j^{(T)}, \nu_j^{(T)}, \kappa_{j,\ell}^{(S)}, \nu_{j,k}^{(S)}$ are now defined as:
	\begin{align}
		\label{kernel_terms}
		\begin{split}
			&\kappa_j^{(T)} = \int_{\left\{u\in \mathbb{R}: (t+hu,\bm{s}+b\bm{v})\in \mathcal{E}\text{ for some } \bm{v}\in \mathbb{R}^d\right\}} u^j K_T(u) \, du < \infty,\\ 
			&\nu_j^{(T)} = \int_{\left\{u\in \mathbb{R}: (t+hu,\bm{s}+b\bm{v})\in \mathcal{E}\text{ for some } \bm{v}\in \mathbb{R}^d\right\}} u^j K_T^2(u) \, du < \infty, \\ 
			&\kappa_{j,\ell}^{(S)} = \int_{\left\{\bm{v}\in \mathbb{R}^d: (t+hu,\bm{s}+b\bm{v})\in \mathcal{E}\text{ for some } u\in \mathbb{R}\right\}} u_{\ell}^j K_S(\bm{u})\, d\bm{u} < \infty, \\ 
			& \text{ and } \quad \nu_{j,k}^{(S)} = \int_{\left\{\bm{v}\in \mathbb{R}^d: (t+hu,\bm{s}+b\bm{v})\in \mathcal{E}\text{ for some } u\in \mathbb{R}\right\}} u_{\ell}^j K_S^2(\bm{u})\, d\bm{u} < \infty.
		\end{split}
	\end{align}
	These terms are again bounded even when $h,b\to 0$. Under the above new definitions, we can redefine $\bm{M}_q, \tilde{\bm{M}}_{q,h,b}, \bm{M}_q^*, \tilde{\bm{M}}_{q,h,b}^*, \bm{\tau}_q, \bm{\tau}_q^*, \tilde{\bm{\tau}}_q$ accordingly, where those partial derivatives of $p(t,\bm{s})$ and the corresponding Taylor's expansion are defined at a nearby interior point by Assumption~\ref{assump:boundary_cond}(a) and taking the limit to the boundary $\partial \mathcal{E}$. More importantly, the matrix $\bm{M}_q$ remains non-singular as $h,b\to 0$ due to Assumption~\ref{assump:boundary_cond}(a) and Lemma 7.1 in \cite{fan2015multivariate}. A similar argument can also be found in Theorem 2.2 in \cite{ruppert1994multivariate}.
	
	However, the matrices $\bm{M}_q$ and $\bm{M}_q^{-1}$ no longer have the interleaving zero structures as in \eqref{M_q_inter_zero} because of the asymmetric integrated ranges for those terms in \eqref{kernel_terms}. Nevertheless, by our Assumption~\ref{assump:boundary_cond}(b), we know that both $\bm{\tau}_q^*=\tilde{\bm{\tau}}_q=\bm{0}$. Thus, no matter $q>0$ is odd or even, the asymptotic conditional bias of $\hat{\beta}_2(t,\bm{s})$ is given by 
	\begin{align*}
		&\mathrm{Bias}\left[\hat{\beta}_2(t,\bm{s}) \big| \mathbb{X} \right] \\
		&= \frac{1}{h\cdot p(t,\bm{s})} \left[h^{q+1} \bm{e}_2^T\bm{M}_q^{-1}\bm{\tau}_q + hb \bm{e}_2^T\bm{M}_q^{-1}\bm{\tau}_q^* + b^2 \bm{e}_2^T\bm{M}_q^{-1}\tilde{\bm{\tau}}_q + O\left(\max\{h,b\}^4\right) + o_P\left(\sqrt{\frac{1}{nhb^d}}\right) \right]\\
		&= \frac{1}{p(t,\bm{s})} \left[h^q \bm{e}_2^T\bm{M}_q^{-1}\bm{\tau}_q + O\left(\frac{\max\{h,b\}^4}{h}\right) + o_P\left(\sqrt{\frac{1}{nhb^d}}\right) \right].
	\end{align*} 
	
	Finally, for any $(t,\bm{s})\in \mathcal{E}$, the unconditional asymptotic rate of convergence of $\hat{\beta}_2(t,\bm{s})$ can be derived by noting that
	$$\mathbb{E}\left[\hat{\beta}_2(t,\bm{s})\right] = \mathbb{E}\left[\mathbb{E}\left(\hat{\beta}_2(t,\bm{s})\big| \mathbb{X}\right)\right] \quad \text{ and } \quad \mathrm{Var}\left[\hat{\beta}_2(t,\bm{s})\right] = \mathbb{E}\left[\mathrm{Var}\left(\hat{\beta}_2(t,\bm{s})\big| \mathbb{X} \right)\right] + \mathrm{Var}\left[\mathbb{E}\left(\hat{\beta}_2(t,\bm{s})\big| \mathbb{X}\right)\right].$$
	Therefore, if $(t,\bm{s})$ is an interior point of $\mathcal{E}$, then we have that
	$$\hat{\beta}_2(t,\bm{s}) -  \frac{\partial}{\partial t} \mu(t,\bm{s}) =
	\begin{cases}
		O\left(h^{q+1} + b^2 + \frac{b^4}{h}\right) + O_P\left(\sqrt{\frac{1}{nh^3 b^d}}\right) & q \text{ is odd},\\
		O\left(h^q + b^2 + \frac{b^4}{h}\right) + O_P\left(\sqrt{\frac{1}{nh^3 b^d}}\right) & q \text{ is even},
	\end{cases}$$
	as $h,b\to 0$ and $nh^3b^d \to \infty$. Otherwise, if $(t,\bm{s})\in \partial \mathcal{E}$, then we have that
	$$\hat{\beta}_2(t,\bm{s}) - \frac{\partial}{\partial t} \mu(t,\bm{s}) = O\left(h^q + \frac{\max\{h,b\}^4}{h}\right) + O_P\left(\sqrt{\frac{1}{nh^3 b^d}}\right)$$
	as $h,b\to 0$ and $nh^3b^d \to \infty$. The results follow.
\end{proof}

\section{Proofs of Results in the Main Paper}
\label{app:proofs}

This section consists of the proofs for theorems, lemmas, and propositions in the main paper as well as other auxiliary results.

\subsection{Proof of Proposition~\ref{prop:add_conf_prop}}
\label{app:proof_prop1}

\begin{customprop}{2}[Properties of the additive confounding model]
	Let $\mu(t,\bm{s})=\mathbb{E}\left(Y|T=t,\bm{S}=\bm{s}\right)=\bar{m}(t)+\eta(\bm{s})$ under the additive confounding model \eqref{add_conf_model}. Assume that $\bar{m}(t)$ is differentiable for any $t\in \mathcal{T}$. The following results hold for any $t\in \mathcal{T}$ under Assumption~\ref{assump:identify_cond}:
	\begin{enumerate}[label=(\alph*)]
		\item $\mathbb{E}\left[\mu(t,\bm{S})\right] = \bar{m}(t)$ when $\E\left[\eta(\bm{S})\right]=0$.
		
		
		\item $\theta(t) \neq \frac{d}{d t} \E\left[\mu(t,\bm{S})|T=t\right]$ when $\frac{d}{dt} \E\left[\eta(\bm{S})|T=t\right] \neq 0$.
		
		\item $\theta(t) =\mathbb{E}\left[\frac{\partial}{\partial t} \mu(t,\bm{S})\right] = \mathbb{E}\left[\frac{\partial}{\partial t} \mu(t,\bm{S}) \big| T=t\right] = \theta_C(t)$.
		
		\item $\E(Y)=\mathbb{E}\left[\mu(T,\bm{S})\right] = \mathbb{E}\left[m(T)\right]$, where $m(t)=\E\left[Y(t)\right]$.
	\end{enumerate}
\end{customprop}

\begin{proof}[Proof of Proposition~\ref{prop:add_conf_prop}]
	All the results follow from some simple calculations and the fact that 
	\begin{equation*}
		\mu(t,\bm{s}) = \E\left(Y|T=t, \bm{S}=\bm{s}\right) = \bar{m}(t) + \eta(\bm{s})
	\end{equation*}
	under the additive confounding model \eqref{add_conf_model}. \\
	
	\noindent (a) Given that $\E\left[\eta(\bm{S})\right]=0$, we have that
	\begin{align*}
		\E\left[\mu(t,\bm{S})\right] = \E\left[\E\left(Y|T=t, \bm{S}\right)\right] = \E\left[\bar{m}(t) + \eta(\bm{S})\right] = \bar{m}(t).
	\end{align*}
	
	
	\noindent (b) First, we notice that $m(t)=\E\left[Y(t)\right] = \bar{m}(t) + \E\left[\eta(\bm{S})\right]$ and $\theta(t)=\bar{m}'(t)$ under model \eqref{add_conf_model}. Then,
	$$\frac{d}{dt} \E\left[\mu(t,\bm{S})|T=t\right] = \frac{d}{dt}\left\{\bar{m}(t) + \E\left[\eta(\bm{S})|T=t\right]\right\} = \bar{m}'(t) + \frac{d}{dt} \E\left(\eta(\bm{S})|T=t\right) \neq \theta(t)$$
	when $\frac{d}{dt} \E\left[\eta(\bm{S})|T=t\right] \neq 0$. \\
	
	\noindent (c) From (b), we know that
	\begin{align*}
		\E\left[\frac{\partial}{\partial t} \mu(t,\bm{S})\right] =\E\left[\bar{m}'(t)\right] = \bar{m}'(t) = \theta(t).
	\end{align*}
	In addition, we also have that
	\begin{align*}
		\theta_C(t) &= \E\left[\frac{\partial}{\partial t} \mu(t,\bm{S}) \Big| T=t \right] = \E\left[\bar{m}'(t) | T=t\right] = \bar{m}'(t) = \theta(t).
	\end{align*}
	
	\noindent (d) On one hand, we know that
	$$\E(Y) = \mathbb{E}\left[\mu(T,\bm{S})\right] = \mathbb{E}\left[\bar{m}(T) + \eta(\bm{S})\right] = \mathbb{E}\left[\bar{m}(T)\right] + \mathbb{E}\left[\eta(\bm{S})\right].$$
	On the other hand, we also have that
	$$\mathbb{E}\left[m(T)\right] = \mathbb{E}\left\{\bar{m}(T) + \mathbb{E}\left[\eta(\bm{S})\right]\right\} = \mathbb{E}\left[\bar{m}(T)\right] + \mathbb{E}\left[\eta(\bm{S})\right].$$
	The proof is thus completed.
\end{proof}

\subsection{Proof of Lemma~\ref{lem:loc_poly_deriv_unif}}
\label{app:proof_beta_2_unif}

\begin{customlem}{3}[Uniform convergence of $\hat{\beta}_2(t,\bm{s})$]
	Let $q>0$. Suppose that Assumptions~\ref{assump:reg_diff}, \ref{assump:den_diff}, \ref{assump:boundary_cond}, and \ref{assump:reg_kernel}(a,b) hold. Let $\hat \beta_2(t,\bm{s})$ be the second element of $\hat{\bm{\beta}}(t,\bm{s}) \in \mathbb{R}^{q+1}$ and $\beta_2(t,\bm{s}) = \frac{\partial}{\partial t} \mu(t,\bm{s})$. Then, as $h,b,\frac{\max\{h,b\}^4}{h}\to 0$ and $\frac{nh^3 b^d}{|\log(hb^d)|},\frac{|\log(hb^d)|}{\log\log n} \to \infty$, we know that 
	$$\sup_{(t,\bm{s})\in \mathcal{E}}\left|\hat{\beta}_2(t,\bm{s}) - \beta_2(t,\bm{s})\right| =
	O\left(h^q\right) + O\left(b^2\right) + O\left(\frac{\max\{h,b\}^4}{h}\right) + O_P\left(\sqrt{\frac{|\log(hb^d)|}{nh^3 b^d}}\right).$$
\end{customlem}

\begin{proof}[Proof of Lemma~\ref{lem:loc_poly_deriv_unif}]
	The proof of Lemma~\ref{lem:loc_poly_deriv_unif} is partially inspired by the proof of Lemma 7 in \cite{cheng2019nonparametric}. Recall from \eqref{localpoly3} that 
	\begin{align*}
		\left(\hat{\bm{\beta}}(t,\bm{s}), \hat{\bm{\alpha}}(t,\bm{s}) \right)^T &= \left[\bm{X}^T(t,\bm{s})\bm{W}(t,\bm{s}) \bm{X}(t,\bm{s})\right]^{-1} \bm{X}(t,\bm{s})^T\bm{W}(t,\bm{s}) \bm{Y} \equiv \left(\bm{X}^T\bm{W} \bm{X}\right)^{-1} \bm{X}^T\bm{W} \bm{Y}.
	\end{align*}
	By direct calculations with $\bm{H} = \Diag(1,h,...,h^q,b,...,b) \in \mathbb{R}^{(q+1+d)\times (q+1+d)}$, we have that
	\begin{align}
		\label{local_poly_exp}
		\begin{split}
			\left(\hat{\bm{\beta}}(t,\bm{s}), \hat{\bm{\alpha}}(t,\bm{s}) \right)^T &= \left(\bm{X}^T\bm{W} \bm{X}\right)^{-1} \bm{X}^T\bm{W} \bm{Y} \\
			&= \frac{1}{nhb^d} \cdot\bm{H}^{-1} \left[\frac{1}{nhb^d} (\bm{X}\bm{H}^{-1})^T \bm{W} (\bm{X}\bm{H}^{-1}) \right]^{-1} \bm{H}^{-1} \bm{X}^T \bm{W} \bm{Y} \\
			&= \bm{H}^{-1} \left[\frac{1}{nhb^d} (\bm{X}\bm{H}^{-1})^T \bm{W} (\bm{X}\bm{H}^{-1}) \right]^{-1} \left[\frac{1}{nhb^d} (\bm{X}\bm{H}^{-1})^T \bm{W} \bm{Y} \right].
		\end{split}
	\end{align}
	We first derive the uniform rate of convergence for $\frac{1}{nhb^d} (\bm{X}\bm{H}^{-1})^T \bm{W} (\bm{X}\bm{H}^{-1})$. Recall from \eqref{X_W_mat} in the proof of Lemma~\ref{lem:loc_poly_deriv} that
	\begin{align}
		\label{X_W_h_mat}
		\begin{split}
			&\left[\frac{1}{nhb^d} (\bm{X}\bm{H}^{-1})^T \bm{W} (\bm{X}\bm{H}^{-1}) \right]_{i,j} \\
			&= \begin{cases}
				\frac{1}{nhb^d}\sum\limits_{k=1}^n \left(\frac{T_k-t}{h}\right)^{i+j-2} K_T\left(\frac{T_k-t}{h}\right) K_S\left(\frac{\bm{S}_k -\bm{s}}{b}\right), \quad\quad 1\leq i,j \leq q+1,\\
				\frac{1}{nhb^d} \sum\limits_{k=1}^n \left(\frac{T_k-t}{h}\right)^{i-1} \left(\frac{S_{k,j-q-1} - s_{j-q-1}}{b}\right) K_T\left(\frac{T_k-t}{h}\right) K_S\left(\frac{\bm{S}_k -\bm{s}}{b}\right), \quad 1\leq i \leq q+1 \text{ and } q+1< j\leq q+1+d,\\
				\frac{1}{nhb^d} \sum\limits_{k=1}^n \left(\frac{S_{k,i-q-1} - s_{i-q-1}}{b}\right) \left(\frac{T_k-t}{h}\right)^{j-1} K_T\left(\frac{T_k-t}{h}\right) K_S\left(\frac{\bm{S}_k -\bm{s}}{b}\right), \quad q+1< i\leq q+1+d \text{ and } 1\leq j \leq q+1,\\
				\frac{1}{nhb^d}\sum\limits_{k=1}^n \left(\frac{S_{k,i-q-1} - s_{i-q-1}}{b}\right) \left(\frac{S_{k,j-q-1} - s_{j-q-1}}{b}\right) K_T\left(\frac{T_k-t}{h}\right) K_S\left(\frac{\bm{S}_k -\bm{s}}{b}\right), \quad q+1< i,j\leq q+1+d.
			\end{cases}
		\end{split}
	\end{align}
	By our notation in \eqref{M_q}, we can obtain from \eqref{X_W_mat_asymp} in the proof of Lemma~\ref{lem:loc_poly_deriv} that
	\begin{align*}
		&\frac{1}{nhb^d} (\bm{X}\bm{H}^{-1})^T \bm{W} (\bm{X}\bm{H}^{-1}) \\
		&= p(t,\bm{s}) \bm{M}_q + O\left(\max\{h,b\}\right) + \underbrace{\frac{1}{nhb^d} (\bm{X}\bm{H}^{-1})^T \bm{W} (\bm{X}\bm{H}^{-1}) - \mathbb{E}\left[\frac{1}{nhb^d} (\bm{X}\bm{H}^{-1})^T \bm{W} (\bm{X}\bm{H}^{-1})\right]}_{\textbf{Term A}},
	\end{align*}
	where we apply an abuse of notation when using $O\left(\max\{h,b\}\right)$ to denote a matrix whose entries are of this order.
	
	By \eqref{X_W_h_mat}, each entry of the matrix in {\bf Term A} is a mean-zero empirical process (scaled by $\frac{1}{\sqrt{nhb^d}}$) over a function in the class $\mathcal{K}_{q,d}$ defined in Assumption~\ref{assump:reg_kernel}(b). By Theorem 2.3 in \cite{gine2002rates} or Theorem 1 in \cite{einmahl2005uniform}, we know that
	$$\sup_{(t,\bm{s})\in \mathcal{E}}\norm{\frac{1}{nhb^d} (\bm{X}\bm{H}^{-1})^T \bm{W} (\bm{X}\bm{H}^{-1}) - \mathbb{E}\left[\frac{1}{nhb^d} (\bm{X}\bm{H}^{-1})^T \bm{W} (\bm{X}\bm{H}^{-1})\right]}_{\max} = O_P\left(\sqrt{\frac{|\log(hb^d)|}{nhb^d}}\right),$$
	where $\norm{\bm{A}}_{\max}=\max_{i,j}|\bm{A}_{ij}|$ for $\bm{A}\in \mathbb{R}^{(q+1+d)\times (q+1+d)}$. Therefore, we conclude that
	$$\sup_{(t,\bm{s})\in \mathcal{E}}\norm{\frac{1}{nhb^d} (\bm{X}\bm{H}^{-1})^T \bm{W} (\bm{X}\bm{H}^{-1}) - p(t,\bm{s}) \bm{M}_q}_{\max} = O\left(\max\{h,b\}\right) + O_P\left(\sqrt{\frac{|\log(hb^d)|}{nhb^d}}\right)$$
	and 
	$$\sup_{(t,\bm{s})\in \mathcal{E}}\norm{\left[\frac{1}{nhb^d} \left(\bm{X}\bm{H}^{-1}\right)^T \bm{W} \left(\bm{X}\bm{H}^{-1}\right)\right]^{-1} - \frac{\bm{M}_q^{-1}}{p(t,\bm{s})}}_{\max} = O\left(\max\{h,b\}\right) + O_P\left(\sqrt{\frac{|\log(hb^d)|}{nhb^d}}\right).$$
	Based on the above result and \eqref{local_poly_exp}, we have that
	\begin{align}
		\label{local_poly_exp2}
		\begin{split}
			\hat{\beta}_2(t,\bm{s}) &= \bm{e}_2^T \begin{bmatrix}
				\hat{\bm{\beta}}(t,\bm{s})\\
				\hat{\bm{\alpha}}(t,\bm{s})
			\end{bmatrix}\\
			&= \frac{1}{h}\cdot \bm{e}_2^T \left[\frac{1}{nhb^d} (\bm{X}\bm{H}^{-1})^T \bm{W} (\bm{X}\bm{H}^{-1}) \right]^{-1} \left[\frac{1}{nhb^d} (\bm{X}\bm{H}^{-1})^T \bm{W} \bm{Y} \right]\\
			&= \frac{1}{h \cdot p(t,\bm{s})}\cdot \bm{e}_2^T \left[\bm{M}_q^{-1} + \bm{A}_{h,b}\right]\left[\frac{1}{nhb^d} (\bm{X}\bm{H}^{-1})^T \bm{W} \bm{Y} \right],
		\end{split}
	\end{align}
	where each entry of $\bm{A}_{h,b}\in \mathbb{R}^{(q+1+d)\times (q+1+d)}$ is uniformly of the order $O\left(\max\{h,b\}\right) + O_P\left(\sqrt{\frac{|\log(hb^d)|}{nhb^d}}\right)$. Notice also that
	\begin{align*}
		\frac{1}{nhb^d} (\bm{X}\bm{H}^{-1})^T \bm{W} \bm{Y} &= 
		\begin{bmatrix}
			\left(\frac{1}{nhb^d} \sum_{k=1}^n Y_k\left(\frac{T_k-t}{h}\right)^{j-1} K_T\left(\frac{T_k-t}{h}\right) K_S\left(\frac{\bm{S}_k-\bm{s}}{b}\right)\right)_{1\leq j\leq q+1}\\
			\left(\frac{1}{nhb^d} \sum_{k=1}^n Y_k\left(\frac{S_{k,j-q-1}-s_{j-q-1}}{b}\right) K_T\left(\frac{T_k-t}{h}\right) K_S\left(\frac{\bm{S}_k-\bm{s}}{b}\right)\right)_{q+1< j\leq q+1+d}
		\end{bmatrix}\\
		&= \mathbb{P}_n\left(\frac{1}{hb^d} \bm{\Psi}_{t,\bm{s}}\right),
	\end{align*}
	where $\mathbb{P}_n$ is the empirical measure associated with observations $\left\{(Y_k,T_k,\bm{S}_k)\right\}_{k=1}^n$ and $\bm{\Psi}_{t,\bm{s}}(y,z,\bm{v}) = y\cdot \bm{\psi}_{t,\bm{s}}(z,\bm{v})$ defined in \eqref{psi_prod}.
	Plugging this result back into \eqref{local_poly_exp2}, we obtain that
	\begin{align}
		\label{local_poly_exp22}
		\begin{split}
			\hat{\beta}_2(t,\bm{s}) &= \frac{1}{h}\cdot \bm{e}_2^T \left[\bm{M}_q^{-1} + \bm{A}_{h,b}\right] \mathbb{P}_n\left(\frac{1}{hb^d} \bm{\Psi}_{t,\bm{s}}\right)\\
			&= \frac{1}{h \cdot p(t,\bm{s})} \cdot \mathbb{P}_n\left(\frac{1}{hb^d} \bm{e}_2^T \bm{M}_q^{-1} \bm{\Psi}_{t,\bm{s}}\right) + \frac{1}{h \cdot p(t,\bm{s})} \cdot  \mathbb{P}_n\left(\frac{1}{hb^d} \bm{e}_2^T\bm{A}_{h,b}\bm{\Psi}_{t,\bm{s}}\right).
		\end{split}
	\end{align}
	Equivalently, we can scale and center $\hat{\beta}_2(t,\bm{s})$ as:
	\begin{align}
		\label{local_poly_exp3}
		\begin{split}
			&\sqrt{nh^3b^d} \left\{\hat{\beta}_2(t,\bm{s}) - \mathbb{E}\left[\hat{\beta}_2(t,\bm{s})\right]\right\} \\
			&= \sqrt{nhb^d} \left[\mathbb{P}_n\left(\frac{1}{hb^d \cdot p(t,\bm{s})} \bm{e}_2^T \bm{M}_q^{-1} \bm{\Psi}_{t,\bm{s}}\right) - \P\left(\frac{1}{hb^d \cdot p(t,\bm{s})} \bm{e}_2^T \bm{M}_q^{-1} \bm{\Psi}_{t,\bm{s}}\right) \right] \\
			&\quad + \sqrt{nhb^d} \left[\mathbb{P}_n\left(\frac{1}{hb^d \cdot p(t,\bm{s})} \bm{e}_2^T\bm{A}_{h,b}\bm{\Psi}_{t,\bm{s}}\right) - \P\left(\frac{1}{hb^d \cdot p(t,\bm{s})} \bm{e}_2^T\bm{A}_{h,b}\bm{\Psi}_{t,\bm{s}}\right) \right]\\
			&= \sqrt{hb^d} \cdot \mathbb{G}_n\left(\frac{1}{hb^d \cdot p(t,\bm{s})} \bm{e}_2^T\bm{M}_q^{-1}\bm{\Psi}_{t,\bm{s}}\right) + \sqrt{hb^d} \cdot \mathbb{G}_n\left(\frac{1}{hb^d \cdot p(t,\bm{s})} \bm{e}_2^T\bm{A}_{h,b}\bm{\Psi}_{t,\bm{s}}\right),
		\end{split}
	\end{align}
	where $\mathbb{G}_n f = \sqrt{n}\left(\mathbb{P}_n -\P\right) f = \frac{1}{\sqrt{n}} \sum_{k=1}^n \left\{f(Y_k,T_k,\bm{S}_k) - \mathbb{E}\left[f(Y_k,T_k,\bm{S}_k)\right]\right\}$. One auxiliary result that we derive here is a form of the uniform Bahadur representation \citep{bahadur1966note,kong2010uniform} as:
	\begin{align*}
		&\sup_{(t,\bm{s})\in \mathcal{E}}\left|\frac{\sqrt{nh^3b^d} \left\{\hat{\beta}_2(t,\bm{s}) - \mathbb{E}\left[\hat{\beta}_2(t,\bm{s})\right]\right\} - \sqrt{hb^d} \cdot \mathbb{G}_n\left(\frac{1}{hb^d \cdot p(t,\bm{s})} \bm{e}_2^T\bm{M}_q^{-1}\bm{\Psi}_{t,\bm{s}}\right)}{\sqrt{hb^d} \cdot \mathbb{G}_n\left(\frac{1}{hb^d \cdot p(t,\bm{s})} \bm{e}_2^T\bm{\Psi}_{t,\bm{s}}\right)} \right|\\
		&= O\left(\max\{h,b\}\right) + O_P\left(\sqrt{\frac{|\log(hb^d)|}{nhb^d}}\right),
	\end{align*}
	because each entry of $\bm{A}_{h,b}\in \mathbb{R}^{(q+1+d)\times (q+1+d)}$ is uniformly of the order $O\left(\max\{h,b\}\right) + O_P\left(\sqrt{\frac{|\log(hb^d)|}{nhb^d}}\right)$.
	
	We remain to derive the uniform rate of convergence for $\sqrt{hb^d} \cdot \mathbb{G}_n\left(\frac{1}{hb^d \cdot p(t,\bm{s})}\cdot  \bm{e}_2^T\bm{M}_q^{-1}\bm{\Psi}_{t,\bm{s}}\right)$ and $\sqrt{hb^d} \cdot \mathbb{G}_n\left(\frac{1}{hb^d \cdot p(t,\bm{s})}\cdot \bm{e}_2^T\bm{\Psi}_{t,\bm{s}}\right)$. Recall from our notation in \eqref{psi_prod} that
	\begin{align*}
		&\sqrt{hb^d} \cdot \mathbb{G}_n\left(\frac{1}{hb^d\cdot p(t,\bm{s})} \cdot \bm{e}_2^T\bm{M}_q^{-1}\bm{\Psi}_{t,\bm{s}}\right) \\
		&= \frac{1}{\sqrt{nhb^d}\cdot p(t,\bm{s})} \sum_{k=1}^n \left\{\bm{e}_2^T \bm{M}_q^{-1} \bm{\Psi}_{t,\bm{s}}(Y_k,T_k,\bm{S}_k) - \mathbb{E}\left[\bm{e}_2^T \bm{M}_q^{-1} \bm{\Psi}_{t,\bm{s}}(Y_k,T_k,\bm{S}_k)\right]\right\}\\
		&= \frac{1}{\sqrt{nhb^d} \cdot p(t,\bm{s})} \sum_{k=1}^n \left\{Y_k\cdot \bm{e}_2^T \bm{M}_q^{-1} \bm{\psi}_{t,\bm{s}}(T_k,\bm{S}_k) - \mathbb{E}\left[Y_k\cdot \bm{e}_2^T \bm{M}_q^{-1} \bm{\psi}_{t,\bm{s}}(T_k,\bm{S}_k)\right]\right\}.
	\end{align*}
	Since each entry of $(z,\bm{v}) \mapsto \bm{\psi}_{t,\bm{s}}(z,\bm{v})$ is simply a linear combination of functions in $\mathcal{K}_{q,d}$ defined in Assumption~\ref{assump:reg_kernel}(b), we can apply Theorem 4 in \cite{einmahl2005uniform} to establish that
	\begin{align}
		\label{local_poly_sup1}
		\begin{split}
			&\sup_{(t,\bm{s})\in \mathcal{E}} \left|\sqrt{hb^d} \cdot \mathbb{G}_n\left(\frac{1}{hb^d \cdot p(t,\bm{s})} \cdot \bm{e}_2^T\bm{M}_q^{-1}\bm{\Psi}_{t,\bm{s}}\right) \right| \\
			&=\sup_{(t,\bm{s})\in \mathcal{E}} \left|\frac{1}{\sqrt{nhb^d}\cdot p(t,\bm{s})}\sum_{k=1}^n \left\{Y_k\cdot \bm{e}_2^T \bm{M}_q^{-1} \bm{\psi}_{t,\bm{s}}(T_k,\bm{S}_k) - \mathbb{E}\left[Y_k\cdot\bm{e}_2^T \bm{M}_q^{-1} \bm{\psi}_{t,\bm{s}}(T_k,\bm{S}_k)\right]\right\} \right|\\
			&=O_P\left(\sqrt{|\log(hb^d)|}\right)
		\end{split}
	\end{align}
	when $\frac{|\log(hb^d)|}{\log\log n} \to \infty$. Similarly, we have that 
	$$\sup_{(t,\bm{s})\in \mathcal{E}}\left|\sqrt{hb^d} \cdot \mathbb{G}_n\left(\frac{1}{hb^d \cdot p(t,\bm{s})} \cdot \bm{e}_2^T\bm{\Psi}_{t,\bm{s}}\right) \right| = O_P\left(\sqrt{|\log(hb^d)|}\right).$$
	Plugging these two rates of convergence back into \eqref{local_poly_exp3}, we obtain that as $h,b\to 0$ and $\frac{nh^3 b^d}{|\log(hb^d)|} \to\infty$,
	\begin{align}
		\label{local_poly_exp4}
		\begin{split}
			&\sup_{(t,\bm{s})\in \mathcal{E}} \left|\hat{\beta}_2(t,\bm{s}) - \mathbb{E}\left[\hat{\beta}_2(t,\bm{s})\right] \right| \\
			&= \sup_{(t,\bm{s})\in \mathcal{E}}\left|\frac{\sqrt{hb^d} \cdot \mathbb{G}_n\left(\frac{1}{hb^d \cdot p(t,\bm{s})} \cdot \bm{e}_2^T\bm{M}_q^{-1}\bm{\Psi}_{t,\bm{s}}\right) + \sqrt{hb^d} \cdot \mathbb{G}_n\left(\frac{1}{hb^d \cdot p(t,\bm{s})}\cdot  \bm{e}_2^T\bm{A}_{h,b}\bm{\Psi}_{t,\bm{s}}\right)}{\sqrt{nh^3b^d}} \right|\\
			&\leq \sup_{(t,\bm{s})\in \mathcal{E}}\left|\frac{\sqrt{hb^d} \cdot \mathbb{G}_n\left(\frac{1}{hb^d \cdot p(t,\bm{s})} \cdot \bm{e}_2^T\bm{M}_q^{-1}\bm{\Psi}_{t,\bm{s}}\right)}{\sqrt{nh^3b^d}} \right| + \sup_{(t,\bm{s})\in \mathcal{E}}\left|\frac{\sqrt{hb^d} \cdot \mathbb{G}_n\left(\frac{1}{hb^d \cdot p(t,\bm{s})} \cdot \bm{e}_2^T\bm{A}_{h,b}\bm{\Psi}_{t,\bm{s}}\right)}{\sqrt{nh^3b^d}} \right| \\
			&= O_P\left(\sqrt{\frac{|\log(hb^d)|}{nh^3b^d}} \right) + O_P\left(\sqrt{\frac{|\log(hb^d)|}{nh^3b^d}} \right) \left[O\left(\max\{h,b\}\right) + O_P\left(\sqrt{\frac{|\log(hb^d)|}{nhb^d}}\right)\right]\\
			&= O_P\left(\sqrt{\frac{|\log(hb^d)|}{nh^3b^d}} \right).
		\end{split}
	\end{align}
	Together with the rate of convergence for the bias term $\mathbb{E}\left[\hat{\beta}_2(t,\bm{s})\right]- \beta_2(t,\bm{s})$ in Lemma~\ref{lem:loc_poly_deriv}, we conclude that
	$$\sup_{(t,\bm{s})\in \mathcal{E}}\left|\hat{\beta}_2(t,\bm{s}) -  \beta_2(t,\bm{s}) \right| = O\left(h^q + b^2 + \frac{\max\{b,h\}^4}{h}\right) + O_P\left(\sqrt{\frac{|\log(hb^d)|}{nh^3 b^d}}\right),$$
	where we combine the odd and even cases for $q>0$ by arguing that $O(h^{q+1})$ is dominated by $O(h^q)$ when $h$ is small. The proof is thus completed.
\end{proof}

\subsection{Proof of \autoref{thm:int_est_rate}}
\label{app:sim_integral_proof}

\begin{customthm}{4}[Convergence of $\hat{\theta}_C(t)$ and $\hat{m}_{\theta}(t)$]
	Let $q>0$ and $\mathcal{T}' \subset \mathcal{T}$ be a compact set so that $p_T(t)$ is uniformly bounded away from 0 within $\mathcal{T}'$. Suppose that Assumptions~\ref{assump:identify_cond}, \ref{assump:diff_inter}, \ref{assump:reg_diff}, \ref{assump:den_diff}, \ref{assump:boundary_cond}, and \ref{assump:reg_kernel} hold. Then, as $h,b,\hslash,\frac{\max\{b,h\}^4}{h}\to 0$, and $\frac{n\max\{h,\hslash\}b^d}{|\log(\hslash hb^d)|}, \frac{|\log(h\hslash b^d)|}{\log\log n}, \frac{nh^3}{|\log h|}\to \infty$, we know that
	\begin{align*}
		\sup_{t\in \mathcal{T}'} \left|\hat{\theta}_C(t) - \theta(t)\right| &= O\left(h^q + b^2 + \frac{\max\{b,h\}^4}{h}\right) + O_P\left(\sqrt{\frac{|\log h|}{nh^3}} + \hslash^2 + \sqrt{\frac{|\log \hslash|}{n\hslash}}\right)
	\end{align*}
	and
	\begin{align*}
		\sup_{t\in \mathcal{T}'}\left|\hat m_\theta(t) - m(t) \right| &= O_P\left(\frac{1}{\sqrt{n}}\right) + O\left(h^q + b^2 + \frac{\max\{b,h\}^4}{h}\right) + O_P\left(\sqrt{\frac{|\log h|}{nh^3}} + \hslash^2 + \sqrt{\frac{|\log \hslash|}{n\hslash}}\right).
	\end{align*}
\end{customthm}

\begin{proof}[Proof of \autoref{thm:int_est_rate}]
	Assume, without loss of generality, that the set $\mathcal{T}'$ is connected. Otherwise, we focus our analysis on a connected component of $\mathcal{T}'$ and take the union afterwards.
	
	Recall from \eqref{simp_integral} and \eqref{integral_part_si} that the integral estimator is defined as:
	$$\hat m_\theta(t) = \frac{1}{n}\sum_{i=1}^n \left[Y_i + \int_{\tilde{t}=T_i}^{\tilde{t}=t} \hat \theta_C(\tilde{t}) \,d\tilde{t} \right] = \frac{1}{n}\sum_{i=1}^n Y_i + \hat{\Delta}_{h,b}(t).$$
	Moreover, under Assumption~\ref{assump:diff_inter}, 
	$$m(t) = \E\left[m(T)\right] + \Delta(t) = \mathbb{E}(Y) + \Delta(t),$$
	where we define 
	$$\Delta(t) \equiv \E\left[\int_{\tilde{t}=T}^{\tilde{t}=t} \theta(\tilde{t})d\tilde{t}\right] = \int \int_{\tilde{t}=\tau}^{\tilde{t}=t} \theta(\tilde{t}) \,d\tilde{t}\, d\P_T(\tau) = \int \int_{\tilde{t}=\tau}^{\tilde{t}=t} \theta_C(\tilde{t}) \,d\tilde{t}\, d\P_T(\tau)$$
	as the population version of $\hat \Delta_{h,b}(t)$ in \eqref{integral_part_si} under Assumption~\ref{assump:diff_inter}. Here, $\P_T$ is the marginal probability distribution function of $T$. 
	Under the boundedness of $\mu(t,\bm{s})$ by Assumption~\ref{assump:reg_diff}, we know that 
	$$\frac{1}{n}\sum_{i=1}^n Y_i - \mathbb{E}\left[m(T)\right] = O_P\left(\frac{1}{\sqrt{n}}\right).$$
	Recall from \eqref{theta_C_est} that $\hat \theta_C(t) = \int \hat \beta_2(t,\bm{s}) d\hat P(\bm{s}|t)$ and $\theta(t)=\theta_C(t) = \int \beta_2(t,\bm{s}) dP(\bm{s}|t)$ with $\beta_2(t,\bm{s}) = \frac{\partial}{\partial t} \mu(t,\bm{s})$ under Assumption~\ref{assump:diff_inter}.
	Then, we have the following decomposition as:
	\begin{equation}
		\label{si_decomp}
		\begin{aligned}
			\hat \Delta_{h,b}(t) - \Delta(t) &= \frac{1}{n}\sum_{i=1}^n \int_{\tilde{t}=T_i}^{\tilde{t}=t} \hat \theta_C(\tilde{t})d\tilde{t} - \int \int_{\tilde{t}=\tau}^{\tilde{t}=t} \theta_C(\tilde{t})d\tilde{t} d\P_T(\tau)\\
			&= \underbrace{\int \int_{\tilde{t}=\tau}^{\tilde{t}=t} \theta_C(\tilde{t})d\tilde{t} \left[d\mathbb{P}_{n,T}(\tau) - d \P_T(\tau)\right]}_{\textbf{Term (I)}} +  \underbrace{\int \int_{\tilde{t}=\tau}^{\tilde{t}=t} \left[\hat \theta_C(\tilde{t}) - \theta_C(\tilde{t})\right] d\tilde{t} d\mathbb{P}_{n,T}(\tau)}_{\textbf{Term (II)}},
		\end{aligned}
	\end{equation}
	where $\mathbb{P}_{n,T}$ is the empirical measure associated with $\{T_1,...,T_n\}$.\\
	
	$\bullet$ {\bf Term (I):} By Dvoretzky-Kiefer-Wolfowitz inequality \citep{dvoretzky1956asymptotic,massart1990tight}, we know that
	$$\sup_{t\in \mathcal{T}'} \left|\mathbb{P}_{n,T}(t) -\P_T(t) \right| = O_P\left(\frac{1}{\sqrt{n}}\right).$$ 
	In addition, under Assumption~\ref{assump:reg_diff} and the compactness of the support $\mathcal{T}'$, we have that
	$$\int \left[\int_{\tilde{t}=\tau}^{\tilde{t}=t} \theta_C(\tilde{t})d\tilde{t}\right]^2 d\P_T(\tau) \leq \int |t-\tau|^2 \sup_{t\in \mathcal{T}}\left|\theta_C(t)\right|^2 d\P_T(\tau) <\infty$$
	for any $t\in \mathcal{T}$. Thus, the (uniform) rate of convergence of Term (I) in \eqref{si_decomp} is
	$$\sup_{t\in \mathcal{T}'}\int \int_{\tilde{t}=\tau}^{\tilde{t}=t} \theta_C(\tilde{t})d\tilde{t} \left[d\mathbb{P}_{n,T}(\tau) - d\P_T(\tau)\right] = O_P\left(\frac{1}{\sqrt{n}}\right).$$

	$\bullet$ {\bf Term (II):} We first focus on establishing the (uniform) rate of convergence for $\hat{\theta}_C(t) -\theta_C(t)$. For any $t\in \mathcal{T}$, we note that
	\begin{align*}
		&\hat{\theta}_C(t) -\theta_C(t) \\
		&= \int \hat{\beta}_2(t,\bm{s}) d\hat{P}(\bm{s}|t) - \int \beta_2(t,\bm{s}) dP(\bm{s}|t)\\
		&= \frac{\sum_{i=1}^n \left[\hat{\beta}_2(t,\bm{S}_i) - \beta_2(t,\bm{S}_i)\right] \bar{K}_T\left(\frac{t-T_i}{\hslash}\right)}{\sum_{j=1}^n \bar{K}_T\left(\frac{t-T_j}{\hslash}\right)} + \frac{\sum_{i=1}^n \beta_2(t,\bm{S}_i) \cdot \bar{K}_T\left(\frac{t-T_i}{\hslash}\right)}{\sum_{j=1}^n \bar{K}_T\left(\frac{t-T_j}{\hslash}\right)} - \int \beta_2(t,\bm{s}) dP(\bm{s}|t)\\
		&= \underbrace{\sum_{(t,\bm{S}_i) \in \mathcal{E} \cap \left(\partial\mathcal{E} \ominus \hslash\right)} \frac{\left[\hat{\beta}_2(t,\bm{S}_i) - \beta_2(t,\bm{S}_i)\right] \bar{K}_T\left(\frac{t-T_i}{\hslash}\right)}{\sum_{j=1}^n \bar{K}_T\left(\frac{t-T_j}{\hslash}\right)}}_{\textbf{Term A}} + \underbrace{\sum_{(t,\bm{S}_i) \in \partial \mathcal{E} \oplus \hslash} \frac{\left[\hat{\beta}_2(t,\bm{S}_i) - \beta_2(t,\bm{S}_i)\right] \bar{K}_T\left(\frac{t-T_i}{\hslash}\right)}{\sum_{j=1}^n \bar{K}_T\left(\frac{t-T_j}{\hslash}\right)}}_{\textbf{Term B}} \\
		&\quad + \underbrace{\int \beta_2(t,\bm{s}) \left[d\hat{P}(\bm{s}|t) - dP(\bm{s}|t)\right]}_{\textbf{Term C}},
	\end{align*}
	where $\partial\mathcal{E} \ominus \hslash = \left\{\bm{z}\in \mathbb{R}^{d+1}: \inf_{\bm{x}\in \partial \mathcal{E}} \norm{\bm{z}-\bm{x}}_2 \geq \hslash\right\}$ and $\partial \mathcal{E} \oplus \hslash = \left\{\bm{z}\in \mathbb{R}^{d+1}: \inf_{\bm{x}\in \partial \mathcal{E}} \norm{\bm{z}-\bm{x}}_2 \leq \hslash\right\}$. In other words, $\mathcal{E}\cap \left(\partial\mathcal{E} \ominus \hslash \right)$ indicates the interior region of $\mathcal{E}$ that is at least $\hslash$ away from its boundary $\partial \mathcal{E}$, while $\partial \mathcal{E} \oplus \hslash$ is the region that is within a distance $\hslash$ from the boundary $\partial \mathcal{E}$.\\
	
	As for \textbf{Term A}, this part of the derivation is inspired by the proof of Theorem 1 in \cite{fan1998direct}. We know from Assumption~\ref{assump:boundary_cond}(c) that the proportion of $(t,\bm{S}_i),i=1,...,n$ that lie in $\partial \mathcal{E} \oplus \hslash$ is of order $O_P(\hslash)$. Thus, with probability tending to 1 as $n\to\infty$, there are at least $n(1-C\hslash)$ points of $(t,\bm{S}_i),i=1,...,n$ that lie in the interior region $\mathcal{E}\cap \left(\partial\mathcal{E} \ominus \hslash \right)$, where $C>0$ is some absolute constant. We denote this collection of points by $\mathcal{I}$ and let 
	$$W_i = \frac{|\mathcal{I}|\cdot \bar{K}_T\left(\frac{t-T_i}{\hslash}\right)}{\sum_{\ell=1}^n \bar{K}_T\left(\frac{t-T_{\ell}}{\hslash}\right)} \quad \text{ for } \quad i\in \mathcal{I},$$
	where $|\mathcal{I}|=n(1-C\hslash)$ is the number of points in $\mathcal{I}$. Recall from \eqref{local_poly_exp22} and the interleaving zero structures of $\bm{M}_q^{-1}$ with an even number $q>0$ that 
	$$\hat{\beta}_2(t,\bm{s}) = \frac{1}{h \cdot p(t,\bm{s})} \sum_{j=1}^n \frac{C_{K_T}\cdot Y_j}{hb^d}\left(\frac{T_j-t}{h}\right) K_T\left(\frac{T_j-t}{h}\right) K_S\left(\frac{\bm{S}_j-\bm{s}}{b}\right) \left[1+A_{h,b}\right],$$
	where $A_{h,b}=O(\max\{h,b\}) + O_P\left(\sqrt{\frac{|\log(hb^d)|}{nhb^d}}\right)$ uniformly over $t\in \mathcal{T}'$ and $C_{K_T}>0$ is a constant that only depends on the kernel $K_T$.
	Now, we can rewrite \textbf{Term A} as:
	\begin{align}
		\label{TermA}
		&\textbf{Term A} \\ \nonumber
		&= \sum_{(t,\bm{S}_i) \in \mathcal{E} \cap \left(\partial\mathcal{E} \ominus \hslash\right)} \frac{\left[\hat{\beta}_2(t,\bm{S}_i) - \beta_2(t,\bm{S}_i)\right] \bar{K}_T\left(\frac{t-T_i}{\hslash}\right)}{\sum_{j=1}^n \bar{K}_T\left(\frac{t-T_j}{\hslash}\right)}\\ \nonumber
		&= \frac{1}{|\mathcal{I}|}\sum_{i\in \mathcal{I}} W_i \left[\hat{\beta}_2(t,\bm{S}_i) - \beta_2(t,\bm{S}_i)\right] \\ \nonumber
		&= \frac{1}{|\mathcal{I}|}\sum_{i\in \mathcal{I}} \frac{W_i}{nh\cdot p(t,\bm{S}_i)} \sum_{j=1}^n \frac{C_{K_T}\left[r_{(t,\bm{S}_i)}(T_j,\bm{S}_j) + \epsilon_j+A_{h,b}\right]}{hb^d} \left(\frac{T_j-t}{h}\right) K_T\left(\frac{T_j-t}{h}\right) K_S\left(\frac{\bm{S}_j-\bm{S}_i}{b}\right)\\ \nonumber
		&= \frac{1}{nh^2} \sum_{j=1}^n C_{K_T} \left(\frac{T_j-t}{h}\right) K_T\left(\frac{T_j-t}{h}\right) \left\{\frac{1}{|\mathcal{I}| b^d} \sum_{i\in \mathcal{I}} \frac{W_i \cdot K_S\left(\frac{\bm{S}_j-\bm{S}_i}{b}\right)}{p(t,\bm{S}_i)} \left[r_{(t,\bm{S}_i)}(T_j,\bm{S}_j) + \epsilon_j+A_{h,b}\right] \right\},\nonumber
	\end{align}
	where $r_{(t,\bm{S}_i)}(T_j,\bm{S}_j)$ is the remainder term of applying Taylor's expansion to $\mu(T_j,\bm{S}_j)$ around $(t,\bm{S}_i)$ as shown in \eqref{taylor_reg}.
	
	We now study the (uniform) rate of convergence of the dominating term $\frac{1}{|\mathcal{I}| b^d} \sum_{i\in \mathcal{I}} \frac{W_i \cdot K_S\left(\frac{\bm{S}_j-\bm{S}_i}{b}\right)}{p(t,\bm{S}_i)}$ inside the curly bracket. Notice that
	\begin{align*}
		\frac{1}{|\mathcal{I}| b^d} \sum_{i\in \mathcal{I}} \frac{W_i \cdot K_S\left(\frac{\bm{S}_j-\bm{S}_i}{b}\right)}{p(t,\bm{S}_i)} &= \frac{1}{b^d} \sum_{i\in \mathcal{I}} \frac{\bar{K}_T\left(\frac{t-T_i}{\hslash}\right) \cdot K_S\left(\frac{\bm{S}_j-\bm{S}_i}{b}\right)}{p(t,\bm{S}_i)\sum_{\ell=1}^n \bar{K}_T\left(\frac{t-T_{\ell}}{\hslash}\right)} \\
		&= \frac{\frac{1}{n\hslash b^d} \sum_{i\in \mathcal{I}} \frac{\bar{K}_T\left(\frac{t-T_i}{\hslash}\right)}{p(t,\bm{S}_i)} \cdot K_S\left(\frac{\bm{S}_j-\bm{S}_i}{b}\right)}{\frac{1}{n\hslash} \sum_{\ell=1}^n \bar{K}_T\left(\frac{t-T_{\ell}}{\hslash}\right)}.
	\end{align*}
	The denominator $\frac{1}{n\hslash}\sum_{\ell=1}^n \bar{K}_T\left(\frac{t-T_{\ell}}{\hslash}\right)$ above is a simple kernel density estimator, and its (uniform) rate of convergence under Assumption~\ref{assump:reg_kernel}(d) is
	$$\sup_{\tilde{t}\in \mathcal{T}'}\left|\frac{1}{n\hslash}\sum_{\ell=1}^n \bar{K}_T\left(\frac{t-T_{\ell}}{\hslash}\right) - p_T(t) \right|= O(\hslash^2) + O_P\left(\sqrt{\frac{|\log \hslash|}{n\hslash}}\right)$$
	when $\hslash\to 0$ and $\frac{n\hslash}{|\log \hslash|}, \frac{|\log \hslash|}{\log\log n} \to \infty$; see Theorem 1 in \cite{einmahl2005uniform}.
	
	\noindent For the numerator $\frac{1}{n\hslash b^d} \sum_{i\in \mathcal{I}} \frac{\bar{K}_T\left(\frac{t-T_i}{\hslash}\right)}{p(t,\bm{S}_i)} \cdot K_S\left(\frac{\bm{S}_j-\bm{S}_i}{b}\right)$, its conditional expectation can be computed as:
	\begin{align*}
		&\mathbb{E}\left[\frac{1}{n\hslash b^d} \sum_{i\in \mathcal{I}} \frac{\bar{K}_T\left(\frac{t-T_i}{\hslash}\right)}{p(t,\bm{S}_i)} \cdot K_S\left(\frac{\bm{S}_j-\bm{S}_i}{b}\right) \Big| \bm{S}_j \right] \\
		&= \frac{K_S(0)}{nb^d\hslash \cdot p(t,\bm{S}_j)} \int_{\mathcal{T}} \bar{K}_T\left(\frac{t_1-t}{\hslash}\right)\cdot p(t_1|\bm{S}_j)\, dt_1 + \frac{|\mathcal{I}|}{nb^d\hslash} \int_{\mathcal{T}\times \mathcal{S}} \frac{\bar{K}_T\left(\frac{t_1-t}{\hslash}\right) \cdot K_S\left(\frac{\bm{s}_1-\bm{S}_j}{b}\right)}{p(t,\bm{s}_1)} \cdot p(t_1,\bm{s}_1) \, dt_1 d\bm{s}_1\\
		&= \frac{K_S(0)}{nb^d \cdot p(t,\bm{S}_j)} \int_{\mathbb{R}} \bar{K}_T\left(u\right)\cdot p(t+u\hslash|\bm{S}_j)\, du + \frac{|\mathcal{I}|}{n} \int_{\mathbb{R}\times \mathbb{R}^d} \frac{\bar{K}_T\left(u\right) \cdot K_S\left(\bm{v}\right)}{p(t,\bm{S}_j +b\bm{v})} \cdot p(t+u\hslash,\bm{S}_j+b\bm{v}) \, du d\bm{v}\\
		&\stackrel{\text{(i)}}{=} \frac{K_S(0)\left[p(t|\bm{S}_j) + \frac{\hslash^2 \bar{\kappa}_2^{(T)}}{2}\cdot \frac{\partial^2}{\partial t^2}p(t|\bm{S}_j) + o(\hslash^2)\right]}{nb^d \cdot p(t,\bm{S}_j)} \\
		&\quad + \frac{|\mathcal{I}|}{n}\left[1 + \frac{\bar{\kappa}_2^{(T)}\hslash^2}{2} \int_{\mathbb{R}^d} \frac{\frac{\partial^2}{\partial t^2}p(t,\bm{S}_j+b\bm{v})}{p(t,\bm{S}_j+b\bm{v})}\cdot K_S(\bm{v})\,d\bm{v} + o(\hslash^2)\right]\\
		&\stackrel{\text{(ii)}}{=} \frac{K_S(0)}{nb^d} + \frac{|\mathcal{I}|}{n} + O_P\left(\hslash^2 + \frac{\hslash^2}{nb^d}\right) \\
		&= 1+ O_P\left(\hslash\right)
	\end{align*}
	when $\hslash,h,b\to 0$ and $\frac{nhb^d}{|\log(\hslash hb^d)|}\to \infty$ as $n\to \infty$, where $\bar{\kappa}_2^{(T)} = \int_{\mathbb{R}} u^2 \bar{K}_T(u)\, du$ as well as we apply Taylor's expansions on $p(\cdot|\bm{S}_j)$ and $p$ and use the boundedness of their derivatives under Assumption~\ref{assump:den_diff} in (i). Additionally, its conditional variance can be upper bounded by 
	\begin{align*}
		&\mathrm{Var}\left[\frac{1}{n\hslash b^d} \sum_{i\in \mathcal{I}} \frac{\bar{K}_T\left(\frac{t-T_i}{\hslash}\right)}{p(t,\bm{S}_i)} \cdot K_S\left(\frac{\bm{S}_j-\bm{S}_i}{b}\right) \Big| \bm{S}_j \right] \\
		&\lesssim \frac{1}{n^2b^{2d} \hslash^2} \sum_{i\in \mathcal{I}} \mathbb{E}\left[\frac{\bar{K}_T^2\left(\frac{t-T_i}{\hslash}\right)}{p^2(t,\bm{S}_i)} \cdot K_S^2\left(\frac{\bm{S}_j-\bm{S}_i}{b}\right) \Big| \bm{S}_j \right]\\
		&= \frac{K_S^2(0)}{n^2b^{2d} \hslash^2} \int_{\mathcal{T}} \frac{\bar{K}_T^2\left(\frac{t_1-t}{\hslash}\right)}{p^2(t,\bm{S}_j)} \cdot p(t_1|\bm{S}_j)\, dt_1 + \frac{|\mathcal{I}|}{n^2 b^{2d}\hslash^2}\int_{\mathcal{T}\times \mathcal{S}} \frac{\bar{K}_T^2\left(\frac{t_1-t}{\hslash}\right) \cdot K_S^2\left(\frac{\bm{s}_1-\bm{S}_j}{b}\right)}{p^2(t,\bm{s}_1)}\cdot p(t_1,\bm{s}_1)\, dt_1 d\bm{s}_1 \\
		&= \frac{K_S^2(0)}{n^2b^{2d} \hslash} \int_{\mathbb{R}} \frac{\bar{K}_T^2\left(u\right)}{p^2(t,\bm{S}_j)} \cdot p(t+u\hslash|\bm{S}_j)\, du + \frac{|\mathcal{I}|}{n^2 b^d\hslash}\int_{\mathbb{R}\times \mathbb{R}^d} \frac{\bar{K}_T^2\left(u\right) \cdot K_S^2\left(\bm{v}\right)}{p^2(t,\bm{S}_j+b\bm{v})}\cdot p(t+u\hslash,\bm{S}_j+b\bm{v})\, du d\bm{v}\\
		&= O_P\left(\frac{1}{nb^d\hslash}\right)
	\end{align*}
	when $\hslash,h,b\to 0$ and $\frac{n\max\{h,\hslash\}b^d}{|\log(\hslash hb^d)|}\to \infty$ as $n\to \infty$, where the last equality again uses the boundedness of $p(t|\bm{S}_j)$ and $p$ under Assumption~\ref{assump:den_diff} and the fact that $(t,\bm{S}_i) \in \mathcal{E} \cap \left(\partial \mathcal{E} \ominus \hslash\right)$ for $i\in \mathcal{I}$. Thus,
	$$\frac{1}{n\hslash b^d} \sum_{i\in \mathcal{I}} \frac{\bar{K}_T\left(\frac{t-T_i}{\hslash}\right)}{p(t,\bm{S}_i)} \cdot K_S\left(\frac{\bm{S}_j-\bm{S}_i}{b}\right) = 1 + O_P(\hslash) + O_P\left(\sqrt{\frac{1}{nb^d\hslash}}\right).$$
	Plugging the above results back into \eqref{TermA}, we obtain that
	\begin{align}
		\label{TermA2}
		\begin{split}
			\textbf{Term A} &= \frac{1}{nh^2} \sum_{j=1}^n C_{K_T} \left(\frac{T_j-t}{h}\right) K_T\left(\frac{T_j-t}{h}\right)\\
			&\quad \times \left\{\left[\frac{1}{P_T(t)} + O_P(\hslash) +O_P\left(\sqrt{\frac{|\log \hslash|}{n\hslash}} + \sqrt{\frac{1}{nb^d\hslash}}\right) \right] \left[\epsilon_j+A_{h,b}+\sum_{i\in \mathcal{I}} r_{(t,\bm{S}_i)}(T_j,\bm{S}_j)\right] \right\}
		\end{split}
	\end{align}
	
	Now, we know from Lemma~\ref{lem:loc_poly_deriv_unif} that 
	$$\sup_{(t,\bm{s})\in \mathcal{E}} \left|\hat{\beta}_2(t,\bm{s}) - \beta_2(t,\bm{s}) \right| = O\left(h^q + b^2 + \frac{\max\{b,h\}^4}{h}\right) + O_P\left(\sqrt{\frac{|\log(hb^d)|}{nh^3 b^d}}\right)$$
	when $q>0$ is even. Since \textbf{Term A} is a weighted sum of $\hat{\beta}_2(t,\bm{S}_i) - \beta_2(t,\bm{S}_i)$, the bias of \textbf{Term A}, which comes from the term involving the residual sum $\sum_{i\in\mathcal{I}} r_{(t,\bm{S}_i)}(T_j,S_j)$, will not be larger than the (uniform) bias of $\hat{\beta}_2(t,\bm{s}) - \beta_2(t,\bm{s})$, \emph{i.e.}, 
	$$\left|\mathrm{Bias}\left[\textbf{Term A}\right]\right| = O\left(h^q + b^2 + \frac{\max\{b,h\}^4}{h}\right).$$
	Thus, it remains to derive the rate of convergence for the stochastic variation of \eqref{TermA2}. Note that when $\hslash,h,b\to 0$ and $\frac{n\max\{h,\hslash\}b^d}{|\log(\hslash hb^d)|}, \frac{|\log(h\hslash b^d)|}{\log\log n}\to \infty$, the variance of \textbf{Term A} is dominated by
	\begin{align*}
		\mathrm{Var}\left[\textbf{Term A}\right] &\lesssim \mathrm{Var}\left[\frac{1}{nh^2} \sum_{j=1}^n \frac{C_{K_T}}{p_T(t)} \left(\frac{T_j-t}{h}\right) K_T\left(\frac{T_j-t}{h}\right)\right]\\
		&\lesssim \frac{C_{K_T}^2}{nh^4\cdot p_T^2(t)}\cdot \mathbb{E}\left[\left(\frac{T_j-t}{h}\right)^2 K_T^2\left(\frac{T_j-t}{h}\right) \right]\\
		&= \frac{C_{K_T}^2}{nh^4\cdot p_T^2(t)} \int_{\mathcal{T}} \left(\frac{t_1-t}{h}\right)^2 K_T^2\left(\frac{t_1-t}{h}\right)\cdot p_T(t_1)\, dt_1\\
		&= \frac{C_{K_T}^2}{nh^3\cdot p_T^2(t)} \int_{\mathbb{R}} u^2 K_T^2\left(u\right)\cdot p_T(t+uh)\, du\\
		&=O\left(\frac{1}{nh^3}\right)
	\end{align*}
	for any $t\in \mathcal{T}'$. Furthermore, under our Assumption~\ref{assump:reg_kernel}(b), the result in Theorem 4 of \cite{einmahl2005uniform} can be applied to derive that
	\begin{align*}
		&\sup_{t\in \mathcal{T}'}\left|\frac{1}{nh^2} \sum_{j=1}^n \frac{C_{K_T}}{p_T(t)} \left(\frac{T_j-t}{h}\right) K_T\left(\frac{T_j-t}{h}\right) - \mathbb{E}\left[\frac{1}{nh^2} \sum_{j=1}^n \frac{C_{K_T}}{p_T(t)} \left(\frac{T_j-t}{h}\right) K_T\left(\frac{T_j-t}{h}\right)\right] \right| \\
		&= O_P\left(\sqrt{\frac{|\log h|}{nh^3}}\right)
	\end{align*}
	when $\frac{nh^3}{|\log h|}, \frac{|\log h|}{\log\log n}\to \infty$; see also \cite{gine2002rates,chacon2011asymptotics}. Therefore, combining with the bias rate of \textbf{Term A}, we obtain that
	$$\sup_{t\in \mathcal{T}'}\left|\textbf{Term A} \right| = O\left(h^q + b^2 + \frac{\max\{b,h\}^4}{h}\right) + O_P\left(\sqrt{\frac{|\log h|}{nh^3}}\right)$$
	when $\hslash,h,b, \frac{\max\{h,b\}^4}{h}\to 0$ and $\frac{n\max\{h,\hslash\}b^d}{|\log(\hslash hb^d)|}, \frac{|\log(h\hslash b^d)|}{\log\log n}, \frac{nh^3}{|\log h|}\to \infty$.\\
	
	As for \textbf{Term B}, we leverage the compact support of $\bar{K}_T$ (Assumption~\ref{assump:reg_kernel}(c)) to argue that $|t-T_i|=O(\hslash)$ for any $(t,\bm{S}_i) \in \partial \mathcal{E} \oplus \hslash$. Then, under Assumptions~\ref{assump:reg_diff} and \ref{assump:reg_kernel}(a), we know that both $\hat{\beta}_2(t,\bm{s})$ and $\beta_2(t,\bm{s})$ are Lipschitz so that
	\begin{align*}
		&\left|\hat{\beta}_2(t,\bm{S}_i) -\beta_2(t,\bm{S}_i)\right| \\
		&\leq \left|\hat{\beta}_2(t,\bm{S}_i) -\hat{\beta}_2(T_i,\bm{S}_i)\right| + \left|\hat{\beta}_2(T_i,\bm{S}_i) -\beta_2(T_i,\bm{S}_i)\right| + \left|\beta_2(T_i,\bm{S}_i) -\beta_2(t,\bm{S}_i)\right|\\
		&\leq O_P\left(\hslash\right) + \left|\hat{\beta}_2(T_i,\bm{S}_i) -\beta_2(T_i,\bm{S}_i)\right|.
	\end{align*}
	In addition, by Assumption~\ref{assump:boundary_cond}(c), we know that the proportion of $(t,\bm{S}_i),i=1,...,n$ that lie in $\partial \mathcal{E} \oplus \hslash$ is of order $O_P(\hslash)$. Thus, $\sum\limits_{(t,\bm{S}_i) \in \partial \mathcal{E} \oplus \hslash} \frac{\bar{K}_T\left(\frac{t-T_i}{\hslash}\right)}{\sum_{j=1}^n \bar{K}_T\left(\frac{t-T_j}{\hslash}\right)} =O_P(\hslash)$ and hence,
	\begin{align*}
		\textbf{Term B} &\leq \sum_{(t,\bm{S}_i) \in \partial \mathcal{E} \oplus \hslash} \frac{\left|\hat{\beta}_2(t,\bm{S}_i) - \beta_2(t,\bm{S}_i)\right| \bar{K}_T\left(\frac{t-T_i}{\hslash}\right)}{\sum_{j=1}^n \bar{K}_T\left(\frac{t-T_j}{\hslash}\right)} \\
		&\leq \sum_{(t,\bm{S}_i) \in \partial \mathcal{E} \oplus \hslash} \frac{\left[\left|\hat{\beta}_2(t,\bm{S}_i) - \hat{\beta}_2(T_i,\bm{S}_i)\right| + \left|\beta_2(T_i,\bm{S}_i) - \beta_2(t,\bm{S}_i)\right| \right] \bar{K}_T\left(\frac{t-T_i}{\hslash}\right)}{\sum_{j=1}^n \bar{K}_T\left(\frac{t-T_j}{\hslash}\right)} \\
		&\quad + \sum_{(t,\bm{S}_i) \in \partial \mathcal{E} \oplus \hslash} \frac{\left|\hat{\beta}_2(T_i,\bm{S}_i) - \beta_2(T_i,\bm{S}_i)\right| \bar{K}_T\left(\frac{t-T_i}{\hslash}\right)}{\sum_{j=1}^n \bar{K}_T\left(\frac{t-T_j}{\hslash}\right)} \\
		& = O_P(\hslash^2) + \sum_{(t,\bm{S}_i) \in \partial \mathcal{E} \oplus \hslash} \frac{\left|\hat{\beta}_2(T_i,\bm{S}_i) - \beta_2(T_i,\bm{S}_i)\right| \bar{K}_T\left(\frac{t-T_i}{\hslash}\right)}{\sum_{j=1}^n \bar{K}_T\left(\frac{t-T_j}{\hslash}\right)}\\
		&\stackrel{\text{(iii)}}{=}  O\left(h^q + b^2 + \frac{\max\{b,h\}^4}{h}\right) + O_P\left(\hslash^2 + \sqrt{\frac{|\log h|}{nh^3}}\right)
	\end{align*}
	when $q>0$ is even, $\hslash,h,b, \frac{\max\{h,b\}^4}{h}\to 0$, and $\frac{n\max\{h,\hslash\}b^d}{|\log(\hslash hb^d)|}, \frac{|\log(h\hslash b^d)|}{\log\log n}, \frac{nh^3}{|\log h|}\to \infty$. Here, in (iii), we can use the same arguments for \textbf{Term A} to argue that the (uniform) rate of convergence of $\sum_{(t,\bm{S}_i) \in \partial \mathcal{E} \oplus \hslash} \frac{\left|\hat{\beta}_2(T_i,\bm{S}_i) - \beta_2(T_i,\bm{S}_i)\right| \bar{K}_T\left(\frac{t-T_i}{\hslash}\right)}{\sum_{j=1}^n \bar{K}_T\left(\frac{t-T_j}{\hslash}\right)}$ is at most $O\left(h^q + b^2 + \frac{\max\{b,h\}^4}{h}\right) + O_P\left(\sqrt{\frac{|\log h|}{nh^3}}\right)$, because $\hat{\beta}_2$ and $\beta_2$ are evaluated on the sample points $(T_i,\bm{S}_i)\in \mathcal{E}$.\\
	
	As for \textbf{Term C}, under Assumptions~\ref{assump:reg_diff} and \ref{assump:reg_kernel}(d), we utilize Theorem 3 in \cite{einmahl2005uniform} to derive its uniform rate of convergence as:
	$$\textbf{Term C} \leq \sup_{(t,\bm{s})\in \mathcal{T}\times \mathcal{S}}|\beta_2(t,\bm{s})| \cdot \sup_{(t,\bm{s})\in \mathcal{T}'\times \mathcal{S}} \left|\hat{P}(\bm{s}|t) - P(\bm{s}|t) \right| \cdot |\mathcal{S}|= O(\hslash^2) + O_P\left(\sqrt{\frac{|\log \hslash|}{n\hslash}}\right),$$
	where the Lebesgue measure $|\mathcal{S}|$ is finite due to the compactness of $\mathcal{S}$. Therefore, 
	\begin{align*}
		&\sup_{t\in \mathcal{T}'} \left|\hat{\theta}_C(t) - \theta_C(t)\right| = O\left(h^q + b^2 + \frac{\max\{b,h\}^4}{h}\right) + O_P\left(\sqrt{\frac{|\log(h)|}{nh^3}} + \hslash^2 + \sqrt{\frac{|\log \hslash|}{n\hslash}}\right)
	\end{align*}
	when $q>0$ is even, $h,b,\hslash,\frac{\max\{b,h\}^4}{h}\to 0$, and $\frac{n\max\{h,\hslash\}b^d}{|\log(\hslash hb^d)|}, \frac{|\log(h\hslash b^d)|}{\log\log n}, \frac{nh^3}{|\log h|}\to \infty$.\\
	
	Finally, plugging the results back into \eqref{si_decomp}, we obtain that
	\begin{align*}
		&\sup_{t\in \mathcal{T}'}\left|\hat{\Delta}_{h,b,\hslash}(t) - \Delta(t)\right| \\
		&\leq \sup_{t\in \mathcal{T}'}\left| \int \int_{\tilde{t}=\tau}^{\tilde{t}=t} \theta_C(\tilde{t})d\tilde{t} \left[d\mathbb{P}_{n,T}(\tau) - d \P_T(\tau)\right] \right| + \sup_{t\in \mathcal{T}'} \left|\frac{1}{n}\sum_{i=1}^n \int_{T_i}^t \left[\hat{\theta}_C(\tilde{t}) - \theta_C(\tilde{t})\right] d\tilde{t}\right|\\
		&\leq \sup_{t\in \mathcal{T}'}\left| \int \int_{\tilde{t}=\tau}^{\tilde{t}=t} \theta_C(\tilde{t})d\tilde{t} \left[d\mathbb{P}_{n,T}(\tau) - d \P_T(\tau)\right] \right| + \sup_{t_1,t_2\in \mathcal{T}}|t_1-t_2| \cdot \sup_{\tilde{t}\in \mathcal{T}'}\left|\hat{\theta}_C(\tilde{t}) - \theta_C(\tilde{t})\right|\\
		&= O_P\left(\frac{1}{\sqrt{n}}\right) + O\left(h^q + b^2 + \frac{\max\{b,h\}^4}{h}\right) + O_P\left(\sqrt{\frac{|\log h|}{nh^3 }} + \hslash^2 + \sqrt{\frac{|\log \hslash|}{n\hslash}}\right)
	\end{align*}
	when $q>0$ is even, $h,b,\hslash,\frac{\max\{b,h\}^4}{h}\to 0$, and $\frac{n\max\{h,\hslash\}b^d}{|\log(\hslash hb^d)|}, \frac{|\log(h\hslash b^d)|}{\log\log n}, \frac{nh^3}{|\log h|}\to \infty$,	where $\sup_{t_1,t_2\in \mathcal{T}}|t_1-t_2| < \infty$ due to the compactness of $\mathcal{T}=\mathrm{proj}_T(\mathcal{E})$. The proof is thus completed.
\end{proof}

\subsection{Proof of Lemma~\ref{lem:asymp_linear}}
\label{app:asymp_linear_proof}

\begin{customlem}{5}[Asymptotic linearity]
	Let $q\geq 2$ in the local polynomial regression for estimating $\frac{\partial}{\partial t}\mu(t,\bm{s})$ and $\mathcal{T}' \subset \mathcal{T}$ be a compact set so that $p_T(t)$ is uniformly bounded away from 0 within $\mathcal{T}'$. Suppose that Assumptions~\ref{assump:identify_cond}, \ref{assump:diff_inter}, \ref{assump:reg_diff}, \ref{assump:den_diff}, \ref{assump:boundary_cond}, and \ref{assump:reg_kernel} hold. Then, if $b\lesssim h\asymp n^{-\frac{1}{\gamma}}$ and $\hslash \asymp n^{-\frac{1}{\varpi}}$ for some $\gamma, \varpi >0$ such that $\frac{nh^5}{\log n} \to c_1$ and $\frac{n\hslash^5}{\log n} \to c_2$ for some finite number $c_1,c_2 \geq 0$ and $\sqrt{nh^3 \max\{h, \hslash\}^4}, \frac{h^3\log n}{\hslash},\frac{\log n}{\sqrt{n\hslash}}, \frac{\log n}{n\max\{h,\hslash\} b^d} \to 0$ as $n\to \infty$, then for any $t\in \mathcal{T}'$, we have that
	$$\sqrt{nh^3} \left[\hat{\theta}_C(t) -\theta(t) \right] = \mathbb{G}_n \bar{\varphi}_t + o_P(1) \quad \text{ and } \quad \sqrt{nh^3} \left[\hat{m}_{\theta}(t) -m(t) \right] = \mathbb{G}_n \varphi_t + o_P(1),$$
	where $\bar{\varphi}_t(Y, T, \bm{S}) = \frac{C_{K_T}[Y -\mu(T,\bm{S})]}{\sqrt{h}\cdot p_T(t)} \left(\frac{T-t}{h}\right) K_T\left(\frac{T-t}{h}\right)$ for some constant $C_{K_T}>0$ that only depends on the kernel $K_T$ and 
	\begin{align*}
		\varphi_t\left(Y,T,\bm{S}\right) &= \mathbb{E}_{T_1}\left[\int_{T_1}^t \bar{\varphi}_{\tilde{t}}(Y, T, \bm{S}) \, d\tilde{t}\right] = \mathbb{E}_{T_1}\left\{\int_{T_1}^t \frac{C_{K_T} \left[Y-\mu(T,\bm{S})\right]}{\sqrt{h}\cdot p_T(\tilde{t})} \left(\frac{T-\tilde{t}}{h}\right) K_T\left(\frac{T-\tilde{t}}{h}\right) d\tilde{t}\right\}.
	\end{align*}
	Furthermore, we have the following uniform results as:
	\begin{align*}
		&\left|\sqrt{nh^3}\sup_{t\in \mathcal{T}'}\left|\hat{\theta}_C(t) -\theta(t)\right| - \sup_{t\in \mathcal{T}'} |\mathbb{G}_n\bar{\varphi}_t|\right| \asymp \left|\sqrt{nh^3} \sup_{t\in \mathcal{T}'}\left|\hat{m}_{\theta}(t) -m(t)\right| - \sup_{t\in \mathcal{T}'}|\mathbb{G}_n\varphi_t|\right| \\
		&= O_P\left(\sqrt{nh^3 \max\{h, \hslash\}^4} + \sqrt{\frac{h^3\log n}{\hslash}} + \frac{\log n}{\sqrt{n\hslash}} + \sqrt{\frac{\log n}{nb^d\hslash}}\right).
	\end{align*}
\end{customlem}

\begin{proof}[Proof of Lemma~\ref{lem:asymp_linear}]
	The entire proof consists of three steps: \textbf{Step 1} establishes the asymptotic linearity of $\sqrt{nh^3} \left[\hat{\theta}_C(t) -\theta_C(t) \right]$ for any $t\in \mathcal{T}'$. \textbf{Step 2} derives the exact rate of convergence for the coupling between $\sqrt{nh^3} \sup\limits_{t\in \mathcal{T}'}\left|\hat{\theta}_C(t) - \theta_C(t) \right|$ and $\sup\limits_{t\in \mathcal{T}'} \left|\mathbb{G}_n\bar{\varphi}_t\right|$. \textbf{Step 3} extends the arguments in the previous two steps using the form of a V-statistics \citep{shieh2014u} to prove the asymptotic linearity of $\sqrt{nh^3} \left[\hat{m}_{\theta}(t) -m(t) \right]$. Note that $\theta(t)=\theta_C(t)$ under Assumption~\ref{assump:diff_inter}.\\
	
	$\bullet$ {\bf Step 1:} Recalling from \eqref{TermA2} and our calculations for \textbf{Term B} and \textbf{Term C} in the proof of \autoref{thm:int_est_rate}, we know that
	\begin{align*}
		\hat{\theta}_C(t) - \theta_C(t) &= \frac{1}{nh^2} \sum_{j=1}^n C_{K_T} \left(\frac{T_j-t}{h}\right) K_T\left(\frac{T_j-t}{h}\right)\\
		&\quad\quad \times \left\{\left[\frac{1}{P_T(t)} + O_P(\hslash) +O_P\left(\sqrt{\frac{|\log \hslash|}{n\hslash}} + \sqrt{\frac{1}{nb^d\hslash}}\right) \right] \left[\epsilon_j+A_{h,b}\right] \right\} \\
		&\quad + O\left(h^q + b^2 + \frac{\max\{b,h\}^4}{h}\right) + O_P\left(\hslash^2+ \sqrt{\frac{|\log \hslash|}{n\hslash}}\right),
	\end{align*}
	where $A_{h,b}=O(\max\{h,b\}) + O_P\left(\sqrt{\frac{|\log(hb^d)|}{nhb^d}}\right)$ uniformly over $t\in \mathcal{T}'$ and $C_{K_T}>0$ is a constant that only depends on the kernel $K_T$. Then, when $q\geq 2$, 
	\begin{align*}
		&\sqrt{nh^3} \left[\hat{\theta}_C(t) - \theta_C(t) \right] \\
		&= \frac{1}{\sqrt{nh}} \sum_{j=1}^n \frac{C_{K_T}\cdot \epsilon_j}{p_T(t)} \left(\frac{T_j-t}{h}\right) K_T\left(\frac{T_j-t}{h}\right) \\
		&\quad + O_P\left(\sqrt{nh^7} + \sqrt{|\log h|}\right)\cdot O_P\left(\hslash + \sqrt{\frac{|\log \hslash|}{n\hslash}} + \sqrt{\frac{1}{nb^d\hslash}}\right) \\
		&\quad + O\left(\sqrt{nh^3\max\{h^{2q}, b^4\}} + \sqrt{nh\max\{h,b\}^8}\right) + O_P\left(\sqrt{nh^3\hslash^4} + \sqrt{\frac{h^3|\log \hslash|}{\hslash}}\right)\\
		&= \frac{1}{\sqrt{nh}} \sum_{j=1}^n \frac{C_{K_T}[Y_j-\mu(T_j,\bm{S}_j)]}{p_T(t)} \left(\frac{T_j-t}{h}\right) K_T\left(\frac{T_j-t}{h}\right) + o_P(1)\\
		&= \mathbb{G}_n \bar{\varphi}_t + o_P(1)
	\end{align*}
	when $n,b,\hslash,nh^3\max\{h^{2q}, b^4, \hslash^4\}, nh\max\{h,b\}^8, nh^7\hslash, \frac{h^3|\log\hslash|}{\hslash} \to 0$ and $\frac{n\hslash}{|\log \hslash|\cdot |\log h|}, \frac{n\max\{h,\hslash\} b^d}{|\log(\hslash h b^d)|} \to \infty$, where 
	$\bar{\varphi}_t(Y,T,\bm{S}) = \frac{C_{K_T}[Y -\mu(T,\bm{S})]}{\sqrt{h}\cdot p_T(t)} \left(\frac{T-t}{h}\right) K_T\left(\frac{T-t}{h}\right)$ and $\mathbb{G}_n=\sqrt{n}\left(\mathbb{P}_n-\P\right)$.\\
	
	$\bullet$ {\bf Step 2:} From our results in \textbf{Step 1}, we know that
	\begin{align*}
		&\sqrt{nh^3} \sup_{t\in \mathcal{T}'}\left|\hat{\theta}_C(t) - \theta_C(t) \right| \\
		&\leq \sup_{t\in \mathcal{T}'}\left|\mathbb{G}_n \bar{\varphi}_t \right| + O_P\left(\sqrt{nh^7} + \sqrt{|\log h|}\right)\cdot O_P\left(\hslash + \sqrt{\frac{|\log \hslash|}{n\hslash}} + \sqrt{\frac{1}{nb^d\hslash}}\right)\\
		&\quad + O\left(\sqrt{nh^3\max\{h^{2q}, b^4\}} + \sqrt{nh\max\{h,b\}^8}\right) + O_P\left(\sqrt{nh^3\hslash^4} + \sqrt{\frac{h^3|\log \hslash|}{\hslash}}\right).
	\end{align*}
	When $q\geq 2$, $b\lesssim h\asymp n^{-\frac{1}{\gamma}}$, and $\hslash \asymp n^{-\frac{1}{\varpi}}$ for some $\gamma, \varpi >0$ such that $\frac{nh^5}{\log n} \to c_1$ and $\frac{n\hslash^5}{\log n} \to c_2$ for some finite number $c_1,c_2 \geq 0$, we can further simplify the above rate of convergence as:
	\begin{align*}
		&\left|\sqrt{nh^3} \sup_{t\in \mathcal{T}'}\left|\hat{\theta}_C(t) - \theta_C(t) \right| -  \sup_{t\in \mathcal{T}'}\left|\mathbb{G}_n \bar{\varphi}_t \right| \right| = O_P\left(\sqrt{nh^3 \max\{h, \hslash\}^4} + \sqrt{\frac{h^3\log n}{\hslash}} + \frac{\log n}{\sqrt{n\hslash}} + \sqrt{\frac{\log n}{nb^d\hslash}}\right),
	\end{align*}
	provided that $\sqrt{nh^3 \max\{h, \hslash\}^4}, \frac{h^3\log n}{\hslash},\frac{\log n}{\sqrt{n\hslash}}, \frac{\log n}{n\max\{h,\hslash\} b^d} \to 0$.\\
	
	$\bullet$ {\bf Step 3:} Recalling from \eqref{simp_integral} for the definition of the integral estimator $\hat{m}_{\theta}(t)$ and \eqref{si_decomp}, we know that
	\begin{align*}
		&\hat m_\theta(t) -m(t) \\
		&= \frac{1}{n}\sum_{i=1}^n \left[Y_i + \int_{\tilde{t}=T_i}^{\tilde{t}=t} \hat \theta_C(\tilde{t})d\tilde{t} \right] - \mathbb{E}\left[m(T)\right] - \mathbb{E}\left[\int_T^t \theta_C(\tilde{t}) \, d\tilde{t}\right]\\
		&= \underbrace{\frac{1}{n}\sum_{i=1}^n Y_i - \mathbb{E}(Y)}_{\textbf{Term I}} + \underbrace{\frac{1}{n}\sum_{i=1}^n \int_{T_i}^t \theta_C(\tilde{t})\, d\tilde{t} - \mathbb{E}\left[\int_T^t \theta_C(\tilde{t}) \, d\tilde{t}\right]}_{\textbf{Term II}} + \frac{1}{n}\sum_{i=1}^n \int_{T_i}^t \left[\hat{\theta}_C(\tilde{t}) - \theta_C(\tilde{t})\right] d\tilde{t}
	\end{align*}
	where we use Assumption~\ref{assump:diff_inter} to obtain the second equality above. 
	
	As for \textbf{Term I} and \textbf{Term II}, under Assumption~\ref{assump:reg_diff} and the condition $\mathbb{E}\left(\epsilon^4\right) < \infty$, it is valid that $\mathrm{Var}(Y) <\infty$ and $\mathrm{Var}\left[\int_T^t \theta_C(\tilde{t})\, d\tilde{t}\right] < \infty$. Thus, by Chebyshev's inequality, we know that
	$$\frac{1}{n}\sum_{i=1}^n Y_i - \mathbb{E}(Y) = O_P\left(\frac{1}{\sqrt{n}}\right) \quad \text{ and } \quad \frac{1}{n} \sum_{i=1}^n \int_{T_i}^t \theta_C(\tilde{t})\, d\tilde{t} - \mathbb{E}\left[\int_T^t \theta_C(\tilde{t})\, d\tilde{t}\right]= O_P\left(\frac{1}{\sqrt{n}}\right).$$
	Furthermore, under Assumption~\ref{assump:reg_diff}, $\left|\int_T^{t_1} \theta_C(\tilde{t})\, d\tilde{t} - \int_T^{t_2} \theta_C(\tilde{t})\, d\tilde{t} \right| \leq \sup_{t\in \mathcal{T}}|\theta_C(t)| \cdot |t_1-t_2|$. Together with the compactness of $\mathcal{T}$ and Example 19.7 in \cite{VDV1998}, we also deduce that
	$$\sup_{t\in \mathcal{T}}\left|\frac{1}{n} \sum_{i=1}^n \int_{T_i}^t \theta_C(\tilde{t})\, d\tilde{t} - \mathbb{E}\left[\int_T^t \theta_C(\tilde{t})\, d\tilde{t}\right] \right|= O_P\left(\frac{1}{\sqrt{n}}\right).$$
	Therefore, together with our results in \textbf{Step 1}, we have that
	\begin{align*}
		&\hat m_\theta(t) -m(t) \\
		&= \frac{1}{n}\sum_{i=1}^n \int_{T_i}^t \left[\hat{\theta}_C(\tilde{t}) - \theta_C(\tilde{t})\right] d\tilde{t} + O_P\left(\frac{1}{\sqrt{n}}\right)\\
		&= \frac{1}{n^2h^2} \sum_{i_1=1}^n \sum_{i_2=1}^n \int_{T_{i_1}}^t \frac{C_{K_T} \cdot \epsilon_{i_2}}{p_T(\tilde{t})}\left(\frac{T_{i_2} -\tilde{t}}{h}\right) K_T\left(\frac{T_{i_2} -\tilde{t}}{h}\right) d\tilde{t} + \frac{1}{n} \sum_{i=1}^n \int_{T_i}^t \bar{A}_{h,\hslash,b} \,d\tilde{t} + O_P\left(\frac{1}{\sqrt{n}}\right),
	\end{align*}
	where $\bar{A}_{h,\hslash,b} = O_P\left(\max\{h^2,\hslash^2\}, \sqrt{\frac{\log n}{n\hslash}} + \frac{\log n}{n\sqrt{\hslash h^3}} + \frac{\sqrt{\log n}}{n\sqrt{b^d\hslash h^3}}\right)$ uniformly over $t\in \mathcal{T}'$ under our conditions on the bandwidth parameters $h,\hslash,b>0$. Since $\frac{1}{n} \sum_{i=1}^n \int_{T_i}^t \bar{A}_{h,\hslash,b} \,d\tilde{t}\leq |\bar{A}_{h,\hslash,b}| \cdot \left(\frac{1}{n}\sum_{i=1}^n |t-T_i|\right) = |\bar{A}_{h,\hslash,b}| \cdot O_P(1)$, we know that
	\begin{align*}
		\sqrt{nh^3}\left[\hat{m}_\theta(t) -m(t)\right]
		&= \underbrace{\frac{1}{n^{\frac{3}{2}}} \sum_{i_1=1}^n \sum_{i_2=1}^n \int_{T_{i_1}}^t \frac{C_{K_T} \left[Y_{i_2} - \mu(T_{i_2},\bm{S}_{i_2})\right]}{\sqrt{h}\cdot p_T(\tilde{t})}\left(\frac{T_{i_2} -\tilde{t}}{h}\right) K_T\left(\frac{T_{i_2} -\tilde{t}}{h}\right) d\tilde{t}}_{\textbf{Term III}}\\
		&\quad + \underbrace{O_P\left(\sqrt{nh^3 \max\{h, \hslash\}^4} + \sqrt{\frac{h^3\log n}{\hslash}} + \frac{\log n}{\sqrt{n\hslash}} + \sqrt{\frac{\log n}{nb^d\hslash}}\right)}_{\textbf{Term IV}},
	\end{align*}
	where \textbf{Term IV} is of order $o_P\left(1\right)$ by our results in \textbf{Step 2}. Now, to derive the asymptotic linearity of $\hat{m}_\theta(t) -m(t)$, it remains to tackle \textbf{Term III} above, which takes a form of V-statistics with a symmetric ``kernel'' defined as:
	\begin{align*}
		\Lambda_t(\bm{U}_{i_1}, \bm{U}_{i_2}) &= \int_{T_{i_1}}^t \frac{C_{K_T} \left[Y_{i_2} - \mu(T_{i_2},\bm{S}_{i_2})\right]}{2\sqrt{h}\cdot p_T(\tilde{t})}\left(\frac{T_{i_2} -\tilde{t}}{h}\right) K_T\left(\frac{T_{i_2} -\tilde{t}}{h}\right) d\tilde{t} \\ 
		&\quad + \int_{T_{i_2}}^t \frac{C_{K_T} \left[Y_{i_1} - \mu(T_{i_1},\bm{S}_{i_1})\right]}{2\sqrt{h}\cdot p_T(\tilde{t})}\left(\frac{T_{i_1} -\tilde{t}}{h}\right) K_T\left(\frac{T_{i_1} -\tilde{t}}{h}\right) d\tilde{t}
	\end{align*}
	with $\bm{U}_i = \left(Y_i,T_i,\bm{S}_i\right)$ for $i=1,...,n$. By Pascal's rule, we know that
	\begin{align*}
		\textbf{Term III} &= 2\sqrt{n}\left(\mathbb{P}_n-\P\right) \P\lambda_t + \sqrt{n}\left(\mathbb{P}_n - \P\right)^2 \Lambda_t + \sqrt{n}\cdot \P^2\Lambda_t \\
		&= \mathbb{E}_{T_1}\left\{\int_{T_1}^t \mathbb{G}_n\left[\frac{C_{K_T} \left[Y-\mu(T,\bm{S})\right]}{\sqrt{h}\cdot p_T(\tilde{t})} \left(\frac{T-\tilde{t}}{h}\right) K_T\left(\frac{T-\tilde{t}}{h}\right)\right] d\tilde{t}\right\} \\
		&\quad + \mathbb{G}_n\left\{\int_{T}^t \mathbb{E}_{\bm{U}_1}\left[\frac{C_{K_T} \left[Y_1-\mu(T_1,\bm{S}_1)\right]}{\sqrt{h}\cdot p_T(\tilde{t})} \left(\frac{T_1-\tilde{t}}{h}\right) K_T\left(\frac{T_1-\tilde{t}}{h}\right)\right]d\tilde{t} \right\} \\
		&\quad + \sqrt{n}\left(\mathbb{P}_n - \P\right)^2\left\{\int_{T_{i_1}}^t \frac{C_{K_T} \left[Y_{i_2} - \mu(T_{i_2},\bm{S}_{i_2})\right]}{2\sqrt{h}\cdot p_T(\tilde{t})}\left(\frac{T_{i_2} -\tilde{t}}{h}\right) K_T\left(\frac{T_{i_2} -\tilde{t}}{h}\right) d\tilde{t} \right\} \\
		&\quad + \sqrt{n} \cdot \mathbb{E}_{T_1}\left\{\int_{T_1}^t \mathbb{E}_{\bm{U}_2}\left[\frac{C_{K_T} \left[Y_2 - \mu(T_2,\bm{S}_2)\right]}{2\sqrt{h}\cdot p_T(\tilde{t})}\left(\frac{T_2 -\tilde{t}}{h}\right) K_T\left(\frac{T_2 -\tilde{t}}{h}\right)\right] d\tilde{t} \right\}\\
		&= \underbrace{\mathbb{G}_n\left\{\mathbb{E}_{T_1}\left[\int_{T_1}^t \frac{C_{K_T} \left[Y-\mu(T,\bm{S})\right]}{\sqrt{h}\cdot p_T(\tilde{t})} \left(\frac{T-\tilde{t}}{h}\right) K_T\left(\frac{T-\tilde{t}}{h}\right) d\tilde{t}\right] \right\}}_{\textbf{Term IIIa}}\\
		&\quad + \underbrace{\sqrt{n}\left(\mathbb{P}_n - \P\right)^2\left\{\int_{T_{i_1}}^t \frac{C_{K_T} \left[Y_{i_2} - \mu(T_{i_2},\bm{S}_{i_2})\right]}{2\sqrt{h}\cdot p_T(\tilde{t})}\left(\frac{T_{i_2} -\tilde{t}}{h}\right) K_T\left(\frac{T_{i_2} -\tilde{t}}{h}\right) d\tilde{t} \right\}}_{\textbf{Term IIIb}},
	\end{align*}
	because $\mathbb{E}_{\bm{U}_2}\left[\frac{C_{K_T} \left[Y_2 - \mu(T_2,\bm{S}_2)\right]}{2\sqrt{h}\cdot p_T(\tilde{t})}\left(\frac{T_2 -\tilde{t}}{h}\right) K_T\left(\frac{T_2 -\tilde{t}}{h}\right)\right]=0$ for any $\tilde{t}\in \mathcal{T}$. By Lemma~\ref{lem:influc_func_moment_bnd}, we know that $\mathrm{Var}\left[\textbf{Term IIIa}\right]$ is a finite positive number, so \textbf{Term IIIa} is the asymptotically dominating term. 
	
	As for \textbf{Term IIIb}, we recall from Assumption~\ref{assump:reg_kernel}(b) by holding $\bm{z}$ that 
	$$\mathcal{K} = \left\{t'\mapsto \left(\frac{t'-t}{h}\right)^{k_1} K\left(\frac{t'-t}{h}\right): t\in \mathcal{T}, h>0, k_1=0,1\right\}$$ 
	is a bounded VC-type class of measurable functions on $\mathbb{R}$. Under Assumption~\ref{assump:reg_diff} and the condition that $\mathbb{E}(\epsilon^4)<\infty$, we deduce by Theorem 4 in \cite{einmahl2005uniform} that with probability 1,
	\begin{align*}
		\sup_{\tilde{t}\in \mathcal{T}'} \left|\mathbb{G}_n\left[\frac{C_{K_T} \left[Y-\mu(T,\bm{S})\right]}{\sqrt{h}\cdot p_T(\tilde{t})} \left(\frac{T-\tilde{t}}{h}\right) K_T\left(\frac{T-\tilde{t}}{h}\right) \right] \right| = O_P\left(\sqrt{\log n}\right)
	\end{align*}
	when $h\asymp n^{-\frac{1}{\gamma}}$. Thus, by Chebyshev's inequality,
	\begin{align*}
		\textbf{Term IIIb} &= \left(\mathbb{P}_n - \P\right)\left\{\int_{T_{i_1}}^t \mathbb{G}_n\left[\frac{C_{K_T} \left[Y_{i_2} - \mu(T_{i_2},\bm{S}_{i_2})\right]}{2\sqrt{h}\cdot p_T(\tilde{t})}\left(\frac{T_{i_2} -\tilde{t}}{h}\right) K_T\left(\frac{T_{i_2} -\tilde{t}}{h}\right) \right] d\tilde{t} \right\}\\
		&= O_P\left(\sqrt{\frac{1}{n} \cdot \mathbb{E}\left\{\left[\int_{T_{i_1}}^t \mathbb{G}_n\left[\frac{C_{K_T} \left[Y_{i_2} - \mu(T_{i_2},\bm{S}_{i_2})\right]}{2\sqrt{h}\cdot p_T(\tilde{t})}\left(\frac{T_{i_2} -\tilde{t}}{h}\right) K_T\left(\frac{T_{i_2} -\tilde{t}}{h}\right) \right] d\tilde{t}\right]^2 \right\}} \right)\\
		&\stackrel{\text{(i)}}{=} O_P\left(\sqrt{\frac{\log n}{n}} \right) = o_P(1).
	\end{align*}
	Here, the equality (i) follows from the calculation that
	\begin{align*}
		&\mathbb{E}\left\{\left[\int_{T_{i_1}}^t \mathbb{G}_n\left[\frac{C_{K_T} \left[Y_{i_2} - \mu(T_{i_2},\bm{S}_{i_2})\right]}{2\sqrt{h}\cdot p_T(\tilde{t})}\left(\frac{T_{i_2} -\tilde{t}}{h}\right) K_T\left(\frac{T_{i_2} -\tilde{t}}{h}\right) \right] d\tilde{t}\right]^2 \right\}\\ 
		&\stackrel{\text{(ii)}}{\leq} \mathbb{E}\left[|t-T_{i_1}|^2 \cdot \sup_{\tilde{t}\in \mathcal{T}'} \left|\mathbb{G}_n\left[\frac{C_{K_T} \left[Y-\mu(T,\bm{S})\right]}{\sqrt{h}\cdot p_T(\tilde{t})} \left(\frac{T-\tilde{t}}{h}\right) K_T\left(\frac{T-\tilde{t}}{h}\right) \right] \right|^2 \right]\\
		&= O\left(\log n\right),
	\end{align*}
	where (ii) applies the mean-value theorem for integrals.\\
	
	As a summary, we conclude that
	\begin{align*}
		&\sqrt{nh^3}\left[\hat{m}_{\theta}(t) - m(t)\right] \\
		&= \mathbb{G}_n\left\{\mathbb{E}_{T_1}\left[\int_{T_1}^t \frac{C_{K_T} \left[Y-\mu(T,\bm{S})\right]}{\sqrt{h}\cdot p_T(\tilde{t})} \left(\frac{T-\tilde{t}}{h}\right) K_T\left(\frac{T-\tilde{t}}{h}\right) d\tilde{t}\right] \right\}\\
		&\quad + O_P\left(\sqrt{nh^3 \max\{h, \hslash\}^4} + \sqrt{\frac{h^3\log n}{\hslash}} + \frac{\log n}{\sqrt{n\hslash}} + \sqrt{\frac{\log n}{nb^d\hslash}}\right) + O_P\left(\sqrt{\frac{\log n}{n}}\right)\\
		&= \mathbb{G}_n \varphi_t +o_P(1)
	\end{align*}
	when $q\geq 2$, $b\lesssim h\asymp n^{-\frac{1}{\gamma}}$, and $\hslash \asymp n^{-\frac{1}{\varpi}}$ for some $\gamma, \varpi >0$ such that $\frac{nh^5}{\log n} \to c_1$ and $\frac{n\hslash^5}{\log n} \to c_2$ for some finite number $c_1,c_2 \geq 0$ as well as $\sqrt{nh^3 \max\{h, \hslash\}^4}, \frac{h^3\log n}{\hslash},\frac{\log n}{\sqrt{n\hslash}}, \frac{\log n}{n\max\{h,\hslash\} b^d} \to 0$, where 
	\begin{align*}
		\varphi_t\left(Y,T,\bm{S}\right) &= \mathbb{E}_{T_1}\left[\int_{T_1}^t \bar{\varphi}_{\tilde{t}}(Y, T, \bm{S}) \, d\tilde{t}\right] = \mathbb{E}_{T_1}\left\{\int_{T_1}^t \frac{C_{K_T} \left[Y-\mu(T,\bm{S})\right]}{\sqrt{h}\cdot p_T(\tilde{t})} \left(\frac{T-\tilde{t}}{h}\right) K_T\left(\frac{T-\tilde{t}}{h}\right) d\tilde{t}\right\}.
	\end{align*}
	
	We also know from our above results that
	\begin{align*}
		&\left|\sqrt{nh^3}\sup_{t\in \mathcal{T}'} \left|\hat{m}_{\theta}(t) - m(t)\right| -\sup_{t\in \mathcal{T}'}\left|\mathbb{G}_n\varphi_t\right| \right| = O_P\left(\sqrt{nh^3 \max\{h, \hslash\}^4} + \sqrt{\frac{h^3\log n}{\hslash}} + \frac{\log n}{\sqrt{n\hslash}} + \sqrt{\frac{\log n}{nb^d\hslash}}\right)
	\end{align*}
	under our conditions on the bandwidth parameters. The proof is thus completed.
\end{proof}

\vspace{3mm}

\begin{lemma}
	\label{lem:influc_func_moment_bnd}
	Let $\mathcal{T}' \subset \mathcal{T}$ be a compact set so that $p_T(t)$ is uniformly bounded away from 0 within $\mathcal{T}'$. Suppose that Assumptions~\ref{assump:reg_diff}, \ref{assump:den_diff}, \ref{assump:boundary_cond}, and \ref{assump:reg_kernel} hold. Then, when $h,b,\hslash$ are sufficiently small as $n\to \infty$, there exist constants $\underline{\sigma},\bar{\sigma} >0$ such that $\underline{\sigma}^2 \leq \mathbb{E}\left[\varphi_t(Y,T,\bm{S})^2\right] \leq \bar{\sigma}^2$ for any $t\in \mathcal{T}'$, where $\varphi_t$ is defined in \eqref{influ_func}. Here, all the constants $\underline{\sigma},\bar{\sigma}>0$ are independent of $h,b,\hslash$ and $n$. Furthermore, the same upper and lower bounds apply to $\mathrm{Var}\left[\varphi_t(Y,T,\bm{S})\right]$. Finally, the same result holds true when we replace $\varphi_t$ with $\bar{\varphi}_t$ for any $t\in \mathcal{T}'$.
\end{lemma}

\begin{proof}[Proof of Lemma~\ref{lem:influc_func_moment_bnd}]
	Recall from \eqref{influ_func} that for any $t\in \mathcal{T}'$
	\begin{align*}
		\mathbb{E}\left[\varphi_t(Y,T,\bm{S})^2\right] 
		&= \mathbb{E}\left\{\left[\mathbb{E}_{T_1}\left\{\int_{T_1}^t \frac{C_{K_T} \left[Y-\mu(T,\bm{S})\right]}{\sqrt{h}\cdot p_T(\tilde{t})} \left(\frac{T-\tilde{t}}{h}\right) K_T\left(\frac{T-\tilde{t}}{h}\right) d\tilde{t}\right\}\right]^2\right\} \\
		&= \frac{C_{K_T}^2\sigma^2}{h}\cdot \mathbb{E}\left\{\left[\mathbb{E}_{T_1}\left\{\int_{T_1}^t \frac{1}{p_T(\tilde{t})} \left(\frac{T-\tilde{t}}{h}\right) K_T\left(\frac{T-\tilde{t}}{h}\right) d\tilde{t}\right\}\right]^2\right\}\\
		&= \frac{C_{K_T}^2\sigma^2}{h} \int_{\mathcal{T}\times \mathcal{S}} \left\{\mathbb{E}_{T_1}\left[\int_{T_1}^t \frac{1}{p_T(\tilde{t})} \left(\frac{t_2-\tilde{t}}{h}\right) K_T\left(\frac{t_2-\tilde{t}}{h}\right) d\tilde{t}\right]\right\}^2 p(t_2,\bm{s}_2) \, d\bm{s}_2 dt_2 \\
		&= \frac{C_{K_T}^2\sigma^2}{h} \int_{\mathcal{T}} \left\{\mathbb{E}_{T_1}\left[\int_{T_1}^t \frac{1}{p_T(\tilde{t})} \left(\frac{t_2-\tilde{t}}{h}\right) K_T\left(\frac{t_2-\tilde{t}}{h}\right) \sqrt{p_T(t_2)}\, d\tilde{t}\right]\right\}^2 dt_2\\
		&\stackrel{\text{(i)}}{=} C_{K_T}^2\sigma^2 \int_{\mathbb{R}} \left\{\mathbb{E}_{T_1}\left[\int_{T_1}^t \frac{u}{p_T(\tilde{t})} \cdot K_T\left(u\right) \sqrt{p_T(\tilde{t}+uh)}\, d\tilde{t}\right]\right\}^2 du\\
		&= C_{K_T}^2\sigma^2 \int_{\mathbb{R}} u^2\cdot K_T^2(u) \left\{\mathbb{E}_{T_1}\left[\int_{T_1}^t \frac{\sqrt{p_T(\tilde{t}+uh)}}{p_T(\tilde{t})} \, d\tilde{t}\right]\right\}^2 du\\
		&\stackrel{\text{(ii)}}{=} C_{K_T}^2\sigma^2 \nu_2^{(T)} \left\{\mathbb{E}_{T_1}\left[\int_{T_1}^t \frac{1}{\sqrt{p_T(\tilde{t})}} \, d\tilde{t}\right]\right\}^2 + O(h)
	\end{align*}
	where (i) applies a change of variable $u=\frac{t_2-\tilde{t}}{h}$ and (ii) utilizes a Taylor's expansion on $p_T(\tilde{t}+uh)$ under Assumption~\ref{assump:den_diff}. Now, since $p_T$ is upper bounded and lower bounded away from 0 within $\mathcal{T}'$, then
	$$\mathbb{E}\left[\varphi_t(Y,T,\bm{S})^2\right] \geq C_{K_T}^2\sigma^2 \nu_2^{(T)} \left[\frac{1}{\mathbb{E}_{T_1}(t-T_1)\cdot \sup\limits_{t\in \mathcal{T}}\sqrt{p_T(t)}}\right]^2 + O(h)\geq \underline{\sigma}^2$$
	and 
	$$\mathbb{E}\left[\varphi_t(Y,T,\bm{S})^2\right] \leq C_{K_T}^2\sigma^2 \nu_2^{(T)} \left\{\mathbb{E}_{T_1}\left[(t-T_1)\right]^2 \cdot \sup_{t\in \mathcal{T}'} \frac{1}{\sqrt{p_T(t)}}\right\} + O(h)\leq \bar{\sigma}^2$$
	for some constants $\underline{\sigma},\bar{\sigma} >0$ that are independent of $h,b,\hslash$ and $n$ when $n$ is sufficiently large.\\
	
	As for the influence function $\bar{\varphi}_t$ of $\hat{\theta}_C(t)$ in \eqref{influ_func}, we use the similar arguments to derive that for any $t\in \mathcal{T}'$,
	\begin{align*}
		\mathbb{E}\left[\bar{\varphi}_t(Y,T,\bm{S})^2\right] &= \mathbb{E}\left\{\left[\frac{C_{K_T}[Y -\mu(T,\bm{S})]}{\sqrt{h}\cdot p_T(t)} \left(\frac{T-t}{h}\right) K_T\left(\frac{T-t}{h}\right) \right]^2 \right\} \\
		&= \frac{C_{K_T}^2\sigma^2}{h}\int_{\mathcal{T}} \frac{1}{p_T^2(t)} \left(\frac{t_1-t}{h}\right)^2 K_T^2\left(\frac{t_1-t}{h}\right) p_T(t_2)\, dt_2\\
		&\stackrel{\text{(iii)}}{=} C_{K_T}^2\sigma^2 \int_{\mathbb{R}} \frac{u^2}{p_T^2(t)} \cdot K_T^2(u)\cdot p_T(t+uh)\, du\\
		&\stackrel{\text{(iv)}}{=} \frac{C_{K_T}^2\sigma^2 \nu_2^{(T)}}{p_T(t)} + O(h),
	\end{align*}
	where (iii) applies a change of variable $u=\frac{t_1-t}{h}$ and (ii) utilizes a Taylor's expansion on $p_T(t+uh)$ under Assumption~\ref{assump:den_diff}. Since $p_T(t)$ is upper bounded and lower bounded away from 0 for any $t\in \mathcal{T}'$, we have that
	$$\underline{\sigma}^2\leq \mathbb{E}\left[\bar{\varphi}_t(Y,T,\bm{S})^2\right] = \frac{C_{K_T}^2\sigma^2 \nu_2^{(T)}}{p_T(t)} + O(h) \leq \bar{\sigma}^2$$
	when $h$ are sufficiently small as $n\to \infty$, where $\underline{\sigma},\bar{\sigma} >0$ can be again chosen to be independent of $h,b,\hslash$ and $n$. The results follow.
\end{proof}

\subsection{Proof of \autoref{thm:gauss_approx_si}}
\label{app:gauss_approx_proof}

Before proving \autoref{thm:gauss_approx_si}, we first study some auxiliary results. Recall that the class $\mathcal{G}$ of measurable functions on $\mathbb{R}^{d+1}$ is VC-type if there exist constants $A_2,\upsilon_2 >0$ such that for any $0 < \epsilon < 1$,
$$\sup_Q N\left(\mathcal{G}, L_2(Q), \epsilon\norm{G}_{L_2(Q)}\right) \leq \left(\frac{A_2}{\epsilon}\right)^{\upsilon_2},$$
where $N\left(\mathcal{G}, L_2(Q), \epsilon\norm{G}_{L_2(Q)}\right)$ is the $\epsilon\norm{G}_{L_2(Q)}$-covering number of the (semi-)metric space $\left(\mathcal{G}, \norm{\cdot}_{L_2(Q)}\right)$, $Q$ is any probability measure on $\mathbb{R}^{d+1}$, $G$ is an envelope function of $\mathcal{G}$, and $\norm{G}_{L_2(Q)}$ is defined as $\left[\int_{\mathbb{R}^{d+1}} \left[G(x)\right]^2 \, dQ(x)\right]^{\frac{1}{2}}$.

\begin{lemma}[VC-type result related to the influence function $\varphi_t$]
	\label{lem:VC_influ_func}
	Let $q\geq 2$ in the local polynomial regression for estimating $\frac{\partial}{\partial t}\mu(t,\bm{s})$ and $\mathcal{T}' \subset \mathcal{T}$ be a compact set so that $p_T(t)$ is uniformly bounded away from 0 within $\mathcal{T}'$. Suppose that Assumptions~\ref{assump:reg_diff}, \ref{assump:den_diff}, and \ref{assump:reg_kernel} hold. Then, the class of scaled influence functions
	$$\tilde{\mathcal{F}} = \left\{(v,x,\bm{z}) \mapsto \sqrt{h^3}\cdot \varphi_t(v,x,\bm{z}): t\in \mathcal{T}'\right\}$$
	has its covering number $N\left(\tilde{\mathcal{F}}, L_2(Q), \epsilon \right)$ and bracketing number $N_{[]}\left(\tilde{\mathcal{F}}, L_2(Q), 2\epsilon \right)$ as:
	$$\sup_Q N\left(\tilde{\mathcal{F}}, L_2(Q), \epsilon \right)\leq \sup_Q N_{[]}\left(\tilde{\mathcal{F}}, L_2(Q), 2\epsilon \right) \leq \frac{C_4}{\epsilon}$$
	for some constant $C_4>0$, where the supremum is taken over all probability measures $Q$ for which the class $\tilde{\mathcal{F}}$ is not identically 0. In other words, $\tilde{\mathcal{F}}$ is a bounded VC-type class of functions with an envelope function $(v,x,\bm{z}) \mapsto F_1(v,x,\bm{z}) = C_5\cdot |v|$ for some constant $C_5>0$. Analogously, 
	$$\tilde{\mathcal{F}}_{\theta} = \left\{(v,x,\bm{z}) \mapsto \sqrt{h^3}\cdot \bar{\varphi}_t(v,x,\bm{z}): t\in \mathcal{T}'\right\}$$
	is also a bounded VC-type class of functions with an envelope function $(v,x,\bm{z}) \mapsto F_1(v,x,\bm{z}) = C_6\cdot |v|$ for some constant $C_6>0$. Here, the constants $C_4,C_5,C_6>0$ are independent of $n,h,b,\hslash$.
\end{lemma}

\begin{proof}[Proof of Lemma~\ref{lem:VC_influ_func}]
	We assume without loss of generality that $\mathcal{T}'$ is connected. Otherwise, we can focus our analysis on a connected component of $\mathcal{T}'$ and take the union afterwards. 
	
	For any $f_{t_1},f_{t_2} \in \tilde{\mathcal{F}}$, we know from \eqref{influ_func} that
	\begin{align*}
		\left|f_{t_1}(v,x,\bm{z}) - f_{t_2}(v,x,\bm{z})\right| &= \left|\sqrt{h^3} \left[\varphi_{t_1}(v,x,\bm{z}) - \varphi_{t_2}(v,x,\bm{z}) \right]\right| \\
		&= \Bigg|\mathbb{E}_{T_1}\left\{\int_{T_1}^{t_1} \frac{h\cdot C_{K_T} \left[v-\mu(x,\bm{z})\right]}{p_T(\tilde{t})} \left(\frac{x-\tilde{t}}{h}\right) K_T\left(\frac{x-\tilde{t}}{h}\right) d\tilde{t}\right\} \\
		&\quad - \mathbb{E}_{T_1}\left\{\int_{T_1}^{t_2} \frac{h\cdot C_{K_T} \left[v-\mu(x,\bm{z})\right]}{p_T(\tilde{t})} \left(\frac{x-\tilde{t}}{h}\right) K_T\left(\frac{x-\tilde{t}}{h}\right) d\tilde{t}\right\} \Bigg|\\
		&= \left|\int_{t_2}^{t_1} \frac{h\cdot C_{K_T} \left[v-\mu(x,\bm{z})\right]}{p_T(\tilde{t})} \left(\frac{x-\tilde{t}}{h}\right) K_T\left(\frac{x-\tilde{t}}{h}\right) d\tilde{t} \right|\\
		&\leq C_3 \left|t_1-t_2\right| |v-\mu(x,z)|,
	\end{align*}
	where the last inequality uses the boundedness of $\frac{1}{p_T(\tilde{t})}$ for any $\tilde{t}\in \mathcal{T}'$ and $\left(\frac{x-\tilde{t}}{h}\right) K_T\left(\frac{x-\tilde{t}}{h}\right)$ for any $\tilde{t}\in \mathcal{T}$ under Assumption~\ref{assump:reg_kernel}. Here, $C_3>0$ is some constant that is independent of $n,h,b,\hslash$. Furthermore, we can also compute the $L_2(Q)$-norm of $(v,x,\bm{z}) \mapsto C_3 |v-\mu(x,z)|$ for any probability measure as:
	\begin{align*}
		C_3\left\{\mathbb{E}_Q\left[|Y-\mu(T,\bm{Z})|^2\right]\right\}^{\frac{1}{2}} \leq C_3\sigma <\infty
	\end{align*}
	under Assumption~\ref{assump:reg_diff}. Since the diameter of $\mathcal{T}'$ is finite, we conclude from Example 19.7 in \cite{VDV1998} that 
	$$\sup_Q N_{[]}\left(\tilde{\mathcal{F}}, L_2(Q), \epsilon \right) \leq \frac{C_3'}{\epsilon}$$
	for some constant $C_3'>0$. Additionally, since any $2\epsilon$ bracket $[f_{t_1}, f_{t_2}]$ of $\mathcal{F}$ is contained in a ball of radius $\epsilon$ centered at $\frac{f_{t_1}+f_{t_2}}{2}$, we know that
	$$\sup_Q N\left(\tilde{\mathcal{F}}, L_2(Q), \epsilon \right)\leq \sup_Q N_{[]}\left(\tilde{\mathcal{F}}, L_2(Q), 2\epsilon \right) \leq \frac{C_4}{\epsilon}$$
	for some constant $C_4>0$. Thus, $\tilde{\mathcal{F}}$ is a bounded VC-type class of functions. Finally, the above calculations tell us that for any $f_t\in \tilde{\mathcal{F}}$, 
	\begin{align*}
		\left|f_t(v,x,\bm{z}) \right| &= \left|\mathbb{E}_{T_1}\left\{\int_{T_1}^t \frac{h\cdot C_{K_T} \left[v-\mu(x,\bm{z})\right]}{p_T(\tilde{t})} \left(\frac{x-\tilde{t}}{h}\right) K_T\left(\frac{x-\tilde{t}}{h}\right) d\tilde{t}\right\} \right|\\
		&\leq C_3\cdot \mathbb{E}_{T_1}\left[|t-T_1| \cdot |v - \mu(x,\bm{z})|\right]\\
		&\leq C_5 \cdot |v - \mu(x,\bm{z})|,
	\end{align*}
	where $C_5>0$ is some absolute constant. Thus, an envelope function of $\tilde{\mathcal{F}}$ can be given by $(v,x,\bm{z}) \mapsto F_1(v,x,\bm{z}) = C_5\cdot |v-\mu(x,\bm{z})|$. The result thus follows for $\tilde{\mathcal{F}}$.\\
	
	As for $\tilde{\mathcal{F}}_{\theta}$, we note that
	\begin{align*}
		g_t(v,x,\bm{z}) &= \sqrt{h^3} \cdot \bar{\varphi}_t(v,x,\bm{z})\\
		&= \frac{h\cdot C_{K_T}[v -\mu(x,\bm{z})]}{p_T(t)} \left(\frac{x-t}{h}\right) K_T\left(\frac{x-t}{h}\right)
	\end{align*}
	for each $g_t\in \tilde{\mathcal{F}}_{\theta}$. Furthermore, by Assumption~\ref{assump:reg_kernel} with some fixed $\bm{z},\bm{s}\in \mathcal{S}$, we know that $\tilde{\mathcal{F}}_{\theta}$ is a bounded VC-type class of functions. Moreover, $\left|g_t(v,x,\bm{z}) \right| \leq C_6|v-\mu(x,\bm{z})|$ for some constant $C_6>0$, which serves as an envelop function of $\tilde{\mathcal{F}}_{\theta}$ with a finite $L_2(Q)$-norm for any probability measure $Q$. The result thus follows for $\tilde{\mathcal{F}}_{\theta}$.
\end{proof}

\vspace{3mm}

\begin{lemma}[Corollary 2.2 in \citealt{chernozhukov2014gaussian}]
	\label{lem:gauss_approx_VC}
	Let $\mathcal{G}$ be a collection of functions that is pointwise measurable and of VC-type (see Assumption~\ref{assump:reg_kernel}) with an envelop function $\tilde{G}$ and constants $A\geq e$ and $\upsilon \geq 1$. Suppose also that for some $\tilde{A}\geq \tilde{\sigma} >0$ and $q'\in [4,\infty]$, we have $\sup_{\tilde{g}\in \mathcal{G}} \mathbb{E}\left[\tilde{g}(\bm{U})^k\right]\leq \tilde{\sigma}^2\cdot \tilde{A}^{k-2}$ for $k=2,3,4$ and $\norm{G}_{P,q'}\leq \tilde{A}$. Let $\mathbb{B}$ be a centered Gaussian process defined on $\mathcal{G}$ with covariance function
	$$\mathrm{Cov}\left(\mathbb{B}(\tilde{g}_1), \mathbb{B}(\tilde{g}_2)\right) = \mathbb{E}\left[\tilde{g}_1(\bm{U}) \cdot \tilde{g}_2(\bm{U})\right],$$
	where $\tilde{g}_1,\tilde{g}_2\in \mathcal{G}$ and $\bm{U}=(Y,T,\bm{S})$. Then, for every $\tilde{\gamma}\in (0,1)$ and sufficiently large, there exists a random variable $\tilde{B} \stackrel{d}{=} \sup_{\tilde{g}\in \mathcal{G}} \mathbb{B}(\tilde{g})$ such that 
	$$\P\left(\left|\sup_{\tilde{g}\in\mathcal{G}} \left|\mathbb{G}_n(\tilde{g})\right| - \tilde{B}\right| > \frac{C_1\cdot \tilde{A}^{\frac{1}{3}} \tilde{\sigma}^{\frac{2}{3}} \log^{\frac{2}{3}} n}{\tilde{\gamma}^{\frac{1}{3}} n^{\frac{1}{6}}}\right) \leq C_2\cdot \tilde{\gamma},$$
	where $C_1,C_2>0$ are two constants that only depend on $q'$. Here, $\tilde{B}_1\stackrel{d}{=}$ $\tilde{B}_2$ for two random variables $\tilde{B}_1,\tilde{B}_2$ means that they have the same distribution.
\end{lemma}

\vspace{3mm}

\begin{lemma}[Lemma 2.3 in \citealt{chernozhukov2014gaussian}]
	\label{lem:gauss_approx_kol_dist}
	Under the same setup for $\mathcal{G}$ as in Lemma~\ref{lem:gauss_approx_VC}, we assume that there exist constants $\underline{\sigma}, \bar{\sigma} >0$ such that $\underline{\sigma}^2 \leq \mathbb{E}\left[\tilde{g}^2(\bm{U})\right] \leq \bar{\sigma}^2$ for all $\tilde{g}\in \mathcal{G}$. Moreover, suppose that there exist constants $r_1,r_2 >0$ such that 
	$$\P\left(\left|\sup_{\tilde{g}\in \mathcal{G}} |\mathbb{G}_n(f)| - \sup_{\tilde{g}\in\mathcal{G}} |\mathbb{B}(f)| \right| > r_1\right) \leq r_2.$$ 
	Then, 
	\begin{align*}
		&\sup_{u\geq 0}\left|\P\left(\sup_{\tilde{g}\in \mathcal{G}} |\mathbb{G}_n(f)| \leq u \right) - \P\left(\sup_{\tilde{g}\in \mathcal{G}} |\mathbb{B}(f)| \geq u\right) \right| \\
		&\leq C_{\sigma} r_1 \left\{\mathbb{E}\left[\sup_{\tilde{g}\in\mathcal{G}} |\mathbb{B}(f)|\right] + \sqrt{\max\left\{1,\, \log\left(\frac{\underline{\sigma}}{r_1}\right)\right\}}\right\} + r_2,
	\end{align*}
	where $C_{\sigma}>0$ is a constant depending only on $\underline{\sigma}$ and $\bar{\sigma}$.
\end{lemma}

\vspace{3mm}

\begin{lemma}
	\label{lem:dominate_rate}
	Let $d\geq 1$. If $b\lesssim h\asymp n^{-\frac{1}{\gamma}}$ and $\hslash \asymp n^{-\frac{1}{\varpi}}$ for some $\gamma\geq \varpi >0$ such that $\frac{nh^5}{\log n} \to c_1$ and $\frac{n\hslash^5}{\log n} \to c_2$ for some finite number $c_1,c_2 \geq 0$ and $\frac{\hslash}{h^3\log n}, \hslash n^{\frac{1}{3}}\log n, \frac{\sqrt{n\hslash}}{\log n}, \frac{n\max\{h,\hslash\} b^d}{\log n} \to \infty$ as $n\to \infty$, then 
	\begin{equation}
		\label{dominate_rate_inq}
		\max\left\{\sqrt{nh^3 \max\{h, \hslash\}^4}, \sqrt{\frac{h^3\log n}{\hslash}}, \frac{\log n}{\sqrt{n\hslash}}\right\} \lesssim \frac{\log^{\frac{2}{3}} n}{\left(nh^3\right)^{\frac{1}{6}}}
	\end{equation}
	when $n$ is sufficiently large.
\end{lemma}

\begin{proof}[Proof of Lemma~\ref{lem:dominate_rate}]
	We consider controlling the four quantities on the left hand side of \eqref{dominate_rate_inq} by the right hand side separately.
	
	$\bullet$ {\bf Quantity I:} Note that $\sqrt{nh^7} \lesssim \frac{\log^{\frac{2}{3}} n}{\left(nh^3\right)^{\frac{1}{6}}}$ is equivalent to $n h^6 \lesssim \log n$. Under our conditions, we know that the slowest rate of convergence for $h$ is $O\left(n^{-\frac{1}{5}}\right)$ up to some logarithmic factors of $n$. Thus, $nh^6 = O\left(n^{-\frac{1}{5}}\right)$ will be dominated by $O(\log n)$. Similarly, $\sqrt{nh^3\hslash^4} \lesssim \frac{\log^{\frac{2}{3}} n}{\left(nh^3\right)^{\frac{1}{6}}}$ is equivalent to $n h^3\hslash^3 \lesssim \log n$. Under our conditions, we know that the slowest rate of convergence for $h$ or $\hslash$ is $O\left(n^{-\frac{1}{5}}\right)$ up to some logarithmic factors of $n$. Thus, $nh^3\hslash^3 = O\left(n^{-\frac{1}{5}}\right)$ will be dominated by $O(\log n)$.
	
	$\bullet$ {\bf Quantity II:} Note that $\sqrt{\frac{h^3\log n}{\hslash}} \lesssim \frac{\log^{\frac{2}{3}} n}{\left(nh^3\right)^{\frac{1}{6}}}$ is equivalent to $\frac{nh^{12}}{\hslash^3} \lesssim \log n$. Again, under our conditions, the slowest rates of convergence for $h$ is $O\left(n^{-\frac{1}{5}}\right)$, and the fastest rate of convergence for $\hslash$ is smaller than the order $O\left(n^{-\frac{1}{3}}\right)$ up to some logarithmic factors of $n$. Under these rates, $\frac{nh^{12}}{\hslash^3} =O\left(n^{-\frac{2}{5}}\right)$ will be dominated by $O\left(\log n \right)$.
	
	$\bullet$ {\bf Quantity III:} Note that $\frac{\log n}{\sqrt{n\hslash}} \lesssim \frac{\log^{\frac{2}{3}} n}{\left(nh^3\right)^{\frac{1}{6}}}$ is equivalent to $\frac{h^3}{n^2 \hslash^3} \lesssim \frac{1}{\log^2 n}$. Again, under our conditions, the slowest rate of convergence for $h$ is $O\left(n^{-\frac{1}{5}}\right)$ up to some logarithmic factors of $n$. Furthermore, given that $\hslash n^{\frac{1}{3}}\log n \to \infty$, the fastest rate of convergence for $\hslash$ should be smaller than the order $O\left(n^{-\frac{1}{3}}\right)$. Thus, under these rates, $\frac{h^3}{n^2 \hslash^3}$ is dominated by $O\left(n^{-\frac{8}{5}}\right)$ and thus by $O\left(\frac{1}{\log^2 n} \right)$.
	
	
	In summary, the result follows from combining the above four cases. 
\end{proof}

\vspace{3mm}

\begin{customthm}{6}[Gaussian approximation]
	Let $q\geq 2$ in the local polynomial regression for estimating $\frac{\partial}{\partial t}\mu(t,\bm{s})$ and $\mathcal{T}' \subset \mathcal{T}$ be a compact set so that $p_T(t)$ is uniformly bounded away from 0 within $\mathcal{T}'$. Suppose that Assumptions~\ref{assump:identify_cond}, \ref{assump:diff_inter}, \ref{assump:reg_diff}, \ref{assump:den_diff}, \ref{assump:boundary_cond}, and \ref{assump:reg_kernel} hold. If $b\lesssim h\asymp n^{-\frac{1}{\gamma}}$ and $\hslash \asymp n^{-\frac{1}{\varpi}}$ for some $\gamma\geq \varpi >0$ such that $\frac{nh^5}{\log n} \to c_1$ and $\frac{n\hslash^5}{\log n} \to c_2$ for some finite number $c_1,c_2 \geq 0$ and $\frac{\hslash}{h^3\log n}, \hslash n^{\frac{1}{3}}\log n, \frac{\sqrt{n\hslash}}{\log n}, \frac{n\max\{h,\hslash\} b^d}{\log n} \to \infty$ as $n\to \infty$, then there exist Gaussian processes $\mathbb{B}, \bar{\mathbb{B}}$ such that 
	\begin{align*}
		&\sup_{u\geq 0} \left|\P\left(\sqrt{nh^3}\cdot \sup_{t\in \mathcal{T}'}\left|\hat{m}_{\theta}(t) -m(t)\right| \leq u\right) - \P\left(\sup_{f\in \mathcal{F}} |\mathbb{B}(f)| \leq u\right) \right|\\
		& \asymp \sup_{u\geq 0} \left|\P\left(\sqrt{nh^3}\cdot \sup_{t\in \mathcal{T}'}\left|\hat{\theta}_C(t) -\theta_C(t)\right| \leq u\right) - \P\left(\sup_{g\in \mathcal{F}_{\theta}} |\bar{\mathbb{B}}(g)| \leq u\right) \right| \\
		&= O\left(\left(\frac{\log^5 n}{nh^3}\right)^{\frac{1}{8}} + \left(\frac{\log^2 n}{nb^d\hslash}\right)^{\frac{3}{8}}\right),
	\end{align*}
	where $\mathcal{F},\mathcal{F}_{\theta}$ are defined in \eqref{func_class}.
\end{customthm}

\begin{proof}[Proof of \autoref{thm:gauss_approx_si}]
	We only prove the Gaussian approximation for $\hat{m}_{\theta}(t)$, since the result for its derivative estimator $\hat{\theta}_C(t)$ follows from an identical argument. At a high level, given the asymptotic linearity of $\hat{m}_{\theta}(t)$ in Lemma~\ref{lem:asymp_linear}, we will use Lemma~\ref{lem:gauss_approx_VC} to establish the coupling between $\sup_{f\in \mathcal{F}}|\mathbb{G}_n(f)|$ and $\sup_{f\in \mathcal{F}}|\mathbb{B}(f)|$ for the Gaussian process $\mathbb{B}$ defined in the theorem statement and then utilize Lemma~\ref{lem:gauss_approx_kol_dist} to translate the coupling to a bound on the Kolmogorov distance between $\sqrt{nh^3} \cdot \sup_{t\in \mathcal{T}'}\left|\hat{m}_{\theta}(t) -m(t)\right|$ and $\sup_{f\in \mathcal{F}}|\mathbb{B}(f)|$. 
	
	By Lemma~\ref{lem:VC_influ_func}, we know that the class of scaled influence functions
	$$\tilde{\mathcal{F}} = \left\{(v,x,\bm{z}) \mapsto \sqrt{h^3}\cdot \varphi_t(v,x,\bm{z}): t\in \mathcal{T}'\right\} = \left\{\sqrt{h^3} \cdot f: f\in \mathcal{F}\right\}$$
	is a VC-type class with an envelope function $(v,x,\bm{z}) \mapsto F_1(v,x,\bm{z}) = C_5\cdot |v-\mu(x,\bm{z})|$ for some constant $C_5>0$ that is independent of $n,h,b,\hslash$. In addition, recalling the definition of $\varphi_t$ in \eqref{influ_func} together with our Assumption~\ref{assump:reg_kernel} on moments of the kernel functions $K_T$ and Lemma~\ref{lem:influc_func_moment_bnd}, we obtain that 
	\begin{align*}
		&\sup_{f\in \tilde{\mathcal{F}}} \mathbb{E}\left[f(\bm{U})^2\right] \leq C_6^2 \cdot h^3 :=\tilde{\sigma}^2 <\infty  \quad \text{ and } \quad \left[\mathbb{E}\left(C_5|Y-\mu(T,\bm{S})|\right)^4\right]^{\frac{1}{4}} \leq C_5' \left(\mathbb{E}|\epsilon|^4\right)^{\frac{1}{4}} :=\tilde{A}< \infty,
	\end{align*}
	where $C_5',C_6>0$ are some constants that are independent of $n,h,b,\hslash$. By Lemma~\ref{lem:gauss_approx_VC}, we know that for any $\tilde{\gamma} \in (0,1)$,
	$$\P\left(\left|\sup_{f\in \tilde{\mathcal{F}}}|\mathbb{G}_n(f)| - \sup_{f\in \tilde{\mathcal{F}}}|\mathbb{B}(f)| \right| > \frac{C_1\cdot \tilde{A}^{\frac{1}{3}} \left(h^3\right)^{\frac{1}{3}} \log^{\frac{2}{3}} n}{\tilde{\gamma}^{\frac{1}{3}} n^{\frac{1}{6}}}\right) \leq \frac{C_2 \cdot \tilde{\gamma}}{2}.$$
	Dividing $\sqrt{h^3}$ on both sides of the inequality inside $\P$ gives us that
	$$\P\left(\left|\sup_{t\in\mathcal{T}'}|\mathbb{G}_n(\varphi_t)| - \sup_{f\in \mathcal{F}}|\mathbb{B}(f)| \right| > \frac{C_1\cdot \tilde{A}^{\frac{1}{3}} \log^{\frac{2}{3}} n}{\tilde{\gamma}^{\frac{1}{3}} \left(nh^3\right)^{\frac{1}{6}}}\right) \leq \frac{C_2 \cdot \tilde{\gamma}}{2}.$$
	On the other hand, we know from Lemma~\ref{lem:asymp_linear} that for any $\tilde{\gamma} \in (0,1)$,
	$$\P\left(\frac{\left|\sqrt{nh^3} \cdot \sup_{t\in \mathcal{T}'}\left|\hat{m}_{\theta}(t) -m(t)\right| - \sup_{t\in \mathcal{T}'}|\mathbb{G}_n\varphi_t|\right|}{\max\left\{\sqrt{nh^3 \max\{h, \hslash\}^4}, \sqrt{\frac{h^3\log n}{\hslash}}, \frac{\log n}{\sqrt{n\hslash}}, \sqrt{\frac{\log n}{nb^d\hslash}}\right\}} > \frac{C_3}{\tilde{\gamma}^{\frac{1}{3}}}\right) \leq \frac{C_2\cdot \tilde{\gamma}}{2},$$
	where $C_3>0$ is some large constant. Combining the above two inequalities yields that
	\begin{align*}
		&\P\Bigg(\left|\sqrt{nh^3} \cdot \sup_{t\in \mathcal{T}'}\left|\hat{m}_{\theta}(t) -m(t)\right| - \sup_{f\in \mathcal{F}}|\mathbb{B}(f)|\right| > \\
		&\quad \frac{C_3}{\tilde{\gamma}^{\frac{1}{3}}}\cdot \max\left\{\sqrt{nh^3 \max\{h, \hslash\}^4}, \sqrt{\frac{h^3\log n}{\hslash}}, \frac{\log n}{\sqrt{n\hslash}}, \sqrt{\frac{\log n}{nb^d\hslash}}\right\} + \frac{C_1'\cdot \log^{\frac{2}{3}} n}{\tilde{\gamma}^{\frac{1}{3}} \left(nh^3\right)^{\frac{1}{6}}}\Bigg) \leq C_2\cdot \tilde{\gamma},
	\end{align*}
	where $C_1'=C_1\cdot \tilde{A}^{\frac{1}{3}} >0$ is again a constant. Now, if $b\lesssim h\asymp n^{-\frac{1}{\gamma}}$ and $\hslash \asymp n^{-\frac{1}{\varpi}}$ for some $\gamma\geq \varpi >0$ such that $\frac{nh^5}{\log n} \to c_1$ and $\frac{n\hslash^5}{\log n} \to c_2$ for some finite number $c_1,c_2 \geq 0$ and $\frac{\hslash}{h^3\log n}, \hslash n^{\frac{1}{3}}\log n, \frac{\sqrt{n\hslash}}{\log n}, \frac{n\max\{h,\hslash\} b^d}{\log n} \to \infty$ as $n\to \infty$, then we know from Lemma~\ref{lem:dominate_rate} that $$\max\left\{\sqrt{nh^3 \max\{h, \hslash\}^4}, \sqrt{\frac{h^3\log n}{\hslash}}, \frac{\log n}{\sqrt{n\hslash}}\right\} \leq \frac{C_1'\cdot \log^{\frac{2}{3}} n}{\left(nh^3\right)^{\frac{1}{6}}}$$ 
	when $n$ is sufficiently large. Hence, we conclude that when $n$ is sufficiently large,
	$$\P\left(\left|\sqrt{nh^3}\cdot \sup_{t\in \mathcal{T}'}\left|\hat{m}_{\theta}(t) -m(t)\right| - \sup_{f\in \mathcal{F}}|\mathbb{B}(f)|\right| > \frac{C_4}{\tilde{\gamma}^{\frac{1}{3}}}\left[\frac{\log^{\frac{2}{3}}}{\left(nh^3\right)^{\frac{1}{6}}} + \sqrt{\frac{\log n}{nb^d\hslash}}\right]\right) \leq C_2\cdot \tilde{\gamma},$$
	for some large constants $C_2,C_4>0$. To upper bound the Kolmogorov distance between $\sqrt{nh^3} \cdot \sup_{t\in \mathcal{T}'}\left|\hat{m}_{\theta}(t) -m(t)\right|$ and $\sup_{f\in \mathcal{F}}|\mathbb{B}(f)|$, we leverage Lemmas~\ref{lem:gauss_approx_kol_dist} and \ref{lem:influc_func_moment_bnd} to obtain that
	\begin{align*}
		&\sup_{u\geq 0} \left|\P\left(\sqrt{nh^3}\cdot \sup_{t\in \mathcal{T}'}\left|\hat{m}_{\theta}(t) -m(t)\right| \leq u\right) - \P\left(\sup_{f\in \mathcal{F}} |\mathbb{B}(f)| \leq u\right) \right| \\
		&\leq \frac{C_5}{\tilde{\gamma}^{\frac{1}{3}}}\left[\frac{\log^{\frac{5}{6}}}{\left(nh^3\right)^{\frac{1}{6}}} + \frac{\log n}{(nb^d\hslash)^{\frac{1}{2}}}\right] + C_2\cdot \tilde{\gamma},
	\end{align*}
	where $C_5 >0$ is some constant that depends only on $\bar{\sigma}\geq \underline{\sigma}>0$ in Lemma~\ref{lem:influc_func_moment_bnd}. Here, we also utilize the fact that $\log\left(\frac{1}{r_1}\right)= \log n$ when $r_1= \frac{C_4}{\tilde{\gamma}^{\frac{1}{3}}}\left[\frac{\log^{\frac{2}{3}}}{\left(nh^3\right)^{\frac{1}{6}}} + \sqrt{\frac{\log n}{nb^d\hslash}}\right]$ and use the Dudley's entropy inequality for Gaussian processes (Corollary 2.2.8 in \citealt{van1996weak}) to argue that $\mathbb{E}\left[\sup_{f\in \mathcal{F}} |\mathbb{B}(f)|\right] = \mathbb{E}\left[\sup_{t\in \mathcal{T}} |\mathbb{B}(\varphi_t)|\right]=O\left(\sqrt{\log n}\right)$. We take $\tilde{\gamma} =O\left(\left(\frac{\log^5 n}{nh^3}\right)^{\frac{1}{8}} + \left(\frac{\log^2 n}{nb^d\hslash}\right)^{\frac{3}{8}}\right)$ to optimize the right hand side of the above inequality and deduce that
	\begin{align*}
		&\sup_{u\geq 0} \left|\P\left(\sqrt{nh^3} \cdot \sup_{t\in \mathcal{T}'}\left|\hat{m}_{\theta}(t) -m(t)\right| \leq u\right) - \P\left(\sup_{f\in \mathcal{F}} |\mathbb{B}(f)| \leq u\right) \right| = O\left(\left(\frac{\log^5 n}{nh^3}\right)^{\frac{1}{8}} + \left(\frac{\log^2 n}{nb^d\hslash}\right)^{\frac{3}{8}}\right).
	\end{align*}
	The result follows.
\end{proof}

\subsection{Proof of \autoref{thm:bootstrap_cons}}
\label{app:boot_consistency}

\begin{customthm}{7}[Bootstrap consistency]
	Let $q\geq2$ in the local polynomial regression for estimating $\frac{\partial}{\partial t}\mu(t,\bm{s})$, $\mathcal{T}' \subset \mathcal{T}$ be a compact set so that $p_T(t)$ is uniformly bounded away from 0 within $\mathcal{T}'$, and $\mathbb{U}_n=\left\{(Y_i,T_i,\bm{S}_i)\right\}_{i=1}^n$ be the observed data. Suppose that Assumptions~\ref{assump:identify_cond}, \ref{assump:diff_inter}, \ref{assump:reg_diff}, \ref{assump:den_diff}, \ref{assump:boundary_cond}, and \ref{assump:reg_kernel} hold. If $b\lesssim h\asymp n^{-\frac{1}{\gamma}}$ and $\hslash \asymp n^{-\frac{1}{\varpi}}$ for some $\gamma, \varpi >0$ such that $\frac{nh^5}{\log n} \to c_1$ and $\frac{n\hslash^5}{\log n} \to c_2$ for some finite number $c_1,c_2 \geq 0$ and $\frac{\hslash}{h^3\log n}, \hslash n^{\frac{1}{3}}\log n, \frac{\sqrt{n\hslash}}{\log n}, \frac{n\max\{h,\hslash\} b^d}{\log n} \to \infty$ as $n\to \infty$, then 
	\begin{align*}
		&\sup_{u\geq 0} \left|\P\left(\sqrt{nh^3} \cdot \sup_{t\in \mathcal{T}'}\left|\hat{m}_{\theta}^*(t) -\hat{m}_{\theta}(t)\right| \leq u \Big| \mathbb{U}_n\right) - \P\left(\sup_{f\in \mathcal{F}} |\mathbb{B}(f)| \leq u\right) \right| \\
		&\asymp \sup_{u\geq 0} \left|\P\left(\sqrt{nh^3}\cdot \sup_{t\in \mathcal{T}'}\left|\hat{\theta}_C^*(t) -\hat{\theta}_C(t)\right| \leq u \Big| \mathbb{U}_n\right) - \P\left(\sup_{g\in \mathcal{F}_{\theta}} |\bar{\mathbb{B}}(g)| \leq u\right) \right| \\
		&= O_P\left(\left(\frac{\log^5 n}{nh^3}\right)^{\frac{1}{8}} + \left(\frac{\log^2 n}{nb^d\hslash}\right)^{\frac{3}{8}}\right),
	\end{align*}
	where $\hat{m}_{\theta}^*(t)$ and $\hat{\theta}_C^*(t)$ are the integral estimator \eqref{simp_integral} and localized derivative estimator \eqref{theta_C_est} based on a bootstrap sample $\mathbb{U}_n^*=\left\{\left(Y_i^*,T_i^*,\bm{S}_i^*\right)\right\}_{i=1}^n$ respectively, and $\mathbb{B}, \bar{\mathbb{B}}$ are the same Gaussian processes as in \autoref{thm:gauss_approx_si}.
\end{customthm}

\begin{proof}[Proof of \autoref{thm:bootstrap_cons}]
	We only prove the bootstrap consistency for $\hat{m}_{\theta}(t)$, since the result for its derivative estimator $\hat{\theta}_C(t)$ follows from an identical argument. Our proof here is similar to the proof of Theorem 7 in \cite{chen2015asymptotic} and the proof of Theorem 4 in \cite{chen2017density}. The key difference is that the functional space remains unchanged as $\mathcal{F}$ in our scenario here for both the original Gaussian approximation (\autoref{thm:gauss_approx_si}) and the bootstrapped version, because the index set $\mathcal{T}'$ of $\mathcal{F}$ is fixed.
	
	Let $\mathbb{U}_n=\left\{(Y_i,T_i,\bm{S}_i)\right\}_{i=1}^n$ be the observed data and $\mathbb{U}_n^*=\left\{(Y_i^*,T_i^*,\bm{S}_i^*)\right\}_{i=1}^n$ be the bootstrap sample. We also denote $\mathbb{G}_n^*(\mathbb{U}_n) = \sqrt{n}\left(\mathbb{P}_n^* - \mathbb{P}_n\right)$, where $\mathbb{P}_n^*$ is the empirical measure defined by the bootstrap sample $\mathbb{U}_n^*$. Assume that $\mathbb{U}_n \subset \mathcal{Y}\times \mathcal{T}\times \mathcal{S}$ is fixed for a moment. Then, we can apply our arguments in Lemma~\ref{lem:asymp_linear} by replacing the probability measure $\P$ by $\mathbb{P}_n$ and obtain that
	\begin{align*}
		&\left|\sqrt{nh^3} \sup_{t\in \mathcal{T}'}\left|\hat{m}_{\theta}^*(t) -\hat{m}_{\theta}(t)\right| - \sup_{t\in \mathcal{T}'}|\mathbb{G}_n^*(\mathbb{U}_n)\varphi_t|\right| \\
		&= O_P\left(\sqrt{nh^3 \max\{h, \hslash\}^4} + \sqrt{\frac{h^3\log n}{\hslash}} + \frac{\log n}{\sqrt{n\hslash}} + \sqrt{\frac{\log n}{nb^d\hslash}}\right).
	\end{align*}
	Following the same argument in \autoref{thm:gauss_approx_si}, we obtain that
	\begin{align}
		\label{gauss_approx_boot}
		\begin{split}
			&\sup_{u\geq 0} \left|\P\left(\sqrt{nh^3} \sup_{t\in \mathcal{T}'}\left|\hat{m}_{\theta}^*(t) -\hat{m}_{\theta}(t)\right| \leq u \Big| \mathbb{U}_n\right) - \P\left(\sup_{f\in \mathcal{F}} |\mathbb{B}_n(f)| \leq u \Big| \mathbb{U}_n\right) \right| \\
			&= O_P\left(\left(\frac{\log^5 n}{nh^3}\right)^{\frac{1}{8}} + \left(\frac{\log^2 n}{nb^d\hslash}\right)^{\frac{3}{8}}\right),
		\end{split}
	\end{align}
	where $\mathbb{B}_n$ is a Gaussian process on $\mathcal{F}$ such that for any $f_1,f_2 \in \mathcal{F}$, it has
	$$\mathbb{E}\left[\mathbb{B}_n(f_1)| \mathbb{U}_n\right] = 0 \quad \text{ and } \quad \mathrm{Cov}\left[\mathbb{B}_n(f_1), \mathbb{B}_n(f_2) \big| \mathbb{U}_n\right]= \frac{1}{n} \sum_{i=1}^n f_1(Y_i,T_i,\bm{S}_i) \cdot f_2(Y_i,T_i,\bm{S}_i).$$ 
	Notice that the difference between $\sup_{f\in \mathcal{F}} |\mathbb{B}_n(f)|$ and $\sup_{f\in \mathcal{F}} |\mathbb{B}(f)|$ is small, because these two Gaussian processes differ in their covariance and 
	\begin{align*}
		\mathrm{Cov}\left[\mathbb{B}_n(f_1), \mathbb{B}_n(f_2) \big| \mathbb{U}_n\right] &= \frac{1}{n} \sum_{i=1}^n f_1(Y_i,T_i,\bm{S}_i) \cdot f_2(Y_i,T_i,\bm{S}_i) \\
		&\to \mathrm{Cov}(\mathbb{B}(f_1), \mathbb{B}(f_2)) = \mathbb{E}\left[f_1(Y,T,\bm{S}) \cdot f_2(Y,T,\bm{S})\right]
	\end{align*}
	as $n\to \infty$. More precisely, by Corollary 9 in \cite{giessing2023gaussian}, we know that
	\begin{align}
		\label{gauss_sup_diff}
		\begin{split}
			&\sup_{u\geq 0} \left|\P\left(\sup_{f\in \mathcal{F}} |\mathbb{B}_n(f)| \leq u \Big| \mathbb{U}_n\right) -\P\left(\sup_{f\in \mathcal{F}} |\mathbb{B}(f)| \leq u\right)\right| \\
			&\leq C_6 \left[\frac{\sup_{f_1,f_2\in \mathcal{F}}\left|\mathrm{Cov}\left[\mathbb{B}_n(f_1), \mathbb{B}_n(f_2) \big| \mathbb{U}_n\right] - \mathrm{Cov}\left[\mathbb{B}(f_1), \mathbb{B}(f_2)\right]\right|}{\max\left\{\mathrm{Var}\left(\sup_{f\in \mathcal{F}}|\mathbb{B}(f)|\right), \mathrm{Var}\left(\sup_{f\in \mathcal{F}}|\mathbb{B}_n(f)|\right)\right\}}\right]^{\frac{1}{3}}\\
			&\leq C_6' \left[ \sup_{f_1,f_2\in \mathcal{F}}\left|\frac{1}{n} \sum_{i=1}^n f_1(Y_i,T_i,\bm{S}_i) \cdot f_2(Y_i,T_i,\bm{S}_i) - \mathbb{E}\left[f_1(Y,T,\bm{S}) \cdot f_2(Y,T,\bm{S})\right]\right| \right]^{\frac{1}{3}}\\
			&= C_6' \left[ \sup_{f\in \mathcal{F}^2}\left|\frac{1}{n} \sum_{i=1}^n f(Y_i,T_i,\bm{S}_i) - \mathbb{E}\left[f(Y,T,\bm{S})\right]\right| \right]^{\frac{1}{3}}
		\end{split}
	\end{align}
	where the last inequality follows from Lemma~\ref{lem:influc_func_moment_bnd} and $C_6,C_6'>0$ are two absolute constants. Here, $\mathcal{F}^2 \equiv \left\{f_1\cdot f_2: f_1,f_2\in \mathcal{F}\right\}$. Now, by symmetrization (Lemma 2.3.1 in \citealt{van1996weak}) and maximal inequality (Corollary 2.2.8 in \citealt{van1996weak}), we also know that
	\begin{align*}
		\sup_{f\in \mathcal{F}^2}\left|\frac{1}{n} \sum_{i=1}^n f(Y_i,T_i,\bm{S}_i) - \mathbb{E}\left[f(Y,T,\bm{S})\right]\right| &\leq 2 \mathbb{E}\left[\sup_{f\in \mathcal{F}^2} \left|\frac{1}{n}\sum_{i=1}^n \chi_i f(Y_i,T_i,\bm{S}_i) \right|\right]\\
		&\leq \frac{C_7}{\sqrt{n}} \int_0^{\infty} \sup_Q \sqrt{\log N\left(\mathcal{F}^2, L_2(Q), \epsilon\right)} \,d\epsilon \\
		&\leq \frac{C_7'}{\sqrt{n}} \int_0^{\infty} \sup_Q \sqrt{\log N\left(\mathcal{F}, L_2(Q), \epsilon\right)} \,d\epsilon\\
		&=\frac{C_7'}{\sqrt{n}} \int_0^{\sqrt{h^3}\cdot M_{\mathcal{F}}} \sup_Q \sqrt{\log N\left(\tilde{\mathcal{F}}, L_2(Q), \sqrt{h^3}\cdot \epsilon\right)} \,d\epsilon\\
		&= \frac{C_7'}{\sqrt{nh^3}} \int_0^{M_{\mathcal{F}}} \sup_Q \sqrt{\log N\left(\tilde{\mathcal{F}}, L_2(Q), u\right)} \,d u\\
		&\leq \frac{C_8}{\sqrt{nh^3}},
	\end{align*}
	where $\chi_1,...,\chi_n$ are Rademacher random variables, $\tilde{\mathcal{F}} = \sqrt{h^3} \cdot \mathcal{F}=\left\{\sqrt{h^3} f: f\in \mathcal{F}\right\}$ by Lemma~\ref{lem:VC_influ_func}, and $C_7,C_7',C_8>0$ are absolute constants. In addition, $M_{\mathcal{F}}= \sup_{f\in \mathcal{F}} \mathbb{E}\left[f(Y_1,T_1,\bm{S}_1)^2\right] <\infty$ by Lemma~\ref{lem:influc_func_moment_bnd} and the last inequality follows from the VC-type property of $\tilde{\mathcal{F}}$ by Lemma~\ref{lem:VC_influ_func}. Combining the above result with \eqref{gauss_sup_diff}, we obtain that
	$$\sup_{u\geq 0} \left|\P\left(\sup_{f\in \mathcal{F}} |\mathbb{B}_n(f)| \leq u \Big| \mathbb{U}_n\right) -\P\left(\sup_{f\in \mathcal{F}} |\mathbb{B}(f)| \leq u\right)\right| = O_P\left(\frac{1}{\left(nh^3\right)^{\frac{1}{6}}}\right).$$
	Together with \eqref{gauss_approx_boot}, the result thus follows.
\end{proof}

\subsection{Proof of Corollary~\ref{cor:boot_ci}}
\label{app:conf_band}

\begin{customcor}{8}[Uniform confidence band]
	Under the setup of \autoref{thm:bootstrap_cons}, we have that
	\begin{align*}
			&\P\left(\theta(t) \in \left[\hat{\theta}_C(t) - \bar{\xi}_{1-\alpha}^*, \hat{\theta}_C(t) + \bar{\xi}_{1-\alpha}^*\right] \text{ for all } t\in \mathcal{T}'\right) = 1-\alpha + O\left(\left(\frac{\log^5 n}{nh^3}\right)^{\frac{1}{8}} + \left(\frac{\log^2 n}{nb^d\hslash}\right)^{\frac{3}{8}}\right), \\
			&\P\left(m(t) \in \left[\hat{m}_{\theta}(t) - \xi_{1-\alpha}^*, \hat{m}_{\theta}(t) + \xi_{1-\alpha}^*\right] \text{ for all } t\in \mathcal{T}'\right) = 1-\alpha + O\left(\left(\frac{\log^5 n}{nh^3}\right)^{\frac{1}{8}} + \left(\frac{\log^2 n}{nb^d\hslash}\right)^{\frac{3}{8}}\right).
	\end{align*}
\end{customcor}

\begin{proof}[Proof of Corollary~\ref{cor:boot_ci}]
	By \autoref{thm:gauss_approx_si} and \autoref{thm:bootstrap_cons}, we have the Berry-Esseen bounds related to bootstrap estimates for the distributions of $\sqrt{nh^3}\cdot \sup_{t\in \mathcal{T}'}\left|\hat{m}_{\theta}(t) -m(t)\right|$ and $\sqrt{nh^3}\cdot \sup_{t\in \mathcal{T}'}\left|\hat{\theta}_C(t) -\theta_C(t)\right|$ as:
	\begin{align*}
		&\sup_{u\geq 0} \left|\P\left(\sqrt{nh^3} \sup_{t\in \mathcal{T}'}\left|\hat{m}_{\theta}^*(t) -\hat{m}_{\theta}(t)\right| \leq u \Big| \mathbb{U}_n\right) - \P\left(\sqrt{nh^3}\cdot \sup_{t\in \mathcal{T}'}\left|\hat{m}_{\theta}(t) -m(t)\right| \leq u\right) \right|\\
		&\quad =O_P\left(\left(\frac{\log^5 n}{nh^3}\right)^{\frac{1}{8}} + \left(\frac{\log^2 n}{nb^d\hslash}\right)^{\frac{3}{8}}\right),
	\end{align*}
	and 
	\begin{align*}
		&\sup_{u\geq 0} \left|\P\left(\sqrt{nh^3}\cdot \sup_{t\in \mathcal{T}'}\left|\hat{\theta}_C^*(t) -\hat{\theta}_C(t)\right| \leq u\right) - \P\left(\sqrt{nh^3}\cdot \sup_{t\in \mathcal{T}'}\left|\hat{\theta}_C(t) -\theta_C(t)\right| \leq u\right) \right|\\
		& = O_P\left(\left(\frac{\log^5 n}{nh^3}\right)^{\frac{1}{8}} + \left(\frac{\log^2 n}{nb^d\hslash}\right)^{\frac{3}{8}}\right).
	\end{align*}
	The results thus follow.
\end{proof}

\end{document}